\documentclass[journal]{IEEEtran}
\usepackage{cite}
\usepackage{amsmath,amssymb,amsfonts}
\usepackage{graphicx}
\graphicspath{{figures/}}

\usepackage{tikz}
\usetikzlibrary{positioning, arrows.meta, calc, decorations.pathreplacing, fit, shapes.geometric, backgrounds}
\usepackage{pgfplots}
\pgfplotsset{compat=1.18}

\usepackage{xcolor}
\usepackage{booktabs}
\usepackage{multirow}
\usepackage{adjustbox}
\usepackage{enumitem}
\usepackage{placeins}
\usepackage{pifont}

\usepackage{algorithm}
\usepackage{algpseudocode}

\usepackage{url}
\usepackage{textcomp}
\usepackage{stfloats}
\usepackage[caption=false,font=normalsize,labelfont=sf,textfont=sf]{subfig}

\usepackage[hidelinks]{hyperref}

\usepackage{amsthm}
\theoremstyle{definition}
\newtheorem{definition}{Definition}
\newtheorem{remark}{Remark}
\newtheorem{proposition}{Proposition}
\newtheorem{theorem}{Theorem}

\newcommand{\cmark}{\ding{51}}
\newcommand{\xmark}{\ding{55}}
\def\BibTeX{{\rm B\kern-.05em{\sc i\kern-.025em b}\kern-.08em
    T\kern-.1667em\lower.7ex\hbox{E}\kern-.125emX}}

\newcommand{\misfit}{\ensuremath{\mathfrak{m}}}

\begin{document}

\title{Learning Selective Merge Policies for Deadline-Constrained Coded Caching via Deep Reinforcement Learning}

\author{Amirhossein~Yousefiramandi}

\markboth{Preprint --- Submitted for peer review}%
{Yousefiramandi: Learning Selective Merge Policies for Deadline-Constrained Coded Caching}

\maketitle

\begin{abstract}
In the coded caching, the server uses the cached information at the users to serve multiple users in parallel with a single coded multi-casting message or packet, that is, a merged packet, and thus mitigates the peak network congestion. In order to deliver the timely messages to the users in the deadline-driven applications like the video streaming, we must determine online the messages to be merged for the delivery, as there is a time limit for each request. It is important to note that while the merging aids the current coded multi-casting packet, it could harm the future deliveries. Our solution employs the deep reinforcement learning to view the coded multi-casting delivery as a masked action-discrete state control problem, and our policy network, trained via the proximal policy optimization, performs better than SACM++. On the uniform-demand benchmark, our policy network reduces the broadcast-packet expiration ratio $\rho$ by $40.9\%$ ($0.208$ vs.\ $0.352$) with respect to the best coded multi-casting baseline (SACM++), while also attaining the best broadcast-efficiency score $\sigma$ across the Track~A battery among the coded multi-casting methods. One noteworthy phenomenon here is that, for the applications with stricter deadlines, the merging becomes selective instead of aggressive, since the policy network selectively merges at approximately $31.8\%$ of the chances, even though the same observation holds across the variations within the same simulator family. The focus of our design is on the efficient pairwise XOR merging, where the higher-order ($K{\ge}3$) coding can be considered as a natural generalization left for future work.
\end{abstract}


\section*{Notation}
\label{sec:notation}
\begin{table}[h!]
\centering
\small
\caption{Notation used throughout the paper.}
\label{tab:notation}
\begin{adjustbox}{max width=\columnwidth}
\begin{tabular}{@{}cl@{\qquad}cl@{}}
\toprule
\textbf{Symbol} & \textbf{Description} & \textbf{Symbol} & \textbf{Description} \\
\midrule
\multicolumn{4}{@{}l}{\textit{System Parameters}} \\
$K$ & Number of edge caches (users) & $N$ & Number of files in library \\
$B$ & Subfiles per file & $F{=}NB$ & Total packets \\
$p_c$ & Cache fraction ($M{=}\lfloor p_c F\rfloor$ packets per cache; $p_c F$ is integer in all evaluated regimes) & $\mathcal{C}_k$ & Cached packets at user $k$ \\
$Q$ & Pending-queue depth & $D$ & Maximum deadline (slots) \\
$H$ & Episode horizon (steps) & $P_{\max}$ & Max candidate merge pairs $\bigl({=}\tbinom{Q}{2}\bigr)$ \\
\midrule
\multicolumn{4}{@{}l}{\textit{Control Model \& Actions}} \\
$r{=}(k_r,f_r,d_r,\mathcal{S}_r)$ & Request tuple & $\mathcal{M}_t$ & Feasible merge set at step $t$ \\
$k_r$ & Dest.\ of fresh $r$; for aggregates, the representative $k_{\mathrm{mg}}$ (Eq.~\eqref{eq:mg_dest}) & $f_r$ & Packet set of $r$: singleton at arrival, union $f_{\mathrm{mg}}$ after merges \\
$d_r$ & Remaining deadline of $r$ & $\mathcal{S}_r$ & Side-info providers of $r$ \\
$\kappa$ & Keep-side decision & $\mathbf{m}_t$ & Action mask vector \\
$\gamma$ & Discount factor & & \\
\midrule
\multicolumn{4}{@{}l}{\textit{Evaluation Metrics (Sec.~\ref{sec:exp:metrics}; auxiliary bookkeeping in Sec.~\ref{sec:model:request_accounting})}} \\
$U_t$ & XOR degree at slot $t$ (packet-set count) & $E_t$ & Expired packet-set mass at slot $t$ \\
$\rho$ & Broadcast-packet expiration ratio (Eq.~\eqref{eq:miss_ratio}) & $\delta$ & Distinct file-identity coverage (Eq.~\eqref{eq:unique_demand_sat}) \\
$\sigma$ & Broadcast-efficiency score (Eq.~\eqref{eq:sys_score}) & $\mu$ & Served packet-set / tx (Eq.~\eqref{eq:served_per_tx}) \\
$\varepsilon$ & Expirations per episode (Eq.~\eqref{eq:exp_per_ep}) & $\mathcal{X}_t$ & Expiration-event set at $t$ (Eq.~\eqref{eq:exp_set}) \\
$\mathcal{A}(r)$ & Per-arrival ID annotation (evaluator bookkeeping; Sec.~\ref{sec:model:request_accounting}) & $\phi$ & Packet$\to$file projection (Eq.~\eqref{eq:phi_packet_to_file}) \\
$C_t,M_t$ & Per-step completed / missed request IDs (Eqs.~\eqref{eq:req_Ct}, \eqref{eq:req_Mt}) & & \\
$\eta_{\mathrm{req}}$ & Request timely-throughput (per-step; Eq.~\eqref{eq:req_eta}) & $m_{\mathrm{req}}$ & Request miss rate (per-step; Eq.~\eqref{eq:req_miss_rate}) \\
$\sigma_{\mathrm{req}}$ & Request selection score (per-step; Eq.~\eqref{eq:req_sigma}) & $\lambda$ & Miss-penalty weight in $\sigma$/$\sigma_{\mathrm{req}}$ \\
\multicolumn{4}{@{}p{0.97\textwidth}}{\footnotesize \emph{Note.} $\sigma$ (broadcast-efficiency score) and $\sigma_{\mathrm{req}}$ (request selection score) are \emph{distinct} metrics computed at different aggregation levels (broadcast/packet-set vs.\ per-arrival); the $\mathrm{req}$ subscript is load-bearing.} \\
\bottomrule
\end{tabular}
\end{adjustbox}
\end{table}

\section{Introduction}

\IEEEPARstart{G}{lobal} mobile data traffic continues to grow at a compound annual rate of approximately 20\%~\cite{ericsson2024mobility}, thereby placing severe pressure on the last-mile link between the edge nodes and the users, while the edge caching alleviates this bottleneck~\cite{liu2019mobile}, considering that the conventional schemes exploit only the local cache contents. With \emph{the coded caching}~\cite{maddah2014fundamental,maddah2015decentralized}, the server can use the information the users have cached to serve multiple users at a time by sending a single coded multi-casting message, i.e., the merged message, thereby relieving the peak network loads while satisfying multiple users at a time by a single coded broadcast.

The coded-caching framework of Maddah-Ali and Niesen~\cite{maddah2014fundamental} and its decentralized extension~\cite{maddah2015decentralized} formalize this XOR-multicasting principle, i.e., the formation of the coded multi-casting message from the cached side information, while Section~\ref{sec:related} provides a detailed review.

The classical coded-caching schemes were designed for a \emph{synchronous, batch-delivery} model in which all requests are collected before any transmission begins, an assumption that is incompatible with the \emph{delay-sensitive} streaming applications of the users, like the video on demand, the IPTV, and the live OTT delivery, where each user request carries a hard deadline by which it must be delivered or it becomes worthless; we address the CDN \emph{scheduling-layer} aspect of these workloads, considering that the PHY adaptation and the link-rate heterogeneity are deferred to Sec.~\ref{sec:conclusion}. Niesen and Maddah-Ali quantified the gain--delay tradeoff, showing that the larger coding gains require aggregating more requests over longer time horizons, at the cost of higher per-request delivery delay~\cite{codedcachingdelaysensitive}. In an online delivery setting with hard per-request expiration deadlines, this tradeoff becomes a \emph{sequential decision} at every time step, considering that the server must either merge a feasible pair of pending requests into a coded XOR packet, thereby yielding an immediate bandwidth reduction while consuming the shared side information and shrinking the intersection available for future merges, or fall back to the unicast, i.e., serving the earliest-deadline request to protect against the expiration while forgoing the multicast gain. The problem, however, is that while the merge is helpful for the formation of the coded multi-casting message, it can be harmful for the subsequent ones: the best action depends on the \emph{entire queue state}, i.e., the remaining deadlines of all the pending requests, their file identities, which of the users have cached which of the files, and how merging now will reshape the feasibility of the future merges, and no closed-form policy captures this dependence in general.

The existing algorithms do not address this deadline-aware online setting, considering that the greedy and the size-aware coded multi-casting heuristics (GCM~\cite{maddah2014fundamental}, SACM~\cite{asghari2019approximation}) always merge whenever any feasible pair exists, irrespective of the deadline urgency of the users (see Section~\ref{sec:related} for a detailed review), while the graph-coloring approaches~\cite{ji2015multiple,vettigli2019efficient} are likewise stateless with respect to the deadlines, and even the fixed-threshold policies ($\tau$-Fit), which merge only when both of the remaining deadlines exceed a parameter~$\tau$, require oracle tuning, cannot adapt to the current queue state, and saturate in performance (as we show in Section~\ref{sec:results}). What is needed is a \emph{state-dependent merge-or-defer policy}, i.e., the policy network, that learns when to exploit a coding opportunity and when to fall back to the unicast, thereby balancing throughput and deadline compliance from experience.

The reinforcement learning (RL) is a natural framework for learning such a state-dependent policy from the interaction, considering that the prior RL work on caching has focused on \emph{placement and replacement}, i.e., deciding what content to pre-fetch or evict under time-varying demand~\cite{sadeghi2019adaptive,somuyiwa2018proactive,zhong2018framework}. The most closely related work is Naderializadeh and Asghari~\cite{naderializadeh2019learning}, who train a deep actor--critic agent for coded-caching \emph{delivery} and demonstrate that the learned policies can match the SACM-family heuristics while reducing the online inference cost. However, that work does not model the hard expiration deadlines, uses a flat multilayer-perceptron (MLP) policy that ignores the graph structure of the mergeability relation, omits the invalid-action masking, and evaluates a single configuration without out-of-distribution (OOD) generalization analysis, while our work directly and systematically addresses all four of the gaps.
Effective learning here requires three ingredients:
(i)~\emph{invalid-action masking}~\cite{huang2020invalidmask,sb3maskableppo}
to restrict policy-gradient updates to feasible XOR merges;
(ii)~a \emph{graph attention network}~\cite{velickovic2018gat} that exploits
the variable-cardinality merge-graph structure with parameter-sharing across
pairs and a relational inductive bias, avoiding a flat monolithic
parameterization of all pairwise relations (although the worst-case
$O(P_{\max})\!=\!O(Q^2)$ pair-construction and edge-MLP cost still applies;
see Sec.~\ref{sec:method:arch}, Appendix~\ref{sec:appendix:arch_details}); and
(iii)~\emph{behavior-cloning warm start and Expert
Iteration}~\cite{anthony2017exit,ross2011dagger} to bootstrap from heuristic
demonstrators and then exceed them.

In this paper, we proposed a deep RL framework for \emph{the deadline-constrained coded caching delivery} in a decentralized random-prefetching model, while providing an empirical evaluation across multiple baselines, holdout seeds, and non-ID regimes (curriculum-seen and unseen-parameter; detailed in Section~\ref{sec:experiments}, where ``OOD-'' is used in the narrow within-family sense of unseen parameter values, i.e., the parameter values never seen during the training of the policy network). In the uniform-demand benchmark, our $\sigma$-selected checkpoint, i.e., the policy network picked by the robust-advantage rule on BE-score $\sigma$ at the validation (see Sec.~\ref{sec:method:selection}), achieves the lowest broadcast-packet expiration ratio $\rho$ among all of the coded multi-casting methods, reducing $\rho$ by $40.9\%$ ($0.208$ vs.\ $0.352$) with respect to SACM{++} (the paired-bootstrap effect sizes in Table~\ref{tab:paired_bootstrap}; the uncoded ED-Unicast baseline still attains a lower $\rho$ of $0.134$ by forfeiting all coding gain, so the ``lowest $\rho$'' headline is explicitly scoped to the coded multi-casting methods), while also attaining the best broadcast-efficiency score~$\sigma$ among the coded multi-casting methods across the full Track~A battery (and the best overall $\sigma$ in 7 of 8 Track~A regimes; ED-Unicast leads on $\sigma$ at OOD-delay10), along with a substantially lower per-step request miss rate $m_{\mathrm{req}}$ than the SACM{++} comparator (the only baseline in the main-text request-level table) on Track~A, even though ED-Unicast remains the best deadline-protection baseline on $m_{\mathrm{req}}$ in the paired-bootstrap analysis (Table~\ref{tab:paired_bootstrap}: $\Delta m_{\mathrm{req}} = +0.074$ for PPO minus ED-Unicast at ID-default, i.e., PPO is worse than ED-Unicast on this metric), so the broader ``lower than all coded baselines'' statement is confined to the SACM{++} comparison while the GCM, SACM, SACM+, and $\tau$-Fit extension is left as future work. The cross-track BE-score picture is more nuanced and is reported separately for the Zipf benchmark in Section~\ref{sec:results:zipf}. The policy network also exhibits a \emph{selective merge strategy}, i.e., the method of selective merging, a qualitative behavior not reproducible by any fixed-threshold heuristic.

\paragraph{Modeling abstraction for unicast slot cost.}
We adopted a one-slot-per-record service model, i.e., each transmission consumes exactly one slot, whether it is a coded XOR broadcast over a packet union $f_{\mathrm{mg}}$ or a unicast of a coding-state record, so that a unicast that resolves an aggregate record (built up from multiple merges) costs the same slot budget as a unicast of a singleton, while the XOR pathway uses an equal-length packet model and the one-gap invariant of Appendix~\ref{sec:appendix:state_sufficiency} for the decodability, and the unicast pathway uses the record-level cost abstraction. We document this choice considering that it determines the throughput--deadline trade-off the policy network optimizes, i.e., a packet-level unicast cost scaling with $|f_r|$ would change both the simulator and the reported gains. The abstraction is reflected in the broadcast-level metrics ($U_t = 1$ for any unicast, regardless of $|f_r|$), while the request-level metrics ($\eta_{\mathrm{req}}$, $m_{\mathrm{req}}$, $\sigma_{\mathrm{req}}$) provide a complementary view that credits each original arrival exactly once.

We evaluate the learned policy along four metric families that use \emph{three different notions of demand}:
(i)~\emph{broadcast-level deadline compliance} ($\rho$: the fraction of broadcast-level packet-set mass that expires before delivery, measuring channel-use waste);
(ii)~\emph{distinct file-identity coverage} ($\delta$: distinct file-identity coverage among the resolved (served or expired) identities under the $H{=}50$ in-episode horizon (Eq.~\eqref{eq:unique_demand_sat}), which conflates all packet arrivals for the same file into a single coverage event, not a measure of multi-packet demand completion, and subject to episode-end right-censoring of identities still pending at step $H$);
(iii)~\emph{request-level accounting} ($\eta_{\mathrm{req}}$, $m_{\mathrm{req}}$, $\sigma_{\mathrm{req}}$, computed from per-arrival identifiers stamped at arrival time), which tracks each original arrival through queue aggregation and credits completion or miss exactly once per request identifier; and
(iv)~\emph{broadcast efficiency} ($\sigma$), a per-slot composite that rewards XOR degree and penalizes expirations.
We anchor our central claims on the primary demand-centric metrics ($\rho$ at the broadcast-packet level and $\delta$ at the file-identity level; both are formally defined in Sec.~\ref{sec:exp:metrics}, while ``demand-centric'' is a name-only umbrella, i.e., not a per-arrival semantic claim), considering $\sigma$ as a broadcast-efficiency diagnostic and the request-level family $(\eta_{\mathrm{req}}, m_{\mathrm{req}}, \sigma_{\mathrm{req}})$ as supplementary metrics whose full-baseline coverage we leave as future work.

The main contributions of this paper are:

\begin{enumerate}

\item \textbf{C1: Deadline-sensitive queue-state control with masked
combinatorial actions.}
We proposed a Gymnasium-compatible masked discrete-action queue-state control formulation for the deadline-constrained coded caching, trained as a stationary policy under a discounted truncated continuing-control surrogate (Sec.~\ref{sec:env:reward}), while the formal model is a contextual POMDP (Definition~\ref{def:problem}), i.e., the ``masked-MDP'' language used elsewhere is shorthand for this surrogate.
The state encodes a finite pending queue of $Q$ requests, each
annotated with target-cache identity (one-hot), side-information
provider flags, normalized remaining deadline, packet-set size,
and merge degree.
The discrete action space tracks up to $P_{\max}=\binom{Q}{2}=45$ feasible merge pairs, while each pair admits two keep-side choices, thereby yielding up to $2P_{\max}$ coded actions (where the keep-side decision controls which queue slot is freed for replenishment by a fresh arrival), plus a deterministic unicast fallback that always serves the earliest-deadline request ($2P_{\max}+1=91$ actions total), so that the policy network learns \emph{which} merge pair to execute and \emph{which side to keep}, or \emph{whether to defer} to the earliest-deadline unicast, considering that it does not select individual unicast targets, while the dynamic feasibility masking enforces hard constraints at every step.
The reward combines a system-aligned $U_t - E_t$ base term, a quality
bonus on coded merges (intersection-size reward minus
union-size penalty), and a \emph{merge-potential} shaping
term~\cite{ng1999shaping} based on the fraction of currently feasible
queue pairs $\Phi(s) = |\mathcal{M}(s)|/\binom{Q}{2}$, \emph{inspired
by} potential-based shaping~\cite{ng1999shaping}; the formal
invariance theorem applies to an idealized fully observed
continuing-state MDP, while in our implemented contextual-POMDP
surrogate (truncated at $H{=}50$, hidden episode cache placement,
aliased aggregate observations) the shaping term is best read as
heuristic credit-assignment guidance rather than as a guarantee of
optimal-policy preservation.
This formulation jointly captures the decentralized cache diversity, the multiple concurrent outstanding requests, the hard per-request expiration deadlines, the keep-side control, and the dynamic invalid-action masking, while to the best of our knowledge (after the literature sweep summarized in Section~\ref{sec:related} and Table~\ref{tab:gap_map}), this specific combination has not previously been integrated in a single RL delivery formulation, even though closely related ingredients exist in adjacent settings, i.e., the claim should be read as a positioning statement rather than as a definitive non-existence result.

\item \textbf{C2: Graph-attentive MaskablePPO policy network.}
We proposed a graph-structured policy network, i.e., the policy network, a structured actor--critic architecture that treats the pending request queue as a node-attributed graph and the candidate merges as the edges, while a two-layer, four-head graph self-attention module~\cite{velickovic2018gat} builds contextual node and edge representations, considering that the shared edge MLPs then score all $(\text{pair},\text{keep-side})$ actions jointly.
This design captures the combinatorial interaction among merge candidates
through message passing with shared per-pair parameters, replacing a flat
monolithic parameterization of all pairwise relations by a relational
encoder (without eliminating the worst-case $O(P_{\max})\!=\!O(Q^2)$ pair
processing intrinsic to scoring all candidate merges; see
Sec.~\ref{sec:method:arch}, Appendix~\ref{sec:appendix:arch_details}).
The policy is trained via Maskable Proximal Policy Optimization
(MaskablePPO)~\cite{schulman2017ppo,sb3maskableppo} with Generalized Advantage
Estimation~\cite{schulman2015gae}, and feasibility is enforced at every
policy-gradient step through the action mask.

\item \textbf{C3: Rollout-improved behavior cloning and ExIt-style online
distillation.}
To handle the sparse rewards and the hard exploration in the early training stages, we adopted a three-phase pipeline, i.e., (i)~a behavior-cloning warm start from a rollout-improved heuristic teacher, (ii)~value-network pre-training to calibrate the critic before the policy updates, and (iii)~MaskablePPO fine-tuning with ExIt/DAgger-style online dataset aggregation~\cite{anthony2017exit,ross2011dagger} and a progressive curriculum that increases the task difficulty~\cite{bengio2009curriculum}, while the full training details are provided in Sections~\ref{sec:method} and~\ref{sec:experiments}, considering that this pipeline lets the policy network first replicate, and then \emph{exceed}, its heuristic teachers.

\item \textbf{C4: Comprehensive evaluation and out-of-distribution
generalization.}
We evaluated the policy network against the baselines covering three strategy families, i.e., the conservative uncoded unicast, the aggressive coded multi-casting, and the oracle-tuned threshold policies, under the holdout evaluation (Section~\ref{sec:experiments}).
In the uniform-demand benchmark (Track~A), under the simulator conventions of Sec.~\ref{sec:env} (notably the one-slot-per-record unicast cost abstraction (A2) and the uniform representative-destination convention $k_{\mathrm{mg}}\sim\mathrm{Unif}\{k_i,k_j\}$ of Eq.~\eqref{eq:mg_dest}, both flagged as first-order limitations in Sec.~\ref{sec:conclusion}), our $\sigma$-selected checkpoint reduces the broadcast-packet expiration ratio $\rho$ by $40.9\%$ ($0.208$ vs.\ $0.352$) and expired-record count per episode ($\varepsilon$) by $50.7\%$ relative to SACM{++}, the highest-throughput coded-multicast baseline (and the heuristic teacher used in BC/ExIt training); on Track~A it also achieves the highest $\sigma$ among coded-multicast methods across all 8~Track~A regimes (and the highest overall $\sigma$ in 7 of 8 Track~A regimes; ED-Unicast leads on $\sigma$ at OOD-delay10 at $0.019$ vs.\ PPO's $-0.002$). Full metrics are in Section~\ref{sec:results}; on the ID-default uniform table at $\rho=0.345$, GCM and SACM tie for the lowest broadcast-packet expiration ratio among coded baselines, and both are reported alongside SACM{++} in the per-metric ranking tables. On the Zipf-demand benchmark (Track~B) the cross-method BE-score lead does not transfer: ED-Unicast tops $\sigma$ overall, and SACM, SACM+, GCM also exceed PPO on $\sigma$ at Zipf ID-default; see Section~\ref{sec:results:zipf}.
The policy network learns a qualitatively different policy, i.e., a \emph{selective merge strategy}, that defers to the earliest-deadline unicast to protect the deadlines while reserving the coded transmissions for the merge opportunities where the side-information intersection is genuinely beneficial. The interesting fact we observed is that these gains transfer across the full Track~A non-ID battery (2~curriculum-seen and 5~unseen-parameter regimes covering the unseen cache fractions and the deadline budgets, along with the parameter-invariance sweeps over the file count under fixed per-packet caching probability $p_c$, all within the same simulator family), with the advantages often \emph{growing} under the high-stress conditions, even though all transfer experiments keep $K$, $Q$, the action dimensionality, and the placement family fixed, so these results demonstrate the \emph{within-family parameter generalization} rather than the broad architectural out-of-distribution robustness. In an extension study (Section~\ref{sec:results:zipf}), a separately trained popularity-aware variant remains competitive with the coded multi-casting baselines across 11~additional Zipf-demand regimes.

\end{enumerate}

The remainder of the paper is organized as follows.
Section~\ref{sec:related} reviews related work.
Section~\ref{sec:model} presents the system model.
Sections~\ref{sec:env}--\ref{sec:method} detail the contextual-POMDP
formulation and the training pipeline.
Sections~\ref{sec:experiments}--\ref{sec:results} report the experimental
setup and the uniform-demand results.
Section~\ref{sec:results:zipf} presents the Zipf-demand extension.
Section~\ref{sec:ablation} reports the full ablation analysis.
Section~\ref{sec:conclusion} concludes.

\section{Related Work}
\label{sec:related}

\subsection{Coded Caching: Foundations and Decentralized Placement}
\label{sec:rel:theory}

The coded-caching framework of Maddah-Ali and Niesen~\cite{maddah2014fundamental} showed that the server can achieve a \emph{global caching gain} scaling with the number of users~$K$ by exploiting the cached side information for the XOR multi-casting, while the decentralized extension~\cite{maddah2015decentralized} removed the need for the coordinated prefetching, i.e., each cache independently stores the subfiles i.i.d.\ at rate~$p_c$, and the server uses the resulting random overlaps for the coded transmissions, thereby achieving an order-optimal memory--rate tradeoff. The mechanism is the \emph{XOR multi-casting}, i.e., if user~$1$ requests file~$A$ and user~$2$ requests file~$B$, and each has pre-cached the other's file, the server broadcasts $A \oplus B$, thereby letting each user cancel its cached file to recover the requested one, while this gain generalizes to larger cliques of the users with overlapping side information. The decentralized random-prefetching model is the placement strategy we adopt, while Pedarsani et al.~\cite{pedarsani2014online} extended the framework to the sequential demand arrivals, characterizing the achievable rate regions when the server dynamically decides how to group the requests, i.e., a step toward the streaming delivery setting we consider.

Our work builds directly on this decentralized model, but departs from it in two ways, i.e., (i) we introduce the hard per-request expiration deadlines absent from~\cite{maddah2015decentralized,pedarsani2014online}, and (ii) we made the \emph{keep-side decision} an explicit \emph{learned action dimension} controlling the queue composition across the time slots, rather than the hand-designed endpoint-retention rule that the earlier coded-caching heuristics use, like the degree-aware endpoint retention of SACM+/SACM++ (the ablation in Appendix~\ref{sec:appendix:keepside_ablation}). To our knowledge, the prior coded-caching work treats the keep-side as a fixed heuristic rather than as an action exposed to a deadline-aware RL scheduler with the invalid-action masking, while we read this as a positioning statement rather than as a definitive non-existence result.

\subsection{Coded Delivery Algorithms and Complexity}
\label{sec:rel:delivery}

With arbitrary cache contents, the optimal coded delivery reduces to an NP-hard clique-cover problem on the side-information graph~\cite{asghari2019approximation}, while the polynomial-time approaches include the graph-coloring-based schemes for specific graph structures (GCLC~\cite{ji2015multiple}) and the heterogeneous-link extensions (HgLC~\cite{vettigli2019efficient}), both reporting substantial coding gains in the deadline-free settings.

For the general decentralized case, Asghari et al.~\cite{asghari2019approximation} provided an $O(\log K)$-approximation polynomial-time algorithm under their model assumptions, i.e., the \emph{Size-Aware Coded Multicast} (SACM) family, including the enhanced SACM+ and SACM++ variants with the degree-aware keep-side selection. We adopted SACM++ as the principal heuristic teacher for the BC and ExIt training considering that it achieves the best served-per-transmission throughput in our deadline-constrained simulator, while we do not claim SACM++ is the best per-metric across all evaluated criteria, and indeed SACM has a lower miss ratio than SACM++ on the ID-default uniform table (Sec.~\ref{sec:results}). All of the SACM variants and GCM maintain a \emph{100\% merge rate}, i.e., whenever any feasible merge exists, they execute it.

Despite the strong coding-gain guarantees, all of these algorithms are \emph{stateless} greedy rules, i.e., they do not account for the remaining deadlines, the effect of the current merge on the future queue composition, or the long-horizon consequences of consuming the shared side information prematurely, while the policy network, i.e., a learned policy that internalizes these future-state effects, could in principle select high-quality pairs selectively rather than merging every feasible candidate, and we test this hypothesis in Sections~\ref{sec:results}--\ref{sec:ablation}.

\subsection{Delay-Sensitive and Online Coded Caching}
\label{sec:rel:deadlines}

Niesen and Maddah-Ali~\cite{codedcachingdelaysensitive} characterized the gain--delay tradeoff, i.e., achieving the large coding gains requires aggregating many of the requests, which is in tension with the hard per-request deadlines, while Pedarsani et al.~\cite{pedarsani2014online} extended this analysis to the online settings with the sequential arrivals and the rate-delay tradeoff regions, even though both of the works operate at the information-theoretic level and do not produce concrete scheduling policies for the streaming queues with the hard deadlines.

Several recent works have addressed the asynchronous and the time-varying aspects of the coded caching, considering that Jiang et al.~\cite{jiang2020decentralized} proposed a decentralized asynchronous coded caching scheme for fog radio access networks, providing both synchronous and asynchronous transmission methods for different delay requirements with closed-form fronthaul-load expressions, while Zhang and Tao~\cite{zhang2021deep} applied deep learning to the wireless coded caching under the unknown and time-variant content popularity, thereby optimizing the cache placement via LSTM-based popularity prediction and a supervised deep deterministic policy gradient. Like our method, this work uses neural networks in the coded-caching pipeline, even though it targets the \emph{placement} rather than the \emph{delivery} phase. Amir et al.~\cite{amir2021coded} studied the coded caching with the time-varying file popularities and the asynchronous delivery, designing schemes that exploit the delivery messages to proactively update the user caches, while Yang et al.~\cite{yang2022optimal} formulated the optimal scheduling for the asynchronous coded caching, thereby deriving the rate-optimal delivery schedules when the user requests arrive at different times with prescribed deadlines. These works advance the understanding of the asynchronous and non-stationary coded caching, even though none of them formulates an RL-based \emph{delivery scheduler} that jointly handles the hard per-request expirations, the keep-side control, and the dynamic invalid-action masking, i.e., the specific gap our work addresses.

A remaining gap is the absence of an \emph{operational scheduling policy} that can be deployed online under the hard deadline constraints, considering that to the best of our knowledge, no prior RL delivery formulation models the following combination as a single integrated learned scheduler, i.e., (i)~a finite queue of pending requests each with a different remaining deadline, (ii)~the XOR feasibility that depends on the current cache state, (iii)~a keep-side decision \emph{exposed as a learned action} (rather than a hand-designed endpoint rule as in SACM+/SACM++) that dynamically reshapes the queue composition, and (iv)~the immediate replacement of the expired requests by the fresh arrivals, while our contextual-POMDP formulation in Section~\ref{sec:env} addresses this gap, thereby enabling deep RL training at scale.

Outside the coded-caching literature, the canonical online-scheduling alternative for the deadline-constrained queueing systems is the Lyapunov-drift / drift-plus-penalty optimization (Tassiulas and Ephremides~\cite{tassiulas1992stability}; Neely~\cite{neely2010stochastic}), i.e., the family of max-weight policies that achieve the queue stability and the bounded-suboptimality guarantees while requiring no learning. These do not directly apply to the coded delivery scheduling, considering that the per-slot broadcast utility $U_t$ depends non-linearly on the queue state through the pairwise side-information intersection $\mathcal{S}_i\cap\mathcal{S}_j$ and the keep-side bit $\kappa$ at the merge step, thereby breaking the standard linear-drift formulation; the extension of the drift framework to this combinatorial setting is identified as future work in Section~\ref{sec:conclusion}.

\subsection{Reinforcement Learning for Caching Systems}
\label{sec:rel:rl_caching}

The reinforcement learning has been extensively applied to the cache \emph{placement and eviction}, i.e., deciding what content to store, rather than to the delivery scheduling, while the representative examples include the PPO-based cache replacement under the dynamic pricing~\cite{sadeghi2019adaptive}, the Q-learning for the proactive caching in the wireless networks~\cite{somuyiwa2018proactive}, the deep RL for the joint caching and user scheduling~\cite{zhong2018framework}, and the survey by Liu et al.~\cite{liu2019mobile} covering 30+ RL-based caching approaches.

A shared limitation of these placement-focused works is that they treat the coded delivery as a black box, i.e., the delivery phase is assumed to be simple unicasting, with no scheduling decisions over the structured XOR merge candidates, while our work is entirely in the \emph{delivery} track, taking the decentralized random placement as given, i.e., a complementary contribution to the placement-focused RL literature, considering that the joint placement--delivery optimization is identified as a future direction in Section~\ref{sec:conclusion}.

\subsection{RL for Coded Delivery, Graph-Attentive Scheduling, and Training Techniques}
\label{sec:rel:rl_delivery}

\textbf{RL for coded delivery.} The most directly related prior work is Naderializadeh and Asghari~\cite{naderializadeh2019learning}, who trained a deep RL agent for the coded caching delivery under arbitrary decentralized cache contents, considering that their actor-critic policy (flat MLP) scores the candidate merge pairs at each step and learns to form the coded transmissions that match or slightly outperform SACM while reducing the online inference complexity, i.e., this work shows that the RL is a viable approach for the coded delivery problem.

Our work differs from~\cite{naderializadeh2019learning} in four respects, i.e.,
\begin{enumerate}[label=(\roman*)]
  \item \textbf{Hard per-request deadlines.} \cite{naderializadeh2019learning} uses a soft delay penalty but does not model the hard expiration deadlines, even though in our setting, an expired slot is refilled immediately, thereby producing the non-stationary feasibility dynamics.
  \item \textbf{Graph-attentive policy architecture.} Their flat MLP encodes all of the pair features into a fixed-length vector, thereby losing the relational structure, while our graph-structured policy network, i.e., the policy network, treats each request as a graph node and each feasible merge as an edge, applying two-layer graph self-attention before scoring the actions.
  \item \textbf{Invalid-action masking.} \cite{naderializadeh2019learning} does not employ the action masking, so the policy gradient wastes capacity on the infeasible actions, while our MaskablePPO formulation masks the infeasible actions at both the sampling and the gradient computation.
  \item \textbf{Evaluation scale.} \cite{naderializadeh2019learning} evaluates a single configuration against few of the baselines without OOD analysis~\cite{henderson2018deep,agarwal2021deep}, while we evaluate against 9~Track~A baselines + PPO ($10$ unique methods) across 50 holdout seeds $\times$ 200 episodes and 7~non-ID conditions, plus a main-body Zipf-demand extension study (Section~\ref{sec:results:zipf}) covering 11~additional regimes.
\end{enumerate}

\textbf{Architectural and training building blocks.}
Our method draws on several established techniques, each adapted to the coded-caching setting, considering that the graph attention networks~\cite{velickovic2018gat} provide the inductive bias for the variable-size, relationally structured inputs, i.e., here, the pending-request queue is a node set with the XOR-feasibility edges, while the invalid-action masking~\cite{huang2020invalidmask} is theoretically justified and empirically necessary when up to 90 of 91 actions may be infeasible at a given step (i.e., only the unicast fallback is valid), so our MaskablePPO formulation~\cite{sb3maskableppo,sb3contrib2022maskable} applies dynamic masks at both the sampling and the gradient computation. The behavior cloning and DAgger~\cite{ross2011dagger} address the distributional shift in the imitation learning, while the Expert Iteration~\cite{anthony2017exit} extends this to a self-improvement loop where the policy network's own rollout planner generates improved labels for the online distillation, and the curriculum learning~\cite{bengio2009curriculum} over the progressive task difficulty shapes the policy network toward a deadline-conscious selective merge strategy, i.e., the method of selective merging, rather than the aggressive-merge policy that the direct training converges to (see the ablation in Section~\ref{sec:ablation}). The PPO with GAE~\cite{schulman2017ppo,schulman2015gae} provides the policy-gradient backbone, implemented via Stable-Baselines3~\cite{raffin2021stable}.

Each of these components exists independently in the literature, even though we are not aware of prior work that integrates them into a single coded-caching delivery agent with the hard deadline constraints, considering that beyond stacking, the integration required two non-obvious couplings, i.e., (i) the dynamic action mask must be wired through the graph-attention edge scorer so that the infeasible pairs are filtered \emph{before} the message passing (rather than zeroed out post-hoc), thereby avoiding the gradient-flow pathology of masking after the softmax (Sec.~\ref{sec:method:arch}), and (ii) the BC/ExIt teacher's pair pre-filter must be coupled to the same per-pair edge scorer used at the inference (top-$K_{\mathrm{pair}}$ candidate pruning), so that the imitation distribution sits on the same support the policy network can express (Sec.~\ref{sec:method:bc}).
We summarize what distinguishes our approach:
\begin{itemize}[leftmargin=*, nosep]
  \item \textbf{Hard-deadline contextual-POMDP surrogate.} We jointly model the per-request expiration
        deadlines, the keep-side decisions, and the masked combinatorial actions, i.e., a combination
        not present in the prior RL-based delivery formulations
        (\S\ref{sec:rel:deadlines}).
  \item \textbf{Graph-attentive policy.} Unlike the flat MLP
        of~\cite{naderializadeh2019learning}, our graph-structured policy network, i.e., the policy network,
        applies two-layer self-attention over the merge graph, thereby capturing
        the relational structure among the pending requests (\S\ref{sec:method:arch}).
  \item \textbf{Behavior-cloning warm start and self-improving distillation.}
        We combined the rollout-improved BC with the Expert Iteration online distillation,
        which lets the policy network exceed its heuristic teacher, i.e., a self-improvement
        loop absent from the prior coded-caching RL work
        (\S\ref{sec:method:bc}--\ref{sec:method:exit}).
  \item \textbf{Large-scale OOD evaluation.} We evaluate across 7~non-ID regimes (2~curriculum-seen + 5~OOD) in the main uniform-demand evaluation,
        along with a main-body Zipf-demand extension study with 11~additional regimes
        (\S\ref{sec:results:zipf}),
        50~holdout seeds, and 9~Track~A baselines + PPO ($10$ unique methods) spanning three policy families, while the
        evaluation scale is qualitatively different from the prior work
        (\S\ref{sec:results}).
\end{itemize}
The integration and its effect on policy quality are validated by the ablation
study in Section~\ref{sec:ablation}.

\medskip\noindent\textbf{Summary of gaps and our contributions.}
Table~\ref{tab:gap_map} maps the four gaps identified above to the contributions of this paper.

\begin{table}[!t]
  \caption{Gap--contribution mapping. Each row identifies a gap in the prior literature and the corresponding contribution that addresses it.}
  \label{tab:gap_map}
  \centering
  \begin{adjustbox}{max width=\columnwidth}
  \begin{tabular}{@{}p{0.48\columnwidth}p{0.48\columnwidth}@{}}
    \toprule
    \textbf{Gap in Prior Work} & \textbf{Our Contribution} \\
    \midrule
    The combination of hard deadlines, masked combinatorial actions, and keep-side decisions does not, to our knowledge, appear together in any prior RL delivery formulation (\S\ref{sec:rel:deadlines}). We frame this as a positioning statement based on the literature sweep documented in this section, not as a definitive non-existence claim, and welcome pointers to closely related prior work
      & \textbf{C1}: Deadline-sensitive masked queue-state control (stationary-policy / truncated continuing-control surrogate) with dynamic feasibility masking (\S\ref{sec:env}) \\[4pt]
    Flat MLP policies ignore the graph structure of the mergeability relation (\S\ref{sec:rel:rl_delivery})
      & \textbf{C2}: Graph-attentive MaskablePPO policy network (\S\ref{sec:method:arch}) \\[4pt]
    Pure RL exploration is slow; we are not aware of a coded-caching agent that uses teacher distillation or self-improvement (\S\ref{sec:rel:rl_delivery})
      & \textbf{C3}: Rollout-improved BC warm start and ExIt online distillation (\S\ref{sec:method:bc}--\ref{sec:method:exit}) \\[4pt]
    Prior evaluations use a single configuration with few baselines and no OOD analysis (\S\ref{sec:rel:rl_delivery})
      & \textbf{C4}: 9 Track~A baselines + PPO = 10 unique methods (12 table rows with two literature-compatible aliases), 50 holdout seeds $\times$ 200 episodes, 7 non-ID evaluation conditions (uniform demand); main-body Zipf-demand extension study adds 11 additional regimes with 6 baselines + PPO-Zipf = 7 unique methods (\S\ref{sec:results:zipf}) \\
    \bottomrule
  \end{tabular}
  \end{adjustbox}
\end{table}

\section{Problem Formulation}\label{sec:formulation}
\label{sec:model}  

We consider a content delivery network consisting of one central server connected to $K$ edge caches, i.e., the users, through a shared broadcast bottleneck link, while the notation is summarized in Table~\ref{tab:notation} preceding the Introduction.

\FloatBarrier

\subsection{Network Model and Decentralized Cache Placement}
\label{sec:model:network}

The server holds a library of $N$ files, and following the standard coded-caching practice~\cite{maddah2014fundamental,maddah2015decentralized}, each file is divided into $B$ equal-sized \emph{subfiles}, i.e., the packets, thereby yielding a total library of $F = NB$ packets.
All packets are assumed to be equal-length, while the XOR operations are bitwise over $\mathrm{GF}(2)$, so that a single coded packet $X = p_1 \oplus p_2$ has the same length as each constituent packet.
In our setting we use $N = 100$, $B = 10$, and thus $F = 1{,}000$ packets, while the evaluation also covers $N \in \{60, 120, 150\}$ to test the generalization.

We adopt the \emph{decentralized fixed-size random placement} inspired by~\cite{maddah2015decentralized}, while at the episode reset, each cache~$k$ independently samples exactly $M = \lfloor p_c F \rfloor$ distinct packet IDs uniformly without replacement from the $F$-packet library, without any coordination with the other caches.\footnote{The standard theoretical model uses i.i.d.\
Bernoulli($p_c$) placement~\cite{maddah2015decentralized}, yielding a random
cache size with mean $p_c F$.  Our implementation fixes the cache size at
$M = p_c F$ for reproducibility (each cache holds exactly $M$ packets, so
side-information statistics do not fluctuate across episodes). For
$F{=}1000$ and $p_c{=}0.30$ the two placement models were empirically
similar in pilot tests on the ID-default regime; we did not run a full
sensitivity comparison across all reported regimes and adopt fixed-size
sampling as a reproducibility-friendly substitute rather than as a
proven equivalence.}
Let $\mathcal{C}_k \subseteq \{0, 1, \ldots, F-1\}$ denote the resulting content of cache~$k$, with $|\mathcal{C}_k| = M$.
Any single packet is present in a given cache with the marginal probability $M / F = p_c$; however, because each cache draws a fixed-size subset, the placements of different packets \emph{within} the same cache are not independent, while the placements \emph{across} different caches are mutually independent.
We use $p_c = 0.30$ (30\%) as the default training configuration and evaluate the generalization to $p_c \in \{0.20, 0.40\}$ out-of-distribution.
The cache placement is \emph{independent of the file-request distribution}, i.e., each packet is cached with the marginal probability~$p_c$ regardless of the file popularity.
This separation means that changing the demand model (e.g., from uniform to Zipf) alters the request queue statistics without affecting the cached content, thereby allowing us to isolate the effect of demand skew on the scheduling performance. However, as described in Section~\ref{sec:res:ood_zipf}, the Zipf-demand agent uses an extended observation space with popularity-aware features and is trained separately from the uniform-demand agent.

The server transmits over a \emph{shared broadcast link}, i.e., every packet sent by the server is received by all $K$ caches simultaneously.
This broadcast nature enables the XOR-coded transmissions that can carry useful information for multiple users in a single slot, while this is the property that the coded caching exploits to achieve the multicast gain.
Fig.~\ref{fig:system_model} illustrates the network topology and the one-step transition timeline.

\begin{figure*}[t]
  \centering
  \resizebox{\textwidth}{!}{
\begin{tikzpicture}[
    >=Stealth,
    node distance=0.55cm,
    server/.style={rectangle, draw=black!80, fill=blue!8, minimum width=2.45cm, minimum height=0.72cm, font=\footnotesize\bfseries, rounded corners=2pt},
    cache/.style={rectangle, draw=black!70, fill=green!8, minimum width=0.95cm, minimum height=0.58cm, font=\scriptsize, rounded corners=2pt},
    qslot/.style={rectangle, draw=black!60, fill=yellow!8, minimum width=2.15cm, minimum height=0.36cm, font=\tiny, align=center},
    phase/.style={rectangle, draw=black!70, fill=blue!10, minimum width=2.25cm, minimum height=0.72cm, font=\scriptsize, rounded corners=3pt, align=center},
    action/.style={diamond, draw=black!70, fill=orange!12, minimum width=0.95cm, minimum height=0.7cm, font=\scriptsize, align=center, inner sep=1pt},
    ann/.style={font=\tiny, text=black!70, align=center},
    note/.style={rectangle, draw=black!45, rounded corners=2pt, fill=white, font=\tiny, align=left},
    arr/.style={->, thick, black!70},
    label/.style={font=\footnotesize\bfseries, text=black!85},
  ]

  \node[label, anchor=south] at (2.6, 4.35) {System model};

  \node[server] (server) at (2.6, 3.65) {Central server};
  \node[ann, below=1pt of server] {library: $100$ files $\times$ $10$ packets $= 1000$ packet IDs};

  \foreach \i/\lab/\x in {1/0/0.4,2/1/1.5,3/2/2.6,4/3/3.7,5/4/4.8} {
    \node[cache] (c\i) at (\x, 2.2) {$\mathcal{C}_{\lab}$};
    \draw[->, thin, blue!55] (server.south) -- (c\i.north);
  }
  \node[ann] at (2.6, 1.55) {each cache stores $|\mathcal{C}_k| = p_c F = 0.3\cdot 1000 = 300$ packet IDs\\sampled uniformly without replacement};

  \node[label, anchor=south] at (6.9, 3.25) {Queue state ($Q=10$)};
  \foreach \j/\y in {0/2.95,1/2.56,2/2.17,3/1.78} {
    \node[qslot] (q\j) at (6.9, \y) {$r_{\j}=(k_{\j}, f_{\j}, d_{\j}, \mathcal{S}_{\j})$};
  }
  \node[ann] at (6.9, 1.2) {$f_j$: requested packet-set\quad $\mathcal{S}_j$: caches that store $f_j$};
  \node[ann] at (6.9, 0.92) {$\vdots$ 10 queue slots};

  \node[label, anchor=south] at (2.8, 0.95) {Merge semantics};
  \node[qslot, fill=red!8, minimum width=2.0cm] (ri) at (1.15, 0.35) {$r_i=(k_i, f_i, d_i, \mathcal{S}_i)$};
  \node[qslot, fill=red!8, minimum width=2.0cm] (rj) at (4.45, 0.35) {$r_j=(k_j, f_j, d_j, \mathcal{S}_j)$};
  \node[font=\scriptsize, circle, draw=black!70, fill=orange!15, inner sep=1.5pt] (xor) at (2.8, 0.35) {$\oplus$};
  \draw[->, thin] (ri.east) -- (xor);
  \draw[->, thin] (rj.west) -- (xor);
  \node[note, anchor=north] (mergebox) at (2.8, -0.32) {feasible iff $f_i \subseteq \mathcal{C}_{k_j}$ and $f_j \subseteq \mathcal{C}_{k_i}$\\merged state: $(k_{\mathrm{mg}},\, f_i\cup f_j,\, \min(d_i,d_j),\, \mathcal{S}_i\cap \mathcal{S}_j)$\\$k_{\mathrm{mg}}\!\sim\!\mathrm{Unif}\{k_i,k_j\}$ independently of the keep-side $\kappa$};

  \draw[dashed, black!30, line width=0.5pt] (9.15, 4.55) -- (9.15, -1.6);

  \node[label, anchor=south] at (14.35, 4.35) {One environment step ($s_t \rightarrow s_{t+1}$)};

  \node[phase, fill=green!10] (obs) at (11.2, 3.55) {Observe\\$o_t=\{\mathbf{x}_t,\mathbf{p}_t\}$\\and mask $\mathbf{m}_t$};
  \node[action] (act) at (14.35, 3.55) {$a_t$};
  \node[phase, fill=blue!10, minimum width=2.8cm] (tx) at (17.6, 3.55) {Transmit / queue update};
  \draw[arr] (obs) -- (act);
  \draw[arr] (act) -- (tx);

  \node[note, anchor=north, text width=5.0cm] (actnote) at (12.4, 2.95) {$a_t\in\{2k,2k+1,90\}$\\$2k/2k{+}1$: merge pair $\mathcal{M}_t[k]=(i,j)$;\\\quad keep-side bit $\kappa{=}0$ keeps $i$, $\kappa{=}1$ keeps $j$\\$90{=}2P_{\max}$: unicast earliest deadline};

  \node[note, anchor=north west, text width=5.0cm] (txmerge) at (16.2, 2.52) {merge: kept slot $\leftarrow$ aggregate;\\dropped slot $\leftarrow$ fresh request;\\packet-set XOR degree $U_t = |f_i\cup f_j|$};
  \node[note, anchor=north west, text width=5.0cm] (txsend) at (16.2, 1.20) {unicast: earliest-deadline record served\\in one slot ($U_t = 1$, one-slot-per-record);\\served slot $\leftarrow$ fresh request};

  \node[phase, fill=yellow!10] (tick) at (14.35, 0.95) {Deadline tick\\$d \leftarrow d-1$ for all queue entries};
  \node[phase, fill=red!8] (expire) at (11.2, -0.35) {Expire / refill\\if $d\le 0$};
  \node[phase, fill=green!10, minimum width=2.8cm] (rew) at (17.6, 0.05) {Reward + next state};

  \draw[arr] (tx.south) -- ++(0,-0.42) -| (tick.north east);
  \draw[arr] (tick) -- (expire);
  \draw[arr] (tick) -- (rew);
  \draw[arr] (expire.east) -| (rew.south);

  \node[note, anchor=north] (rbox) at (17.6, -0.72) {$R_t = 1\cdot U_t - 1\cdot E_t \;+\; \mathbb{1}[\text{coded merge}]\bigl(0.75\,|\mathcal{S}_i\cap \mathcal{S}_j| - 0.15\,\max(0,|f_i\cup f_j|-2)\bigr)$\\$\qquad\;\; +\; 0.20\,[\gamma_\Phi\Phi(s_{t+1})-\Phi(s_t)]$,\quad $\Phi(s)=|\mathcal{M}(s)|/\binom{Q}{2}$,\quad $\gamma_\Phi = 0.995$};

  \node[note, anchor=north] (endbox) at (14.35, -1.35) {episode ends when $t=H=50$; metrics are reported in the final $\texttt{info}$ dict};

\end{tikzpicture}}
  \caption{System model and one-step timeline for deadline-constrained coded caching.
    \textbf{Left:} A central server broadcasts over a shared link to $K{=}5$ edge caches,
    each storing exactly $M{=}p_c F{=}300$ library packets sampled uniformly without replacement. A queue of $Q{=}10$
    outstanding requests (each with a per-request deadline) drives the scheduling decisions.
    An XOR-coded merge combines two requests $r_i, r_j$ into a single broadcast $X$
    whenever the feasibility condition is met.
    \textbf{Right:} At each step the agent observes the queue state and feasible merge
    set~$\mathcal{M}_t$, selects either a coded merge (with keep-side choice~$\kappa$)
    or an earliest-deadline unicast, after which all deadlines decrement and expired
    requests are refilled.}
  \label{fig:system_model}
\end{figure*}
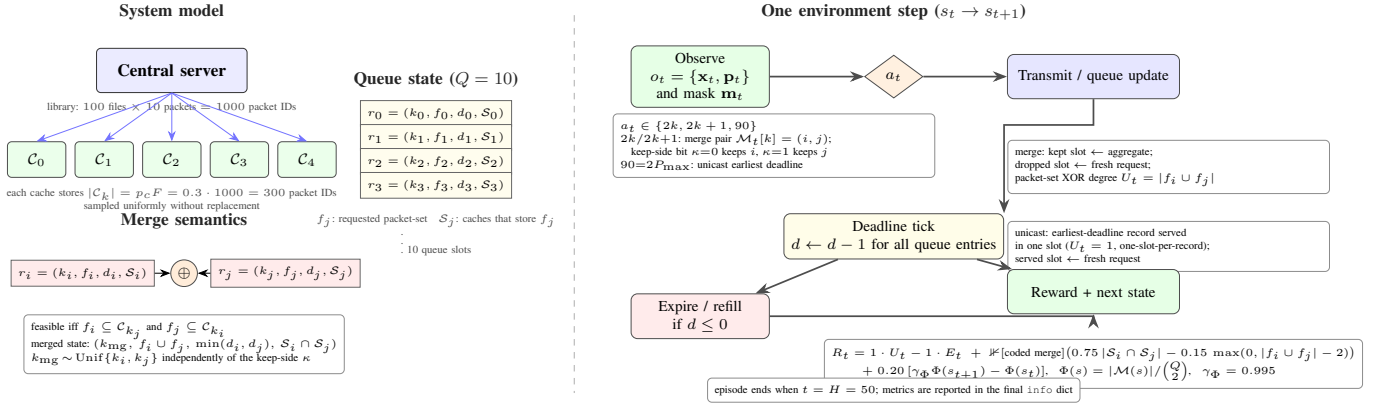

\subsection{Online Request Arrival and Queue Dynamics}
\label{sec:model:queue}

The users issue subfile, i.e., the packet, requests continuously over time.
We model the set of outstanding requests as a finite queue of depth $Q$ that is maintained at full occupancy throughout an episode, i.e., whenever a slot is vacated by the delivery or expiration, a new request arrives immediately, thereby modeling the steady-state operation of the delay-sensitive streaming system.

\paragraph{Request Representation}

Each active request $r$ in the queue is a four-tuple
\begin{equation}
  r = \bigl(k_r,\; f_r,\; d_r,\; \mathcal{S}_r\bigr),
  \label{eq:request_tuple}
\end{equation}
where
\begin{itemize}[leftmargin=*, nosep]
  \item $k_r \in \{0,\ldots,K-1\}$ is the \emph{destination cache}
        when $r$ is a fresh singleton request (the user making the
        request). When $r$ is an aggregate coding-state record
        produced by a chained merge, the same field carries the
        \emph{representative destination} $k_{\mathrm{mg}}$ sampled
        independently of $\kappa$ via Eq.~\eqref{eq:mg_dest}, a
        bookkeeping representative rather than residual user demand.
        A single record is in exactly one of these two states at any
        time, so the symbol $k_r$ is unambiguous in context; for
        clarity we write $k_{\mathrm{mg}}$ explicitly whenever the
        record in question is an aggregate (e.g., in unicast
        bookkeeping for aggregates, in the chained-merge proof of
        Appendix~\ref{sec:appendix:state_sufficiency}, and in the
        observation feature tables of
        Appendix~\ref{sec:appendix:features});
  \item $f_r \subseteq \{0,\ldots,F-1\}$ is the \emph{packet set}
        to be delivered (a single subfile at request creation; may grow to
        a union after an XOR merge, as detailed in Section~\ref{sec:model:coded});
  \item $d_r \in \{1,\ldots,D\}$ is the \emph{remaining deadline} in
        integer time slots; and
  \item $\mathcal{S}_r \subseteq \{0,\ldots,K-1\}$ is the
        \emph{side-information set}, i.e., the subset of caches that
        currently hold all packets in $f_r$ and can therefore cancel them
        from any XOR-coded packet.
\end{itemize}

\paragraph{Request Generation}

A fresh request $r_{\mathrm{new}}$ is generated by:
(i)~\emph{rejection-sampling} a candidate file/packet pair
$(n,b) \sim \mathrm{Uniform}\{0,\ldots,N-1\}\times\mathrm{Uniform}\{0,\ldots,B-1\}$
until the resulting packet $f_r = \{nB+b\}$ is \emph{not} held by
every cache simultaneously, i.e., until
$\{k:\, f_r\subseteq\mathcal{C}_k\}\neq\{0,\ldots,K-1\}$;
under independent placement with marginal probability $p_c$, the
per-draw rejection probability is $p_c^{K}$, which is about
$0.30^{5}\!\approx\!0.24\%$ at the default $p_c=0.30$ and about
$0.40^{5}\!\approx\!1.0\%$ at $p_c=0.40$, so the expected number of
resamples is essentially~1;
(ii) sampling the destination cache $k_r$ uniformly from the
admissible set
$\{0,\ldots,K-1\}\setminus\{k:\,f_r\subseteq\mathcal{C}_k\}$, which
guarantees $f_r \not\subseteq \mathcal{C}_{k_r}$ (no unsolicited
delivery);
(iii) computing $\mathcal{S}_r = \{k : f_r \subseteq \mathcal{C}_k\}$ (the
set of caches that can provide side information); and
(iv) drawing the initial deadline uniformly,
$d_r \sim \mathrm{Uniform}\{1,\ldots,D\}$.
A request is generated whenever the destination cache $k_r$ does not already hold $f_r$ (no unsolicited delivery), while we explicitly allow $\mathcal{S}_r=\emptyset$, in which case no XOR partner can cancel the request and the only feasible delivery action is the unicast.
The contextual-POMDP surrogate therefore covers both the codable arrivals ($\mathcal{S}_r\neq\emptyset$) and the unicast-only arrivals ($\mathcal{S}_r=\emptyset$), while the latter must still be served before their deadline, thereby contributing to the deadline misses on equal footing with the codable arrivals.
Each arrival therefore corresponds to a single packet, i.e., the subfile, request, while a full-file retrieval can be modelled as a stream of $B$ individual packet requests, each carrying its own deadline.

\paragraph{Modeling Assumptions}
\label{sec:model:assumptions}

We state five assumptions before introducing the transmission model and the reward, considering that they determine the slot-cost and deadline-bookkeeping semantics of the contextual-POMDP surrogate:
\begin{enumerate}[label=(A\arabic*),leftmargin=2.2em,nosep]
  \item \textbf{Equal-length packets.} All packets have the same
    length, so one XOR broadcast carries one packet's worth of data
    and feasibility is governed by the side-information set
    $\mathcal{S}_r$ rather than packet size.
  \item \textbf{One-slot-per-record service.} Each environment step
    delivers exactly one record: either a coded XOR over a feasible
    pair (one slot, $U_t = |f_{\mathrm{mg}}|$ as a packet-set
    cardinality, see Sec.~\ref{sec:exp:metrics}) or a
    unicast of one coding-state record (one slot, $U_t = 1$,
    \emph{regardless of $|f_r|$}). The unicast pathway is therefore a
    record-level service abstraction, not a packet-level model: a
    physical packet-level model would charge $|f_r|$ slots for the
    unicast of an aggregate. We adopt the record-level abstraction so
    that the merge-or-defer decision faces a uniform slot cost across
    its two branches; we revisit its broadcast-efficiency implications
    in Sec.~\ref{sec:exp:metrics}.
  \item \textbf{Static random cache placement.} Caches are populated
    once at episode start by uniform random sampling without
    replacement (cache fraction $p_c$) and held fixed for the episode;
    placement updates are out of scope.
  \item \textbf{EDF unicast.} The unicast action is bound to the
    earliest-deadline record; per-request unicast targeting is not
    part of the action space, isolating the merge-or-defer decision
    from unicast scheduling.
  \item \textbf{Refill-timing convention.} A queue slot vacated by a
    Phase-1 transmission is refilled \emph{before} the Phase-2 deadline
    decrement, so a fresh request sampled with $d_r \sim
    \mathrm{Uniform}\{1,\ldots,D\}$ enters the next decision step with
    effective remaining deadline in $\{0,\ldots,D-1\}$. A slot vacated
    by a Phase-3 expiration is refilled \emph{after} the decrement and
    therefore retains its sampled deadline in $\{1,\ldots,D\}$. As a
    consequence, a Phase-1 fresh request that draws $d_r=1$ joins the
    same step's expiration set $\mathcal{X}_t$
    (Eq.~\eqref{eq:exp_set}) and is reported under the standard
    expiration counters. This induces a small \emph{policy-independent}
    floor in $E_t$ and the related metrics ($\rho$, $\sigma$,
    $\varepsilon$, $E_t^{\mathrm{uniq}}$, $m_{\mathrm{req}}$): under
    $d_r \sim \mathrm{Uniform}\{1,\ldots,D\}$ the per-Phase-1-arrival
    same-step-expiration probability is exactly $1/D$, which upper-bounds
    the contribution at $1/D \approx 5\%$ for the default $D{=}20$,
    $\le 10\%$ for the tight-deadline regime $D{=}10$, and
    $\le 3.3\%$ for $D{=}30$. All inter-method comparisons in
    Sec.~\ref{sec:results} share the same episode seeds and the same
    initial cache placements, so this floor is shared in expectation
    up to the first-divergent-action RNG drift. The simulator does
    \emph{not} implement a strict common-random-numbers (CRN) scheme
    across methods (RNG draws within an episode are sequential, and
    policy-dependent refill timings cause the realized arrival
    sequences to diverge across methods after the first divergent
    action), so the floor does not cancel pathwise. The paired-bootstrap comparisons of
    Table~\ref{tab:paired_bootstrap} therefore inherit shared-context
    pairing rather than the stronger same-arrivals pairing. Absolute
    interpretations of $\rho, \varepsilon, m_{\mathrm{req}}$, and
    $\delta$ are \emph{benchmark-specific} to this refill convention
    and should be qualified accordingly when compared across
    simulators or papers using a different (e.g., symmetric) refill
    rule.
    We adopt this asymmetric convention deliberately, since it is the
    actual simulator behavior, and a faithful re-implementation
    must place the Phase-1 refill before the Phase-2 decrement and
    the Phase-3 refill after, exactly as in
    Algorithm~\ref{alg:one_step}.
\end{enumerate}

\subsection{Coded Transmission Model}
\label{sec:model:coded}

At each time slot the server performs exactly one of the two transmission types.

\paragraph{Unicast Transmission}

The server transmits the packet set $f_r$ of a single coding-state record~$r$ to its representative destination~$k_r$ in one slot, while in our formulation, the unicast action always selects the record with the shortest remaining deadline, $r^* = \arg\min_r d_r$, so that the unicast is a deterministic earliest-deadline-first (EDF) fallback rather than a free scheduling choice. Because the deadlines are integer-valued in $\{1,\ldots,D\}$ and the queue depth is~$Q$, ties on $d_r$ occur regularly, and we resolve them deterministically in favor of the \emph{smallest queue index}~$r$ among the tied records, where the queue index is the FIFO insertion order maintained by the simulator (the reference implementation realizes this via an insertion-ordered map keyed by queue index, even though the rule itself is language-independent, i.e., the tied records are served in arrival order). The same tie-break rule is applied consistently to the ED-Unicast, the PPO unicast fallback, and every lookahead teacher used in BC and ExIt.
Upon reception, $k_r$ obtains all packets in $f_r$ (which is a singleton at arrival and may be a union after one or more merges).
We adopt this \emph{one-slot-per-record} convention as a transmission-cost abstraction, i.e., the broadcast-level packet-set XOR degree is fixed to $U_t = 1$ for any unicast action, regardless of $|f_r|$.

\smallskip\noindent
\emph{Bookkeeping for aggregate-record unicasts.}
When~$r$ is an aggregate (built up from $n\ge 2$ prior merges), the underlying $n$ original arrivals were already \emph{decoded} at their respective merge broadcasts (Sec.~\ref{sec:model:coded}, ``Broadcast Semantics of Aggregation''), and the surviving $r$ is a residual \emph{coding-state record} that tracks the union $f_{\mathrm{mg}}$ for the future XOR-feasibility checks rather than the unserved demand.
A later unicast on~$r$ therefore physically clears that residual record, even though the request-level newly-completed set $C_t = \mathcal{A}(\mathrm{tx}_t) \setminus (D_{t-1} \cup L_{t-1})$ defined in Sec.~\ref{sec:model:request_accounting} excludes any arrival identifier already credited in $D_{t-1}$, so no original arrival is counted twice. The two readings of an aggregate-record unicast are therefore consistent: at the broadcast level it is a single record-clearing slot ($U_t = 1$, $|\mathcal{X}_t|$ unaffected if the deadline has not expired), while at the request level it adds at most the not-yet-credited subset of $\mathcal{A}(r)$ to $C_t$.
The one-gap invariant proved in
Appendix~\ref{sec:appendix:state_sufficiency} is a \emph{decodability}
property for the XOR pathway and does not constrain the unicast
slot-cost model.

The unicast target selection is fixed to the earliest-deadline-first rather than learned, thereby isolating the merge-or-defer decision and controlling the action-space growth, while learning per-request unicast targets is left as future work.

\paragraph{XOR-Coded Multicast}

Given two distinct requests $r_i$ and $r_j$ in the queue, the server may
broadcast a single XOR-coded packet $X = \bigoplus_{p \in f_i \cup f_j} p$
if and only if the following \emph{XOR feasibility condition} holds:
\begin{equation}
  \boxed{
  f_i \subseteq \mathcal{C}_{k_j}
  \quad \text{and} \quad
  f_j \subseteq \mathcal{C}_{k_i}.
  }
  \label{eq:xor_feasibility}
\end{equation}
Condition~\eqref{eq:xor_feasibility} ensures that the cache $k_j$ already holds all packets in $f_i$ and can therefore cancel them from $X$ to recover $f_j$, and symmetrically for the cache $k_i$.
A single broadcast thus simultaneously serves two users, thereby yielding $U_t = 2$ in the simplest case (or $U_t = |f_i \cup f_j|$ after chained merges).

\paragraph{Post-Merge Queue Update and Keep-Side Decision}

An XOR merge is a \emph{compound action}, i.e., beyond deciding \emph{which pair} $(r_i, r_j)$ to merge, the agent also selects a binary \emph{keep-side bit} $\kappa \in \{0, 1\}$, where $\kappa = 0$ retains slot $i$ and $\kappa = 1$ retains slot $j$. We use $\kappa$ exclusively for this binary indicator throughout the paper, while the semantic ``which slot is kept'' is recovered by the deterministic decode $i$ if $\kappa{=}0$ else $j$, and is never re-introduced as a separate symbol.
The merge produces a new \emph{aggregate request}
\begin{align}
  f_{\mathrm{mg}}   &= f_i \cup f_j,              \label{eq:mg_file} \\
  d_{\mathrm{mg}}   &= \min(d_i,\, d_j),           \label{eq:mg_deadline} \\
  \mathcal{S}_{\mathrm{mg}} &= \mathcal{S}_i \cap \mathcal{S}_j, \label{eq:mg_side} \\
  k_{\mathrm{mg}}   &\sim \mathrm{Unif}\{k_i, k_j\} \text{\ \ independently of $\kappa$}. \label{eq:mg_dest}
\end{align}

\smallskip\noindent
\emph{Status of the representative-destination rule
\eqref{eq:mg_dest}.}
Both original arrivals are already \emph{decoded} at the merge
broadcast (Sec.~``Broadcast Semantics of Aggregation''), so the
residual record's destination $k_{\mathrm{mg}}$ is a bookkeeping
representative rather than a remaining physical demand: a future
broadcast that includes the residual record is feasible iff partner
caches hold $f_{\mathrm{mg}}$, and the representative only enters
through the symmetric one-gap invariant of
Appendix~\ref{sec:appendix:state_sufficiency}.
We separate two distinct claims about \eqref{eq:mg_dest}:
\emph{(a) Correctness requirement.} The chained-merge state
sufficiency proof (Proposition~\ref{prop:one_gap},
Appendix~\ref{sec:appendix:state_sufficiency}) only requires
$k_{\mathrm{mg}} \in \{k_i, k_j\}$; uniformity and independence from
$\kappa$ are \emph{not} used in the proof, so a deterministic
keep-side-tied or earliest-deadline-tied rule would also satisfy the
one-gap invariant.
\emph{(b) Chosen simulator convention.} On top of the correctness requirement, we adopt the uniform sampling \emph{independent of} $\kappa$ as a deliberate design choice, so that the keep-side bit $\kappa$ is free to encode the queue-evolution control without also overloading the residual destination identity, while alternatives such as the keep-side-tied $k_{\mathrm{mg}}=k_{\kappa}$ or higher-degree / earliest-deadline representatives are plausible and would change the transition kernel, even though a focused sensitivity analysis is left for future work considering that every method we report here is evaluated under the \emph{same} convention, so the paired comparisons are not biased by~\eqref{eq:mg_dest}.
The queue slot selected by the keep-side bit $\kappa$ (slot $i$ if $\kappa = 0$; slot $j$ if $\kappa = 1$) is updated to hold the aggregate request, while the other slot is immediately replaced by a fresh request $r_{\mathrm{new}}$.
Considering that $d_{\mathrm{mg}}$ and $\mathcal{S}_{\mathrm{mg}}$ are fully determined by~\eqref{eq:mg_deadline}--\eqref{eq:mg_side}, the keep-side choice~$\kappa$ cannot alter the merged request's deadline, the side-information set, or the destination cache.
What $\kappa$ does control is the \emph{queue evolution}, i.e., it selects which slot is freed and immediately replenished by a fresh arrival~$r_{\mathrm{new}}$.
The identity of~$r_{\mathrm{new}}$, i.e., its file, deadline, and side-information providers, is drawn stochastically, thereby giving the agent indirect influence over the queue's future composition and mergeability.
Appendix~\ref{sec:appendix:keepside_ablation} provides an empirical ablation confirming that the learned keep-side decision contributes measurably to the performance beyond the pair selection alone.

\paragraph{Broadcast Semantics of Aggregation}

Each XOR merge constitutes an \emph{actual broadcast}, i.e., the server transmits $X = \bigoplus_{p\in f_i\cup f_j} p$ in one channel use, while every user whose individual request was folded into the merge decodes its requested packet immediately via its cached side information.
The metric $U_t = |f_{\mathrm{mg}}|$ therefore records the
\emph{packet-set-cardinality XOR degree} of the coded
packet, i.e., the number of \emph{distinct packet identities} carried in
the broadcast at that step (Section~\ref{sec:exp:metrics}). It is
\emph{not} the number of original singleton requests folded into the
aggregate: under the request generator (Sec.~\ref{sec:model:queue})
two pending records can share packet identities without contradicting
XOR feasibility, so $|f_{\mathrm{mg}}|$ can equal the number of
folded original arrivals only when no such overlaps occur (see the
non-gap-overlap remark in
Appendix~\ref{sec:appendix:coded_remarks}). Whenever original-arrival
semantics are intended we use the request-level metrics
$\eta_{\mathrm{req}}, m_{\mathrm{req}}, \sigma_{\mathrm{req}}$ of
Section~\ref{sec:exp:metrics} instead.
The aggregate request $r_{\mathrm{mg}}$ that remains in the queue after
the broadcast is not unserved demand; it is a
\emph{coding-state record} that
(i)~tracks the growing packet union $f_{\mathrm{mg}}$ for future XOR
feasibility checks via~\eqref{eq:xor_feasibility}, and
(ii)~carries the tightened side-information set
$\mathcal{S}_{\mathrm{mg}} = \mathcal{S}_i \cap \mathcal{S}_j$.
A chained merge of $(r_{\mathrm{mg}}, r_z)$ is a \emph{second} channel
use: the server broadcasts $X' = \bigoplus_{p\in f_{\mathrm{mg}}\cup
f_z} p$.
The two \emph{active participants} ($k_{\mathrm{mg}}$ and $k_z$) can
always decode~$X'$
(Proposition~\ref{prop:one_gap},
Theorem~\ref{thm:state_sufficiency};
Appendix~\ref{sec:appendix:state_sufficiency});
previously served users in the chain may or may not be able to decode
$X'$, but need not do so as they were already served at their
respective merge steps
(Remark~\ref{rem:scope_decodability}).
Because each chained broadcast is a separate transmission, a user that participates in multiple links of the chain contributes to $U_t$ at every link, while this reflects the per-broadcast XOR degree, not the redundant delivery.
If the coding-state record expires before a subsequent merge or unicast resolves it, the expiration penalty $E_t = |f_{\mathrm{mg}}|$ captures the wasted time slot and the side-information investment, not the undelivered data (Section~\ref{sec:model:deadlines}).
Considering that $\mathcal{S}_{\mathrm{mg}} = \mathcal{S}_i \cap \mathcal{S}_j$, each successive merge demands that a prospective partner cache holds \emph{all} packets in the growing set $f_{\mathrm{mg}}$, which is a \emph{typically} harder condition, with strict worsening exactly when the next merge contributes a packet identity not already in $f_{\mathrm{mg}}$ (Appendix~\ref{sec:appendix:coded_remarks}).
The method of aggressive merging therefore narrows the pool of future merge partners, thereby producing a throughput--flexibility trade-off that the scheduling agent must learn to manage.
This construction is analogous to incrementally growing a clique in the side-information conflict graph, i.e., each new edge corresponds to one additional coded broadcast exploiting the enlarged cache overlap~\cite{asghari2019approximation}.

Two supplementary remarks on chained-merge scope and per-broadcast
accounting are provided in Appendix~\ref{sec:appendix:coded_remarks}.
A worked three-request merge example appears in
Appendix~\ref{sec:appendix:example}.

\paragraph{Unique-Demand Accounting}
A separate evaluation concern is whether the scheduler covers the distinct file identities at least once during the episode, independent of which cache requested the packet. We define a \emph{unique-demand}, i.e., the file-identity-coverage, accounting at the granularity of the file identities.
Each packet ID $p \in \{0,\ldots,F-1\}$ projects onto a file identity via the deterministic map
\begin{equation}
  \phi(p) \;=\; \bigl\lfloor p / B \bigr\rfloor \;\in\; \{0,\ldots,N-1\},
  \label{eq:phi_packet_to_file}
\end{equation}
so the file identity carried by a packet set $f_r$ is
$\phi(f_r) = \{\phi(p) : p \in f_r\} \subseteq \{0,\ldots,N-1\}$.
Let $\mathcal{F}_{\mathrm{served}}^{(t-1)} \subseteq \{0,\ldots,N-1\}$
be the set of file identities that have been delivered by any
transmission strictly before step~$t$
(initialized to $\varnothing$ at episode start). Define
\begin{align}
  U_t^{\mathrm{uniq}} &\;=\; \bigl| \phi(f_{\mathrm{tx}_t}) \setminus \mathcal{F}_{\mathrm{served}}^{(t-1)} \bigr|, \label{eq:Ut_uniq} \\
  E_t^{\mathrm{uniq}} &\;=\; \biggl| \biggl( \bigcup_{r \in \mathcal{X}_t} \phi(f_r) \biggr) \setminus \mathcal{F}_{\mathrm{served}}^{(H)} \biggr|, \label{eq:Et_uniq}
\end{align}
where $f_{\mathrm{tx}_t}$ is the packet set delivered at step~$t$
(either $f_{\mathrm{mg}}$ for a coded merge or $f_r$ for a unicast on
record~$r$), $\mathcal{X}_t$ is the expiration-event set of
Eq.~\eqref{eq:exp_set}, and $\mathcal{F}_{\mathrm{served}}^{(H)}$ is
the cumulative served-file-identity set at episode end (so an expired
record contributes only those file identities never delivered during
the episode).
Under this accounting, each distinct file identity is credited exactly once per episode, i.e., a coding-state record that expires after its constituent file identities were already delivered earlier in the episode incurs $E_t^{\mathrm{uniq}} = 0$.
The broadcast-efficiency view $(U_t, E_t)$ and the unique-demand view $(U_t^{\mathrm{uniq}}, E_t^{\mathrm{uniq}})$ are complementary, i.e., the former quantifies how effectively each channel use exploits the coded-multicast gain at the granularity of the broadcast-level packet-set XOR degree, while the latter isolates how many distinct file identities, under the projection $\phi$, are delivered before their deadlines.
Section~\ref{sec:exp:metrics} formalizes all three metric families, and all
results tables report all three views
(Sections~\ref{sec:results}--\ref{sec:ablation}).
The distinct file-identity coverage $\delta$ becomes a co-primary evaluation criterion alongside the broadcast-packet expiration ratio $\rho$ in Section~\ref{sec:results} (formal definitions in Section~\ref{sec:exp:metrics}); ``co-primary'' here means jointly headline within the symbol-retention convention of Sec.~\ref{sec:exp:metrics}, not a claim of per-arrival miss-probability or per-user completion semantics.

\paragraph{Request-Level Accounting}
\label{sec:model:request_accounting}
The unique-demand view tracks the \emph{file identities}, i.e., two requests targeting the same file from different caches are conflated as a single demand.
To distinguish the individual arrivals through the queue aggregation, we introduce a finer-grained \emph{request-level} accounting.

\emph{Status of $\mathcal{A}(r)$.} The annotation set $\mathcal{A}(r)$
introduced below is \emph{auxiliary evaluator bookkeeping} maintained
alongside the latent simulator state for the sole purpose of computing
the request-level metrics M8--M10 of Sec.~\ref{sec:exp:metrics}; it is
\emph{not} part of the control-relevant queue record
$r=(k_r,f_r,d_r,\mathcal{S}_r)$ used by the contextual-POMDP transition
kernel of Sec.~\ref{sec:env}, and it is never observed by the policy
$\pi_\theta(a\mid o)$. The control-relevant record remains the
four-tuple stated above; $\mathcal{A}(r)$ is carried through merges
(Eq.~\eqref{eq:req_merge} below) and refills only so that completion
and miss can be credited per original request identifier at episode end.

Each new arrival~$r$ is assigned a unique request identifier~$u$ and carries the singleton annotation set $\mathcal{A}(r) = \{u\}$, while when two queue entries merge (Section~\ref{sec:model:coded}), the resulting record inherits the union of both annotation sets:
\begin{equation}
  \mathcal{A}(r_{\mathrm{mg}}) = \mathcal{A}(r_i) \cup \mathcal{A}(r_j).
  \label{eq:req_merge}
\end{equation}
At each step~$t$, let $\mathcal{A}(\mathrm{tx}_t)$ denote the
annotation set of the transmitted entry (or the union for coded
transmissions involving both endpoints), and let
$\mathcal{A}(\mathrm{exp}_t)$ denote the union of annotation sets
over all entries that expire at step~$t$.
We define the \emph{newly completed} and \emph{newly missed} request
sets as
\begin{align}
  C_t &= \mathcal{A}(\mathrm{tx}_t) \setminus
         \bigl(D_{t-1} \cup L_{t-1}\bigr),
  \label{eq:req_Ct} \\
  M_t &= \mathcal{A}(\mathrm{exp}_t) \setminus
         \bigl(D_t \cup L_{t-1}\bigr),
  \label{eq:req_Mt}
\end{align}
with cumulative sets $D_t = D_{t-1} \cup C_t$ and
$L_t = L_{t-1} \cup M_t$ ($D_0 = L_0 = \varnothing$).
By construction, each request identifier is credited exactly once as completed, missed, or still pending at episode end, while Section~\ref{sec:exp:metrics} defines the derived metrics, considering that the $\lambda$-sensitivity of the composite request-level score is analyzed in Appendix~\ref{sec:appendix:req_lambda_sensitivity}.

\paragraph{Feasible Merge Set}

The complete set of feasible pairs at time $t$ is
\begin{equation}
  \mathcal{M}_t = \bigl\{(i,j) \;\big|\; 0 \le i < j < Q,\;
    f_i \subseteq \mathcal{C}_{k_j},\;
    f_j \subseteq \mathcal{C}_{k_i}\bigr\},
  \label{eq:merge_set}
\end{equation}
with cardinality $|\mathcal{M}_t| \leq P_{\max} = \binom{Q}{2} = 45$.
Setting $P_{\max} = \binom{Q}{2}$ ensures that no feasible pair is ever discarded; the merge set $\mathcal{M}_t$ is therefore exact.
We index $\mathcal{M}_t$ in \emph{lexicographic order on queue indices
with $i<j$}: the same enumeration rule is applied identically by the
environment, the pair-feature tensor (whose rows correspond to the
listed pairs in this order, then zero-padded to $P_{\max}$), and the
action decoder $(i_k,j_k) = \mathcal{M}_t[k]$
(Sec.~\ref{sec:model:problem}).
This rule is the order produced by
\texttt{itertools.combinations(range(Q), 2)} in our reference
implementation and guarantees that policy logits, action masks, and
pair features remain aligned across processes.
This set changes at every step as the deadlines decrement, new requests arrive, and the side-information sets evolve through the merges.
Constructing $\mathcal{M}_t$ from scratch requires $O(Q^2)$ pairwise feasibility checks, each verifying the two subset-inclusion conditions on the cache contents.
In practice, the environment maintains $\mathcal{M}_t$ incrementally, i.e., only the pairs involving a newly arrived or recently merged request are re-evaluated, thereby reducing the per-step cost to $O(Q)$ checks plus $O(|\mathcal{M}_t|)$ bookkeeping.
The feasibility predicate~\eqref{eq:merge_set} has \emph{no size cap} on $|f_i \cup f_j|$, while the constant $U_{\max}=6$ used in the observation feature tables of Appendix~\ref{sec:appendix:features} is purely an observation clip on the size feature, not a merge-feasibility cap.

\paragraph{Impact on Future Mergeability}

A direct consequence of~\eqref{eq:mg_side} is that the merging \emph{shrinks} the side-information set, i.e., $\mathcal{S}_{\mathrm{mg}} \subseteq \mathcal{S}_i$ and $\mathcal{S}_{\mathrm{mg}} \subseteq \mathcal{S}_j$, while the subsequent merges involving $r_{\mathrm{mg}}$ require a cache to hold \emph{all} packets in the now-larger set $f_{\mathrm{mg}}$, which is less likely than holding $f_i$ or $f_j$ individually.
The method of aggressive merging therefore progressively reduces the pool of future merge partners, which is a structural reason why the classical ``always merge'' strategy is suboptimal under deadlines.

\subsection{Deadline Dynamics and Expiration}
\label{sec:model:deadlines}

The time is slotted, and after every transmission, whether unicast or coded, all deadlines in the queue are decremented by one:
\begin{equation}
  d_r \leftarrow d_r - 1, \quad \forall\, r \in \text{queue}.
  \label{eq:deadline_dec}
\end{equation}
A request $r$ \emph{expires} if $d_r \leq 0$ before it is served, while the expiration is a \emph{hard constraint}, i.e., an expired request yields zero throughput benefit to any user.
Upon expiration, the slot is immediately refilled by a fresh request $r_{\mathrm{new}}$ so that the queue remains at depth $Q$.
An aggregate request, formed by merging $r_i$ and $r_j$, that expires wastes the side-information investment of the prior merge, thereby making the method of overly aggressive merging more costly under tight deadlines.

The aggregate-expiration accounting is detailed in the transition dynamics of Section~\ref{sec:env:transition}.

\medskip
The performance metrics for evaluating the scheduling policies are formally defined alongside the experimental protocol in Section~\ref{sec:exp:metrics}, while the \emph{opportunity rate} $o_{\text{rate}} = \bigl|\{t : |\mathcal{M}_t|>0\}\bigr|/H$ measures the fraction of steps at which at least one feasible merge pair exists and is used to contextualize the merge-rate values across the scenarios.

\subsection{Problem Statement}
\label{sec:model:problem}

We state the scheduling problem studied in this paper below.

\begin{definition}[Deadline-Constrained Coded Caching Scheduling]
\label{def:problem}
Given a network with $K$ edge caches, a library of $N$ files ($B$ subfiles
each), decentralized random-prefetching placement with cache probability $p_c$,
a queue of $Q$ outstanding requests each with a per-request hard deadline
$d_r \leq D$, and an episode horizon of $H$ slots, the reference scheduling
problem of interest is to find an online scheduling policy
$\pi\colon \mathcal{S} \to \mathcal{A}$ that maximizes the expected
broadcast-efficiency score:
\begin{equation}
  \pi^{\sigma} = \arg\max_{\pi}\;
  \mathbb{E}_{\pi}\!\left[\sigma\right],
  \label{eq:opt_problem}
\end{equation}
subject to the constraint that any request still unserved after $D$ time slots expires and yields zero throughput.
The action space studied in this paper is the restricted \emph{merge-or-defer with EDF unicast fallback} class formalized in (A4) of Sec.~\ref{sec:model:assumptions} and parameterized in Sec.~\ref{sec:env:action}, i.e., the policy~$\pi$ chooses among the coded-merge actions, with a binary keep-side bit, and a deterministic earliest-deadline-first unicast, rather than over the arbitrary per-request transmission targets. The optimization of~\eqref{eq:opt_problem} is therefore over this restricted policy class, not over the unrestricted online scheduling problem.
We do not optimize~\eqref{eq:opt_problem} directly, i.e., the agent is trained on a shaped surrogate $R_t = R_{\mathrm{base}} + R_{\mathrm{quality}} + R_{\mathrm{shape}}$ (Sec.~\ref{sec:env:reward}) chosen to improve the sample efficiency and the deadline-aware behavior, while the trained policy $\pi_{\theta}^{\mathrm{surr}}$ is then evaluated on the demand-centric and request-level metrics of Sec.~\ref{sec:exp:metrics}. Among the three reward components, the potential-based shaping $R_{\mathrm{shape}}$ is \emph{inspired by} the policy-invariance theorem of Ng et al.~\cite{ng1999shaping}, while that theorem holds exactly only for an idealized, fully observed, continuing-state MDP with terminal $\Phi(s_H)=0$ and time-to-go information. The implemented training environment is the contextual-POMDP surrogate of Definition~\ref{def:problem} (truncated at $H{=}50$, hidden episode cache placement, aliased aggregate observations, no time-to-go feature), so we do \emph{not} claim a formal invariance guarantee for the implemented setup, even though in this paper $R_{\mathrm{shape}}$ should be read as the heuristic credit-assignment shaping that empirically improves the sample efficiency without distorting the empirical headline rankings in our ablation. The merge-quality bonus $R_{\mathrm{quality}}$ does change the optimized objective with respect to $\sigma$, so $\pi^{\sigma}$ in~\eqref{eq:opt_problem} is a reference target rather than the exact attractor of training. We report this distinction here so that the empirical claims in Sec.~\ref{sec:results} are read as those of $\pi_{\theta}^{\mathrm{surr}}$ on the demand-centric and request-level metrics, not as those of an exact maximizer of $\sigma$.

\smallskip\noindent
\textbf{Formal status: contextual POMDP, not a strict MDP.}
We treat the agent as a stationary policy for a discounted continuing-control \emph{POMDP} truncated at the horizon $H = 50$ (with the discount $\gamma = 0.995$ in PPO/GAE), rather than as the optimal non-stationary policy of a strict finite-horizon MDP, while the ``masked-MDP'' language used elsewhere in the paper is shorthand for this contextual-POMDP surrogate, considering that the queue-only observation $o_t$ is \emph{not} a strict Markov state for at least three reasons.
\emph{(i) Hidden episode context.} The cache placement $\{\mathcal{C}_k\}_{k=0}^{K-1}$ is sampled once at the episode start (A3) and remains hidden in the sense that it is not directly fed to the policy, while the future arrival distributions depend on it via the rejection-sampling rule of Sec.~\ref{sec:model:queue}, so two episodes with the same observed queue features can induce different refill distributions even though the latent placements differ.
This makes the formal model a \emph{contextual MDP / POMDP} with the hidden episode context.
\emph{(ii) Aggregate-state partial observability.} For an aggregate record produced by the chained merges, Appendix~\ref{sec:appendix:state_sufficiency} proves the sufficiency of $(k_{\mathrm{mg}}, f_{\mathrm{mg}}, \mathcal{S}_{\mathrm{mg}})$ for the future feasibility and decodability, while the policy observation retains $k_{\mathrm{mg}}$ and $\mathcal{S}_{\mathrm{mg}}$ explicitly even though it exposes $f_{\mathrm{mg}}$ only through its size $\min(|f_{\mathrm{mg}}|,U_{\max})/U_{\max}$, so two observation-equivalent queue states can yield different merge union sizes $|f_i \cup f_j|$ when the records share non-gap packets, thereby resulting in different immediate $U_t$ rewards and next aggregate sizes (Appendix~\ref{sec:appendix:coded_remarks}). The packet identities are deliberately discarded for the scalability of the graph-attention encoder, not because they are theoretically redundant, while reducing the remaining bias is an open question we leave to future work.
\emph{(iii) Time-to-go.} $s_t$ omits a remaining-horizon feature, so the finite-horizon-optimal policy can in principle differ from the stationary fixed point the agent learns, while the chosen $\gamma = 0.995$ has effective discount horizon $1/(1-\gamma) = 200 > H$, thereby making the surrogate closer to a continuing average-return objective than to a finite-horizon objective.
All empirical comparisons in Sec.~\ref{sec:results} are paired across the methods at fixed $H$ on shared seeds and shared cache placements, so (i)--(iii) are absorbed into the within-benchmark comparison, while the paper's claims should be read as conclusions about the learned stationary policy under this contextual-POMDP surrogate, not as the optimality claims for a strict finite-horizon MDP. Exposing the remaining-horizon and packet-identity features, or reframing the agent input to include a sufficient summary of the cache placement, would close (i)--(iii) at the cost of a wider observation, and we flag this as future work in Sec.~\ref{sec:conclusion}.
\end{definition}

\begin{remark}[Training objective vs.\ evaluation metrics]
\label{rem:training_vs_eval}
The optimization in~\eqref{eq:opt_problem} maximizes the expected broadcast-efficiency score~$\sigma$, which matches the base training reward $R_{\mathrm{base}} = U_t - E_t$.
However, because $U_t$ is the packet-set-cardinality XOR degree of the broadcast, and grows with each chained merge, $\sigma$ is best interpreted as a \emph{broadcast-efficiency} metric measuring the per-slot packet utilization rather than the distinct file-identity coverage.
The demand-centric evaluation metrics, i.e., the broadcast-packet expiration ratio~$\rho$ and the distinct file-identity coverage~$\delta$, are computed post-hoc from the same trajectories, even though they are not directly optimized.
We anchor our main empirical claims on these demand-centric metrics (Sections~\ref{sec:results}--\ref{sec:ablation}).
\end{remark}

Several features of Problem~\ref{def:problem} make it challenging:

\begin{itemize}[leftmargin=*, nosep]

  \item \textbf{Combinatorial action space.}
  The feasible merge set $\mathcal{M}_t$ changes at every step, while the agent must select a pair \emph{and} a keep side, thereby inducing up to $2P_{\max} + 1 = 91$ discrete actions of which only a state-dependent subset is valid at any given step.

  \item \textbf{Long-horizon mergeability consequences.}
  Each merge action simultaneously (i)~delivers one coded packet, (ii)~shrinks the side-information set for the aggregate request via $\mathcal{S}_{\mathrm{mg}} = \mathcal{S}_i \cap \mathcal{S}_j$, and (iii)~changes the future feasibility structure of the queue, while evaluating these cascading effects requires look-ahead beyond the next step.

  \item \textbf{Hard deadline and expiration dynamics.}
  The penalty for the expiration is discontinuous, i.e., a request delivers the full value if served one slot before its deadline and zero value if served one slot after, while the policies that optimize only the immediate throughput tend to defer the urgent requests in favor of the coding opportunities, thereby producing large expiration penalties.

  \item \textbf{Stochastic request arrivals.}
  New requests arrive with random file identities, deadlines, and side-information sets, so the future queue state is unknown at the decision time, while a useful policy must generalize across the distribution of the arriving requests.

\end{itemize}

These properties make the problem intractable for the exact dynamic programming at realistic scales and motivate the reinforcement learning formulation developed in the following sections.
Specifically, we model Problem~\ref{def:problem} as a \emph{masked discrete-action Markov Decision Process} (Section~\ref{sec:env}), in which the state encodes the full queue configuration, while the action space is dynamically masked to the feasible set $\mathcal{M}_t \cup \{\text{unicast}\}$, and the reward is shaped to match~\eqref{eq:opt_problem}.

\paragraph{Assumptions and Scope.}
To aid the reproducibility we summarize the modelling boundaries.
\emph{Fixed within each episode:} the network topology (one server, $K$ edge caches, shared broadcast link), the cache contents (drawn once via the uniform-without-replacement placement with fraction~$p_c$), the file library ($N$ files, $B$ equal-length subfiles each), and the file-popularity distribution (uniform for Track~A, while a separately trained agent handles the Zipf demand in Track~B, Section~\ref{sec:res:ood_zipf}).
\emph{Stochastic:} the request arrivals are drawn each time a queue slot is filled or refilled by (i)~sampling the file index from the within-track popularity law (i.e., uniform on Track~A, while Zipf with $\alpha{=}0.8$ or Mandelbrot--Zipf on Track~B), (ii)~sampling the packet index uniformly within the file, (iii)~rejecting any sampled (file, packet) that is held by every cache, (iv)~drawing the destination cache uniformly from the \emph{admissible} subset of caches that do not already store the sampled packet, and (v)~drawing the deadline $d_r \sim \mathrm{Uniform}\{1,\ldots,D\}$, while ``Request Generation'' above provides the full conditional procedure.
\emph{Not modelled:} the physical-layer errors (i.e., the error-free broadcast is assumed), the time-varying link rates (i.e., all packets are equal-length), the queueing or propagation delay beyond the discrete deadline mechanism, the cache-size heterogeneity, and the wireless fading.


\subsection{Contextual-POMDP Formulation}\label{sec:env}

We cast the deadline-constrained coded caching scheduling problem (Definition~\ref{def:problem}) as a \emph{masked discrete-action contextual POMDP} that can be optimized by a deep RL agent, while we use ``masked discrete-action MDP'' interchangeably below as a shorthand for this contextual-POMDP surrogate, considering that the formal-status caveats of Definition~\ref{def:problem}, i.e., the hidden episode cache context, the aggregate-state partial observability, and the no time-to-go feature, apply throughout.
The environment is implemented as a Gymnasium-compatible class~\cite{towers2024gymnasium} and is trained using $32$ parallel environment workers.

\subsubsection{Contextual-POMDP Surrogate}
\label{sec:env:mdp}

The implemented training environment is the
contextual-POMDP surrogate of Definition~\ref{def:problem}, with
latent-state space $\mathcal{S}$, observation space $\mathcal{O}$,
and observation map $o_t \in \mathcal{O}$ defined in
Sec.~\ref{sec:env:state}. We summarize it by the tuple
\begin{equation}
  \mathcal{M} = \bigl(\mathcal{S},\, \mathcal{O},\, \mathcal{A},\, R,\, P,\, \gamma\bigr),
\end{equation}
where $\mathcal{A}$ is the discrete action space,
$R\colon \mathcal{S} \times \mathcal{A} \times \mathcal{S} \to \mathbb{R}$
is the \emph{latent-state} reward kernel
$R(s_t, a_t, s_{t+1})$ used by the simulator (it depends on the
realized expiration-event set $\mathcal{X}_t$ and the merged-record
union $|f_{\mathrm{mg}}|$, which are functions of the latent queue and
not of the aliased observation alone; see
Sec.~\ref{sec:env:reward}), $P$ is the latent-state transition kernel,
and $\gamma = 0.995$ is the discount factor. The agent receives the scalar reward $R_t$ alongside the observation stream during the training, i.e., its policy $\pi_\theta(a \mid o)$ conditions on the observation only, while $R_t$ is computed from the latent state by the environment and delivered as part of the standard Gymnasium step tuple. We do not treat $R$ as a function of the observation, considering that the observation-equivalent aggregate states can yield different immediate $U_t$ rewards and different next aggregate sizes (Definition~\ref{def:problem}, point~(ii); also discussed in Appendix~\ref{sec:appendix:coded_remarks}). Whenever the ``masked-MDP'' shorthand is used elsewhere in the paper, it refers to this contextual-POMDP surrogate, not to a strict fully observed MDP.
Each episode runs for a fixed horizon of $H = 50$ transmission steps, and no early termination is applied.

\subsubsection{State Space and Observation Encoding}
\label{sec:env:state}

The latent state $s_t$ at step $t$ is the complete queue configuration $\{r_0, r_1, \ldots, r_{Q-1}\}$ together with the fixed cache assignment $\{\mathcal{C}_0, \ldots, \mathcal{C}_{K-1}\}$ for the current episode, while the agent's observation $o_t$ encodes only the queue, thereby projecting the cache assignment through the per-request side-information sets $\mathcal{S}_r$ and through the aggregate side-information set $\mathcal{S}_{\mathrm{mg}}$ for any chained-merge record. This projection is \emph{lossy} in two ways that are spelled out in the formal-status block of Definition~\ref{def:problem}, i.e., the placement is treated as hidden episode context (the policy sees only its $\mathcal{S}_r$ summaries, not the cache contents themselves), while the aggregate packet set $f_{\mathrm{mg}}$ is summarized only by its size feature $\min(|f_{\mathrm{mg}}|,U_{\max})/U_{\max}$ rather than by the explicit packet identities. We discard the packet identities deliberately for the scalability of the graph-attention encoder, not because they are theoretically redundant, and consequently the model is best read as a contextual POMDP whose observation is a sufficient statistic for the feasibility but not in general for the immediate $U_t$ reward when the records share non-gap packets.

The observation $o_t$ is a \emph{structured dictionary} with two components
whose dimensions are summarized in Table~\ref{tab:obs_dims}.

\begin{table}[!t]
  \caption{Observation Space Components ($K=5$, $Q=10$, $P_{\max}=45$).
    Track~B adds popularity-aware features (Section~\ref{sec:res:ood_zipf}).}
  \label{tab:obs_dims}
  \centering
  \begin{adjustbox}{max width=\columnwidth}
  \begin{tabular}{@{}lccc@{}}
    \toprule
    \textbf{Component} & & \textbf{Track~A} & \textbf{Track~B} \\
    \midrule
    Per-request feature vector & --- & $(130,)$ & $(140,)$ \\
    Pairwise feature tensor & --- & $(45, 8)$ & $(45, 11)$ \\
    \bottomrule
  \end{tabular}
  \end{adjustbox}
\end{table}
\FloatBarrier

The per-request feature vector $\phi_i \in [0,1]^{d_{\text{req}}}$ encodes the target cache identity (one-hot), the side-information providers, the normalized deadline, the file-set size, and the merge degree ($d_{\text{req}}{=}13$ for Track~A, $14$ for Track~B with an additional popularity-mass feature), while the per-pair feature vector $\psi_{ij}$ encodes the intersection size, the partner degrees, the combined urgency, and the aggregate sizes ($d_{\text{pair}}{=}8$ for Track~A, $11$ for Track~B). All features are bounded in $[0,1]$, while the full feature definitions are provided in Appendix~\ref{sec:appendix:features}.

\subsubsection{Action Space and Keep-Side Parameterization}
\label{sec:env:action}

The action space is $\mathcal{A} = \{0, 1, \ldots, 2P_{\max}\}$, i.e., a discrete set of $2P_{\max} + 1 = 91$ actions, while each action encodes both the \emph{pair index} and the \emph{keep-side decision} via the bijection
\begin{equation}
  a \;\longmapsto\; \bigl(k,\; \kappa\bigr)
  \;=\; \bigl(\lfloor a/2 \rfloor,\; a \bmod 2\bigr),
  \label{eq:action_decode}
\end{equation}
where $k \in \{0,\ldots,P_{\max}{-}1\}$ indexes the candidate pair
$(i_k, j_k) = \mathcal{M}_t[k]$ and $\kappa \in \{0, 1\}$ selects the kept
slot ($\kappa = 0$ retains slot $i_k$; $\kappa = 1$ retains slot $j_k$).
The special action $a = 2P_{\max} = 90$ triggers a unicast transmission.
The three action categories are summarized below:
\begin{itemize}[leftmargin=*, nosep]
  \item $a \in \{0, 2, 4, \ldots, 2P_{\max}{-}2\}$ (even): coded merge,
        keep slot $i_k$.
  \item $a \in \{1, 3, 5, \ldots, 2P_{\max}{-}1\}$ (odd): coded merge,
        keep slot $j_k$.
  \item $a = 2P_{\max}$: unicast, serve earliest-deadline request.
\end{itemize}

The role of the keep-side parameter in governing the queue evolution, exposed here as a learned action dimension rather than as the hand-designed endpoint-retention rules used by SACM+/SACM++, is discussed in Appendix~\ref{sec:appendix:coded_remarks}.

\subsubsection{Transition Dynamics}
\label{sec:env:transition}

Given the state $s_t$ and action $a_t$, the transition to $s_{t+1}$ proceeds in three sequential phases.

\smallskip\noindent
\textbf{One-step timeline.}
At step~$t$ the agent observes the queue and the feasible merge set~$\mathcal{M}_t$, selects an action, i.e., the coded merge or the unicast, while the environment executes the transmission and updates the affected queue slots (Phase~1), decrements all remaining deadlines by one (Phase~2), then expires and refills any request whose deadline has reached zero (Phase~3), thereby yielding the latent state~$s_{t+1}$. The observation $o_t$ does \emph{not} carry a time-to-go feature, and consistent with the truncated continuing-control framing in Definition~\ref{def:problem}, the agent is trained as a stationary policy on the contextual-POMDP observation defined there.

\smallskip\noindent
\textbf{Expiration event set.}
Let $\widetilde{\mathrm{queue}}_t$ denote the queue \emph{after} the Phase-1 transmission update and the Phase-2 deadline decrement, even though \emph{before} any Phase-3 refill, while the \emph{expiration-event set} at step~$t$ is
\begin{equation}
  \mathcal{X}_t \;=\; \bigl\{ r \in \widetilde{\mathrm{queue}}_t :\, d_r \le 0 \bigr\},
  \label{eq:exp_set}
\end{equation}
i.e., the records whose deadline reached or fell below zero during this step. All expiration-derived counters in this paper, i.e., $E_t$, the request-level miss set $M_t$, and the file-identity expiration set $E_t^{\mathrm{uniq}}$, are defined in terms of $\mathcal{X}_t$, never in terms of the post-refill queue~$s_{t+1}$ in which the corresponding slots have already been replaced by fresh arrivals.

\smallskip\noindent
\textbf{Refill timing convention (cross-reference to A5).}
The asymmetric refill ordering (Phase-1 refills before the
Phase-2 decrement, Phase-3 refills after) is recorded as
modeling assumption~(A5) above
(Sec.~\ref{sec:model:assumptions}). The two refill paths occupy
different positions in the one-step timeline, and we adopt this
asymmetry deliberately rather than patching it: a slot vacated by the
Phase-1 transmission is refilled \emph{before} the Phase-2 global
decrement, so the fresh request sampled with $d_r \sim
\mathrm{Uniform}\{1,\ldots,D\}$ is decremented on this step and enters
the next decision step with effective remaining deadline in
$\{0,\ldots,D-1\}$; a slot vacated by Phase-3 expiration is refilled
\emph{after} the Phase-2 decrement and so retains its sampled
remaining deadline in $\{1,\ldots,D\}$. As a consequence, a Phase-1
fresh request that draws $d_r=1$ joins $\mathcal{X}_t$ and expires in
the same step without ever being selectable; we treat this as the
model's representation of an arrival whose initial slack already fits
inside one transmission slot, and report it under the standard
expiration counters. Any reader re-implementing the simulator must
place the Phase-1 refill before
the Phase-2 decrement and the Phase-3 refill after, exactly as in
Algorithm~\ref{alg:one_step}.

The three phases are formalized in Algorithm~\ref{alg:one_step}.

\begin{algorithm}[h]
\caption{One-Step Environment Transition}
\label{alg:one_step}
\begin{algorithmic}[1]
\Require State $s_t = \{r_0,\ldots,r_{Q-1}\}$, action $a_t$, mask $\mathbf{m}_t$
\Statex \textit{// Phase 1 --- Transmission (Phase-1 refills happen \textbf{before} the Phase-2 decrement, so a fresh request sampled with $d_r=1$ here will expire in Phase~3 of the same step; see ``Refill timing convention'' above.)}
\If{$a_t \neq 2P_{\max}$ \textbf{and} $\lfloor a_t/2\rfloor < |\mathcal{M}_t|$}
    \State Decode $(k,\kappa) \leftarrow (\lfloor a_t/2\rfloor,\; a_t \bmod 2)$
    \State $(i_k,j_k) \leftarrow \mathcal{M}_t[k]$;\; compute $(f_\mathrm{mg}, d_\mathrm{mg}, \mathcal{S}_\mathrm{mg})$ via Eqs.~\eqref{eq:mg_file}--\eqref{eq:mg_side}
    \State Sample $k_\mathrm{mg} \sim \mathrm{Unif}\{k_{i_k}, k_{j_k}\}$ \emph{independently of} $\kappa$ \Comment{Eq.~\eqref{eq:mg_dest}; correctness requires only $k_{\mathrm{mg}}\in\{k_{i_k},k_{j_k}\}$ for Prop.~\ref{prop:one_gap}, uniformity is a chosen simulator convention}
    \State Form aggregate queue record $r_\mathrm{mg} \leftarrow (k_\mathrm{mg}, f_\mathrm{mg}, d_\mathrm{mg}, \mathcal{S}_\mathrm{mg})$ \Comment{four-tuple, matches the canonical control-relevant record}
    \State Set $\mathcal{A}(r_\mathrm{mg}) \leftarrow \mathcal{A}(r_{i_k}) \cup \mathcal{A}(r_{j_k})$ \Comment{auxiliary evaluator bookkeeping (Sec.~\ref{sec:model:request_accounting}); not part of the queue tuple}
    \State \textbf{if} $\kappa = 0$ \textbf{then} slot $i_k \leftarrow r_\mathrm{mg}$; refill slot $j_k$ with $r_\mathrm{new}$
    \State \textbf{else} slot $j_k \leftarrow r_\mathrm{mg}$; refill slot $i_k$ with $r_\mathrm{new}$ \Comment{$\kappa$ is the keep-side bit, not a slot index}
\Else \Comment{unicast}
    \State Serve $r^* = \arg\min_{r}\bigl(d_r,\, \mathrm{queue\text{-}index}(r)\bigr)$; refill its slot with $r_\mathrm{new}$ \Comment{ties on $d_r$ broken by smallest queue index, matching the reference simulator}
\EndIf
\Statex \textit{// Phase 2 --- Deadline Decrement}
\State $d_r \leftarrow d_r - 1$ for all $r$ in queue
\Statex \textit{// Phase 3 --- Expiration Handling}
\State $\mathcal{X}_t \leftarrow \{r \in \widetilde{\mathrm{queue}}_t : d_r \le 0\}$ \Comment{event set, Eq.~\eqref{eq:exp_set}}
\State $E_t \leftarrow \sum_{r \in \mathcal{X}_t} |f_r|$ \Comment{expired packet-identity count}
\For{each $r \in \mathcal{X}_t$}
    \State Refill slot with $r_\mathrm{new}$
\EndFor
\State Compute $\mathbf{m}_{t+1}$, observation $o_{t+1}$, reward $R_t$
\State \Return $(o_{t+1}, \mathbf{m}_{t+1}, R_t)$
\end{algorithmic}
\end{algorithm}

\subsubsection{Reward Function}
\label{sec:env:reward}

The per-step reward $R(s_t, a_t)$ decomposes additively into the three components:

\begin{equation}
  R(s_t, a_t) = R_{\mathrm{base}} + R_{\mathrm{quality}}
                + R_{\mathrm{shape}}.
  \label{eq:reward_total}
\end{equation}

\paragraph{R1: System-Aligned Base Reward}

\begin{equation}
  R_{\mathrm{base}} = w_{\mathrm{served}} \cdot U_t
                     - w_{\mathrm{exp}} \cdot E_t,
  \label{eq:reward_base}
\end{equation}
with $w_{\mathrm{served}} = 1.0$ and $w_{\mathrm{exp}} = 1.0$, while $U_t = |f_{\mathrm{mg}}|$ is the XOR degree of a coded merge (or $U_t = 1$ for a unicast transmission), and $E_t = \sum_{r \in \mathcal{X}_t} |f_r|$ is the total expired packet-set mass at step~$t$, computed over the \emph{expiration-event set} $\mathcal{X}_t$ defined in Eq.~\eqref{eq:exp_set} (i.e., the records whose deadline reached zero after the Phase-2 decrement, evaluated before the Phase-3 refill replaces those slots with fresh arrivals).
The per-broadcast packet-set formulation matches the training reward with the broadcast-efficiency score $\sigma = H^{-1}\sum_t(U_t - \lambda E_t)$ (Section~\ref{sec:exp:metrics}), so that the higher reward correlates with the higher broadcast-level packet utilization, while Remark~\ref{rem:training_vs_eval} discusses the distinction between this training-aligned broadcast-level metric and the file-identity / request-level evaluation metrics.

\paragraph{R2: Quality Bonus (Coded Merge Only)}

\begin{equation}
  R_{\mathrm{quality}} =
  \begin{cases}
    \begin{aligned}[t]
    &w_{\mathrm{inter}} \cdot |\mathcal{S}_{\mathrm{mg}}| \\
    &- w_{\mathrm{union}} \cdot \max\!\bigl(0,\, |f_{\mathrm{mg}}| - 2\bigr)
    \end{aligned}
    & \text{coded merge,} \\[8pt]
    0 & \text{unicast,}
  \end{cases}
  \label{eq:reward_quality}
\end{equation}
with $w_{\mathrm{inter}} = 0.75$ and $w_{\mathrm{union}} = 0.15$, while the \emph{intersection bonus} $w_{\mathrm{inter}} \cdot |\mathcal{S}_{\mathrm{mg}}|$ rewards the pairs whose merged side-information set is large, considering that a larger $\mathcal{S}_{\mathrm{mg}}$ means the aggregate request can participate in future XOR merges with more potential partners.
The \emph{union penalty} $w_{\mathrm{union}} \cdot \max(0, |f_{\mathrm{mg}}| - 2)$ discourages repeatedly merging the already-merged aggregates, i.e., as $|f_{\mathrm{mg}}|$ grows, the XOR feasibility condition~\eqref{eq:xor_feasibility} becomes harder to satisfy (a partner must cache the entire union), thereby making the deeply chained aggregates increasingly difficult to serve before their deadline.

\paragraph{R3: Potential-Based Shaping}

\begin{equation}
  R_{\mathrm{shape}} = w_\Phi \cdot
  \bigl(\gamma_\Phi \cdot \Phi(s_{t+1}) - \Phi(s_t)\bigr),
  \label{eq:reward_shape}
\end{equation}
where the potential function is
\begin{equation}
  \Phi(s) = \frac{|\mathcal{M}(s)|}{\binom{Q}{2}},
  \label{eq:potential}
\end{equation}
i.e., the fraction of all $\binom{Q}{2}$ possible queue pairs that are currently the feasible merges, while we use $w_\Phi = 0.20$ and $\gamma_\Phi = 0.995$.

This shaping is \emph{inspired by} the potential-based shaping~\cite{ng1999shaping}, whose formal invariance theorem holds for an idealized fully observed continuing-state MDP, while in the implemented contextual-POMDP surrogate of Definition~\ref{def:problem} (truncated at $H{=}50$, hidden episode cache placement, aliased aggregate observations, no time-to-go feature), the theorem does not apply directly, and we therefore present $R_{\mathrm{shape}}$ as the heuristic credit-assignment guidance and do not claim that it leaves the optimal policy of the surrogate exactly unchanged. The empirical effect of $R_{\mathrm{shape}}$ is examined in the reward ablation in Sec.~\ref{sec:ablation}, while an exact finite-horizon-MDP invariance result would additionally require a terminal $\Phi(s_H)=0$, a time-to-go feature, and an observation exposing the hidden episode context and the packet identities.

\begin{remark}[Reward design and selective merging]
\label{rem:reward_shaping}
The potential-based term $R_{\mathrm{shape}}$ is inspired by~\cite{ng1999shaping}, even though the formal invariance theorem applies only to an idealized fully observed continuing-state MDP, and we therefore treat $R_{\mathrm{shape}}$ as the heuristic shaping for the implemented contextual-POMDP and make no formal optimal-policy preservation claim.
The quality bonus $R_{\mathrm{quality}}$, specifically the intersection bonus and the union penalty, is the heuristic shaping that explicitly incentivizes the high-overlap merges and discourages the deep chaining, thereby giving rise to the selective merge behavior reported in Section~\ref{sec:results}, which is learned under this specific reward design, not in a reward-free setting.
\end{remark}

\begin{remark}[Training reward vs.\ evaluation metrics]
The shaped reward $R(s_t, a_t)$ defined above is used \emph{exclusively during training} to guide the policy optimization, while all performance comparisons in Sections~\ref{sec:results}--\ref{sec:ablation} are based on the communications-system metrics defined in Section~\ref{sec:exp:metrics}, i.e., the broadcast-efficiency score, the broadcast-packet expiration ratio $\rho$, the distinct file-identity coverage $\delta$, the served/tx, the coding gain, and the expirations, and the per-step reward is reported only as a training diagnostic and is never used for the policy ranking.
\end{remark}

\subsubsection{Dynamic Action Masking}
\label{sec:env:masking}

At every step $t$, a boolean mask $\mathbf{m}_t \in \{0, 1\}^{2P_{\max}+1}$ is computed alongside the observation as follows:
\begin{equation}
  m_t(a) =
  \begin{cases}
    1 & \text{if } \lfloor a/2 \rfloor < |\mathcal{M}_t|
        \text{ (pair } \lfloor a/2 \rfloor \text{ exists)}, \\
    1 & \text{if } a = 2P_{\max} \text{ (unicast, always valid)}, \\
    0 & \text{otherwise.}
  \end{cases}
  \label{eq:mask}
\end{equation}
The mask is exposed to the learning algorithm via the Gymnasium \texttt{action\_masks()} interface, while MaskablePPO~\cite{schulman2017ppo,sb3maskableppo} applies it at the two points as follows:

\begin{enumerate}[leftmargin=*, nosep]
  \item \textbf{Sampling:} the infeasible actions are assigned zero probability by replacing their logits with $-\infty$ before the softmax, so they are never selected during the rollout collection.
  \item \textbf{Gradient update:} the infeasible actions are excluded from the policy-gradient loss computation, thereby preventing the optimizer from assigning the probability mass to them.
\end{enumerate}

Huang and Onta\~{n}\'{o}n~\cite{huang2020invalidmask} show that this masking introduces no gradient bias, while the details are provided in Appendix~\ref{sec:appendix:masking_details}.

\subsubsection{Episode Termination}
\label{sec:env:termination}

An episode terminates deterministically after $H = 50$ transmission steps, and no early termination is applied.
At termination, the environment returns the episode-level aggregate counters from which the evaluation metrics of Section~\ref{sec:exp:metrics} are derived, while the full evaluation protocol, i.e., the holdout seeds, the episode seeding, and the confidence intervals, is described in Section~\ref{sec:exp:setup}.

\section{Learning Framework}
\label{sec:method}

This section describes the learning system we developed for the
deadline-constrained coded caching delivery, i.e., the delivery of the
coded multi-casting messages under the strict deadlines of the users.
The system has three components, while each of them addresses a
distinct aspect of the problem: a structured policy
network that uses the graph topology of the pending request queue
(Section~\ref{sec:method:arch}), a three-phase training pipeline that
bootstraps from the heuristics and then refines through the
self-improved online experience
(Sections~\ref{sec:method:bc}--\ref{sec:method:exit}),
and a conservative model-selection criterion that reduces the sensitivity
to the training-seed variance~\cite{henderson2018deep}
(Section~\ref{sec:method:selection}), thereby giving us a more reliable
final policy.

\subsection{Policy Architecture}
\label{sec:method:arch}

The policy must map a variable-cardinality request queue (whose merge
structure changes at every step) to a probability distribution over the
$91$ discrete actions, while at the same time respecting the
feasibility constraints of the merge graph.
A flat multilayer perceptron (MLP) would require padding to a worst-case
size and could not share the computation across the merge candidates,
even though the merge candidates share most of the structural
information. We therefore use a graph-structured policy network, i.e.,
a two-stage architecture that treats the queue as a node-attributed graph
and each candidate merge as an attributed edge.

We chose the graph attention~\cite{velickovic2018gat} over a flat MLP
for two reasons:
(i)~the merge topology $\mathcal{M}_t$ changes at every step, so we need
a variable-structure encoder; and
(ii)~the merge decisions are coupled across the pairs, since the value
of merging $(r_i, r_j)$ depends on its effect on the other pairs'
feasibility. The ablation in Section~\ref{sec:ablation} confirms that
replacing the graph attention with a flat MLP degrades the miss ratio
across all the eight evaluation regimes, i.e., the regimes used in our
evaluation.

\subsubsection*{Stage 1: Graph-Pair Encoder}

The extractor receives the structured observation from
Section~\ref{sec:env:state}, i.e., a per-request feature matrix
$\mathbf{X} \in \mathbb{R}^{Q \times d_{\mathrm{req}}}$
($d_{\mathrm{req}} = 13$ for Track~A; $14$ for Track~B)
and a per-pair feature matrix
$\mathbf{E} \in \mathbb{R}^{P_{\max} \times d_{\mathrm{pair}}}$
($d_{\mathrm{pair}} = 8$ for Track~A; $11$ for Track~B;
zero-padded to $P_{\max}=45$).

The extractor processes the observation in four stages:
(i)~a \emph{NodeMLP} ($d_{\mathrm{req}} \to 256 \to 128$) lifts each per-request vector to a 128-d embedding;
(ii)~two \emph{GraphAttentionBlocks}~\cite{velickovic2018gat} with four heads and $d_{\mathrm{model}}{=}128$ propagate information along feasible merge edges to produce contextual node embeddings;
(iii)~a \emph{ContextNet} mean-pools all node embeddings through an MLP ($128 \to 256 \to 128$) to produce a global queue context vector $\mathbf{c}$; and
(iv)~an \emph{EdgeNet} forms each pair embedding $\mathbf{e}_{ij} \in \mathbb{R}^{64}$ by concatenating the two node embeddings, their element-wise product, the \emph{absolute} pairwise difference $|\mathbf{h}_i - \mathbf{h}_j|$, and the \emph{base} pair features (the first $d_{\mathrm{pair}} - 2$ entries; the last two pair features encode normalized queue indices and are used only to build the merge-graph adjacency, not as EdgeNet inputs), then maps through an MLP ($(4{\times}128 + d_{\mathrm{pair}} - 2) \to 256 \to 64$, i.e., input width $518$ for Track~A and $521$ for Track~B).
Full equations and inference cost analysis are in Appendix~\ref{sec:appendix:arch_details}.

\subsubsection*{Stage 2: Pairwise Action Head}

Given the context vector $\mathbf{c}$ and all pair embeddings
$\{\mathbf{e}_{ij}\}$, the action head scores the 91-dimensional action
space.
For each of the (up to) 45 candidate pairs, two logits are produced by a
shared MLP ($[\mathbf{c} \| \mathbf{e}_{ij}] \to 128 \to 2$), corresponding
to keep-side~$i$ and keep-side~$j$, respectively.
The unicast action receives a separate logit from an MLP
($\mathbf{c} \to 64 \to 1$).
The $91$ raw logits are concatenated, the infeasible entries are set to
$-\infty$ via the dynamic action mask (Section~\ref{sec:env:masking}),
and a softmax produces the action distribution
$\pi_\theta(a \mid s)$ (here and below, $s$ is the shorthand for the
agent observation $o_t$ of Definition~\ref{def:problem}; the underlying
environment is a contextual POMDP and we never condition the policy on
the latent state directly), thereby producing a feasible action
distribution at every step.

The \emph{value head} uses a separate encoder from the actor, i.e., it
runs an independent graph-pair encoder of identical architecture (with
its own weights), \emph{concatenates} the global queue context vector
$\mathbf{c} \in \mathbb{R}^{128}$ with the $P_{\max}=45$ per-pair
embeddings $\mathbf{e}_{ij} \in \mathbb{R}^{64}$ to form the
$128 + 45{\cdot}64 = 3008$-dimensional value latent
$\mathbf{z}_v = [\mathbf{c}\,\|\,\mathbf{e}_1\,\|\,\cdots\,\|\,\mathbf{e}_{45}]$,
and maps it through a critic MLP ($3008 \to 256 \to 128 \to 1$),
exactly matching Table~\ref{tab:model_arch} and
Fig.~\ref{fig:policy_arch}. This separation lets the critic learn its
own task-relevant features, independent of the actor's action-scoring
head, thereby reducing the interference between the actor and the
critic.

Table~\ref{tab:model_arch} and Fig.~\ref{fig:policy_arch} summarize the
parameter count and the dataflow by module. The total is
${\approx}1.73$M parameters for the Track~A (uniform) agent and
${\approx}1.75$M for the Track~B (Zipf) agent (i.e., the larger input
feature dimensions account for the difference), thereby giving enough
representational capacity without overfitting on the
$10^4$-episode evaluation regime.

\begin{table}[t]
\centering
\caption{Policy network architecture. $Q{=}10$ requests, $K{=}5$ caches, $P_{\max}{=}45$ candidate pairs.}
\label{tab:model_arch}
\begin{adjustbox}{max width=\columnwidth}
\begin{tabular}{l l r}
\toprule
\textbf{Component} & \textbf{Architecture} & \textbf{Params} \\
\midrule
Node MLP & $13 \to 256 \to 128$ (ReLU) & 36.5K \\
Graph Self-Attention & $L{=}2$ layers, $h{=}4$ heads, $d{=}128$ & 198K \\
Context MLP & $128 \to 256 \to 128$ (ReLU) & 66K \\
Edge MLP & $518 \to 256 \to 64$ (ReLU) & 149K \\
\midrule
Actor (pairwise action head) & per-pair: $192 \to 128 \to 2$ & 25K \\
Critic graph-pair encoder & (same architecture, separate weights) & 449.5K \\
Critic value MLP & $3008 \to 256 \to 128 \to 1$ & 803K \\
\midrule
\textbf{Total} & & \textbf{$\sim$1.73M} \\
\bottomrule
\end{tabular}
\end{adjustbox}
\end{table}

The Track~B (Zipf) agent uses slightly larger input dimensions
($d_{\text{req}}{=}14$, $d_{\text{pair}}{=}11$), thereby yielding
${\approx}1.75$M parameters, while all the other architectural choices
are identical (Appendix~\ref{sec:appendix:arch_details}).

\begin{figure*}[t]
  \centering
  \resizebox{\textwidth}{!}{
\begin{tikzpicture}[
    >=Stealth,
    node distance=0.4cm,
    block/.style={rectangle, draw=black!70, rounded corners=3pt, minimum height=0.62cm, font=\scriptsize, align=center},
    input/.style={block, fill=gray!12, minimum width=1.9cm},
    stage1/.style={block, fill=blue!10},
    stage2/.style={block, fill=orange!10},
    output/.style={block, fill=green!10},
    note/.style={rectangle, draw=black!45, rounded corners=2pt, fill=white, font=\tiny, align=left},
    arr/.style={->, thick, black!70},
    darr/.style={->, thick, dashed, black!55},
    dim/.style={font=\tiny, text=black!55},
    label/.style={font=\scriptsize\bfseries, text=black!80},
  ]

  \node[input, minimum width=2.1cm] (feat) at (0, 1.5) {Per-request features\\$\mathbf{X}\in\mathbb{R}^{10\times 13}$};
  \node[input, minimum width=2.1cm] (pair) at (0, -0.9) {Pair features\\$\mathbf{P}\in\mathbb{R}^{45\times 8}$};
  \node[label, anchor=south] at (0, 2.45) {Observation dict};

  \node[stage1, minimum width=1.9cm] (nodemlp) at (2.8, 1.5) {Node MLP\\$13\to 256\to 128$};
  \node[dim, above=1pt of nodemlp] {ReLU};
  \draw[arr] (feat) -- (nodemlp);

  \node[stage1, minimum width=2.1cm] (adj) at (2.8, -0.25) {Build adjacency\\from $(i_{\mathrm{norm}},j_{\mathrm{norm}})$\\+ valid-pair mask};
  \draw[arr] (pair) -- (adj);
  \node[dim, below=1pt of adj] {$\mathbf{A}\in\{0,1\}^{10\times 10}$ with self-loops};

  \node[stage1, fill=blue!18, minimum width=2.2cm] (gat) at (5.7, 1.0) {Masked self-attention\\$\times 2$ blocks\\heads$=4$, $d=128$};
  \draw[arr] (nodemlp) -- (gat) node[midway, above, dim] {$\mathbf{h}^{(0)}$};
  \draw[arr] (adj.east) |- ($(gat.south)+(0,-0.1)$);

  \node[stage1, fill=cyan!12, minimum width=1.25cm] (pool) at (8.6, 2.45) {Mean\\pool};
  \node[stage1, fill=cyan!15, minimum width=2.0cm] (ctxmlp) at (10.9, 2.45) {Context MLP\\$128\to 256\to 128$};
  \node[stage1, fill=cyan!20, minimum width=0.9cm] (ctx) at (13.0, 2.45) {$\mathbf{c}$};
  \draw[arr] (gat.north east) |- (pool);
  \draw[arr] (pool) -- (ctxmlp);
  \draw[arr] (ctxmlp) -- (ctx);
  \node[dim, above=1pt of ctx] {$\mathbb{R}^{128}$};

  \node[stage1, fill=blue!12, minimum width=2.5cm] (gather) at (8.7, -0.35) {Gather $(\mathbf{h}_i,\mathbf{h}_j)$ and concat\\$[\mathbf{h}_i,\mathbf{h}_j,\mathbf{h}_i\odot\mathbf{h}_j,$\\$|\mathbf{h}_i{-}\mathbf{h}_j|,\psi_{ij}]$};
  \node[stage1, fill=blue!15, minimum width=2.0cm] (edgemlp) at (11.4, -0.35) {Edge MLP\\$518\to 256\to 64$};
  \node[stage1, fill=blue!20, minimum width=0.95cm] (eij) at (13.4, -0.35) {$\mathbf{e}_{ij}$};
  \draw[arr] (gat.south east) |- (gather.west);
  \draw[arr] (pair.east) -| ($(gather.south west)+(0.15,0)$);
  \draw[arr] (gather) -- (edgemlp);
  \draw[arr] (edgemlp) -- (eij);
  \node[dim, below=1pt of eij] {$\mathbb{R}^{64}$, padded to 45 pairs};

  \node[note, anchor=north west] (pfnote) at (7.2, -1.45) {$\psi_{ij}$ uses only the first 6 pair features:\\inter, deg$_i$, deg$_j$, min-deadline, size$_i$, size$_j$.\\The last 2 entries $(i_{\mathrm{norm}},j_{\mathrm{norm}})$ are used to build $\mathbf{A}$.};

  \node[note, anchor=north west] (latnote) at (7.2, -2.38) {Actor latent: $\mathbf{z}_{\pi}=[\mathbf{c}\|\mathbf{e}_1\|\cdots\|\mathbf{e}_{45}]$\\$\in\mathbb{R}^{128+45\cdot 64}=\mathbb{R}^{3008}$.};

  \begin{scope}[on background layer]
    \node[fit=(nodemlp)(adj)(gat)(pool)(ctxmlp)(ctx)(gather)(edgemlp)(eij)(pfnote)(latnote),
          draw=blue!40, dashed, rounded corners=5pt, fill=blue!3,
          inner xsep=6pt, inner ysep=8pt] (s1box) {};
  \end{scope}
  \node[label, text=blue!70, anchor=south] at (s1box.north) {Stage 1: Graph-pair encoder (actor branch shown)};

  \node[stage2, minimum width=1.55cm] (pconcat) at (15.2, 1.5) {$[\mathbf{c}\|\mathbf{e}_{ij}]$};
  \node[stage2, fill=orange!15, minimum width=1.55cm] (pairmlp) at (17.2, 1.5) {Pair scorer\\$192\to 128\to 2$};
  \draw[arr] (ctx.east) -| ($(pconcat.north)+(-0.25,0)$);
  \draw[arr] (eij.east) -| ($(pconcat.south)+(-0.25,0)$);
  \draw[arr] (pconcat) -- (pairmlp);
  \node[dim, above=1pt of pairmlp] {$\times 45$ candidate pairs};

  \node[stage2, fill=orange!15, minimum width=1.55cm] (sendmlp) at (17.2, -0.25) {Send scorer\\$128\to 64\to 1$};
  \draw[arr] (ctx.east) -- ++(0.35,0) |- (sendmlp.west);

  \node[stage2, fill=orange!20, minimum width=1.15cm] (logits) at (19.3, 0.65) {91 logits};
  \draw[arr] (pairmlp.east) -| (logits.north) node[near start, above, dim] {90};
  \draw[arr] (sendmlp.east) -| (logits.south) node[near start, below, dim] {1};

  \node[output, fill=red!10, minimum width=1.2cm] (mask) at (21.0, 0.65) {Apply action\\mask $\mathbf{m}_t$};
  \node[output, fill=green!15, minimum width=1.0cm] (softmax) at (22.7, 0.65) {Softmax};
  \node[output, fill=green!20, minimum width=1.25cm] (pi) at (24.4, 0.65) {$\pi_\theta(a\mid s)$};
  \draw[arr] (logits) -- (mask);
  \draw[arr] (mask) -- (softmax) node[midway, above, dim] {invalid $\mapsto -\infty$};
  \draw[arr] (softmax) -- (pi);

  \node[note, anchor=north west] (criticnote) at (14.5, -1.65) {Critic: a separate graph-pair encoder with the same architecture\\but independent weights, followed by a value MLP $3008\to 256\to 128\to 1$.};
  \node[stage2, fill=purple!10, minimum width=1.75cm] (vf) at (21.4, -1.65) {Value head\\$3008\to 256\to 128\to 1$};
  \node[output, fill=purple!18, minimum width=0.9cm] (vout) at (24.0, -1.65) {$V_\phi(s)$};
  \draw[darr] (criticnote.east) -- (vf.west);
  \draw[arr] (vf) -- (vout);

  \begin{scope}[on background layer]
    \node[fit=(pconcat)(pairmlp)(sendmlp)(logits)(mask)(softmax)(pi)(criticnote)(vf)(vout),
          draw=orange!40, dashed, rounded corners=5pt, fill=orange!3,
          inner xsep=6pt, inner ysep=8pt] (s2box) {};
  \end{scope}
  \node[label, text=orange!70, anchor=south] at (s2box.north) {Stage 2: Pairwise action head + critic};

\end{tikzpicture}}
  \caption{Architecture of the graph-structured policy network.
    \textbf{Stage~1 (graph-pair encoder):} Per-request features are lifted by
    a node MLP and refined by two graph-attention blocks ($h{=}4$ heads,
    $d{=}128$).  A context network mean-pools all node embeddings to a global
    context vector~$\mathbf{c} \in \mathbb{R}^{128}$.  An edge network combines node embeddings,
    element-wise interactions, and raw pair features into per-pair
    embeddings~$\mathbf{e}_{ij} \in \mathbb{R}^{64}$.
    \textbf{Stage~2 (pairwise action head):} Each pair embedding is concatenated
    with~$\mathbf{c}$ and mapped to two keep-side logits; a separate
    unicast MLP produces one additional logit.  Dynamic action masking
    zeroes infeasible entries before softmax produces~$\pi_\theta(a \mid s)$.
    The value head uses a separate graph-pair encoder (not shown).}
  \label{fig:policy_arch}
\end{figure*}
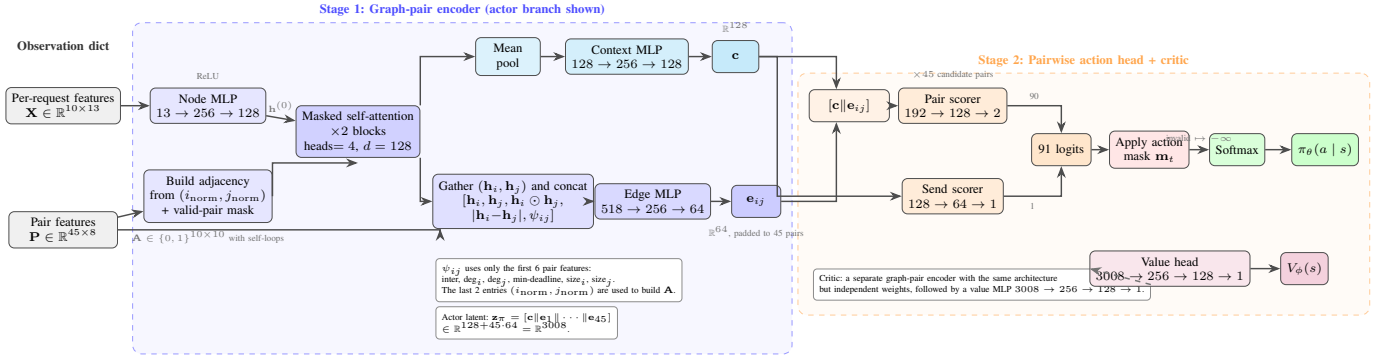

\FloatBarrier

\subsection{Phase 1: Rollout-Improved Behavior Cloning}
\label{sec:method:bc}

A randomly initialized policy almost never discovers the selective merging
from the exploration alone, considering that the reward signal is sparse
while the consequences of the premature merging (i.e., fewer future
coding opportunities and more expirations) are delayed across many
steps. We therefore start with a supervised warm-up using the
\emph{behavior cloning} (BC) from a rollout-improved
teacher~\cite{ross2011dagger}, thereby giving the policy a useful
initial starting point.

\subsubsection*{Teacher construction}
For each encountered state $s$, the teacher constructs a candidate
action set as follows.
We rank all currently feasible merge pairs $(i,j) \in \mathcal{M}_t$
by the lexicographic key
\begin{align*}
\mathrm{rank}(i,j) \;=\; \bigl(\,
  &|\mathcal{S}_i\cap\mathcal{S}_j|,\;
   -\min(d_i,d_j),\\
  &\deg(r_i)+\deg(r_j),\;
   -k(i,j)\bigr),
\end{align*}
sorted in descending order, where the primary key is the
side-information intersection size, the secondary key prefers more
urgent (smaller minimum-deadline) pairs, the tertiary key prefers
larger combined degree, and the final key $-k(i,j)$ negates the
pair's position under the lexicographic queue-index enumeration
($i<j$) used to build the action mask, so that on full ties the
\emph{earliest} lexicographic position wins (i.e., descending sort
on the negated key implements ascending lexicographic
candidate-enumeration order on ties). We collectively call this
the \emph{raw merge score}; the same key is implemented in the
reference simulator's \texttt{candidate\_actions\_topk} routine.
We retain the top $K_{\mathrm{pair}}=16$ pairs under this ranking;
\emph{both} keep-side variants $\kappa\in\{0,1\}$ are expanded for each
selected pair, yielding up to $2K_{\mathrm{pair}}=32$ coded candidates.
The deterministic EDF unicast action is \emph{always} included as an
additional candidate (so the planner can always choose to defer to
unicast on urgent states), giving a candidate set of size at most
$2K_{\mathrm{pair}}+1=33$, further restricted to mask-feasible
actions.
Each candidate $a$ is then evaluated by a short-horizon simulation:
\begin{enumerate}[label=\roman*.]
  \item \textbf{Clone the environment} in state $s$. Cloning copies
        the queue, the per-episode cache placement, and the RNG
        bit-generator state, so the cloned environment carries the
        same hidden episode context as the original.
  \item Execute candidate action $a$.
  \item Run SACM++ for $d=4$ lookahead steps; the per-step return
        accumulated along the rollout is the \emph{full shaped
        reward} $R_{\mathrm{total}} = R_{\mathrm{base}} +
        R_{\mathrm{quality}} + R_{\mathrm{shape}}$ used in PPO
        training (Sec.~\ref{sec:env:reward}), discounted with
        $\gamma = 0.995$.
  \item Repeat for $M=4$ independent Monte Carlo seeds. For each
        candidate, every seed re-clones the original environment and
        re-seeds its RNG with a fixed per-seed value before applying
        $a$, so all candidates are scored against the \emph{same
        common-random-numbers} stream per seed and the
        candidate-vs-candidate comparison is paired across rollouts.
        The discounted returns are averaged across the $M$ seeds.
\end{enumerate}
The action with the highest average discounted return is selected as
the teacher label $a^* = \arg\max_a \hat{V}(s, a)$, while the ties on
$\hat{V}(s,a)$ are broken by the candidate-enumeration order above
(i.e., lexicographic $(i,j,\kappa)$ for the coded actions, with the
unicast action ranked last). Because the lookahead uses the rollouts
rather than a learned value function, this teacher is model-free and
\emph{approximately rollout-improved} over the SACM++ heuristic on the
retained candidate set, i.e., each retained candidate is scored by its
average $d{=}4$ Monte-Carlo discounted return under the SACM++
continuation, thereby allowing the teacher to outperform the SACM++
action under its own finite-horizon, pruning-restricted estimate
$\hat V(s,a)$. We do \emph{not} claim a strict per-step improvement
guarantee, even though the rollout-improved teacher is empirically
stronger than SACM++, since $\hat V$ is finite-horizon, stochastic,
and sensitive to the top-$K_{\mathrm{pair}}$ pre-filter (see the
pruning-sensitivity caveat below).

\smallskip\noindent
\emph{Pruning-sensitivity caveat.} The candidate-set pre-filter
keeps only the top $K_{\mathrm{pair}}$ feasible pairs, so the
supervisory signal is restricted to the actions that survive a heuristic
pre-filter closely related to the SACM++ ranking, i.e., the same
ranking that we used for the teacher construction. We do not report
the binding rate of this filter or a full-enumeration-vs-pruned
teacher comparison on a small instance in this version; that
sensitivity sweep would isolate how much of the BC/ExIt advantage
comes from the rollout-improved planner with respect to the pre-filter
itself, and we leave it for future work.

\subsubsection*{Dataset and training}
We collected a dataset of $N_{\mathrm{BC}} = 150{,}000$ state--action
pairs $(s, a^*)$ by rolling out the teacher across the randomized
episodes. We then trained the policy network by minimizing the
cross-entropy loss
$\mathcal{L}_{\mathrm{BC}} = -\mathbb{E}_{(s,a^*)}[\log \pi_\theta(a^* \mid s)]$
using Adam with the learning rate $\eta_{\mathrm{BC}} = 3\times10^{-4}$,
for $6$~epochs with minibatch size $2{,}048$ and gradient clipping
at $1.0$.

After the BC, the policy already reproduces the selective merge
behavior, i.e., it unicasts the urgent requests and prefers the pairs
with large cached side-information intersections, thereby shortening
the self-improvement phase that follows.

\subsection{Phase 2: Value Network Warm-Up}
\label{sec:method:warmup}

The policy gradient methods need a well-calibrated value function to
produce useful advantage estimates, considering that the advantage
estimates drive the policy updates. Right after the BC, the critic's
parameters are still random and produce inaccurate baselines, which
gives high-variance gradients that can destabilize the BC-warmed policy.

To address this, we froze the action-scoring layer and ran $50{,}000$
MaskablePPO steps with $32$~parallel environments, while only the
value network (and the shared lower-layer graph-pair encoder of the
critic) receives the gradient updates. The critic thus learns to
predict the returns under the fixed BC policy before the policy
optimization begins, thereby giving the subsequent PPO phase a
calibrated advantage estimator. Once the $50{,}000$ steps are complete,
we unfroze the action-scoring layer and Phase~3 begins with an
already-calibrated advantage estimator.

\subsection{Phase 3: MaskablePPO with Curriculum Learning}
\label{sec:method:ppo}

Phase 3 runs Maskable Proximal Policy
Optimization~\cite{schulman2017ppo,sb3maskableppo} with Generalized
Advantage Estimation (GAE,~$\lambda=0.95$)~\cite{schulman2015gae} for
$6{,}000{,}000$ Phase-3 environment steps per training seed.
Combined with the $50{,}000$-step Phase-2 warm-up
(Sec.~\ref{sec:method:warmup}), this gives $6{,}050{,}000$ total
RL environment steps per seed; the training-curve x-axes in
Fig.~\ref{fig:training_curves} (and its Zipf counterpart in
Fig.~\ref{fig:training_curves_zipf}) include the warm-up offset, so
the curriculum boundaries plotted there are at $T = 50\text{K} +
500\text{K} = 550\text{K}$ and $T = 50\text{K} + 1\text{M} =
1{.}05\text{M}$.
Training is distributed across 32~parallel environment instances.

\subsubsection*{Curriculum}
Training directly on the target domain ($N=100$ files, $p_c=0.30$) is
challenging, considering that with more files the side-information
intersection of any two random requests is small, so the feasible merge
pairs are rare while the reward is sparse. Following the curriculum
learning~\cite{bengio2009curriculum}, we increase the task difficulty
across the three stages:
\begin{enumerate}[label=\Roman*.]
  \item \textbf{Stage~I} ($500{,}000$ steps): $N=60$ files,
        $p_c=0.50$, dense merge opportunities, strong shaping signal.
  \item \textbf{Stage~II} ($500{,}000$ steps): $N=80$ files,
        $p_c=0.40$, intermediate difficulty.
  \item \textbf{Stage~III} ($500{,}000+$ steps): $N=100$ files,
        $p_c=0.30$, target evaluation domain.
\end{enumerate}
Each stage is initialized from the final checkpoint of the previous
stage. The policy thus encounters the monotonically increasing task
difficulty, while the learned selective merge behaviors transfer from
the easier to the harder domains, thereby producing a more
generalization-capable final policy.

\subsubsection*{PPO hyperparameters}
The main hyperparameters are as follows: the learning rate
$5{\times}10^{-4} \to 1{\times}10^{-4}$ (linear decay), the entropy
coefficient $0.010 \to 0.001$, $32$ parallel environments, $6$M
Phase-3 environment steps per seed (plus the $50$K warm-up of
Phase~2, for $6.05$M total RL steps). The training runs in
$250{,}000$-step chunks, thereby letting us interleave the ExIt
distillation between the chunks (Section~\ref{sec:method:exit}). The
full hyperparameter set is in the Appendix
Table~\ref{tab:hparams}.

\subsection{Expert Iteration with Online Distillation}
\label{sec:method:exit}

Even with a BC warm start, pure PPO can plateau if the policy's value
function under-estimates the benefit of the rare but highly rewarding
selective merge decisions, i.e., the decisions that matter most for
the final policy. We address this through the \emph{Expert Iteration}
(ExIt)~\cite{anthony2017exit} with the Dataset Aggregation
(DAgger)~\cite{ross2011dagger}, interleaved with PPO between the
training chunks, thereby periodically re-injecting a planner-based
supervisory signal into the policy. The trigger is implemented as follows: PPO runs in
$250{,}000$-step chunks and a distillation iteration fires at the end
of any chunk for which a $300{,}000$-step interval threshold has been
crossed (the next-fire counter advances by $300{,}000$ each time).
Concretely, the first fire occurs at $T = 500{,}000$ Phase-3 steps
(the first chunk boundary at or after the start threshold of
$300{,}000$), and subsequent fires snap to the chunk boundary at or
after each $+300{,}000$ increment of the threshold. Because a
$+300{,}000$ threshold advance lands between chunk boundaries every
fourth chunk (the chunks at $T \in \{1.75\text{M},\,3.25\text{M},\,
4.75\text{M}\}$ fall \emph{below} the threshold and skip), this
realizes \emph{$20$ distillation iterations} over the $24$ Phase-3
chunks, with an average inter-fire interval of $\approx 300{,}000$
Phase-3 steps that matches the threshold parameter.

\subsubsection*{Distillation procedure}
Each ExIt iteration proceeds as follows:
\begin{enumerate}[label=\roman*.]
  \item \textbf{Roll-in.} Collect 8{,}192 states by executing the
        \emph{current policy} with probability $1-p_{\mathrm{expert}}=0.80$
        and the heuristic expert (SACM++) with probability
        $p_{\mathrm{expert}}=0.20$.
        The mixed roll-in gives sufficient coverage of high-value state
        regions that the current policy may still visit infrequently.
  \item \textbf{Critic-bootstrapped labeling.} For each collected state
        $s$, the candidate set is constructed by the same raw merge
        score key as Section~\ref{sec:method:bc} (intersection size,
        then $-\min(d_i,d_j)$, then combined degree, then queue-index
        position), retaining the top $K_{\mathrm{pair}}=12$ pairs,
        expanding both keep-side variants, and always including the
        EDF unicast action, giving up to $2K_{\mathrm{pair}}+1=25$
        mask-feasible candidates. For each candidate $a$, clone the
        environment (queue, hidden cache placement, and RNG
        bit-generator state) and execute the action; run SACM++ for
        $d=5$ lookahead steps using the full shaped reward
        $R_{\mathrm{total}}$ at $\gamma=0.995$, and bootstrap with
        the critic's value estimate $\hat{V}_\theta(s')$ at the
        terminal state. Average over $M=3$ Monte Carlo seeds with
        common random numbers across candidates per seed, and select
        $\hat{a} = \arg\max_a \hat{Q}(s, a)$, with the same
        candidate-enumeration tie-break as in BC. The
        pruning-sensitivity caveat from BC carries over verbatim.
  \item \textbf{Buffer aggregation.} Append the $(s, \hat{a})$ pairs to a
        DAgger replay buffer capped at 80{,}000 entries (uniform random
        eviction when full).
  \item \textbf{Distillation.} Train the policy on the buffer for 2~epochs
        via cross-entropy loss, Adam, $\eta=1\times10^{-4}$, batch size 2{,}048.
\end{enumerate}

Because the ExIt teacher bootstraps from the agent's own improving
critic $\hat{V}_\theta$, the labels improve as the training progresses,
thereby letting the agent surpass its heuristic teachers
(Appendix~\ref{sec:appendix:training_details}).

Algorithm~\ref{alg:pipeline} summarizes the complete three-phase training pipeline.

\begin{algorithm}[t]
\caption{BC $\to$ PPO $\to$ ExIt Training Pipeline}
\label{alg:pipeline}
\begin{algorithmic}[1]
\Statex \textbf{Input:} Environment $\mathcal{E}$, heuristic teacher SACM++, hyperparams (Table~\ref{tab:hparams})
\Statex \textbf{Output:} Trained policy $\pi_\theta$
\Statex
\Statex \textit{// Phase 1: Rollout-Improved Behavior Cloning (\S\ref{sec:method:bc})}
\State Collect $N_{\mathrm{BC}} = 150\text{K}$ state--action pairs $(s, a^*)$ using rollout teacher
\State Train $\pi_\theta$ via cross-entropy loss for 6 epochs (Adam, $\eta = 3\!\times\!10^{-4}$)
\Statex
\Statex \textit{// Phase 2: Value Network Warm-Up (\S\ref{sec:method:warmup})}
\State Freeze pairwise action head of $\pi_\theta$
\State Run MaskablePPO for $50\text{K}$ steps (critic only receives gradients)
\State Unfreeze pairwise action head
\Statex
\Statex \textit{// Phase 3: MaskablePPO + Curriculum + ExIt (\S\ref{sec:method:ppo}--\ref{sec:method:exit})}
\Statex \textit{// Single global Phase-3 step counter $T$, $250\text{K}$-step chunks; $6\text{M}$ total Phase-3 steps across all stages (24 chunks). Adding the $50\text{K}$ Phase-2 warm-up gives $6.05\text{M}$ total RL env steps per seed.}
\Statex \textit{// Curriculum advance schedule: I $\to$ II at $T = 500\text{K}$ (after 2 chunks); II $\to$ III at $T = 1\text{M}$ (after 4 chunks); III runs the remaining 20 chunks ($5\text{M}$ steps).}
\State $T \leftarrow 0$;\; $T_{\mathrm{next}} \leftarrow 300\text{K}$;\; stage $\leftarrow$ I with $(N{=}60, p_c{=}0.50)$
\While{$T < 6\text{M}$}
    \If{$T = 500\text{K}$} stage $\leftarrow$ II with $(N{=}80, p_c{=}0.40)$ \EndIf
    \If{$T = 1\text{M}$}   stage $\leftarrow$ III with $(N{=}100, p_c{=}0.30)$ \EndIf
    \State Run MaskablePPO with GAE ($\gamma{=}0.995$, $\lambda{=}0.95$) for one $250\text{K}$-step chunk
    \State $T \leftarrow T + 250\text{K}$
    \If{$T \geq T_{\mathrm{next}}$} \Comment{ExIt fires only when $T$ has crossed the next $+300\text{K}$ threshold}
        \State \textit{// ExIt distillation (\S\ref{sec:method:exit})}
        \State Collect $8{,}192$ states via mixed roll-in ($p_{\mathrm{expert}}{=}0.20$)
        \State Label each state: critic-bootstrapped lookahead planner
        \State Append $(s, \hat{a})$ to DAgger buffer (cap $80\text{K}$)
        \State Distill: 2 epochs cross-entropy on buffer
        \While{$T_{\mathrm{next}} \le T$} $T_{\mathrm{next}} \leftarrow T_{\mathrm{next}} + 300\text{K}$ \EndWhile \Comment{advance counter past current $T$}
    \EndIf
\EndWhile
\State \Return $\pi_\theta$ \Comment{schedule realizes $20$ ExIt iterations across the $24$ chunks}
\end{algorithmic}
\end{algorithm}

Fig.~\ref{fig:training_curves} shows the evolution of evaluation reward and
training losses across the three curriculum stages.

\begin{figure}[t]
  \centering
  \includegraphics[width=\columnwidth]{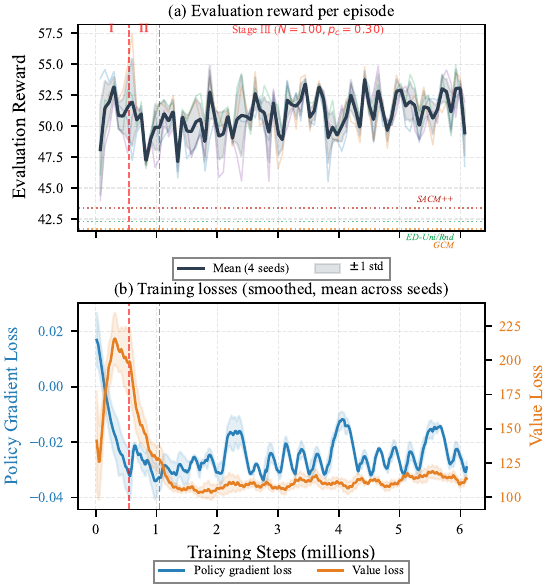}
  \caption{Training dynamics of the Track~A (uniform-demand) agent across 4 seeds.
    The x-axis shows total RL environment steps (6.05M per seed, including the 50K
    Phase-2 warm-up); curriculum stage transitions therefore appear at
    $T = 50\text{K} + 500\text{K} = 550\text{K}$ and
    $T = 50\text{K} + 1.0\text{M} = 1.05\text{M}$ on this axis (corresponding to
    $T_{\mathrm{Phase\,3}} = 500\text{K}$ and $1.0\text{M}$ Phase-3 steps).
    \textbf{(a)}~Evaluation reward (mean $\pm$ std across seeds);
    vertical dashed lines mark curriculum stage transitions
    at $550\text{K}$ and $1.05\text{M}$ total RL steps. Horizontal dotted
    lines show baseline policy rewards (SACM{++}, ED-Unicast/Random,
    GCM) evaluated under the same training reward configuration.
    \textbf{(b)}~Policy-gradient loss (left axis, blue) and value loss
    (right axis, orange), both smoothed and averaged across seeds.
    Both losses stabilize after $\sim$2M total RL steps, with periodic
    perturbations from ExIt distillation events.}
  \label{fig:training_curves}
\end{figure}

\subsection{Model Selection}
\label{sec:method:selection}

We ran four independent training seeds (each with distinct network
initializations and environment seeds), thereby giving us four
candidate models. To control for the variance in the training
outcomes~\cite{henderson2018deep}, we selected among the four
candidates using a \emph{robust advantage criterion} evaluated on
$50$ validation seeds that are entirely separate from the $50$ holdout
seeds used in the final evaluation reported in
Section~\ref{sec:results}.

For each candidate model $m$ and validation seed $v$, the per-seed advantage
over the highest-throughput coded baseline / BC--ExIt teacher (SACM++) is
\begin{equation}
  a_v^{(m)} = \sigma^{(m)}_v - \sigma^{\mathrm{SACM++}}_v,
\end{equation}
where $\sigma_v$ denotes BE-score on seed $v$.
The selection criterion is the \emph{robust advantage}:
\begin{equation}
  \Omega^{(m)} = \bar{a}^{(m)} - 0.5\,\hat{\sigma}^{(m)}_a,
\end{equation}
with $\bar{a}^{(m)}$ the mean per-seed advantage and $\hat{\sigma}^{(m)}_a$
its standard deviation across the 50 validation seeds.
The $-0.5\,\hat{\sigma}$ penalty favors the models with both the high
average advantage \emph{and} the low variance, while it penalizes the
models that perform well only on a subset of the validation seeds.
We pick the model with the highest $\Omega^{(m)}$ and evaluate it once
on the holdout seeds, considering that these holdout seeds are never
observed during the training or the model selection, thereby giving us
an unbiased estimate of the generalization performance.
Appendix~\ref{sec:appendix:seed_stability} reports the holdout
performance of all four seed models, confirming the low inter-seed
variance.

\smallskip\noindent
\emph{Methodological caveat: $\sigma$-based selection vs.\ the primary
demand-centric metrics.} The selection criterion is aligned with the
broadcast-efficiency score $\sigma$, while
Sec.~\ref{sec:exp:metrics} declares the broadcast-packet expiration
ratio $\rho$ and the distinct file-identity coverage $\delta$ as the
\emph{primary} user-facing metrics, i.e., the metrics we report in the
headline tables; the request-level family
$(\eta_{\mathrm{req}}, m_{\mathrm{req}}, \sigma_{\mathrm{req}})$ is
\emph{supplementary} in the current version (full-baseline coverage
left as future work). We chose $\sigma$ as the pre-registered
selection criterion because it is the only metric that matches the
base training reward $R_{\mathrm{base}} = U_t - E_t$
(Sec.~\ref{sec:env:reward}, Remark~\ref{rem:training_vs_eval}),
thereby keeping the validation--training comparison consistent. The extended
seed-stability table at Appendix~\ref{sec:appendix:seed_stability}
reports $\sigma$, $\rho$, served/tx, expirations, $\delta$,
$\eta_{\mathrm{req}}$, $m_{\mathrm{req}}$, and $\sigma_{\mathrm{req}}$
across the four seeds and shows narrow inter-seed spread on every
column. The selected seed is the per-seed best on both of the primary
demand-centric metrics ($\rho{=}0.203$ and $\delta{=}0.824$) as well
as on $\sigma$, while on the supplementary request-level family it is
competitive even though it is not the per-seed leader (i.e., seed~1
reaches $\sigma_{\mathrm{req}}{=}0.853$ with respect to the selected
seed's $0.833$), which reflects a broadcast-efficiency /
request-level deadline-compliance trade-off rather than a weakness on
the selected seed's own primary objective.
We do not report a focused robustness sweep over alternative selection
rules (best-$\rho$, best-$m_{\mathrm{req}}$, or a Pareto multi-metric
rule) in this version; that controlled selection-sensitivity sweep
remains future work.

\section{Experimental Protocol}
\label{sec:experiments}

This section describes the two evaluation tracks, the complete set of
methods (Track~A: $9$ baselines + PPO = $10$ unique methods, $12$ table
rows with the two literature-compatible aliases; Track~B: $6$ baselines
+ PPO-Zipf = $7$ unique methods), the non-ID generalization regimes
(split into the curriculum-seen and the unseen-parameter ones; see
Sec.~\ref{sec:exp:ood} for the narrow within-family sense in which we
use ``OOD-'' throughout), and the statistical reporting protocol, i.e.,
the protocol that we follow for all the comparisons in this paper.
The quantitative results are deferred to Section~\ref{sec:results}.

\subsection{Environment Parameters and Training Configuration}
\label{sec:exp:setup}

All the experiments use the environment described in
Section~\ref{sec:env} with the default parameters listed in
Table~\ref{tab:env_params}. We trained two agents independently, i.e.,
the \textbf{Track~A agent} on the uniform demand (${\approx}1.73$M
parameters) with $N{=}100$, $p_c{=}0.30$, $D{=}20$; and the
\textbf{Track~B agent} on the Zipf demand ($\alpha{=}0.8$,
${\approx}1.75$M parameters, with the additional popularity-aware
features). Within each track, all the evaluation conditions beyond
the training domain are \emph{non-ID}, while we further split the
non-ID into the \emph{curriculum-seen} regimes (i.e., the parameter
values the agent encountered during the curriculum stages of
Sec.~\ref{sec:method:ppo}) and the \emph{out-of-distribution} (OOD)
regimes (i.e., the parameter values never seen during the training).
The exact split is given in Sec.~\ref{sec:exp:ood} (Track~A:
$2$~curriculum-seen, $5$~OOD; Track~B: $2$~curriculum-seen, $9$~OOD).

\begin{table}[!t]
  \renewcommand{\arraystretch}{1.2}
  \caption{Default Environment Parameters (Training Domain)}
  \label{tab:env_params}
  \centering
  \begin{tabular}{clr}
    \toprule
    Symbol & Description & Value \\
    \midrule
    $N$           & File catalog size             & 100   \\
    $B$           & Subfiles per file             & 10    \\
    $K$           & Edge caches (users)           & 5     \\
    $Q$           & Queue capacity                & 10    \\
    $D$           & Max deadline (slots)          & 20    \\
    $H$           & Episode horizon (steps)       & 50    \\
    $p_c$ & Decentralized cache fraction & 0.30 \\
    $P_{\max}$    & Max candidate merge pairs $\bigl({=}\tbinom{Q}{2}\bigr)$     & 45    \\
    \bottomrule
  \end{tabular}
\end{table}

The queue size $Q=10$ and the number of caches $K=5$ are fixed across
all the experiments. Varying $K$ changes the observation dimension
($d_{\mathrm{req}} = 2K+3$) and would require retraining, even though
the rest of the architecture remains the same; extending to
$K \in \{10, 20\}$ is left for the future work
(Section~\ref{sec:conclusion}).

We ran four independent training seeds, each for $6{,}000{,}000$
Phase-3 environment steps (plus a $50{,}000$-step Phase-2 warm-up,
for $6{,}050{,}000$ total RL environment steps per seed) using
MaskablePPO with $32$ parallel environments and the three-phase
curriculum described in Section~\ref{sec:method}. The model selection
uses the robust advantage criterion
$\Omega = \bar{a} - 0.5\hat{\sigma}_a$ on a validation split of $50$
seeds, entirely separate from the holdout set. The selected model is
evaluated on $50$ holdout seeds ($200$ episodes/seed), thereby giving
$\mathbf{10{,}000}$ \textbf{episodes per method per evaluation
condition}. All the methods share the identical seed set under each
condition, thereby giving a fully \emph{paired} design. The seed
formulas, the wall-clock times, and the inference latency are
reported in Appendix~\ref{sec:appendix:reproducibility}.

\subsection{Baseline Implementations}
\label{sec:exp:baselines}

We compare against the $9$ Track~A baselines (ED-Unicast, GCM, SACM,
SACM+, SACM{++}, TauFit-0/1/2/3) and the $6$ Track~B baselines (i.e.,
the five Track~A coded methods reused under the Zipf demand plus the
popularity-aware SACM{++}-Pop). Counting PPO as the trained agent, the
Track~A evaluation reports $9 + 1 = 10$ unique methods, displayed as
the $12$ table rows once the two literature-compatible aliases
(Perfect-Fit $\equiv$ TauFit-0 and First-Fit $\equiv$ TauFit-3) are
shown for the compatibility with the prior literature, while the
Track~B evaluation reports $6 + 1 = 7$ unique methods. The baselines
are organized into the three strategy families.
Table~\ref{tab:baselines} summarizes all the baselines with their
merge rates on the ID-default condition.

\subsubsection*{Family 1 — Conservative Uncoded}

\textbf{ED-Unicast} transmits the earliest-deadline pending request as
a unicast packet at every step, irrespective of the available merge
set. It is the deadline-safe lower bound, i.e., zero expirations from
the excessive merging, even though there is no coding gain (merge
rate $= 0$).

\subsubsection*{Family 2 — Aggressive Coded Multicast}

All four methods in this family always merge whenever
$\mathcal{M}_t \neq \emptyset$, giving merge rate $= 100\%$.

\textbf{GCM (Greedy Coded Merge)} selects the first element of
$\mathcal{M}_t$ without further discrimination.

\textbf{SACM} selects the pair that maximizes the cached side-information
intersection $|\mathcal{S}_i \cap \mathcal{S}_j|$, consistently keeping
endpoint~$i$~\cite{asghari2019approximation}.

\textbf{SACM+} refines the keep-side decision: after selecting the
maximum-intersection pair, it retains the endpoint with the higher merge
degree, improving future connectivity.

\textbf{SACM++} applies a lexicographic priority key
$(|\mathcal{S}_i \cap \mathcal{S}_j|,\, -\min(d_i, d_j))$ that
maximizes the side-information overlap and then prioritizes the more
time-critical pair, while keeping the degree-aware endpoint retention.
SACM++ is the highest-throughput coded-multicast baseline in our
reference simulator (i.e., the highest Served/Tx, $\mu$, on the
ID-default uniform table) and is the heuristic teacher used for the
BC and the ExIt, considering that it maximizes the broadcast-level
packet throughput; even though it is the highest-throughput baseline,
it is not the per-metric best baseline (e.g., GCM and SACM tie SACM++
on $\rho$ at $0.345$ on the ID-default, while SACM has a higher
$\sigma$ than SACM++).

\smallskip\noindent
\emph{Baseline keep-side semantics.} For all baseline methods (SACM, SACM+, SACM++, and the SACM++-Pop variant introduced below), the keep-side bit selects only which queue slot is vacated and refilled (Sec.~\ref{sec:model:coded}); the representative destination $k_{\mathrm{mg}}$ follows Eq.~\eqref{eq:mg_dest} independently of $\kappa$, exactly as it does for the learned policy. ``Endpoint retention'' here is therefore a queue-slot operation, not a destination-identity operation.

\subsubsection*{Family 3 — Oracle-Tuned Threshold Policies ($\tau$-Fit)}

This family uses the \emph{misfit}
metric~\cite{codedcachingdelaysensitive}, while the threshold
$\tau \in \{0,1,2,3\}$ controls the merge aggressiveness, from the
\textbf{Perfect-Fit} ($\tau{=}0$, merge rate $13.3\%$) through the
\textbf{TauFit-1/2/3} ($33.2\%$--$41.6\%$) to the \textbf{First-Fit}
($\tau{=}K{-}2$, identical to TauFit-3). The Perfect-Fit and TauFit-0
are identical ($\tau{=}0$), while the First-Fit and TauFit-3
($\tau{=}K{-}2$) are likewise identical, i.e., the same policy with
two different labels. Both of the alias labels are retained in the
tables for the compatibility with the prior literature, thereby
giving $12$ table rows from the $10$ unique methods (i.e., the $9$
Track~A baselines plus PPO). For each evaluation metric, we select
the oracle-optimal $\tau^*$ by the exhaustive search on the validation
split, considering that this oracle is unavailable at the deployment
and is the ceiling of any fixed-threshold policy. The full misfit
definition and the per-metric oracle selections are in
Appendix~\ref{sec:appendix:exp_details}.

\subsubsection*{Track~B Additional Baseline}

For the Zipf-demand track (Section~\ref{sec:exp:ood}), we add
\textbf{SACM++-Pop}: a popularity-aware extension of SACM++ that uses
the lexicographic priority key
\[
\bigl(\,|\mathcal{S}_i\cap\mathcal{S}_j|,\;
       \hat m(f_i\cup f_j),\;
       -\min(d_i,d_j)\bigr),
\]
sorted in descending order, where $\hat m(\cdot)$ is the
popularity-mass projection of Appendix~\ref{sec:appendix:features}
($\hat m(f) = \sum_{p\in f}\hat p_{\phi(p)}$ under the assumed Zipf
file-popularity law). Ties on the full key are broken by the
lexicographic queue-index order $(i,j)$ with $i<j$. The keep-side
bit is set to $\kappa{=}0$ when
$\deg(r_i) + 2\hat m(f_i) \ge \deg(r_j) + 2\hat m(f_j)$ and to
$\kappa{=}1$ otherwise (popularity-weighted degree-aware endpoint
retention); the representative destination still follows
Eq.~\eqref{eq:mg_dest} as for every method (Sec.~\ref{sec:model:coded}).
This baseline tests whether a heuristic that explicitly exploits the
Zipf skew can close the gap with respect to the learned policy.

\begin{table}[!t]
  \renewcommand{\arraystretch}{1.2}
  \caption{Baseline Summary (Merge Rate on ID-Default, Uniform Demand)}
  \label{tab:baselines}
  \centering
  \begin{adjustbox}{max width=\columnwidth}
  \begin{tabular}{llcc}
    \toprule
    Method & Family & Merge Rate & Key property \\
    \midrule
    ED-Unicast    & Uncoded         & 0\%    & Never codes \\
    GCM           & Coded (greedy)  & 100\%  & First feasible pair \\
    SACM          & Coded           & 100\%  & Max intersection \\
    SACM+         & Coded           & 100\%  & Max inter. + degree keep \\
    SACM++        & Coded           & 100\%  & Lex (inter., $-d$) + degree \\
    Perfect-Fit   & $\tau$-Fit ($\tau$=0) & 13.3\% & Misfit=0 only \\
    TauFit-1      & $\tau$-Fit ($\tau$=1) & 33.2\% & Misfit$\leq$1 \\
    TauFit-2      & $\tau$-Fit ($\tau$=2) & 40.6\% & Misfit$\leq$2 \\
    TauFit-3      & $\tau$-Fit ($\tau$=3) & 41.6\% & Misfit$\leq$3 \\
    First-Fit     & $\tau$-Fit ($\tau$=3) & 41.6\% & $\equiv$ TauFit-3 \\
    \midrule
    SACM++-Pop    & Coded (Track~B) & 100\%  & Pop.-weighted inter. \\
    \midrule
    \textbf{PPO-Agent} & Learned & 31.8\% & Selective (adaptive) \\
    \bottomrule
  \end{tabular}
  \end{adjustbox}
\end{table}
\FloatBarrier

\subsection{Out-of-Distribution Evaluation Regimes}
\label{sec:exp:ood}

Each track's trained agent is evaluated zero-shot across multiple
non-ID conditions, thereby giving $7$~non-ID regimes for the Track~A
($2$~curriculum-seen, $5$~OOD) and $11$~non-ID regimes for the
Track~B ($2$~curriculum-seen, $9$~OOD). Within each track, the same
trained model is applied directly, without the retraining, the
fine-tuning, or the parameter updates. The observation and action
space shapes are invariant \emph{within} each track ($Q=10$, $K=5$,
$P_{\max}=45$), while only the environment dynamics change. The two
tracks use the different observation spaces
(Table~\ref{tab:obs_dims}) and independently trained models, so the
cross-track comparisons reflect the two distinct policies rather than
the zero-shot transfer of a single model.

\paragraph{Curriculum-seen vs.\ unseen-parameter conditions.}
Two of the non-ID evaluation conditions, $N{=}60$ (Stage~I of the
curriculum) and $p_c{=}0.40$ (Stage~II), were encountered during the
training, even though at the earlier curriculum stages rather than
at the convergence. We label these as the \emph{curriculum-seen
transfer} conditions (prefix \textbf{Curr-}) to distinguish them from
the \emph{unseen-parameter} conditions (prefix \textbf{OOD-}) that
were never seen during the training. The \textbf{OOD-} prefix is used
throughout this paper in this narrow sense, i.e., the unseen parameter
values (the catalog size, the cache fraction, the deadline budget,
the Zipf exponent) within the \emph{same} simulator family ($K$, $Q$,
the action dimensionality $2P_{\max}{+}1$, and the placement family
are held fixed across all the evaluations), so all the transfer
claims should be read as the within-family parameter generalization
rather than as a broader distribution shift. The same convention
applies to the Zipf track as well.

\subsubsection*{Track~A --- Uniform Demand (2 Curr + 5 OOD conditions)}

Starting from the training domain (ID-default: $N=100$,
$p_c=0.30$, $D=20$, uniform demand), we vary one
parameter at a time:

\begin{itemize}[leftmargin=*]
  \item \textbf{Catalog size (sanity check):} $N \in \{60, 120, 150\}$
        (Curr-file60, OOD-file120, OOD-file150).
        Under the decentralized placement with fixed $p_c$, the
        per-packet caching probability is preserved across the catalog
        sizes, while this axis therefore verifies that the policy's
        queue-state features are invariant to the library cardinality
        rather than testing a true distribution shift
        (Section~\ref{sec:res:ood_uniform}).
  \item \textbf{Cache density:} $p_c \in \{0.20, 0.40\}$
        (OOD-pcache0.20, Curr-pcache0.40).
        The lower density shrinks the feasible merge set, while the
        higher density makes the aggressive merging likely to cause
        the expirations through the queue saturation.
  \item \textbf{Deadline budget:} $D \in \{10, 30\}$
        (OOD-delay10, OOD-delay30).
        The tighter deadlines stress the expiration avoidance, while
        the looser deadlines reward the throughput maximization.
\end{itemize}

Track~B adds $11$~non-ID regimes spanning the Zipf tail weight
($\alpha \in \{0.6, 1.0, 1.2\}$), an alternative demand law
(Mandelbrot-Zipf), and the same cross-axis shifts as the Track~A. The
full conditions are listed in Table~\ref{tab:ood}, while the detailed
descriptions are in Appendix~\ref{sec:appendix:exp_details}. Each
track uses a separately trained agent, considering that the
cross-track results are not aggregated. The Track~B results are
reported in Section~\ref{sec:results:zipf}.

\begin{table}[!t]
  \renewcommand{\arraystretch}{1.15}
  \caption{Evaluation Conditions (2 ID + 4 Curr + 14 OOD)}
  \label{tab:ood}
  \centering
  \begin{adjustbox}{max width=\columnwidth}
  \begin{tabular}{llccc}
    \toprule
    Regime & Axis & $N$ & $p_c$ & $D$ \\
    \midrule
    \multicolumn{5}{l}{\textit{Track~A --- Uniform demand (8 conditions, 9 baselines + PPO = 10 unique methods)}} \\
    \midrule
    ID-default      & —             & 100 & 0.30 & 20 \\
    Curr-file60      & Catalog       &  60 & 0.30 & 20 \\
    OOD-file120     & Catalog       & 120 & 0.30 & 20 \\
    OOD-file150     & Catalog       & 150 & 0.30 & 20 \\
    OOD-pcache0.20  & Cache density & 100 & 0.20 & 20 \\
    Curr-pcache0.40  & Cache density & 100 & 0.40 & 20 \\
    OOD-delay10     & Deadline      & 100 & 0.30 & 10 \\
    OOD-delay30     & Deadline      & 100 & 0.30 & 30 \\
    \midrule
    \multicolumn{5}{l}{\textit{Track~B --- Zipf demand ($\alpha$=0.8 base, 1 ID + 2 Curr + 9 OOD, 6 baselines + PPO-Zipf = 7 unique methods, separately trained agent)}} \\
    \midrule
    Zipf-ID         & —             & 100 & 0.30 & 20 \\
    OOD-alpha0.6    & Tail weight   & 100 & 0.30 & 20 \\
    OOD-alpha1.0    & Tail weight   & 100 & 0.30 & 20 \\
    OOD-alpha1.2    & Tail weight   & 100 & 0.30 & 20 \\
    OOD-mandelbrot  & Demand law    & 100 & 0.30 & 20 \\
    Curr-file60 (Z)  & Catalog       &  60 & 0.30 & 20 \\
    OOD-file120 (Z) & Catalog       & 120 & 0.30 & 20 \\
    OOD-file150 (Z) & Catalog       & 150 & 0.30 & 20 \\
    OOD-pcache0.20 (Z) & Cache density & 100 & 0.20 & 20 \\
    Curr-pcache0.40 (Z) & Cache density & 100 & 0.40 & 20 \\
    OOD-delay10 (Z)    & Deadline   & 100 & 0.30 & 10 \\
    OOD-delay30 (Z)    & Deadline   & 100 & 0.30 & 30 \\
    \bottomrule
  \end{tabular}
  \end{adjustbox}
\end{table}

\subsection{Performance Metrics}
\label{sec:exp:metrics}\label{sec:model:metrics}

We evaluate the scheduling policies on the thirteen metrics over an
episode of $H$ steps (ten performance, three diagnostic), organized
into the four families (i.e., the broadcast-level, the file-identity,
the request-level, and the diagnostic ones) and enumerated in detail
below.
Throughout this paper we adopt a single descriptor for the
broadcast-level demand quantity: $U_t$ is the
\emph{packet-set-cardinality XOR degree at step $t$}, defined as
$U_t = |f_{\mathrm{mg}}|$ for a coded transmission (the number of
distinct packet identities carried in the aggregate) and $U_t = 1$
for a unicast. Similarly, $E_t = \sum_{r \in \mathcal{X}_t} |f_r|$ is
the broadcast-level packet-identity count of expired records at
step~$t$. We deliberately avoid attaching ``users served'',
``user-equivalents'', or ``number of original singleton requests'' to
$U_t$ in the headline metric language: the request generator does not
enforce uniqueness of packet identities across active records, so
$|f_r|$ in general does not equal the number of original singleton
arrivals folded into~$r$ (see the per-broadcast accounting remark in
Appendix~\ref{sec:appendix:coded_remarks} for an explicit
counterexample). Whenever original-arrival semantics are intended, we
use the request-level metrics
$\eta_{\mathrm{req}}, m_{\mathrm{req}}, \sigma_{\mathrm{req}}$
(M8--M10), which are computed from the per-arrival identifier sets
$\mathcal{A}(r)$ stamped at arrival time and credit each original
arrival exactly once.
Because a packet identity folded into a coding-state record may
participate in the subsequent broadcasts, $U_t$ measures the
\emph{per-slot broadcast utilization} (i.e., the packet throughput)
at the broadcast level, not the unique demand satisfaction.
The three families therefore use \emph{three different notions of
demand}, and we use disjoint terminology to refer to them:
(a) the broadcast-level family ($\rho$, $\sigma$, $\mu$, $g$,
$\varepsilon$) is computed from $(U_t, E_t)$ and counts the
\emph{packet-set XOR degree} per broadcast (i.e., distinct packet
identities carried in the aggregate), where a packet identity folded into a
chain can be counted at every link;
(b) the file-identity family ($\delta$, $\rho^{\mathrm{uniq}}$) is
computed from $(U_t^{\mathrm{uniq}}, E_t^{\mathrm{uniq}})$ and credits
each \emph{distinct file identity} exactly once per episode (this is
a coverage indicator, not a multi-packet completion metric);
(c) the request-level family ($\eta_{\mathrm{req}}$, $m_{\mathrm{req}}$,
$\sigma_{\mathrm{req}}$) is computed from the original arrival
identifiers~$\mathcal{A}(r)$ defined in Sec.~\ref{sec:model:request_accounting}
and credits each \emph{individual user arrival} exactly once.
We therefore report four complementary views:
(i)~\textbf{primary metrics} (M1--M2):
the broadcast-packet expiration ratio~$\rho$ (a packet-set waste
ratio at the broadcast level) and distinct file-identity
coverage~$\delta$ (a coverage indicator at the file-identity level),
which capture two different facets of timely delivery;
(ii)~\textbf{request-level metrics} (M8--M10, \emph{supplementary in the current version}; the displayed request-level tables compare PPO against the SACM{++} and ED-Unicast comparators, with full-baseline coverage over GCM, SACM, SACM+, and the $\tau$-Fit family left as future work, so this family is not a co-primary headline in the present manuscript): based
on per-request ID tracking (Section~\ref{sec:model:request_accounting}),
which follow each original arrival through queue aggregation and credit
completion or miss exactly once per request identifier; this is the
only family in which the numerator objects are arrival-level
identifiers (the equations remain $H$-normalized rather than
total-arrivals-normalized; see M8--M10);
(iii)~\textbf{broadcast-efficiency metrics} (M3--M7, \emph{secondary}):
based on $U_t$ and $E_t$ as packet-set counts, which measure how well
each channel use exploits the coded-multicast gain at the packet-set
level; and
(iv)~\textbf{diagnostic metrics} (M11--M13): merge rate, opportunity rate,
and reward per step, which characterize policy behavior rather than
performance.

\smallskip\noindent
\textbf{Symbol-retention convention.}
We retain the symbol names $\rho$ and $\delta$ for continuity with
prior tables, but their semantics throughout this paper are the
formal definitions of Eqs.~\eqref{eq:miss_ratio}--\eqref{eq:unique_demand_sat},
namely the broadcast-packet expiration ratio (a packet-set waste
ratio at the broadcast level) and the distinct file-identity coverage
(a once-per-file coverage indicator under the projection $\phi$).
``Demand-centric'' is therefore a shorthand for those formal
definitions rather than a claim of per-arrival miss-probability or
per-user completion semantics; the request-level family
$\eta_{\mathrm{req}}, m_{\mathrm{req}}, \sigma_{\mathrm{req}}$
(M8--M10) is computed from per-arrival identifiers but normalizes
by $H$ rather than by the total number of admitted arrivals, so it
provides per-step request-level rates rather than per-arrival
probabilities.
The compact table-header labels ``Miss Ratio'' and ``Coverage''
used in Tables~\ref{tab:nm14_baselines}, \ref{tab:ablation_id},
\ref{tab:zipf_id}, and elsewhere are the same symbol-retention
shorthand: ``Miss Ratio'' is the broadcast-packet expiration
ratio~$\rho$~\eqref{eq:miss_ratio} and ``Coverage'' is the distinct
file-identity coverage~$\delta$~\eqref{eq:unique_demand_sat}; neither
implies per-arrival or per-user satisfaction.

\smallskip\noindent
\textbf{Right-censoring at episode end.}
All the headline metrics ($\rho$, $\delta$, $\sigma$, $\mu$, $g$,
$\varepsilon$, and the request-level
$\eta_{\mathrm{req}}, m_{\mathrm{req}}, \sigma_{\mathrm{req}}$) are
computed from the $H = 50$ in-episode steps only. The records still
pending at step~$H$ are neither served nor expired in the counters,
which is consistent with the truncated continuing-control framing of
Sec.~\ref{sec:model:problem}, even though it does produce a
policy-dependent right-censoring, i.e., a policy that defers the
difficult records past the horizon will look better on $\rho$ and
$m_{\mathrm{req}}$ with respect to a tail-resolved evaluation. The
inter-method comparisons are paired across the shared episode seeds
and the shared initial cache placements, while the simulator does
\emph{not} implement strict common random numbers, i.e., the RNG
draws within an episode are sequential, so the policy-dependent
refill events cause the realized arrival sequences to diverge across
the methods after the first divergent action. The censoring therefore
is shared in expectation up to the first-divergent-action RNG drift,
rather than cancelling pathwise (notably the paired-bootstrap
intervals of Table~\ref{tab:paired_bootstrap}); the absolute values of
$\rho$, $m_{\mathrm{req}}$, and $\delta$ should therefore be read with
this caveat. A sensitivity sweep over $H \in \{50, 100, 200\}$ and a
tail-flush variant are deferred to a follow-up extension.

\begin{enumerate}[label=\textbf{M\arabic*.}, leftmargin=*, nosep]

  \item \textbf{Broadcast-Packet Expiration Ratio ($\downarrow$, primary).}
  The fraction of broadcast-level packet-set mass that expires before
  delivery is
  \begin{equation}
    \rho = \frac{\displaystyle\sum_{t=1}^{H} E_t}
               {\displaystyle\sum_{t=1}^{H} U_t + \sum_{t=1}^{H} E_t}.
    \label{eq:miss_ratio}
  \end{equation}
  Both the numerator and the denominator are computed from the
  packet-set counts $(U_t, E_t)$, not from the original arrival
  identifiers, so $\rho$ is a broadcast-channel waste ratio (i.e., the
  fraction of the packet-set mass on the channel that is lost to the
  expiration) rather than a per-arrival miss probability, while the
  request-level counterpart is $m_{\mathrm{req}}$ (M9), itself a
  per-step rate rather than a per-arrival probability.
  We retain the symbol $\rho$ for continuity with prior tables
  but use the broader name ``broadcast-packet expiration ratio'' to
  avoid implying per-user miss-rate semantics. Lower values indicate
  better timely delivery at the broadcast level.

  \item \textbf{Distinct File-Identity Coverage ($\uparrow$, primary).}
  Fraction of distinct \emph{file identities} that appeared in the
  served set at least once during the episode, under unique-demand
  (file-identity) accounting (Section~\ref{sec:model:coded}). Under
  the packet-level request model (each arrival is a singleton packet
  request), $\delta$ is therefore a \emph{coverage indicator}: a file
  identity is credited once any of its packets is served, and
  subsequent packet requests for the same file no longer contribute
  to $U_t^{\mathrm{uniq}}$. We refer to this as ``coverage'' rather
  than ``demand satisfaction'' or ``demand completion'' because $\delta$
  does not measure completion of a multi-packet user-file request:
  \begin{equation}
    \delta = \frac{\displaystyle\sum_{t=1}^{H} U_t^{\mathrm{uniq}}}
                  {\displaystyle\sum_{t=1}^{H} U_t^{\mathrm{uniq}}
                   + \sum_{t=1}^{H} E_t^{\mathrm{uniq}}}.
    \label{eq:unique_demand_sat}
  \end{equation}
  Here $U_t^{\mathrm{uniq}}$ and $E_t^{\mathrm{uniq}}$ are file-identity
  counts defined formally via the packet-to-file projection
  $\phi(p) = \lfloor p/B \rfloor$ in
  Eqs.~\eqref{eq:Ut_uniq}--\eqref{eq:Et_uniq}: $U_t^{\mathrm{uniq}}$
  counts only file identities not previously delivered in the episode,
  and $E_t^{\mathrm{uniq}}$ counts only file identities in an expired
  record that were never delivered during the episode.
  A coding-state record whose constituent file identities were already
  covered earlier in the episode contributes $E_t^{\mathrm{uniq}} = 0$
  upon expiration.
  Unlike $\sigma$ (M3), each distinct file identity is counted exactly
  once; unlike $\eta_{\mathrm{req}}$/$m_{\mathrm{req}}$ (M8/M9), all
  packet arrivals targeting the same file are conflated into a single
  coverage event.
  This is the primary file-identity coverage metric; for
  request-level completion accounting (computed from per-arrival
  identifiers, $H$-normalized) see the request-level metrics
  M8--M10.

  \item \textbf{Broadcast-Efficiency Score (BE-Score, $\uparrow$, secondary).}
  A composite metric that jointly rewards per-broadcast XOR degree and penalizes expirations:
  \begin{equation}
    \sigma = \frac{1}{H}\sum_{t=1}^{H}\bigl(U_t - \lambda E_t\bigr),
    \quad \lambda = 1,
    \label{eq:sys_score}
  \end{equation}
  with higher values preferred.
  Because $U_t$ is the packet-set-cardinality XOR degree of the
  broadcast at step~$t$ (counting distinct packet identities, not
  original singleton arrivals), $\sigma$ measures
  \emph{broadcast-slot packet utilization} rather than unique demand
  satisfaction.
  A policy that creates the deep merge chains can achieve a
  high~$\sigma$, even though some of the underlying demands are
  served redundantly. For this reason, $\sigma$ is reported as a
  secondary broadcast-efficiency metric, while the primary evaluation
  uses the demand-centric metrics $\rho$~(M1) and $\delta$~(M2). The
  model selection during the training uses~$\sigma$ because it matches
  the training reward (Remark~\ref{rem:training_vs_eval}). We note
  that the selection on $\sigma$ is \emph{not} the same as the
  selection on the primary metric $\rho$, i.e., the $\sigma$-selected
  checkpoint is the model that maximizes the broadcast-efficiency
  composite at the validation, while the headline $\rho$ results in
  Sec.~\ref{sec:results} are then read off that checkpoint rather
  than off a separately $\rho$-selected one. We adopt $\sigma$ because
  it (i) coincides with the shaped training surrogate, thereby
  producing stable validation rankings during the checkpoint sweep,
  and (ii) penalizes the expirations ($-E_t$) directly, so the
  $\sigma$-selection is not blind to the deadline compliance, even
  though it is not formally identical to the $\rho$-selection. A
  formal per-metric selection-criterion sensitivity sweep remains
  the future work.

  The composite score uses $\lambda{=}1$, while
  Appendix~\ref{sec:appendix:lambda_sensitivity} reports a sensitivity
  sweep over $\lambda\in\{0.5,1,2,3\}$. PPO retains the rank~$1$ among
  all the coded-multicast methods for $\lambda \geq 1$, while it ranks
  fourth at $\lambda{=}0.5$, where the low expiration penalty favors
  the aggressive always-merge baselines.

  \item \textbf{Served Packet-Set per Transmission ($\uparrow$).}
  Average broadcast-level packet-set cardinality delivered per
  transmission slot (the ``Served/Tx'' column in our tables):
  \begin{equation}
    \mu = \frac{1}{H}\sum_{t=1}^{H} U_t.
    \label{eq:served_per_tx}
  \end{equation}
  We retain the column header ``Served/Tx'' for table compactness, but
  $\mu$ is a packet-set count per slot, not a count of original
  arrivals served per slot; the request-level counterpart is
  $\eta_{\mathrm{req}}$ (M8), itself a per-step request-level rate.

  \item \textbf{Average XOR Degree ($\uparrow$).}
  Mean packet-set cardinality per \emph{coded} broadcast,
  restricted to time steps where a coded action is taken:
  \begin{equation}
    g = \frac{\displaystyle\sum_{t\,:\,\text{coded}} U_t}
             {\bigl|\{t : \text{coded action at } t\}\bigr|}.
    \label{eq:avg_coding_gain}
  \end{equation}
  We label this ``average XOR degree'' rather than ``coding gain'' to
  match the packet-set semantics of $U_t$; tables retain the legacy
  ``Coding Gain'' column header for continuity with prior literature.

  \item \textbf{Expirations per Episode ($\downarrow$).}
  Total count of distinct queue-entry expiration events per episode
  (the count of records in $\bigcup_t \mathcal{X}_t$, \emph{not} the
  packet-set mass $\sum_t E_t$):
  \begin{equation}
    \varepsilon = \sum_{t=1}^{H}\,|\mathcal{X}_t|.
    \label{eq:exp_per_ep}
  \end{equation}
  In the tables, ``Exp/Episode'' always refers to $\varepsilon$.
  Whenever the reward-decomposition discussion accompanying
  Fig.~\ref{fig:reward_decomp} or the appendix refers to
  ``expired packet-set mass,'' that is the distinct quantity
  $\sum_t E_t$ (the sum of $|f_r|$ over expired records), and not
  $\varepsilon$.

  \item \textbf{Unique Miss Ratio ($\downarrow$).}
  Complement of~$\delta$:
  \begin{equation}
    \rho^{\mathrm{uniq}} = 1 - \delta.
    \label{eq:unique_miss_ratio}
  \end{equation}

  \item \textbf{Request Timely-Throughput ($\uparrow$, supplementary).}
  Average number of \emph{newly completed} original request identifiers
  per step, under request-level accounting
  (Section~\ref{sec:model:request_accounting}; computed from the per-arrival
  identifier sets $\mathcal{A}(r)$ stamped at arrival time):
  \begin{equation}
    \eta_{\mathrm{req}} = \frac{1}{H}\sum_{t=1}^{H} |C_t|,
    \label{eq:req_eta}
  \end{equation}
  where $C_t$ is the set of request IDs completed at step~$t$ that
  were not previously completed or missed.

  \item \textbf{Request Miss Rate ($\downarrow$, supplementary).}
  Average number of newly missed original request identifiers per step,
  computed from the per-arrival identifier sets $\mathcal{A}(r)$:
  \begin{equation}
    m_{\mathrm{req}} = \frac{1}{H}\sum_{t=1}^{H} |M_t|.
    \label{eq:req_miss_rate}
  \end{equation}
  This is the finest-grained demand-centric metric: it measures the
  per-step rate at which \emph{original user arrivals} miss their
  deadlines, tracking each arrival through all merge aggregations
  (and is therefore distinct from both the broadcast-level $\rho$ and
  the file-identity $\rho^{\mathrm{uniq}}$).

  \item \textbf{Request Selection Score ($\uparrow$).}
  Composite request-level metric analogous to~$\sigma$ (M3):
  \begin{equation}
    \sigma_{\mathrm{req}} = \frac{1}{H}\sum_{t=1}^{H}
      \bigl(|C_t| - \lambda\,|M_t|\bigr),
    \quad \lambda = 1.
    \label{eq:req_sigma}
  \end{equation}
  The $\lambda$-sensitivity of $\sigma_{\mathrm{req}}$ is analyzed in
  Appendix~\ref{sec:appendix:req_lambda_sensitivity}.

  \item \textbf{Merge Rate (diagnostic).}
  Fraction of steps with at least one feasible merge pair at which the
  agent elects to merge:
  \begin{equation}
    m_{\text{rate}} = \frac{\bigl|\{t : a_t \in \text{coded},\,|\mathcal{M}_t|>0\}\bigr|}
                          {\bigl|\{t : |\mathcal{M}_t|>0\}\bigr|}.
    \label{eq:merge_rate}
  \end{equation}
  This metric is neither universally better high nor low, considering
  that it characterizes the policy's \emph{selectivity} rather than
  its performance.

  \item \textbf{Opportunity Rate (diagnostic).}
  Fraction of steps at which at least one feasible merge pair exists,
  measured along the queue trajectory \emph{visited by the evaluated
  policy}:
  \begin{equation}
    o_{\text{rate}} = \frac{\bigl|\{t : |\mathcal{M}_t|>0\}\bigr|}{H}.
    \label{eq:opp_rate}
  \end{equation}
  Even though $|\mathcal{M}_t|$ is a function of the queue state, the
  queue trajectory itself is shaped by the policy (i.e., every merge
  clears one slot and refills the other with a fresh arrival), so
  $o_{\text{rate}}$ is a \emph{visited-state} (i.e., policy-induced)
  diagnostic, not an exogenous environment constant. It is reported
  to contextualize the merge-rate values in the same row of each
  results table.

  \item \textbf{Reward per Step (diagnostic, training only).}
  Mean shaped reward per step; reported as a training diagnostic only and never used for method ranking (definition in Appendix~\ref{sec:appendix:exp_details}).

\end{enumerate}

For the $\tau$-Fit family, each metric is optimized independently,
i.e., the oracle threshold $\tau^*$ for a given metric is selected by
the exhaustive search on the validation split, while the resulting
policy is then evaluated on the holdout set. This procedure ensures
that each $\tau$-Fit variant is evaluated at its best for that metric,
thereby giving each $\tau$-Fit variant its strongest possible
appearance in our comparisons.

\textbf{Statistical reporting.}\label{sec:exp:stat}
All the methods share the same holdout seeds and the same initial
cache placements, i.e., a \emph{shared-seed paired design at the
seed/context level} rather than a strict common-random-numbers (CRN)
scheme, considering that within an episode the RNG draws are
sequential, so the policy-dependent refill timing causes the realized
arrival sequences to diverge across the methods after the first
divergent action (see the no-CRN caveat in
Sec.~\ref{sec:exp:metrics}). Every comparison is therefore computed
on the per-seed paired differences, with the understanding that the
pairing holds at the seed/context level rather than pathwise.
We report the mean $\pm$ $95\%$ CI across the seed-level means, while
for the main comparisons we report the $95\%$ percentile bootstrap CI
of the paired difference ($10{,}000$ resamples). The
paired-bootstrap procedure is as follows: for each of the $50$
evaluation seeds we compute the per-seed mean of each metric under
each method (i.e., averaging over the $200$ episodes/seed), form the
per-seed paired difference
$\Delta_s = m^{\text{PPO}}_s - m^{\text{baseline}}_s$, and report the
$95\%$ percentile-bootstrap CI of the across-seed average
$\bar\Delta$ over $10{,}000$ resamples (drawing with replacement from
the $\{\Delta_s\}_{s=1}^{50}$ values). Following the
\cite{henderson2018deep,agarwal2021deep}, we report the effect sizes
with the calibrated uncertainty rather than the binary significance
claims. Additional procedural details are in
Appendix~\ref{sec:appendix:exp_details}.

\subsection{Reproducibility}
\label{sec:exp:repro}

The implementation uses PyTorch and the
Stable-Baselines3~\cite{raffin2021stable} with the
MaskablePPO~\cite{sb3contrib2022maskable}. All the random number
generators are seeded deterministically, while the evaluation uses
the argmax action selection. The full source code will be released
upon acceptance. The complete reproducibility details, the hardware
specifications, and the computational costs are in
Appendix~\ref{sec:appendix:reproducibility}.

\section{Main Results: Uniform-Demand Benchmark}\label{sec:results}

We report results for the uniform-demand benchmark (Track~A). All values are mean $\pm$ 95\% CI across 50 holdout seeds $\times$ 200 episodes per seed (10\,000 episodes per cell); for key comparisons we also report bootstrap CIs of the paired difference (see Section~\ref{sec:exp:stat}). Primary evaluation uses the demand-centric metrics broadcast-packet expiration ratio~$\rho$ (M1) and distinct file-identity coverage~$\delta$ (M2), supplemented by the broadcast-efficiency score~$\sigma$ (M3), served/tx~$\mu$, coding gain~$g$, expirations~$\varepsilon$, and request-level views (M8--M10); the shaped training reward is excluded from comparison tables and is never used for method ranking. Throughout this section, $\rho$ and $\delta$ (the primary demand-centric metrics; $m_{\mathrm{req}}$ is supplementary, see Sec.~\ref{sec:exp:metrics}) are paired within-benchmark quantities at fixed $H{=}50$ with episode-end right-censoring; absolute values are simulator-specific (Sec.~\ref{sec:exp:metrics}). All headline numerical results in this section are reported for the \emph{$\sigma$-selected checkpoint} (Sec.~\ref{sec:method:selection}); a primary-metric-aligned selection rule is left as future work. The reported gains are also conditional on the one-slot-per-record unicast cost abstraction (A2), on the uniform representative-destination convention $k_{\mathrm{mg}}\!\sim\!\mathrm{Unif}\{k_i,k_j\}$ (Eq.~\eqref{eq:mg_dest}), and on the aggregate-size observation clip $\min(|f_r|,U_{\max})/U_{\max}$ at $U_{\max}=6$ (Sec.~\ref{sec:env:state}); see the first-order limitations paragraph in Sec.~\ref{sec:conclusion}.

\paragraph{Summary of findings.}
We summarize the main findings and the mechanisms that drive them
before presenting the detailed tables:
\begin{itemize}[leftmargin=*, nosep]
  \item \textbf{The method of selective merging achieves the lowest broadcast-packet expiration ratio among the coded multi-casting methods.}
    The results of our experiment show that the policy network outperforms the best baseline of the coded multi-casting method, i.e., SACM{++}, by $40.9\%$ with respect to the broadcast-packet expiration ratio
    ($\rho = 0.208$ vs.\ $0.352$; paired-bootstrap $95\%$ CI on the relative reduction:
    $[40.6\%, 41.5\%]$, computed from $\Delta\rho \in [-0.146, -0.143]$ in
    Table~\ref{tab:paired_bootstrap} normalized by $\rho_{\mathrm{SACM{++}}}=0.352$),
    while reducing the expirations by $50.7\%$
    ($14.17$ vs.\ $28.76$ per episode), i.e., the policy network achieves all of this without compromising the satisfaction of the users' demands.
    Among all methods, only ED-Unicast (which never codes) achieves a
    lower broadcast-packet expiration ratio ($\rho = 0.134$).
    The distinct file-identity coverage of the policy network is competitive, i.e., $\delta = 0.824$ for the policy network
    vs.\ $0.866$ for ED-Unicast and $0.797$ for SACM{++}.
    The policy network also achieves the highest
    \emph{composite} broadcast-efficiency score ($\sigma = 0.976$), i.e.,
    the method of selective merging lowers the broadcast-packet
    expiration ratio and improves the composite BE-score even though
    the raw Served/Tx ($\mu$) of the policy network is
    lower than that of the always-merge coded baselines (e.g., the policy network at $\mu=1.323$
    vs.\ SACM{++} at $\mu=1.590$ at ID-default).
    \emph{Why:} by restricting the coded transmissions to
    the high-intersection pairs (avg.\ $0.589$ vs.\ $0.393$ for SACM{++}), the
    policy network preserves the bandwidth headroom for unicasting the deadline-critical
    requests.

  \item \textbf{The policy network outperforms even the oracle-tuned thresholds.}
    The $\tau$-Fit family saturates at $\sigma{=}0.915$ (broadcast-efficiency score) regardless of
    the threshold, while the policy network outperforms this ceiling by $+6.61\%$ to $+11.25\%$
    with respect to the oracle's tuning criterion.  \emph{Why:} a fixed
    threshold commits to a static throughput--reliability tradeoff,
    whereas the policy network conditions on the instantaneous queue configuration and
    adapts the merge decision per step, i.e., it is not bound by the
    saturation curve.

  \item \textbf{The advantages of the policy network grow under stress, while the expirations show the largest relative improvement.}
    The broadcast-efficiency advantage of the policy network over SACM{++} grows from $+34.4\%$ at ID-default to
    $+79.0\%$ at high cache density ($p_c{=}0.40$) and an absolute
    $+0.449$ BE-score units under tight deadlines ($D{=}10$), i.e., the same observation is true for the variations of the cache fraction and deadline budgets as well.
    \emph{Why:} the abundant merge opportunities cause the blind-merge policies to
    flood the channel with low-quality coded packets, thereby producing the deadline
    cascades that the method of selective merging avoids.
    The policy network reduces the expirations by $50.7\%$ with respect to SACM{++}
    ($14.17$ vs.\ $28.76$ per episode), i.e., the single largest per-metric
    gain, which directly follows from the selective-merge mechanism, i.e.,
    fewer low-overlap merges means fewer wasted transmission slots and
    more timely deliveries before the deadlines expire.
\end{itemize}

\subsection{In-Distribution Performance (Track~A)}
\label{sec:res:id}

Table~\ref{tab:nm14_baselines} reports the six core methods on the ID-default
uniform-demand condition, sorted by broadcast-packet expiration ratio (\textbf{M1}) ascending;
the full comparison including $\tau$-Fit threshold rules is deferred
to Appendix Table~\ref{tab:nm14_full}; per-metric rankings of all
methods are provided in Appendix Table~\ref{tab:nm14_rankings}.
The results of our experiment show that the \textbf{PPO-Agent}, i.e., the policy network, achieves the lowest broadcast-packet expiration ratio among all the
coded multi-casting methods ($\rho = 0.208 \pm 0.001$), while reducing the deadline
misses by $40.9\%$ with respect to the best baseline of the coded multi-casting method, i.e., SACM{++} ($\rho = 0.352$).
Only ED-Unicast achieves a lower broadcast-packet expiration ratio ($\rho = 0.134$), while at
the cost of zero coding gain, i.e., the policy network outperforms ED-Unicast by $15.5\%$
on the broadcast-efficiency score ($\sigma = 0.976$ vs.\ $0.845$,
denominator ED-Unicast).
The distinct file-identity coverage of the policy network is $\delta = 0.824$, while below ED-Unicast's
$0.866$ (i.e., a structural consequence of unicast serving exactly one unique
request per slot with no merge-induced expirations) yet above all of the SACM
variants ($0.797$ for SACM{++}).
No baseline of the coded multi-casting method achieves higher~$\delta$ than the policy network.
The policy network also achieves the highest broadcast-efficiency score on
the ID-default Track~A regime ($\sigma = 0.976 \pm 0.002$, i.e., the highest
among all evaluated methods) and the highest average XOR degree
($g = 2.162$). Across the full Track~A non-ID battery, the policy network leads on
$\sigma$ among the coded multi-casting methods in all 8~regimes, and among
\emph{all} evaluated methods, the policy network leads in 7 of 8 Track~A regimes,
while ED-Unicast takes the top $\sigma$ at OOD-delay10
(Appendix Table~\ref{tab:uniform_ood_all_methods_sys}; ED-Unicast
$0.019$ vs.\ the policy network $-0.002$). The Zipf-track BE-score picture is different and
is discussed separately in Section~\ref{sec:results:zipf}.
Table~\ref{tab:paired_bootstrap} reports the paired bootstrap 95\% CIs of the
metric differences between the policy network (PPO-Agent) and two reference baselines, i.e., ED-Unicast (the conservative miss-ratio leader) and SACM{++} (the highest-throughput coded baseline and the BC/ExIt teacher), computed on the 50 shared holdout seeds (10\,000 resamples).
Against the best baseline of the coded multi-casting method, i.e., SACM{++}, the policy network achieves a significantly lower broadcast-packet expiration ratio
($\Delta\rho = -0.144\,[-0.146,\,-0.143]$), a higher distinct file-identity coverage
($\Delta\delta = +0.027\,[+0.026,\,+0.028]$), a $29.8\%$ lower request
miss rate ($\Delta m_{\mathrm{req}} = -0.097\,[-0.099,\,-0.095]$), and
a higher broadcast-efficiency score
($\Delta\sigma = +0.250\,[+0.248,\,+0.253]$).
The tradeoff is a lower request throughput
($\Delta\eta_{\mathrm{req}} = -0.217\,[-0.219,\,-0.215]$), which
follows from the always-merge strategy of SACM{++}.
Against ED-Unicast, the policy network trades a moderate increase in the broadcast-packet expiration ratio
($\Delta\rho = +0.074$) and a small decrease in the file-identity coverage
($\Delta\delta = -0.043$) for a substantially higher broadcast efficiency
($\Delta\sigma = +0.131$), i.e., all intervals exclude zero.

\begin{table}[t]
\centering
\caption{Paired bootstrap 95\% CI of metric differences
  (PPO-Agent minus baseline) on 50 shared holdout seeds
  (10\,000 resamples). Negative $\Delta\rho$,
  $\Delta m_{\mathrm{req}}$ / positive $\Delta\delta$,
  $\Delta\eta_{\mathrm{req}}$, $\Delta\sigma$ favor PPO-Agent.}
\label{tab:paired_bootstrap}
\begin{adjustbox}{max width=\columnwidth}
\begin{tabular}{lccccc}
\toprule
\textbf{Baseline} & $\Delta\rho$ (95\% CI) & $\Delta\delta$ (95\% CI) & $\Delta m_{\mathrm{req}}$ (95\% CI) & $\Delta\eta_{\mathrm{req}}$ (95\% CI) & $\Delta\sigma$ (95\% CI) \\
\midrule
ED-Unicast & $+0.074\,[+0.073,\,+0.075]$ & $-0.043\,[-0.044,\,-0.042]$ & $+0.074\,[+0.073,\,+0.076]$ & $+0.069\,[+0.068,\,+0.069]$ & $+0.131\,[+0.129,\,+0.133]$ \\
SACM{++}   & $-0.144\,[-0.146,\,-0.143]$ & $+0.027\,[+0.026,\,+0.028]$ & $-0.097\,[-0.099,\,-0.095]$ & $-0.217\,[-0.219,\,-0.215]$ & $+0.250\,[+0.248,\,+0.253]$ \\
\bottomrule
\end{tabular}
\end{adjustbox}
\end{table}

\begin{table*}[t]
\centering
\caption{In-distribution (ID-default), uniform demand. 50 holdout seeds $\times$ 200 episodes. Mean $\pm$ 95\% CI. Bold = best per column. Underlining is used \emph{only on the PPO-Agent row} to mark cells where PPO ranks second after an uncoded ED-Unicast best (Miss Ratio, Coverage $\delta$, Exp/Episode); other rows do not carry second-best annotations because the focus of the comparison is PPO's per-metric placement relative to the column leader. $^{\dagger}$Naderializadeh-style (learned): a pure flat-MLP MaskablePPO policy with no BC warm start, no Expert Iteration, and no curriculum, i.e., the closest reproducible analog of Naderializadeh \& Asghari~\cite{naderializadeh2019learning} adapted to our deadline-aware environment, trained for 4 seeds and selected via the same robust-advantage criterion as the headline. The action masking is retained for the numerical stability. The coverage~$\delta$ was not extracted in this evaluation run and is left blank.}
\label{tab:nm14_baselines}
\begin{adjustbox}{max width=\textwidth}
\begin{tabular}{lcccccc}
\toprule
\textbf{Method} & \textbf{Miss Ratio} (↓) & \textbf{Coverage $\delta$} (↑) & \textbf{BE-Score $\sigma$} (↑) & \textbf{Served/Tx} (↑) & \textbf{Coding Gain} (↑) & \textbf{Exp/Episode} (↓) \\
\midrule
ED-Unicast & $\mathbf{0.134 \pm 0.001}$ & $\mathbf{0.866 \pm 0.001}$ & $0.845 \pm 0.001$ & $1.000 \pm 0.000$ & --- & $\mathbf{7.74 \pm 0.06}$ \\
\textbf{PPO-Agent} & $\underline{0.208 \pm 0.001}$ & $\underline{0.824 \pm 0.001}$ & $\mathbf{0.976 \pm 0.002}$ & $1.323 \pm 0.002$ & $\mathbf{2.162 \pm 0.002}$ & $\underline{14.17 \pm 0.07}$ \\
Naderializadeh-style (learned)$^{\dagger}$ & $0.236 \pm 0.004$ & --- & $0.923 \pm 0.006$ & $1.337 \pm 0.008$ & $2.158 \pm 0.008$ & $15.75 \pm 0.34$ \\
GCM & $0.345 \pm 0.001$ & $0.784 \pm 0.001$ & $0.733 \pm 0.003$ & $1.549 \pm 0.002$ & $2.095 \pm 0.001$ & $28.42 \pm 0.15$ \\
SACM & $0.345 \pm 0.001$ & $0.782 \pm 0.001$ & $0.745 \pm 0.003$ & $1.575 \pm 0.002$ & $2.131 \pm 0.001$ & $28.64 \pm 0.13$ \\
SACM+ & $0.348 \pm 0.001$ & $0.787 \pm 0.001$ & $0.731 \pm 0.003$ & $1.572 \pm 0.002$ & $2.124 \pm 0.001$ & $28.71 \pm 0.14$ \\
SACM{++} & $0.352 \pm 0.001$ & $0.797 \pm 0.001$ & $0.726 \pm 0.003$ & $\mathbf{1.590 \pm 0.003}$ & $2.131 \pm 0.001$ & $28.76 \pm 0.14$ \\
\bottomrule
\end{tabular}
\end{adjustbox}
\end{table*}

\FloatBarrier

With respect to SACM{++}, i.e., the highest-throughput heuristic,
the policy network reduces the broadcast-packet expiration ratio by $40.9\%$ ($0.208$ vs.\ $0.352$)
and the expired-record count per episode ($\varepsilon$) by $50.7\%$ ($14.17$ vs.\ $28.76$),
while also improving the average XOR degree by $+1.5\%$
($2.162$ vs.\ $2.131$).
The single metric where the policy network does not lead is the served packet-set count per
transmission ($\mu = 1.323$ vs.\ $1.590$ for SACM{++}), i.e., the policy network transmits fewer coded
packets overall, while those coded packets are of substantially higher quality
and impose far fewer broadcast-packet expirations.

ED-Unicast achieves the highest distinct file-identity coverage ($\delta = 0.866$ vs.\ the policy network's $0.824$).
This advantage is structural-by-design rather than an artifact of the metric definition, i.e., unicast serves at most one new file identity per transmission and never incurs merge-induced expirations, which empirically yields the highest~$\delta$ in our evaluation. We avoid claiming that ED-Unicast maximizes~$\delta$ ``by construction'' (a unicast can contribute zero to $U_t^{\mathrm{uniq}}$ when its target file identity has already been covered earlier in the episode), while the absence of merge-induced expirations and the avoidance of multi-file aggregates explain why this empirical advantage persists across the regimes.
However, this unicast-only strategy forfeits all coding gain, i.e., the policy network outperforms ED-Unicast by $15.5\%$ on the broadcast-efficiency score ($\sigma = 0.976$ vs.\ $0.845$, denominator ED-Unicast).
The method of selective merging of the policy network trades a moderate $\delta$ reduction ($-4.9\%$ with respect to ED-Unicast) for substantially better channel utilization.
No baseline of the coded multi-casting method achieves higher~$\delta$ than the policy network.

A useful diagnostic is the merge rate, i.e., the policy network merges only $31.8\%$ of the
available opportunities, whereas all of the SACM variants merge $100\%$.
No single \emph{coded multi-casting} baseline achieves better values than the policy network across the
joint objective $\{\rho,\, \delta,\, \sigma,\, g\}$, i.e., SACM{++} wins on $\mu$ while
losing on $\rho$, $\delta$, and $\sigma$. The uncoded ED-Unicast baseline wins on $\rho$ and $\delta$ while losing on
$\sigma$ and $g$ (by definition, i.e., the coding gain~$g=0$ for any unicast-only scheme). This multi-objective dominance is visualized in
Fig.~\ref{fig:pareto}, i.e., the policy network (PPO-Agent) occupies the upper-left (most
favorable) region of the $\sigma$--$\rho$ trade-off space, with
a higher broadcast-efficiency score and a lower broadcast-packet expiration ratio than any point on the
$\tau$-Fit threshold sweep.
Fig.~\ref{fig:pareto_demand} shows the complementary demand-centric view, i.e.,
the policy network achieves both a higher~$\delta$ and a lower~$\rho$ than all of the SACM-family baselines.

\begin{figure}[t]
  \centering
  \includegraphics[width=\columnwidth]{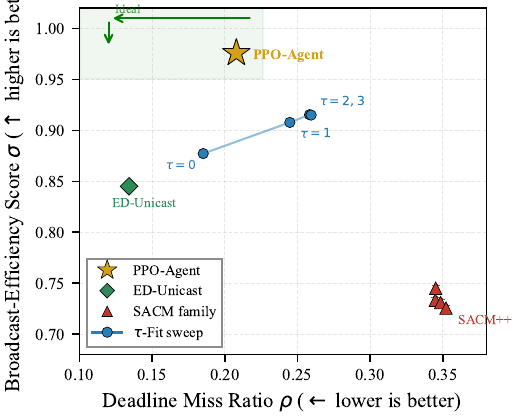}
  \caption{Multi-objective trade-off between broadcast-efficiency score~$\sigma$ ($\uparrow$)
    and broadcast-packet expiration ratio~$\rho$ ($\downarrow$) on the ID-default uniform
    benchmark (50 holdout seeds, 10\,000 episodes).  PPO-Agent (star) achieves
    both higher $\sigma$ and lower $\rho$ than any point on
    the $\tau$-Fit threshold sweep (connected circles), dominating the entire
    efficiency--compliance frontier.  Error bars: 95\% CI.  The SACM family (triangles)
    clusters at high $\rho$ / low $\sigma$; ED-Unicast (diamond) achieves
    low $\rho$ but at substantially lower $\sigma$.}
  \label{fig:pareto}
\end{figure}

\begin{figure}[t]
  \centering
  \includegraphics[width=\columnwidth]{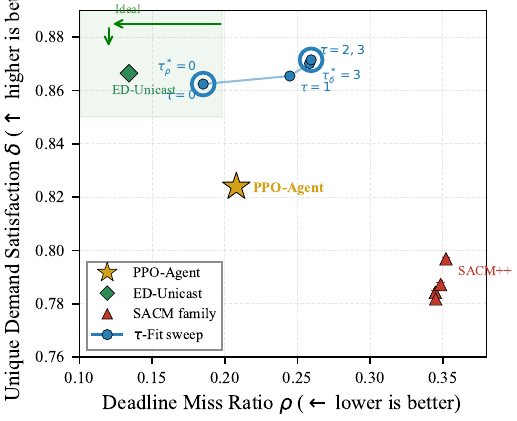}
  \caption{Demand-centric trade-off between distinct file-identity coverage~$\delta$ ($\uparrow$)
    and broadcast-packet expiration ratio~$\rho$ ($\downarrow$) on the ID-default uniform
    benchmark.
    PPO-Agent (star) achieves higher~$\delta$ and lower~$\rho$
    than all SACM-family baselines.  ED-Unicast (diamond) is the
    empirically strongest method on $\rho$ and $\delta$ in this
    benchmark; its structural advantage on $\delta$ comes from serving
    at most one new file identity per transmission and avoiding
    merge-induced expirations (see body discussion of
    Sec.~\ref{sec:res:id}).  Error bars: 95\% CI.}
  \label{fig:pareto_demand}
\end{figure}

\subsection{Oracle-Tuned Threshold Analysis}
\label{sec:res:taufit}

The $\tau$-Fit family saturates at $\sigma{=}0.915$ regardless of the threshold ($\tau \geq 2$), while the policy network outperforms this ceiling by $+6.61\%$ to $+11.25\%$ with respect to the oracle's tuning criterion (full $\tau$-sweep and per-metric oracle analysis in Appendix~\ref{sec:appendix:taufit}).
A fixed threshold commits to a static throughput--reliability tradeoff regardless of the queue state, whereas the policy network conditions on the instantaneous configuration and achieves both a higher $\sigma$ and a lower $\rho$ than any point on the $\tau$-sweep Pareto frontier.

\subsection{Transfer and OOD Generalization}
\label{sec:res:ood_uniform}

Table~\ref{tab:uniform_ood} presents the zero-shot transfer results across
all seven non-ID conditions (2~curriculum-seen, 5~OOD) for the uniform-demand track, i.e., the policy network is
evaluated even outside the region of the training data, without any retraining, hyperparameter tuning, or environmental
adaptation.

\begin{table*}[t]
\centering
\caption{Uniform demand, OOD generalization: PPO-Agent vs.~SACM{++}. Mean $\pm$ 95\% CI.}
\label{tab:uniform_ood}
\begin{adjustbox}{max width=\textwidth}
\begin{tabular}{lcccccccccc}
\toprule
 & \multicolumn{2}{c}{Miss Ratio (↓)} & \multicolumn{2}{c}{BE-Score $\sigma$ (↑)} & \multicolumn{2}{c}{Served/Tx (↑)} & \multicolumn{2}{c}{Coding Gain (↑)} & \multicolumn{2}{c}{Exp/Episode (↓)} \\
\textbf{Regime} & PPO & SACM{++} & PPO & SACM{++} & PPO & SACM{++} & PPO & SACM{++} & PPO & SACM{++} \\
\midrule
ID-default & $\mathbf{0.208 \pm 0.001}$ & $0.352 \pm 0.001$ & $\mathbf{0.976 \pm 0.002}$ & $0.726 \pm 0.003$ & $1.323 \pm 0.002$ & $\mathbf{1.590 \pm 0.003}$ & $\mathbf{2.162 \pm 0.002}$ & $2.131 \pm 0.001$ & $\mathbf{14.17 \pm 0.07}$ & $28.76 \pm 0.14$ \\
Curr-file60 & $\mathbf{0.208 \pm 0.001}$ & $0.352 \pm 0.001$ & $\mathbf{0.976 \pm 0.002}$ & $0.726 \pm 0.002$ & $1.324 \pm 0.002$ & $\mathbf{1.592 \pm 0.003}$ & $\mathbf{2.162 \pm 0.001}$ & $2.132 \pm 0.001$ & $\mathbf{14.22 \pm 0.09}$ & $28.76 \pm 0.11$ \\
OOD-file120 & $\mathbf{0.208 \pm 0.001}$ & $0.353 \pm 0.001$ & $\mathbf{0.976 \pm 0.002}$ & $0.725 \pm 0.003$ & $1.325 \pm 0.002$ & $\mathbf{1.593 \pm 0.002}$ & $\mathbf{2.163 \pm 0.002}$ & $2.132 \pm 0.001$ & $\mathbf{14.24 \pm 0.08}$ & $28.87 \pm 0.13$ \\
OOD-file150 & $\mathbf{0.209 \pm 0.001}$ & $0.353 \pm 0.001$ & $\mathbf{0.975 \pm 0.002}$ & $0.725 \pm 0.003$ & $1.324 \pm 0.001$ & $\mathbf{1.592 \pm 0.002}$ & $\mathbf{2.163 \pm 0.002}$ & $2.131 \pm 0.001$ & $\mathbf{14.28 \pm 0.08}$ & $28.83 \pm 0.12$ \\
OOD-pcache0.20 & $\mathbf{0.176 \pm 0.001}$ & $0.265 \pm 0.001$ & $\mathbf{0.919 \pm 0.002}$ & $0.828 \pm 0.002$ & $1.168 \pm 0.001$ & $\mathbf{1.296 \pm 0.002}$ & $\mathbf{2.070 \pm 0.002}$ & $2.048 \pm 0.001$ & $\mathbf{11.14 \pm 0.08}$ & $18.25 \pm 0.09$ \\
Curr-pcache0.40 & $\mathbf{0.238 \pm 0.001}$ & $0.412 \pm 0.001$ & $\mathbf{1.038 \pm 0.002}$ & $0.580 \pm 0.003$ & $1.511 \pm 0.002$ & $\mathbf{1.948 \pm 0.002}$ & $\mathbf{2.266 \pm 0.002}$ & $2.235 \pm 0.001$ & $\mathbf{17.35 \pm 0.08}$ & $39.72 \pm 0.12$ \\
OOD-delay10 & $\mathbf{0.500 \pm 0.000}$ & $0.555 \pm 0.000$ & $\mathbf{-0.002 \pm 0.002}$ & $-0.451 \pm 0.003$ & $1.465 \pm 0.002$ & $\mathbf{1.833 \pm 0.002}$ & $\mathbf{2.126 \pm 0.001}$ & $2.110 \pm 0.001$ & $\mathbf{60.22 \pm 0.12}$ & $81.48 \pm 0.15$ \\
OOD-delay30 & $\mathbf{0.064 \pm 0.001}$ & $0.186 \pm 0.002$ & $\mathbf{1.204 \pm 0.002}$ & $1.141 \pm 0.002$ & $1.293 \pm 0.002$ & $\mathbf{1.481 \pm 0.002}$ & $\mathbf{2.186 \pm 0.003}$ & $2.135 \pm 0.001$ & $\mathbf{3.68 \pm 0.05}$ & $10.93 \pm 0.11$ \\
\bottomrule
\end{tabular}
\end{adjustbox}
\end{table*}

Complete per-method breakdowns of broadcast-packet expiration ratio and broadcast-efficiency score
across all uniform-demand non-ID regimes are provided in
Appendix~\ref{sec:appendix:uniform_ood_detail}
(Tables~\ref{tab:uniform_ood_all_methods_miss}
and~\ref{tab:uniform_ood_all_methods_sys}).

\begin{table}[t]
\centering
\caption{PPO-Agent advantage over SACM{++} (Uniform demand).
  $\Delta = \text{PPO} - \text{SACM{++}}$, $\pm$ paired bootstrap 95\%
  CI half-width (10\,000 resamples, 50 shared holdout seeds).}
\label{tab:delta_uniform}
\begin{adjustbox}{max width=\columnwidth}
\begin{tabular}{lccccc}
\toprule
\textbf{Regime} & $\Delta$ BE-Score & $\Delta$ Miss Ratio & $\Delta$ Served/Tx & $\Delta$ Coding Gain & $\Delta$ Exp/Episode \\
\midrule
ID-default & $\mathbf{+0.250 \pm 0.003}$ & $\mathbf{-0.144 \pm 0.001}$ & $-0.267 \pm 0.003$ & $\mathbf{+0.031 \pm 0.002}$ & $\mathbf{-14.59 \pm 0.14}$ \\
Curr-file60 & $\mathbf{+0.249 \pm 0.003}$ & $\mathbf{-0.144 \pm 0.001}$ & $-0.268 \pm 0.003$ & $\mathbf{+0.030 \pm 0.002}$ & $\mathbf{-14.54 \pm 0.13}$ \\
OOD-file120 & $\mathbf{+0.251 \pm 0.003}$ & $\mathbf{-0.144 \pm 0.001}$ & $-0.268 \pm 0.002$ & $\mathbf{+0.031 \pm 0.002}$ & $\mathbf{-14.62 \pm 0.12}$ \\
OOD-file150 & $\mathbf{+0.250 \pm 0.003}$ & $\mathbf{-0.144 \pm 0.001}$ & $-0.268 \pm 0.002$ & $\mathbf{+0.032 \pm 0.002}$ & $\mathbf{-14.55 \pm 0.13}$ \\
OOD-pcache0.20 & $\mathbf{+0.091 \pm 0.002}$ & $\mathbf{-0.089 \pm 0.001}$ & $-0.127 \pm 0.002$ & $\mathbf{+0.022 \pm 0.002}$ & $\mathbf{-7.11 \pm 0.09}$ \\
Curr-pcache0.40 & $\mathbf{+0.458 \pm 0.003}$ & $\mathbf{-0.174 \pm 0.001}$ & $-0.436 \pm 0.003$ & $\mathbf{+0.031 \pm 0.003}$ & $\mathbf{-22.37 \pm 0.13}$ \\
OOD-delay10 & $\mathbf{+0.449 \pm 0.003}$ & $\mathbf{-0.054 \pm 0.000}$ & $-0.368 \pm 0.003$ & $\mathbf{+0.016 \pm 0.002}$ & $\mathbf{-21.26 \pm 0.16}$ \\
OOD-delay30 & $\mathbf{+0.062 \pm 0.003}$ & $\mathbf{-0.122 \pm 0.002}$ & $-0.189 \pm 0.002$ & $\mathbf{+0.052 \pm 0.002}$ & $\mathbf{-7.25 \pm 0.11}$ \\
\bottomrule
\end{tabular}
\end{adjustbox}
\end{table}

\begin{table}[t]
\centering
\caption{Distinct file-identity coverage~$\delta$ (uniform demand).
  Each distinct file identity counted at most once; coding-state record
  expirations of already-covered file identities contribute
  $E_t^{\mathrm{uniq}}{=}0$.
  Mean $\pm$ 95\% CI.}
\label{tab:unique_demand_uniform}
\begin{adjustbox}{max width=\columnwidth}
\begin{tabular}{lcc}
\toprule
\textbf{Regime} & \textbf{PPO-Agent $\delta$} (↑) & \textbf{SACM{++} $\delta$} (↑) \\
\midrule
ID-default      & $\mathbf{0.824 \pm 0.001}$ & $0.797 \pm 0.001$ \\
Curr-file60     & $\mathbf{0.823 \pm 0.001}$ & $0.796 \pm 0.001$ \\
OOD-file120     & $\mathbf{0.823 \pm 0.001}$ & $0.796 \pm 0.001$ \\
OOD-file150     & $\mathbf{0.822 \pm 0.001}$ & $0.797 \pm 0.001$ \\
OOD-pcache0.20  & $\mathbf{0.839 \pm 0.001}$ & $0.806 \pm 0.001$ \\
Curr-pcache0.40 & $\mathbf{0.817 \pm 0.001}$ & $0.806 \pm 0.001$ \\
OOD-delay10     & $0.560 \pm 0.001$ & $\mathbf{0.601 \pm 0.001}$ \\
OOD-delay30     & $\mathbf{0.944 \pm 0.001}$ & $0.907 \pm 0.001$ \\
\bottomrule
\end{tabular}
\end{adjustbox}
\end{table}

The policy network (PPO-Agent) achieves a higher~$\delta$ than SACM{++} in 7 of 8~regimes.
The single exception is \textbf{OOD-delay10} ($\delta = 0.560$ vs.\ $0.601$), i.e.,
under the tightest deadline budget ($D{=}10$), all methods suffer deadline
broadcast-packet expiration ratios near $50\%$, while the always-merge strategy of SACM{++} serves more
total request identifiers per step ($\eta_{\mathrm{req}} = 1.604$ vs.\
$1.240$), thereby translating into higher unique demand coverage even though
its miss rate is higher.
The policy network achieves a substantially lower broadcast-packet expiration ratio ($-5.5$~pp) and
fewer expirations ($60.22$ vs.\ $81.48$, i.e., a $26.1\%$ reduction), and a
substantially higher broadcast-efficiency score ($\sigma = -0.002$ vs.\
$-0.451$).
The $\delta$~loss at OOD-delay10 thus reflects a throughput--compliance
tradeoff that the method of selective merging of the policy network resolves in favor of the
deadline protection.

\paragraph{Multi-metric summary.}
The policy network (PPO-Agent) achieves the lowest broadcast-packet expiration ratio among the coded multi-casting
methods in all 8~uniform-demand regimes and the highest unique demand
satisfaction among the coded methods in 7 of 8~regimes (Table~\ref{tab:unique_demand_uniform}).
ED-Unicast (which never codes) achieves the lowest broadcast-packet expiration ratio and the highest
$\delta$ overall in most regimes, while at the cost of zero coding gain.
SACM{++} achieves a higher~$\delta$ than the policy network only at OOD-delay10 and a higher
request throughput in all regimes, i.e., an outcome of its always-merge strategy.
No single baseline dominates the policy network across all the primary demand-centric metrics ($\rho$ and $\delta$; the supplementary request-level $m_{\mathrm{req}}$ is reported separately in Sec.~\ref{sec:res:request_level}), i.e.,
ED-Unicast forfeits all of the broadcast efficiency, while the policy network outperforms SACM{++} by $40.9\%$ with respect to the
broadcast-packet expiration ratio~$\rho$
($\rho_{\mathrm{PPO}}{=}0.208$ vs.\ $\rho_{\mathrm{SACM{++}}}{=}0.352$, i.e.,
a $0.144$ absolute reduction off the SACM{++} baseline).
Under the request-level accounting (Section~\ref{sec:res:request_level}),
the policy network achieves a substantially lower request miss
rate~$m_{\mathrm{req}}$ than the SACM{++} comparator (i.e., the only
baseline in the main-text request-level table) in all 8~regimes, while
the paired-bootstrap analysis of Table~\ref{tab:paired_bootstrap}
shows that ED-Unicast remains the strongest deadline-protection
baseline on $m_{\mathrm{req}}$ ($\Delta m_{\mathrm{req}}=+0.074$ for
the policy network minus ED-Unicast at ID-default, i.e., the policy network is worse on this
metric), so we restrict the headline ``lower than'' phrasing to the
SACM{++} comparison and leave the broader extension to GCM, SACM,
SACM+, and the $\tau$-Fit family as future work. For the request
selection score~$\sigma_{\mathrm{req}}$, the policy network leads the
\emph{coded} baselines (SACM{++}) in the majority of regimes for
$\lambda \geq 2$ and in all eight regimes by $\lambda \geq 5$
(i.e., per-regime crossover thresholds in Table~\ref{tab:req_crossover},
with the worst case at $\lambda^{\star} = 4.92$ for OOD-delay10 and
$\lambda^{\star}=3.60$ for Curr-pcache0.40), while the uncoded ED-Unicast
baseline remains the highest $\sigma_{\mathrm{req}}$ at every
$\lambda$ shown on the ID-default sensitivity table of
Appendix~\ref{sec:appendix:req_lambda_sensitivity} (e.g., at
$\lambda=2,3,5$, ED-Unicast scores $0.690, 0.535, 0.226$ versus
the policy network's $0.610, 0.381, -0.077$), i.e., at the cost of zero coding gain.
The $\sigma_{\mathrm{req}}$ advantage of the policy network is therefore over the coded
baselines, not over the uncoded EDF policy.
The always-merge strategy of SACM{++} yields a higher per-step
throughput~$\eta_{\mathrm{req}}$ and a higher $\sigma_{\mathrm{req}}$ at
$\lambda{=}1$, while this advantage inverts once the deadline misses carry
a moderate penalty.
The policy network occupies a distinct position on the multi-objective
frontier, i.e., it is the only method that achieves a low broadcast-level
broadcast-packet expiration ratio~$\rho$, a competitive distinct file-identity
coverage~$\delta$, a low per-step request miss rate~$m_{\mathrm{req}}$
(i.e., per the per-step definition $m_{\mathrm{req}}{=}H^{-1}\sum_t|M_t|$,
not normalized by total arrivals), and a high broadcast efficiency~$\sigma$.

\subsection{Request-Level Accounting (Supplementary)}
\label{sec:res:request_level}

The request-level metrics provide the finest-grained user-centric
evaluation, tracking each original arrival through all merge aggregations
and crediting completion or miss exactly once per request identifier
(Section~\ref{sec:model:request_accounting}).
Table~\ref{tab:request_level_uniform} reports these supplementary request-level metrics
(M8--M10) for PPO-Agent and SACM{++} across all eight regimes.
Unlike the unique-demand satisfaction~$\delta$ (M2), which is a
\emph{ratio} metric, the request selection score
$\sigma_{\mathrm{req}} = H^{-1}\sum_t(|C_t| - \lambda|M_t|)$ is an
\emph{absolute} per-step quantity whose ranking depends on the penalty
weight~$\lambda$.

At $\lambda{=}1$, SACM{++} achieves a higher~$\sigma_{\mathrm{req}}$ because
its always-merge strategy produces more completed request identifiers per
step, however, this ranking inverts at $\lambda \geq 2.24$, i.e., the policy network's
${\sim}30\%$ lower request miss rate becomes the dominant factor when
the deadline violations are penalized more heavily.
For the operationally relevant penalty weights ($\lambda \ge 2.24$ on the
ID-default sensitivity analysis,
see Appendix~\ref{sec:appendix:req_lambda_sensitivity}), the policy network achieves
the highest $\sigma_{\mathrm{req}}$ \emph{among the coded baselines
shown here}, while the uncoded ED-Unicast baseline remains higher than every
coded method on $\sigma_{\mathrm{req}}$ at every $\lambda$ in the
ID-default table, i.e., the $\sigma_{\mathrm{req}}$ advantage reported
here is over the coded baselines, not over the EDF unicast policy.
This crossover shows that the ``best'' policy depends on the system
operator's valuation of the deadline compliance versus the throughput.

\begin{table}[t]
\centering
\caption{Request-level metrics (uniform demand).
  Each original arrival tracked through aggregation; completion and miss
  credited exactly once per request ID.
  Mean $\pm$ 95\% CI (50 seeds $\times$ 200 episodes).}
\label{tab:request_level_uniform}
\begin{adjustbox}{max width=\columnwidth}
\begin{tabular}{lcccccc}
\toprule
 & \multicolumn{2}{c}{$\eta_{\mathrm{req}}$ (↑)} & \multicolumn{2}{c}{$m_{\mathrm{req}}$ (↓)} & \multicolumn{2}{c}{$\sigma_{\mathrm{req}}\;(\lambda{=}1)$ (↑)} \\
\textbf{Regime} & PPO & SACM{++} & PPO & SACM{++} & PPO & SACM{++} \\
\midrule
ID-default      & $1.069 \pm 0.001$ & $\mathbf{1.285 \pm 0.002}$ & $\mathbf{0.229 \pm 0.001}$ & $0.326 \pm 0.001$ & $0.839 \pm 0.001$ & $\mathbf{0.959 \pm 0.002}$ \\
Curr-file60     & $1.069 \pm 0.001$ & $\mathbf{1.285 \pm 0.002}$ & $\mathbf{0.230 \pm 0.001}$ & $0.327 \pm 0.001$ & $0.839 \pm 0.001$ & $\mathbf{0.958 \pm 0.002}$ \\
OOD-file120     & $1.069 \pm 0.001$ & $\mathbf{1.285 \pm 0.002}$ & $\mathbf{0.231 \pm 0.001}$ & $0.329 \pm 0.002$ & $0.838 \pm 0.001$ & $\mathbf{0.957 \pm 0.002}$ \\
OOD-file150     & $1.068 \pm 0.001$ & $\mathbf{1.286 \pm 0.002}$ & $\mathbf{0.231 \pm 0.002}$ & $0.328 \pm 0.002$ & $0.837 \pm 0.002$ & $\mathbf{0.958 \pm 0.002}$ \\
OOD-pcache0.20  & $1.033 \pm 0.001$ & $\mathbf{1.117 \pm 0.001}$ & $\mathbf{0.198 \pm 0.001}$ & $0.269 \pm 0.001$ & $0.835 \pm 0.001$ & $\mathbf{0.848 \pm 0.001}$ \\
Curr-pcache0.40 & $1.115 \pm 0.001$ & $\mathbf{1.489 \pm 0.002}$ & $\mathbf{0.251 \pm 0.001}$ & $0.354 \pm 0.002$ & $0.865 \pm 0.002$ & $\mathbf{1.135 \pm 0.002}$ \\
OOD-delay10     & $1.240 \pm 0.001$ & $\mathbf{1.604 \pm 0.002}$ & $\mathbf{0.974 \pm 0.002}$ & $1.048 \pm 0.002$ & $0.266 \pm 0.003$ & $\mathbf{0.556 \pm 0.003}$ \\
OOD-delay30     & $1.030 \pm 0.000$ & $\mathbf{1.145 \pm 0.001}$ & $\mathbf{0.061 \pm 0.001}$ & $0.117 \pm 0.001$ & $0.969 \pm 0.001$ & $\mathbf{1.027 \pm 0.001}$ \\
\bottomrule
\end{tabular}
\end{adjustbox}
\end{table}

Two patterns are evident.
First, the policy network (PPO-Agent) achieves a \textbf{substantially lower request miss
rate than the SACM{++} comparator (i.e., the only baseline in the
main-text request-level table)} in every tested regime, with
reductions of $29.8\%$ at ID-default ($m_{\mathrm{req}} = 0.229$
vs.\ $0.326$ for SACM{++}) and ranging from $7.1\%$ (OOD-delay10) to
$48.0\%$ (OOD-delay30), while the paired-bootstrap analysis at
Table~\ref{tab:paired_bootstrap} shows that ED-Unicast remains the
strongest deadline-protection baseline on $m_{\mathrm{req}}$ at
ID-default ($\Delta m_{\mathrm{req}} = +0.074$ for the policy network minus
ED-Unicast), so the headline ``lower than'' statement is restricted
to the SACM{++} comparison, i.e., the broader extension to GCM, SACM,
SACM+, and the $\tau$-Fit family is left as future work and is not
directly supported by the displayed evidence.
Second, SACM{++} achieves a higher request timely-throughput~$\eta_{\mathrm{req}}$
in every regime, because its always-merge strategy transmits more coded
packets per step.
This crossover mirrors the broadcast-efficiency analysis
(Appendix~\ref{sec:appendix:lambda_sensitivity}), where the policy network retains
rank~1 among the coded methods for $\lambda \geq 1$.
A full $\lambda$-sensitivity sweep for $\sigma_{\mathrm{req}}$
is presented in Appendix~\ref{sec:appendix:req_lambda_sensitivity}.
Fig.~\ref{fig:pareto_request} visualizes this throughput--compliance
trade-off, i.e., the policy network (PPO-Agent) occupies the intermediate position with the
lowest request miss rate among the displayed comparators (the policy network and SACM{++} as the coded comparators, with ED-Unicast as the uncoded EDF reference point).

\begin{figure}[t]
  \centering
  \includegraphics[width=\columnwidth]{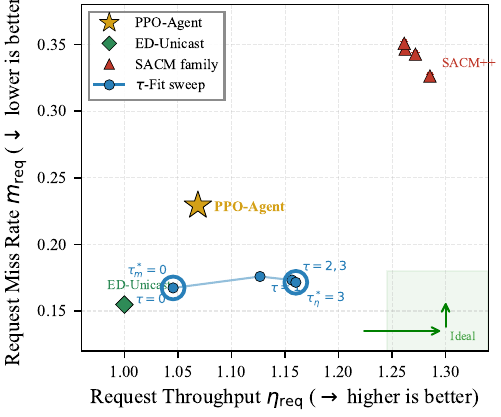}
  \caption{Request-level trade-off between request throughput~$\eta_{\mathrm{req}}$ ($\rightarrow$)
    and request miss rate~$m_{\mathrm{req}}$ ($\downarrow$) on the ID-default uniform benchmark.
    SACM{++} achieves the highest throughput at the cost of higher miss rates.
    PPO-Agent occupies the intermediate position with the lowest miss rate
    among coded-multicast methods.  Error bars: 95\% CI.}
  \label{fig:pareto_request}
\end{figure}

\FloatBarrier

Fig.~\ref{fig:ood_advantage} visualizes these advantages, including
the larger gains under stress conditions (Curr-pcache0.40, OOD-delay10).

\begin{figure}[t]
  \centering
  \includegraphics[width=\columnwidth]{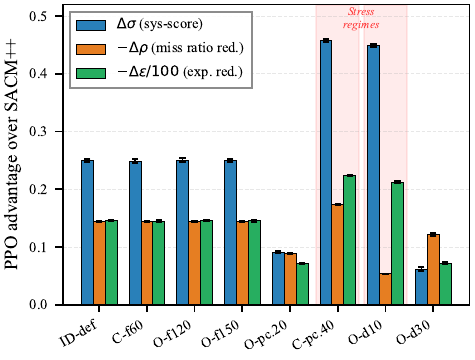}
  \caption{PPO-Agent advantage over SACM{++} across all evaluation regimes
    (uniform demand).  Bars show $\Delta\sigma$ (blue, BE-score improvement),
    $-\Delta\rho$ (orange, miss-ratio reduction), and
    $-\Delta\varepsilon/100$ (green, expiration reduction, scaled for
    readability).  All bars are positive, indicating PPO dominance.
    Shaded columns mark the stress regimes where advantages grow largest.
    Error bars: paired bootstrap 95\% CI.}
  \label{fig:ood_advantage}
\end{figure}

\FloatBarrier

\textbf{Catalog size (Curr-file60 / OOD-file120 / OOD-file150), sanity check.}
Varying the library from 60 to 150 files has a negligible effect on any
method's performance ($\sigma \in [0.975, 0.976]$ for the policy network and
$\sigma \approx 0.725$ for SACM{++}), i.e., the same observation is true for the variations of the number of files as well.
This near-invariance is expected, i.e., under decentralized placement with
fixed~$p_c$, the per-packet caching probability is preserved regardless
of~$N$, so the marginal side-information statistics that determine
the queue-state feature distributions do not change with the catalog size.
We therefore treat this axis as a \emph{sanity check} confirming
the feature invariance rather than a stress test.
A genuinely structural catalog-size shift would require altering the
placement model itself (e.g., correlated or capacity-constrained
caching) and is deferred to future work.

\textbf{Cache density (OOD-pcache0.20 and Curr-pcache0.40).}
These two regimes are the most informative. At \emph{low cache density}
(OOD-pcache0.20), the advantage of the policy network narrows, i.e., the policy network $\sigma {=} 0.919$ vs.\
SACM{++} $\sigma {=} 0.828$ ($+11.0\%$), i.e., fewer cached subfiles means fewer
feasible merge pairs at any step, while the blind-merge strategy of SACM{++} causes
fewer expirations ($18.25$ vs.\ $28.76$ at ID-default), thereby shrinking the benefit
of the method of selective merging. At \emph{high cache density} (Curr-pcache0.40),
the effect reverses, i.e., the policy network $\sigma {=} 1.038$ vs.\
SACM{++} $\sigma {=} 0.580$ (\textbf{$+79.0\%$}), i.e., with the abundant merge
opportunities, SACM{++} merges every step and incurs $39.72$
expirations/episode, while the policy network maintains its selective strategy at only
$17.35$ expirations/episode. The BE-score exceeds $1.0$ because the
loosened expiration rate allows the served packet-set mass ($\sum_t U_t$)
to accumulate faster than the episode horizon penalizes.

\textbf{Deadline budget (OOD-delay10 and OOD-delay30).}
Under the tightest deadline ($D{=}10$), \emph{all} methods suffer
large expirations, while the policy network's $\sigma {=} {-}0.002$
outperforms SACM{++}'s $\sigma {=} {-}0.451$, i.e., an absolute advantage of
$0.449$ BE-score units. SACM{++} incurs $81.48$ expirations/episode
vs.\ the policy network's $60.22$, i.e., a $26.1\%$ reduction even under extreme stress.
Under the loosest deadline ($D{=}30$), the advantage narrows to $+5.5\%$
($\sigma = 1.204$ vs.\ $1.141$), i.e., longer time budgets reduce the
penalty for aggressive merging, thereby partially recovering the performance of SACM{++}.

The interesting fact we observed is that the advantages of the policy network grow under
the conditions that \emph{penalize the blind merging} (i.e., high cache density,
tight deadlines) and shrink when those conditions relax, i.e., this is
the regime where the adaptive, state-dependent decision making
provides the most value.

\subsection{Emergent Selective Merge Strategy}
\label{sec:res:behavior}

To characterize the qualitative behavior of the learned policy, we
examine three diagnostic metrics alongside performance: \emph{merge rate}
(fraction of steps where a coded merge is selected, given at least one
feasible pair), \emph{opportunity rate} (fraction of steps where at least
one feasible pair exists), and \emph{average intersection per merge}
(mean $|\mathcal{S}_i \cap \mathcal{S}_j|$ of the selected pair, normalized
by $K-2$, computed only over merging steps; under the no-self-cache
generation rule and the one-gap invariant of
Appendix~\ref{sec:appendix:state_sufficiency}, the maximum admissible
overlap is $K-2$ because neither destination cache can sit inside
$\mathcal{S}_i \cap \mathcal{S}_j$).
Table~\ref{tab:diagnostics} reports these for the ID-default uniform
condition.

\begin{table}[t]
\centering
\caption{Policy behavior diagnostics, uniform demand, ID-default.
  Avg.\ intersection computed over merging steps only; higher values
  indicate selection of higher-quality pairs.}
\label{tab:diagnostics}
\begin{adjustbox}{max width=\columnwidth}
\begin{tabular}{lcccc}
\toprule
\textbf{Method} & \textbf{Merge Rate} & \textbf{Opp Rate} & \textbf{Avg Inter/Merge} \\
\midrule
\textbf{PPO-Agent} & $0.318 \pm 0.001$ & $0.876 \pm 0.002$ & $0.589 \pm 0.004$ \\
TauFit-2           & $0.406 \pm 0.001$ & $0.925 \pm 0.001$ & $0.359 \pm 0.003$ \\
Perfect-Fit        & $0.133 \pm 0.001$ & $0.959 \pm 0.001$ & $0.819 \pm 0.007$ \\
ED-Unicast         & $0.000$             & $0.965 \pm 0.001$ & --- \\
GCM                & $1.000$             & $0.501 \pm 0.002$ & $0.274 \pm 0.002$ \\
SACM{++}           & $1.000$             & $0.522 \pm 0.002$ & $0.393 \pm 0.002$ \\
\bottomrule
\end{tabular}
\end{adjustbox}
\end{table}

We make three observations.

\textbf{(1) The method of selective merging on the high-quality pairs.}
The policy network merges only $31.8\%$ of the opportunities, while achieving an average
intersection of $0.589$ per merge, i.e., well above SACM{++}'s $0.393$ and
TauFit-2's $0.359$, i.e., the policy network merges less often, and on substantially better
pairs. By restricting the coded transmissions
to the pairs with large mutual side-information overlap, the policy network preserves more
bandwidth headroom for unicasting the urgent requests before their deadlines
expire.
Fig.~\ref{fig:coding_gain_merge} visualizes this relationship, i.e., the policy network
achieves the highest per-merge coding gain at a merge rate
substantially below the always-merge coded baselines (GCM and
SACM{++} at $1.000$), while among the coded family Perfect-Fit merges even
less aggressively yet at the cost of a far higher
broadcast-level expiration ratio.

\begin{figure}[t]
  \centering
  \includegraphics[width=\columnwidth]{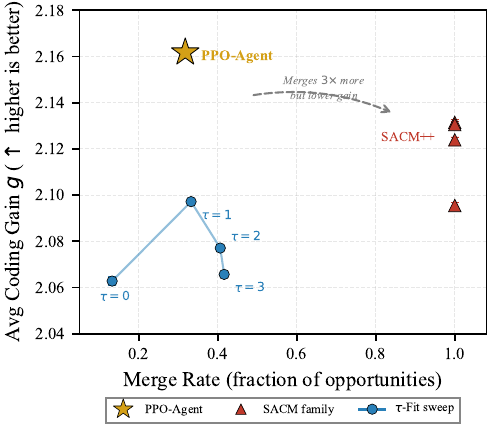}
  \caption{Coding gain vs.\ merge rate (ID-default, uniform demand).
    PPO-Agent merges only $31.8\%$ of opportunities yet achieves the
    highest average coding gain ($g{=}2.162$).  The SACM family merges
    every opportunity ($100\%$) but with lower per-merge quality.  The
    $\tau$-Fit sweep (connected circles) shows that relaxing the
    threshold increases merge rate but \emph{decreases} coding gain.}
  \label{fig:coding_gain_merge}
\end{figure}

Fig.~\ref{fig:reward_decomp} decomposes the per-episode \emph{training}
reward into its constituent terms (i.e., served packet-set mass,
expiration penalty, intersection bonus, union penalty). The policy network (PPO-Agent)
achieves the highest total reward ($54.6$) among all five policies,
thereby surpassing SACM{++} ($43.4$), ED-Unicast ($42.3$), and GCM ($41.7$).
The advantage stems from two sources, i.e., the policy network earns $50\%$ higher
intersection bonus \emph{per merge} ($0.44$ vs.\ $0.29$) and incurs
$59\%$ less expired packet-set mass per episode
($\sum_t E_t \approx 17.5$ vs.\ $43.2$, i.e., this is distinct from the
``Exp/Episode'' count $\varepsilon = 14.17$ vs.\ $28.76$ reported in
Table~\ref{tab:nm14_baselines}, which counts the expired
\emph{records} rather than the expired packet-set mass).
This indicates that the selective strategy of the policy network maximizes the reward through
the merge \emph{quality} and the deadline management rather than the merge volume.

\begin{figure}[t]
  \centering
  \includegraphics[width=\columnwidth]{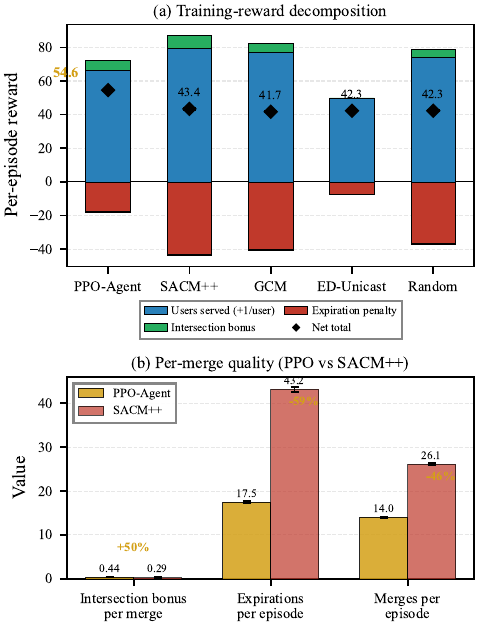}
  \caption{Training-reward decomposition (ID-default, uniform demand, 50 holdout seeds).
    \textbf{(a)}~Per-episode reward stacked by component (5 methods):
    PPO-Agent achieves the highest net reward ($54.6$) by accumulating
    more served packet-set mass while incurring less expired
    packet-set mass ($\sum_t E_t$) than always-merge baselines.
    \textbf{(b)}~Per-merge quality (PPO vs.\ SACM{++}): PPO achieves
    $50\%$ higher intersection bonus per merge and $59\%$ less
    expired packet-set mass per episode ($\sum_t E_t$), despite
    merging $46\%$ less often. ``Expired packet-set mass'' is the
    reward-side quantity $\sum_t E_t$ and is distinct from the
    ``Exp/Episode'' count $\varepsilon$ reported in
    Table~\ref{tab:nm14_baselines}.}
  \label{fig:reward_decomp}
\end{figure}

Fig.~\ref{fig:eval_comparison} visualises the per-episode reward
distribution across all 50 holdout seeds, i.e., the policy network (PPO-Agent) consistently
achieves higher per-episode rewards than every baseline, with clear
separation visible at the individual-episode level.

\begin{figure}[t]
  \centering
  \includegraphics[width=\columnwidth]{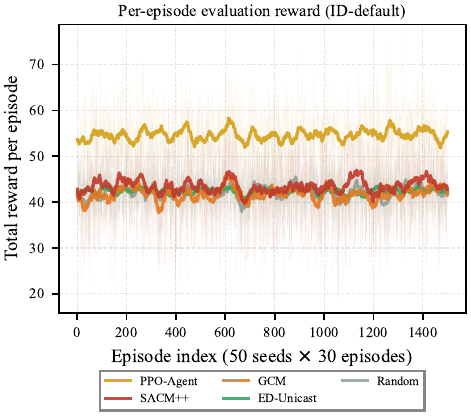}
  \caption{Per-episode total reward visualization on a 30-episode
    subset of the holdout protocol (50 holdout seeds $\times$ 30
    episodes; uniform demand, training reward config). The 30-episode
    subset is used here purely to keep the per-episode plot legible;
    all numerical headline metrics in this paper come from the full
    50 seeds $\times$ 200 episodes evaluation protocol of
    Sec.~\ref{sec:experiments}. Faint lines show raw values; bold
    lines show the rolling mean (window $=30$). PPO-Agent maintains a
    consistent advantage over all baselines on this subset.}
  \label{fig:eval_comparison}
\end{figure}

\textbf{(2) Deliberate unicast-first strategy.}
Even though the opportunity rate is $87.6\%$ (i.e., a feasible merge pair
is almost always available), the policy network chooses to unicast approximately $68\%$
of the time, i.e., this is not passive inaction due to infeasibility, it is
a learned preference for unicasting when the merge quality or
the deadline urgency does not justify the coding.

\textbf{(3) Comparison to Perfect-Fit.}
Perfect-Fit (TauFit-$\tau{=}0$) is the most conservative fixed-threshold
policy, merging only when $\misfit(r_i, r_j) = 0$ (i.e., exact side-information
match, i.e., the pairwise misfit metric $\misfit$ is defined in
Appendix~\ref{sec:appendix:exp_details}, distinct from the
broadcast-level expiration ratio $\rho$), i.e., it achieves a high average intersection ($0.819$) yet a
broadcast-efficiency score of only $0.877$, because it forgoes many beneficial merges
with slight misfit. The policy network merges $2.4\times$ more often ($31.8\%$ vs.\
$13.3\%$) with moderately lower pair quality, i.e., a balance that yields a
BE-score improvement of $+11.3\%$ over Perfect-Fit. No
fixed-$\tau$ rule can achieve this balance without the oracle access to $\tau^{\star}$,
while the policy network discovers it directly from the queue state.

\begin{remark}[Behavioral contrast across demand regimes]
A separately trained Zipf-demand agent (Section~\ref{sec:results:zipf})
merges $99.7\%$ of the available opportunities, compared to the
$31.8\%$ selective rate observed here, even though the cache placement is identical.
This demand-adaptive behavior shows that the architecture discovers
qualitatively different strategies for different demand regimes, i.e., a
flexibility no fixed-threshold policy can achieve.
\end{remark}

Taken together, the diagnostic evidence supports the following mechanistic
interpretation, i.e., the policy network has learned a \emph{selective merge
heuristic} (i.e., merge when the intersection is large and the deadlines are
comfortable, unicast otherwise) and applies it in a fully
state-dependent manner. The $\tau$-Fit family approximates this
heuristic with a fixed threshold, while the policy network instantiates it adaptively
for every queue configuration, with consistent performance
across the tested distributions, catalog sizes, cache densities, and
deadline budgets, i.e., the same observation is true for the variations of the number of files, cache fraction, and deadline budgets as well.

\section{Extension Study: Zipf-Demand Variant}\label{sec:results:zipf}

This section evaluates a separately trained Zipf-demand agent on the skewed demand distributions to show that the same architectural and training framework extends beyond the uniform demand. The policy network remains competitive with the baselines of the coded multi-casting method, while it does not dominate all methods, i.e., ED-Unicast achieves the highest overall broadcast-efficiency score ($0.847$ vs.\ the policy network's $0.732$) by avoiding the merge-induced expirations entirely. All values follow the same statistical reporting protocol as in Section~\ref{sec:results}.

\paragraph{Summary.} The interesting fact we observed is that the demand structure shapes the effective merge behavior, i.e., the uniform-demand policy network merges $31.8\%$ of the opportunities, whereas the Zipf-demand policy network merges $99.7\%$, even though the cache placement remains popularity-blind (i.e., uniform-without-replacement with fraction~$p_c$) in both cases. The descriptive diagnostic in Section~\ref{sec:res:ood_zipf} is consistent with a policy-mediated component to this difference, while we treat the analysis as suggestive rather than confirmatory because it is computed at a pooled-pair granularity (i.e., see the diagnostic's caveats in Section~\ref{sec:res:ood_zipf}).

\begin{figure}[t]
  \centering
  \includegraphics[width=\columnwidth]{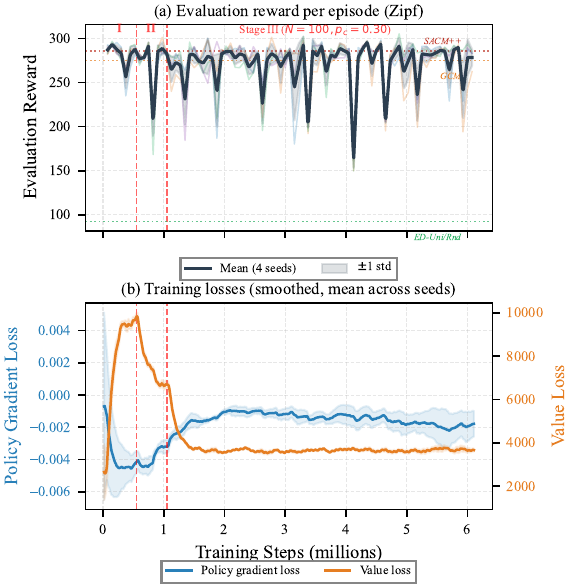}
  \caption{Training dynamics of the Track~B (Zipf-demand, $\alpha{=}0.8$) agent
    across 4 seeds. The x-axis shows total RL environment steps (6.05M per
    seed, including the 50K Phase-2 warm-up); curriculum stage transitions
    therefore appear at $550\text{K}$ and $1.05\text{M}$ total RL steps
    (corresponding to $T_{\mathrm{Phase\,3}} = 500\text{K}$ and
    $1.0\text{M}$ Phase-3 steps).
    \textbf{(a)}~Evaluation reward (mean $\pm$ std across seeds);
    vertical dashed lines mark curriculum stage transitions
    at $550\text{K}$ and $1.05\text{M}$ total RL steps. Horizontal dotted
    lines show baseline policy rewards evaluated under the same
    training reward configuration.
    \textbf{(b)}~Policy-gradient loss (left axis, blue) and value loss
    (right axis, orange), both smoothed and averaged across seeds.}
  \label{fig:training_curves_zipf}
\end{figure}

\subsection{OOD Generalization --- Zipf Demand (Track~B)}
\label{sec:res:ood_zipf}

Track~B evaluates a \emph{separately trained} Zipf-demand agent on
Zipf-distributed file requests (file-popularity law
$P(F=n) \propto \mathrm{rank}(n)^{-\alpha}$, $\alpha{=}0.8$).
This agent uses a popularity-aware observation augmentation
($d_{\text{req}}{=}14$, $d_{\text{pair}}{=}11$;
${\approx}1.75$M~parameters) and is trained on Zipf($\alpha{=}0.8$)
demand.
The Zipf-ID condition is in-distribution for this agent; the remaining
11~conditions are non-ID (2~curriculum-seen, 9~OOD).

Table~\ref{tab:zipf_id} shows the Zipf ID-default results. The main
cross-track observation is that the Zipf-trained policy network
merges $99.72\%$ of the available opportunities, compared with $31.8\%$ for
the uniform-trained policy network.
We caution that this comparison does \emph{not} isolate the demand law, i.e.,
the Track~B policy network uses a wider observation space than the Track~A
policy network (i.e., per-request popularity-mass feature, $d_{\text{req}} = 14$
versus $13$, and popularity features at the pair level, $d_{\text{pair}}
= 11$ versus $8$). The behavioral contrast therefore conflates the
demand law (uniform vs.\ Zipf), the feature augmentation
(popularity-aware features vs.\ none), and the joint training
distribution. A controlled ablation that holds the observation space
fixed and varies only the demand law is left for future work;
concretely, training a PPO-Uniform agent with the Track~B observation
shape and evaluating zero-shot on Zipf-ID (and the symmetric
experiment) would decouple the demand law from the observation
channel and the training distribution, even though the
unified-observation refactor required is outside the present scope.
The discussion below should be read as descriptive evidence under
matched-architecture, mismatched-feature configurations rather than
as a clean causal claim about the demand structure alone.
Because the cache placement is uniform-without-replacement with fraction~$p_c$,
each packet's marginal caching probability is identical regardless of
the file popularity, i.e., Zipf demand therefore does not alter the per-packet
side-information statistics, while Zipf demand does change which
file indices co-occur in the queue. We observe that the opportunity
rate drops from $0.876$ (uniform) to $0.528$ (Zipf), i.e., the Zipf-trained
policy network merges nearly every available pair, and the performance under
the near-full merging is comparable to that of the SACM family rather than
substantially worse. The precise mechanism linking the popularity-blind
placement and the Zipf request arrivals to a regime where
the near-full merging is advantageous has not been empirically
characterized, while we outline a diagnostic analysis below.

\paragraph{Diagnostic: visited-state intersection distributions.}
To move from the observational account above toward a mechanistic
explanation, a natural diagnostic is to compare the distribution of
candidate-pair intersection sizes $|\mathcal{S}_i \cap \mathcal{S}_j|/K$
in the queue states \emph{actually visited} by the two PPO agents (uniform-trained
and Zipf-trained) under their respective demand regimes.  Under uniform-without-replacement ($p_c{=}0.30$) placement, the
\emph{single-request marginal} cache-residency law is
packet-exchangeable, so the marginal distribution of the
side-information set $\mathcal{S}_r$ for one independent singleton
request is the same regardless of demand law. The same is \emph{not}
true at the two-request joint level: under skewed demand the same
packet (and therefore the same $\mathcal{S}_r$) is more likely to be
requested twice, so the \emph{unconditional} distribution of
$|\mathcal{S}_i\cap\mathcal{S}_j|$ for two requests is itself shifted
by the demand law, even before any policy effect. With the paper's
default $N{=}100$, $B{=}10$, the per-pair same-packet probability is
$1/(NB)=0.001$ under uniform demand versus
$\frac{1}{B}\sum_f p_f^2 \approx 0.0033$ under Zipf$(\alpha{=}0.8)$,
and same-packet pairs induce identical $\mathcal{S}_r$, so they push
mass toward the high-intersection tail. This is a baseline shift, not
a policy effect; the diagnostic we report below mixes this baseline
shift with the visited-state policy effect, and the
``policy-mediated'' interpretation in the SACM++ comparison should be
read as suggestive descriptive evidence rather than a clean
decomposition.
With that caveat, Zipf demand also changes which file--cache
combinations co-occur in the queue at decision time, potentially
shifting the \emph{conditional} (visited-state) intersection
distribution beyond the baseline.
Holdout evaluation traces (50~seeds $\times$ 200~episodes, logging all
candidate-pair intersections at every decision step) yield a mean
visited-state intersection of $0.058$ under uniform demand versus
$0.065$ under Zipf demand
(Kolmogorov--Smirnov statistic $D = 0.021$ at the pooled-pair
granularity; $n_{\text{unif}} \approx 1.59\text{M}$ and
$n_{\text{Zipf}} \approx 0.62\text{M}$ candidate-pair observations).
We report $D$ as a descriptive effect size only: candidate-pair
observations within a queue state, and across consecutive decision
steps within an episode, are heavily dependent, so a standard
two-sample KS test is mis-specified at this pooled granularity and a
nominal small-sample $p$-value would overstate the evidence. We
therefore avoid quoting a $p$-value here and read the diagnostic as
descriptive evidence rather than as a hypothesis test; a future
seed-level or episode-level reanalysis would be needed to support an
inferential claim at an independent unit.
The distribution does shift upward under Zipf demand, i.e., the fraction of
the zero-intersection pairs drops from $73.7\%$ to $71.6\%$, while the pairs with
$|\mathcal{S}_i \cap \mathcal{S}_j|/K \geq 0.4$ nearly double
($2.7\%$ $\to$ $4.1\%$). This suggests a concrete data-driven mechanism, i.e.,
the Zipf-induced file-index concentration produces queues in which
a larger fraction of the candidate pairs have high intersection,
so the near-full merging is less costly in terms of wasted
side-information even though each individual packet's caching
probability is unchanged.
Running the same diagnostic with the deterministic SACM++ baseline
(which fixes the policy and varies only the demand law) shows no
comparable shift in pooled effect size ($D = 0.002$). This is
consistent with, but does not by itself prove, a
policy-mediated component to the PPO shift, because the SACM++
comparison is invariant to all PPO-specific visited-state effects yet
already absorbs whatever baseline two-request joint shift the
same-packet calculation above would predict. We therefore read the
diagnostic as descriptive evidence that PPO under Zipf reshapes queue
composition toward higher-intersection states, while flagging that a
clean decomposition into (a) the demand-law baseline shift and (b)
the policy-induced visited-state shift would require either
conditioning on distinct requested packet IDs or a controlled
counterfactual rollout that we leave to future work.
We leave a full decomposition by queue depth and aggregate size
to future work.

\begin{table*}[t]
\centering
\caption{Zipf demand ($\alpha{=}0.8$), in-distribution. 50 seeds $\times$ 200 episodes. Bold = best.}
\label{tab:zipf_id}
\begin{adjustbox}{max width=\textwidth}
\begin{tabular}{lcccccccc}
\toprule
\textbf{Method} & \textbf{BE-Score $\sigma$} (↑) & \textbf{Miss Ratio} (↓) & \textbf{Served/Tx} (↑) & \textbf{Coding Gain} (↑) & \textbf{Exp/Episode} (↓) & \textbf{Merge Rate} (↑) & \textbf{Opp Rate} (↑) & \textbf{Coverage $\delta$} (↑) \\
\midrule
\textbf{PPO-Agent} & $0.732 \pm 0.002$ & $0.351 \pm 0.001$ & $\mathbf{1.590 \pm 0.003}$ & $2.122 \pm 0.002$ & $28.30 \pm 0.13$ & $0.997 \pm 0.000$ & $0.528 \pm 0.002$ & $0.804 \pm 0.001$ \\
ED-Unicast & $\mathbf{0.847 \pm 0.001}$ & $\mathbf{0.133 \pm 0.001}$ & $1.000 \pm 0.000$ & --- & $\mathbf{7.65 \pm 0.06}$ & $0.000 \pm 0.000$ & $\mathbf{0.964 \pm 0.001}$ & $\mathbf{0.869 \pm 0.001}$ \\
GCM & $0.737 \pm 0.002$ & $0.343 \pm 0.001$ & $1.544 \pm 0.002$ & $2.095 \pm 0.001$ & $28.14 \pm 0.11$ & $\mathbf{1.000 \pm 0.000}$ & $0.496 \pm 0.002$ & $0.784 \pm 0.001$ \\
SACM & $0.749 \pm 0.003$ & $0.343 \pm 0.001$ & $1.569 \pm 0.003$ & $2.130 \pm 0.001$ & $28.35 \pm 0.14$ & $\mathbf{1.000 \pm 0.000}$ & $0.504 \pm 0.002$ & $0.781 \pm 0.001$ \\
SACM+ & $0.738 \pm 0.003$ & $0.346 \pm 0.001$ & $1.565 \pm 0.003$ & $2.124 \pm 0.001$ & $28.31 \pm 0.14$ & $\mathbf{1.000 \pm 0.000}$ & $0.503 \pm 0.002$ & $0.787 \pm 0.001$ \\
SACM{++} & $0.728 \pm 0.002$ & $0.351 \pm 0.001$ & $1.587 \pm 0.003$ & $2.131 \pm 0.002$ & $28.61 \pm 0.10$ & $\mathbf{1.000 \pm 0.000}$ & $0.518 \pm 0.002$ & $0.795 \pm 0.001$ \\
SACM{++}-Pop & $0.731 \pm 0.003$ & $0.348 \pm 0.001$ & $1.571 \pm 0.002$ & $\mathbf{2.133 \pm 0.001}$ & $28.37 \pm 0.13$ & $\mathbf{1.000 \pm 0.000}$ & $0.504 \pm 0.002$ & $0.789 \pm 0.001$ \\
\bottomrule
\end{tabular}
\end{adjustbox}
\end{table*}

Fig.~\ref{fig:pareto_tradeoff_zipf} and Fig.~\ref{fig:pareto_demand_zipf}
visualize the multi-objective trade-offs under Zipf demand, i.e.,
while in the uniform track the methods spread out, all of the multicast methods cluster tightly
in the high-$\rho$ region, whereas ED-Unicast occupies the upper-left
corner alone, under the near-full merging regime.

\begin{figure}[t]
  \centering
  \includegraphics[width=\columnwidth]{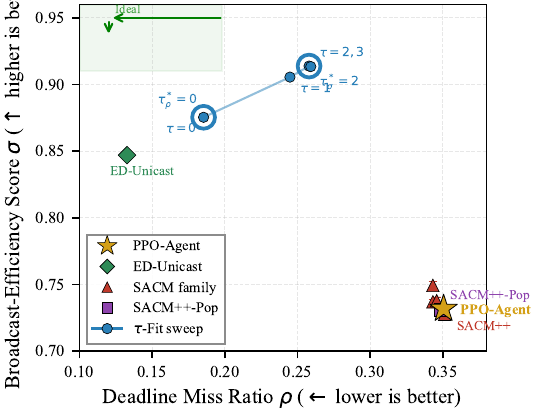}
  \caption{Multi-objective trade-off between broadcast-efficiency
    score~$\sigma$ ($\uparrow$) and broadcast-packet expiration ratio~$\rho$
    ($\downarrow$) under Zipf demand ($\alpha{=}0.8$, ID-default).
    Under near-full merging, all multicast methods cluster tightly;
    ED-Unicast dominates on both axes by avoiding merge-induced
    expirations.}
  \label{fig:pareto_tradeoff_zipf}
\end{figure}

\begin{figure}[t]
  \centering
  \includegraphics[width=\columnwidth]{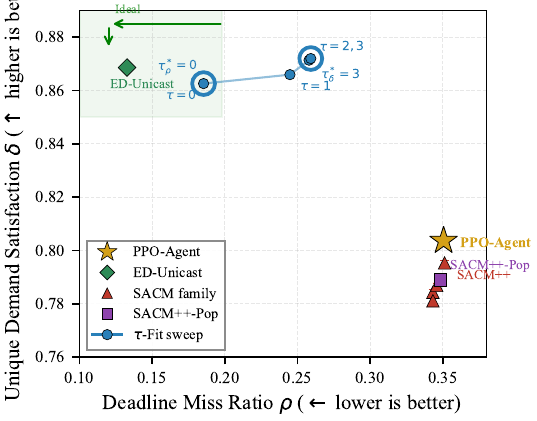}
  \caption{Distinct file-identity coverage~$\delta$ ($\uparrow$) vs.\
    broadcast-packet expiration ratio~$\rho$ ($\downarrow$) under Zipf demand.
    PPO-Agent achieves the highest~$\delta$ among multicast methods,
    while ED-Unicast leads overall.  SACM{++}-Pop does not consistently
    improve over SACM{++} despite privileged access to the demand
    distribution.}
  \label{fig:pareto_demand_zipf}
\end{figure}

Because all of the coded multi-casting methods merge near-maximally under Zipf
demand, the performance differences are smaller than in the uniform track.
The policy network retains a focused advantage in the
PPO-vs-SACM{++}-vs-SACM{++}-Pop comparison that motivates the Track~B
ablation, i.e., on the BE-score, the mean of the policy network is comparable to that of SACM{++} ($\sigma = 0.732$ vs.\ $0.728$,
i.e., the two $\pm 0.002$ CIs overlap, so we do not claim a strict ``outperforms'' on this metric)
and to the popularity-aware oracle SACM{++}-Pop, while the policy network attains a
competitive broadcast-packet expiration ratio at Zipf ID-default ($\rho = 0.351$,
i.e., tied with SACM{++}'s $0.351$ and within $+0.003$ of SACM{++}-Pop at
$0.348$), and incurs the fewest expirations per episode of those
three methods ($28.30$ vs.\ $28.61$ for SACM{++}).
We do not claim a global BE-score lead over the broader multicast
baseline set, i.e., SACM ($0.749$), SACM+ ($0.738$), and GCM ($0.737$) all
achieve a higher~$\sigma$ than the policy network ($0.732$) at Zipf ID-default
(Table~\ref{tab:zipf_id}), and the policy network instead trades the raw BE-score for
a lower broadcast-packet expiration ratio, fewer expirations, and the request-level metrics
discussed below.
ED-Unicast achieves the highest overall BE-score ($0.847$) by avoiding
the merge-induced expirations entirely, as in the uniform-demand case, i.e.,
its zero-merge strategy forgoes all coding gain.
SACM{++}-Pop, which uses the popularity-weighted pair scoring to exploit the
known Zipf distribution, does not consistently improve over SACM{++} on the
deadline metrics and does not outperform the policy network on the throughput, even though having
privileged access to the demand distribution unavailable to the deployed policy network.

\begin{figure}[t]
  \centering
  \includegraphics[width=\columnwidth]{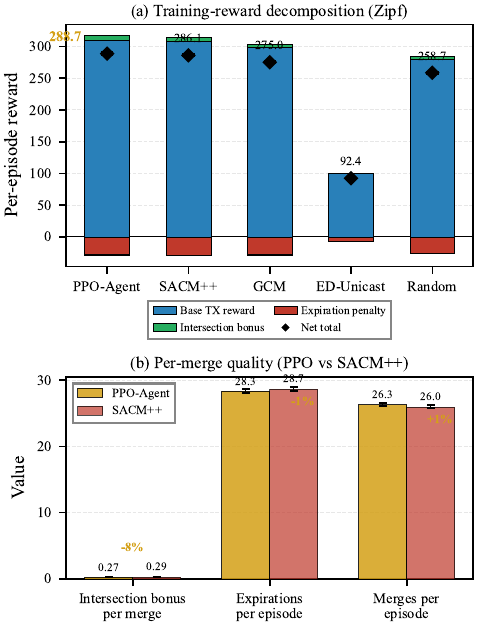}
  \caption{Training-reward decomposition (ID-default, Zipf demand, 50 holdout seeds).
    \textbf{(a)}~Per-episode reward stacked by component (5 methods):
    under near-full merging, PPO-Agent and SACM{++} achieve similar
    net rewards. By contrast, on the uniform track selective
    merging yields a clear PPO advantage.
    \textbf{(b)}~Per-merge quality (PPO vs.\ SACM{++}).}
  \label{fig:reward_decomp_zipf}
\end{figure}

\begin{figure}[t]
  \centering
  \includegraphics[width=\columnwidth]{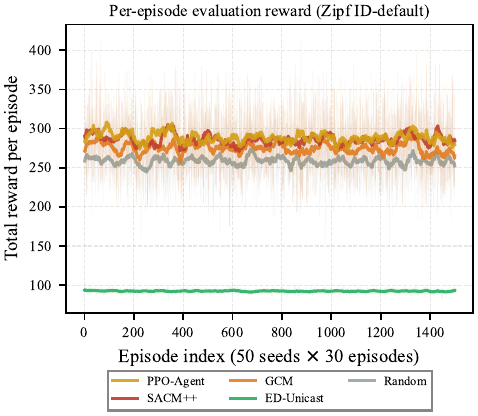}
  \caption{Per-episode total reward visualization on a 30-episode
    subset of the holdout protocol (50 holdout seeds $\times$ 30
    episodes; Zipf demand, training reward config). The 30-episode
    subset is used here purely to keep the per-episode plot legible;
    all numerical headline metrics in this paper come from the full
    50 seeds $\times$ 200 episodes evaluation protocol of
    Sec.~\ref{sec:experiments}. Faint lines show raw values; bold
    lines show the rolling mean (window $=30$).}
  \label{fig:eval_comparison_zipf}
\end{figure}

To enable direct cross-track comparison, Table~\ref{tab:zipf_sys_score}
reports the composite broadcast-efficiency score $\sigma$ for the three primary methods
across all Zipf regimes, computed from the same evaluation data using
$\sigma = H^{-1}\sum_t (U_t - \lambda E_t)$ (Section~\ref{sec:exp:metrics};
$H$ is the episode horizon, not the global Phase-3 step counter $T$ of
Sec.~\ref{sec:method:ppo}).
The policy network achieves the highest BE-score among the three methods in most
regimes, while in OOD-alpha0.6 the popularity-aware baseline SACM{++}-Pop
edges the policy network by $0.001$ (i.e., Pop $= 0.730$ vs.\ the policy network $= 0.729$), and in
OOD-alpha1.0 the two methods tie at $0.732$ (i.e., within 95\% CI of each
other). The largest advantages of the policy network occur under high cache density
(i.e., Curr-pcache0.40: $\sigma_{\text{PPO}} = 0.608$ vs.\
$\sigma_{\text{SACM++}} = 0.583$, $+4.3\%$)
and tight deadlines
(i.e., OOD-delay10: $\sigma_{\text{PPO}} = -0.426$ vs.\
$\sigma_{\text{SACM++}} = -0.449$, $+0.023$ absolute), i.e., the same observation is true for the variations of the cache fraction and deadline budgets as well.
As in the uniform track, the catalog-size axis (Curr-file60 / OOD-file120 / OOD-file150)
shows near-invariance ($\sigma \in [0.730, 0.733]$ for the policy network), i.e.,
consistent with the preserved per-packet caching probability under
the uniform-without-replacement ($p_c$) placement
(cf.\ Section~\ref{sec:res:ood_uniform}).

\begin{table*}[t]
\centering
\caption{Zipf demand, BE-score summary: PPO-Agent vs.~SACM{++} vs.~SACM{++}-Pop. Mean $\pm$ 95\% CI. Bold = best per regime; ties (within rounding to 3 decimal places) are bolded jointly.}
\label{tab:zipf_sys_score}
\begin{adjustbox}{max width=\textwidth}
\begin{tabular}{lccc}
\toprule
 & \multicolumn{3}{c}{BE-Score $\sigma$ (↑)} \\
\textbf{Regime} & PPO & SACM{++} & Pop \\
\midrule
ID-default ($\alpha{=}0.8$) & $\mathbf{0.732 \pm 0.002}$ & $0.728 \pm 0.002$ & $0.731 \pm 0.003$ \\
OOD-alpha0.6 & $0.729 \pm 0.003$ & $0.725 \pm 0.003$ & $\mathbf{0.730 \pm 0.003}$ \\
OOD-alpha1.0 & $\mathbf{0.732 \pm 0.002}$ & $0.728 \pm 0.003$ & $\mathbf{0.732 \pm 0.003}$ \\
OOD-alpha1.2 & $\mathbf{0.739 \pm 0.003}$ & $0.732 \pm 0.003$ & $0.733 \pm 0.002$ \\
OOD-mandelbrot & $\mathbf{0.732 \pm 0.002}$ & $0.727 \pm 0.003$ & $0.730 \pm 0.003$ \\
Curr-file60 (Z) & $\mathbf{0.730 \pm 0.003}$ & $0.728 \pm 0.003$ & $0.729 \pm 0.002$ \\
OOD-file120 (Z) & $\mathbf{0.731 \pm 0.003}$ & $0.727 \pm 0.003$ & $0.730 \pm 0.003$ \\
OOD-file150 (Z) & $\mathbf{0.733 \pm 0.003}$ & $0.728 \pm 0.003$ & $0.730 \pm 0.003$ \\
OOD-pcache0.20 (Z) & $\mathbf{0.829 \pm 0.002}$ & $0.826 \pm 0.002$ & $0.821 \pm 0.002$ \\
Curr-pcache0.40 (Z) & $\mathbf{0.608 \pm 0.003}$ & $0.583 \pm 0.004$ & $0.597 \pm 0.003$ \\
OOD-delay10 (Z) & $\mathbf{-0.426 \pm 0.003}$ & $-0.449 \pm 0.003$ & $-0.445 \pm 0.003$ \\
OOD-delay30 (Z) & $\mathbf{1.146 \pm 0.002}$ & $1.143 \pm 0.002$ & $1.143 \pm 0.002$ \\
\bottomrule
\end{tabular}
\end{adjustbox}
\end{table*}

Across the 11 non-ID Zipf conditions (Table~\ref{tab:zipf_ood}),
the policy network reduces the expirations per episode with respect to SACM{++} in every regime,
with the largest absolute reductions under the high-stress conditions, i.e.,
OOD-delay10 (the policy network $79.65$ vs.\ SACM{++} $81.30$, i.e., $-1.65$ expirations/ep)
and Curr-pcache0.40 (the policy network $38.04$ vs.\ SACM{++} $39.45$, i.e., $-1.41$ expirations/ep).
SACM{++}-Pop does not improve over SACM{++} on the OOD-delay10 stress regime, and improves only marginally on Curr-pcache0.40 (i.e., BE-score $0.597$ vs.\ $0.583$, broadcast-packet expiration ratio $0.408$ vs.\ $0.411$, expirations $39.41$ vs.\ $39.45$), i.e., the popularity-aware refinement does not deliver a consistent advantage under the stress.

\begin{table*}[t]
\centering
\caption{Zipf demand, OOD generalization: PPO-Agent vs.~SACM{++} vs.~SACM{++}-Pop. Bold = best per regime, assigned on unrounded internal means (so two cells with identical displayed precision may still rank differently).}
\label{tab:zipf_ood}
\begin{adjustbox}{max width=\textwidth}
\begin{tabular}{lcccccccccccc}
\toprule
 & \multicolumn{3}{c}{Miss Ratio (↓)} & \multicolumn{3}{c}{Served/Tx (↑)} & \multicolumn{3}{c}{Coding Gain (↑)} & \multicolumn{3}{c}{Exp/Episode (↓)} \\
\textbf{Regime} & PPO & SACM{++} & Pop & PPO & SACM{++} & Pop & PPO & SACM{++} & Pop & PPO & SACM{++} & Pop \\
\midrule
ID-default & $0.351 \pm 0.001$ & $0.351 \pm 0.001$ & $\mathbf{0.348 \pm 0.001}$ & $\mathbf{1.590 \pm 0.003}$ & $1.587 \pm 0.003$ & $1.571 \pm 0.002$ & $2.122 \pm 0.002$ & $2.131 \pm 0.002$ & $\mathbf{2.133 \pm 0.001}$ & $\mathbf{28.30 \pm 0.13}$ & $28.61 \pm 0.10$ & $28.37 \pm 0.13$ \\
OOD-alpha0.6 & $0.352 \pm 0.001$ & $0.352 \pm 0.001$ & $\mathbf{0.349 \pm 0.001}$ & $\mathbf{1.594 \pm 0.003}$ & $1.589 \pm 0.002$ & $1.577 \pm 0.003$ & $2.122 \pm 0.001$ & $2.131 \pm 0.001$ & $\mathbf{2.136 \pm 0.001}$ & $\mathbf{28.49 \pm 0.13}$ & $28.71 \pm 0.13$ & $28.57 \pm 0.13$ \\
OOD-alpha1.0 & $0.350 \pm 0.001$ & $0.350 \pm 0.001$ & $\mathbf{0.348 \pm 0.001}$ & $\mathbf{1.586 \pm 0.002}$ & $1.580 \pm 0.002$ & $1.566 \pm 0.003$ & $2.121 \pm 0.001$ & $2.129 \pm 0.001$ & $\mathbf{2.131 \pm 0.002}$ & $\mathbf{28.19 \pm 0.13}$ & $28.45 \pm 0.13$ & $28.22 \pm 0.16$ \\
OOD-alpha1.2 & $0.347 \pm 0.001$ & $0.348 \pm 0.001$ & $\mathbf{0.346 \pm 0.001}$ & $\mathbf{1.575 \pm 0.003}$ & $1.571 \pm 0.003$ & $1.560 \pm 0.003$ & $2.119 \pm 0.001$ & $2.127 \pm 0.001$ & $\mathbf{2.130 \pm 0.001}$ & $\mathbf{27.72 \pm 0.12}$ & $28.08 \pm 0.13$ & $27.97 \pm 0.13$ \\
OOD-mandelbrot & $0.350 \pm 0.001$ & $0.351 \pm 0.001$ & $\mathbf{0.348 \pm 0.001}$ & $\mathbf{1.587 \pm 0.003}$ & $1.580 \pm 0.003$ & $1.566 \pm 0.003$ & $2.120 \pm 0.001$ & $2.128 \pm 0.001$ & $\mathbf{2.131 \pm 0.001}$ & $\mathbf{28.19 \pm 0.13}$ & $28.44 \pm 0.14$ & $28.23 \pm 0.14$ \\
Curr-file60 & $0.351 \pm 0.001$ & $0.351 \pm 0.001$ & $\mathbf{0.349 \pm 0.001}$ & $\mathbf{1.590 \pm 0.003}$ & $1.584 \pm 0.002$ & $1.571 \pm 0.002$ & $2.121 \pm 0.001$ & $2.129 \pm 0.001$ & $\mathbf{2.131 \pm 0.001}$ & $\mathbf{28.36 \pm 0.12}$ & $28.52 \pm 0.12$ & $28.39 \pm 0.13$ \\
OOD-file120 & $0.351 \pm 0.001$ & $0.352 \pm 0.001$ & $\mathbf{0.349 \pm 0.001}$ & $\mathbf{1.593 \pm 0.003}$ & $1.590 \pm 0.003$ & $1.573 \pm 0.002$ & $2.122 \pm 0.001$ & $\mathbf{2.132 \pm 0.001}$ & $2.132 \pm 0.001$ & $\mathbf{28.40 \pm 0.14}$ & $28.70 \pm 0.14$ & $28.45 \pm 0.14$ \\
OOD-file150 & $0.351 \pm 0.001$ & $0.351 \pm 0.001$ & $\mathbf{0.349 \pm 0.001}$ & $\mathbf{1.594 \pm 0.003}$ & $1.587 \pm 0.002$ & $1.572 \pm 0.003$ & $2.122 \pm 0.002$ & $2.130 \pm 0.001$ & $\mathbf{2.133 \pm 0.001}$ & $\mathbf{28.37 \pm 0.14}$ & $28.59 \pm 0.13$ & $28.40 \pm 0.14$ \\
OOD-pcache0.20 & $\mathbf{0.265 \pm 0.001}$ & $0.266 \pm 0.001$ & $0.267 \pm 0.001$ & $1.294 \pm 0.002$ & $\mathbf{1.294 \pm 0.002}$ & $1.289 \pm 0.001$ & $2.043 \pm 0.001$ & $2.048 \pm 0.001$ & $\mathbf{2.049 \pm 0.001}$ & $\mathbf{18.23 \pm 0.10}$ & $18.30 \pm 0.11$ & $18.25 \pm 0.10$ \\
Curr-pcache0.40 & $\mathbf{0.408 \pm 0.001}$ & $0.411 \pm 0.001$ & $0.408 \pm 0.001$ & $\mathbf{1.964 \pm 0.003}$ & $1.936 \pm 0.003$ & $1.914 \pm 0.003$ & $2.234 \pm 0.002$ & $2.231 \pm 0.001$ & $\mathbf{2.238 \pm 0.001}$ & $\mathbf{38.04 \pm 0.12}$ & $39.45 \pm 0.16$ & $39.41 \pm 0.14$ \\
OOD-delay10 & $\mathbf{0.552 \pm 0.000}$ & $0.555 \pm 0.000$ & $0.555 \pm 0.000$ & $\mathbf{1.845 \pm 0.002}$ & $1.827 \pm 0.002$ & $1.801 \pm 0.002$ & $2.098 \pm 0.001$ & $2.108 \pm 0.001$ & $\mathbf{2.114 \pm 0.001}$ & $\mathbf{79.65 \pm 0.13}$ & $81.30 \pm 0.11$ & $81.59 \pm 0.12$ \\
OOD-delay30 & $0.183 \pm 0.001$ & $0.184 \pm 0.001$ & $\mathbf{0.181 \pm 0.001}$ & $\mathbf{1.479 \pm 0.002}$ & $1.477 \pm 0.002$ & $1.469 \pm 0.002$ & $2.126 \pm 0.001$ & $2.134 \pm 0.001$ & $\mathbf{2.136 \pm 0.002}$ & $\mathbf{10.63 \pm 0.07}$ & $10.73 \pm 0.08$ & $10.63 \pm 0.08$ \\
\bottomrule
\end{tabular}
\end{adjustbox}
\end{table*}

The complete per-method miss-ratio breakdown for all Zipf-demand
non-ID regimes is provided in Appendix
Table~\ref{tab:zipf_ood_all_methods_miss}.

\begin{table}[t]
\centering
\caption{PPO-Agent advantage over SACM{++} (Zipf demand). $\Delta = \text{PPO} - \text{SACM{++}}$.}
\label{tab:delta_zipf}
\begin{adjustbox}{max width=\columnwidth}
\begin{tabular}{lccccc}
\toprule
\textbf{Regime} & $\Delta$ BE-Score & $\Delta$ Miss Ratio & $\Delta$ Served/Tx & $\Delta$ Coding Gain & $\Delta$ Exp/Episode \\
\midrule
ID-default & $\mathbf{+0.004}$ & $\mathbf{-0.001}$ & $\mathbf{+0.004}$ & $-0.009$ & $\mathbf{-0.312}$ \\
OOD-alpha0.6 & $\mathbf{+0.004}$ & $\mathbf{-0.000}$ & $\mathbf{+0.006}$ & $-0.009$ & $\mathbf{-0.224}$ \\
OOD-alpha1.0 & $\mathbf{+0.005}$ & $\mathbf{-0.001}$ & $\mathbf{+0.006}$ & $-0.008$ & $\mathbf{-0.261}$ \\
OOD-alpha1.2 & $\mathbf{+0.007}$ & $\mathbf{-0.001}$ & $\mathbf{+0.003}$ & $-0.008$ & $\mathbf{-0.361}$ \\
OOD-mandelbrot & $\mathbf{+0.005}$ & $\mathbf{-0.001}$ & $\mathbf{+0.006}$ & $-0.007$ & $\mathbf{-0.243}$ \\
Curr-file60 & $\mathbf{+0.002}$ & $+0.000$ & $\mathbf{+0.007}$ & $-0.008$ & $\mathbf{-0.162}$ \\
OOD-file120 & $\mathbf{+0.004}$ & $\mathbf{-0.001}$ & $\mathbf{+0.003}$ & $-0.010$ & $\mathbf{-0.303}$ \\
OOD-file150 & $\mathbf{+0.005}$ & $\mathbf{-0.000}$ & $\mathbf{+0.007}$ & $-0.008$ & $\mathbf{-0.224}$ \\
OOD-pcache0.20 & $\mathbf{+0.003}$ & $\mathbf{-0.001}$ & $+0.000$ & $-0.005$ & $\mathbf{-0.061}$ \\
Curr-pcache0.40 & $\mathbf{+0.025}$ & $\mathbf{-0.003}$ & $\mathbf{+0.028}$ & $\mathbf{+0.003}$ & $\mathbf{-1.412}$ \\
OOD-delay10 & $\mathbf{+0.022}$ & $\mathbf{-0.003}$ & $\mathbf{+0.019}$ & $-0.010$ & $\mathbf{-1.648}$ \\
OOD-delay30 & $\mathbf{+0.003}$ & $\mathbf{-0.001}$ & $\mathbf{+0.002}$ & $-0.008$ & $\mathbf{-0.101}$ \\
\bottomrule
\end{tabular}
\end{adjustbox}
\end{table}

\begin{remark}[Paired uncertainty on Zipf deltas]
All deltas in Table~\ref{tab:delta_zipf} are computed on the same 50
shared holdout seeds.
Paired bootstrap 95\% CIs (10\,000 resamples) yield half-widths of
${\approx}0.002$--$0.004$ for BE-score,
${\approx}0.001$ for broadcast-packet expiration ratio, ${\approx}0.1$--$0.15$ for
expirations, and ${\approx}0.002$ for coding gain.
The BE-score, miss-ratio, and expirations deltas favoring PPO, and the
coding-gain deltas favoring SACM{++}, all have intervals that do not
cross zero; we report these as descriptive effect-size bounds rather
than formal significance tests.
\end{remark}

\paragraph{Request-level accounting (Zipf).}
Table~\ref{tab:request_level_zipf} reports the request-level metrics
(M8--M10) under Zipf demand, tracking each original arrival through
the queue aggregation as in the uniform analysis
(Section~\ref{sec:res:request_level}).
The interesting fact we observed is that the result contrasts with the uniform track, i.e.,
the policy network (PPO-Agent) achieves \emph{both} a higher request timely-throughput~$\eta_{\mathrm{req}}$
\emph{and} a lower request miss rate~$m_{\mathrm{req}}$ than SACM{++} in
11 of 12~regimes, so the request selection score~$\sigma_{\mathrm{req}}$
favors the policy network at \emph{every} tested penalty weight
$\lambda \in \{0.5,\,1,\,\ldots,\,10\}$ \emph{under Zipf demand}, i.e., no crossover is needed in this regime.
While on the uniform track
(Section~\ref{sec:res:request_level}), SACM{++}'s ${\sim}17\%$
higher~$\eta_{\mathrm{req}}$ required $\lambda \geq 2.24$ before
the policy network's lower miss rate became decisive.
Under Zipf demand, the near-full merging strategy of the policy network ($99.7\%$ merge
rate) matches or exceeds the throughput of the deterministic SACM family
while also maintaining lower deadline misses, i.e., a combination
not achievable under the selective-merge strategy optimal for the
uniform demand.
The largest advantages appear under high cache density
(Curr-pcache0.40: $\sigma_{\mathrm{req}} = 1.202$ vs.\ $1.129$,
$+6.5\%$) and tight deadlines
(OOD-delay10: $0.649$ vs.\ $0.551$, $+17.8\%$), i.e., the same observation is true for the variations of the cache fraction and deadline budgets as well,
while the sole near-tie is OOD-pcache0.20, where all methods converge to
similar performance.
Fig.~\ref{fig:pareto_request_zipf} visualizes this dominance, i.e.,
the policy network (PPO-Agent) sits strictly below and to the right of SACM{++} in
the $\eta_{\mathrm{req}}$--$m_{\mathrm{req}}$ plane, so
no throughput--compliance trade-off exists under Zipf demand.
Appendix~\ref{sec:appendix:req_lambda_zipf} presents the full
$\lambda$-sensitivity analysis.

\begin{figure}[t]
  \centering
  \includegraphics[width=\columnwidth]{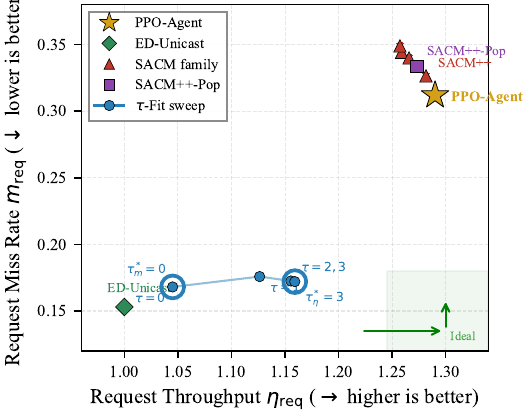}
  \caption{Request throughput~$\eta_{\mathrm{req}}$ ($\rightarrow$)
    vs.\ request miss rate~$m_{\mathrm{req}}$ ($\downarrow$) under
    Zipf demand. PPO-Agent achieves both higher throughput and lower
    miss rate than the displayed SACM{++} and SACM{++}-Pop
    comparators on this figure: no trade-off exists between PPO
    and these displayed comparators, unlike the uniform track
    (cf.\ Fig.~\ref{fig:pareto_request}); a request-level extension
    over GCM, SACM, SACM+, and the $\tau$-Fit family is left as
    future work.}
  \label{fig:pareto_request_zipf}
\end{figure}

\begin{table*}[t]
\centering
\caption{Request-level metrics (Zipf demand, $\alpha{=}0.8$).
  Each original arrival tracked through aggregation; completion and miss
  credited exactly once per request ID.
  Mean $\pm$ 95\% CI (50 seeds $\times$ 200 episodes).  Bold = best.}
\label{tab:request_level_zipf}
\begin{adjustbox}{max width=\textwidth}
\begin{tabular}{lcccccccccc}
\toprule
 & \multicolumn{3}{c}{$\eta_{\mathrm{req}}$ (↑)} & \multicolumn{3}{c}{$m_{\mathrm{req}}$ (↓)} & \multicolumn{3}{c}{$\sigma_{\mathrm{req}}\;(\lambda{=}1)$ (↑)} \\
\textbf{Regime} & PPO & SACM{++} & Pop & PPO & SACM{++} & Pop & PPO & SACM{++} & Pop \\
\midrule
ID-default ($\alpha{=}0.8$)
  & $\mathbf{1.290 \pm 0.002}$ & $1.282 \pm 0.002$ & $1.273 \pm 0.002$
  & $\mathbf{0.312 \pm 0.001}$ & $0.326 \pm 0.002$ & $0.334 \pm 0.002$
  & $\mathbf{0.978 \pm 0.002}$ & $0.955 \pm 0.002$ & $0.939 \pm 0.002$ \\
OOD-alpha0.6
  & $\mathbf{1.293 \pm 0.002}$ & $1.284 \pm 0.002$ & $1.275 \pm 0.002$
  & $\mathbf{0.313 \pm 0.002}$ & $0.327 \pm 0.002$ & $0.336 \pm 0.001$
  & $\mathbf{0.980 \pm 0.002}$ & $0.957 \pm 0.002$ & $0.939 \pm 0.002$ \\
OOD-alpha1.0
  & $\mathbf{1.288 \pm 0.002}$ & $1.279 \pm 0.002$ & $1.270 \pm 0.002$
  & $\mathbf{0.312 \pm 0.002}$ & $0.326 \pm 0.002$ & $0.333 \pm 0.002$
  & $\mathbf{0.976 \pm 0.002}$ & $0.953 \pm 0.002$ & $0.937 \pm 0.002$ \\
OOD-alpha1.2
  & $\mathbf{1.280 \pm 0.002}$ & $1.274 \pm 0.002$ & $1.267 \pm 0.002$
  & $\mathbf{0.309 \pm 0.001}$ & $0.323 \pm 0.001$ & $0.330 \pm 0.002$
  & $\mathbf{0.971 \pm 0.002}$ & $0.952 \pm 0.002$ & $0.937 \pm 0.002$ \\
OOD-mandelbrot
  & $\mathbf{1.289 \pm 0.002}$ & $1.281 \pm 0.002$ & $1.271 \pm 0.002$
  & $\mathbf{0.311 \pm 0.001}$ & $0.324 \pm 0.002$ & $0.332 \pm 0.002$
  & $\mathbf{0.978 \pm 0.003}$ & $0.957 \pm 0.002$ & $0.939 \pm 0.002$ \\
Curr-file60 (Z)
  & $\mathbf{1.291 \pm 0.002}$ & $1.281 \pm 0.002$ & $1.274 \pm 0.002$
  & $\mathbf{0.313 \pm 0.002}$ & $0.325 \pm 0.001$ & $0.333 \pm 0.002$
  & $\mathbf{0.979 \pm 0.003}$ & $0.956 \pm 0.002$ & $0.942 \pm 0.002$ \\
OOD-file120 (Z)
  & $\mathbf{1.292 \pm 0.002}$ & $1.284 \pm 0.002$ & $1.275 \pm 0.002$
  & $\mathbf{0.312 \pm 0.002}$ & $0.326 \pm 0.002$ & $0.334 \pm 0.002$
  & $\mathbf{0.979 \pm 0.002}$ & $0.958 \pm 0.002$ & $0.941 \pm 0.002$ \\
OOD-file150 (Z)
  & $\mathbf{1.292 \pm 0.002}$ & $1.282 \pm 0.002$ & $1.274 \pm 0.002$
  & $\mathbf{0.312 \pm 0.002}$ & $0.326 \pm 0.002$ & $0.334 \pm 0.002$
  & $\mathbf{0.979 \pm 0.003}$ & $0.956 \pm 0.002$ & $0.940 \pm 0.002$ \\
OOD-pcache0.20 (Z)
  & $1.116 \pm 0.001$ & $1.116 \pm 0.001$ & $\mathbf{1.119 \pm 0.001}$
  & $0.269 \pm 0.001$ & $0.270 \pm 0.002$ & $\mathbf{0.268 \pm 0.002}$
  & $0.848 \pm 0.001$ & $0.847 \pm 0.002$ & $\mathbf{0.852 \pm 0.002}$ \\
Curr-pcache0.40 (Z)
  & $\mathbf{1.504 \pm 0.002}$ & $1.484 \pm 0.002$ & $1.458 \pm 0.002$
  & $\mathbf{0.303 \pm 0.001}$ & $0.354 \pm 0.002$ & $0.383 \pm 0.002$
  & $\mathbf{1.202 \pm 0.002}$ & $1.129 \pm 0.002$ & $1.075 \pm 0.003$ \\
OOD-delay10 (Z)
  & $\mathbf{1.631 \pm 0.002}$ & $1.600 \pm 0.002$ & $1.570 \pm 0.002$
  & $\mathbf{0.982 \pm 0.002}$ & $1.048 \pm 0.002$ & $1.088 \pm 0.002$
  & $\mathbf{0.649 \pm 0.003}$ & $0.551 \pm 0.003$ & $0.482 \pm 0.003$ \\
OOD-delay30 (Z)
  & $\mathbf{1.146 \pm 0.001}$ & $1.142 \pm 0.001$ & $1.140 \pm 0.001$
  & $\mathbf{0.110 \pm 0.001}$ & $0.115 \pm 0.001$ & $0.118 \pm 0.001$
  & $\mathbf{1.036 \pm 0.001}$ & $1.027 \pm 0.001$ & $1.022 \pm 0.001$ \\
\bottomrule
\end{tabular}
\end{adjustbox}
\end{table*}

\FloatBarrier

\section{Ablations and Analysis}
\label{sec:ablation}

The training pipeline introduced in Section~\ref{sec:method} combines four components, i.e., the behavior-cloning warm start (BC), the graph-attention architecture (GraphAttn), the Expert Iteration distillation (ExIt), and the curriculum learning. This section isolates the contribution of each of these components by training five ablated variants, each with exactly one component removed, plus one combined removal (BC + ExIt), while evaluating all six of the models across the full eight-regime battery (i.e., ID-default plus the seven non-ID conditions) used in Sections~\ref{sec:results}.

\subsection{Ablation Setup}
\label{sec:ablation:setup}

\textbf{Variants.} Starting from the full model, we train five single-component ablations and one combined ablation:
\begin{enumerate}[label=(\roman*)]
  \item \textbf{w/o ExIt}: removes the Expert Iteration distillation loops (Sections~\ref{sec:method:exit}); PPO training proceeds without any online teacher distillation after the BC warm start.
  \item \textbf{w/o BC Warm-Start}: removes the rollout-improved behavior cloning phase (Section~\ref{sec:method:bc}); the policy head is randomly initialized before the value warm-up and PPO stages.
  \item \textbf{w/o Graph-Attention}: replaces the graph-pair encoder with a flat two-layer MLP of matching parameter count (target ${\sim}1.73$M parameters, matched to the Track~A full model; see Appendix~\ref{sec:appendix:arch_details}) that receives the concatenated per-request and per-pair feature tensors as input; the $\cmark$/$\xmark$ component columns below capture the presence of the graph-attentive pair encoder.
  \item \textbf{w/o BC + ExIt}: removes both the BC warm start and the ExIt distillation loops, leaving only random initialization followed by MaskablePPO with curriculum.
  \item \textbf{w/o Curriculum}: removes all three curriculum stages; the policy trains directly on the full-difficulty environment ($N=100$, $p_c=0.30$) from the start.
\end{enumerate}
\begin{table}[t]
\centering
\caption{Ablation study: component legend. Each variant removes one or two training components from the full PPO-Agent pipeline.}
\label{tab:ablation_legend}
\begin{tabular}{lcccc}
\toprule
\textbf{Variant} & \textbf{BC} & \textbf{Graph-Attn} & \textbf{ExIt} & \textbf{Curriculum} \\
\midrule
Full Model & \checkmark & \checkmark & \checkmark & \checkmark \\
w/o ExIt & \checkmark & \checkmark & --- & \checkmark \\
w/o BC Warm-Start & --- & \checkmark & \checkmark & \checkmark \\
w/o Graph-Attention & \checkmark & --- (MLP) & \checkmark & \checkmark \\
w/o BC + ExIt & --- & \checkmark & --- & \checkmark \\
w/o Curriculum & \checkmark & \checkmark & \checkmark & --- \\
\bottomrule
\end{tabular}
\end{table}

Each of the variants is trained with 4 independent random seeds, while the model selection follows the same robust-advantage criterion used for the full model (Section~\ref{sec:method:selection}), i.e., using 50 validation seeds withheld from the evaluation. All of the reported numbers use the same 50 \emph{evaluation} holdout seeds ($\times\,200$ episodes/seed = 10,000 episodes per method per regime) shared with Section~\ref{sec:results}; each ablation variant is itself trained independently with its own 4 training seeds, however, the evaluation-set seeds and the per-episode contexts are common across all of the variants and the full model, thereby making the comparison fair.

Consistent with Section~\ref{sec:exp:metrics}, the \emph{primary} metrics are the broadcast-packet expiration ratio ($\rho$, $\downarrow$) and the distinct file-identity coverage ($\delta$, $\uparrow$). Broadcast-efficiency metrics (BE-Score~$\sigma$, served/tx, coding gain, expirations) are reported for completeness, the request-level family ($\eta_{\mathrm{req}}, m_{\mathrm{req}}, \sigma_{\mathrm{req}}$) is reported as a \emph{supplementary} diagnostic (the displayed request-level tables compare PPO only to a subset of the baselines; full-baseline coverage is left as future work, see Sec.~\ref{sec:exp:metrics}), and diagnostic metrics (merge rate, opportunity rate, reward per step) are reported only as policy-behavior summaries.

\subsection{Component Contributions at ID-Default}
\label{sec:ablation:components}

Table~\ref{tab:ablation_id} presents the ablation results on the ID-default regime, while Table~\ref{tab:ablation_delta_id} summarizes the per-component impact as signed deltas with respect to the full model. The largest single-component ID-default reliability degradation comes from removing the \emph{curriculum} ($\Delta\rho{=}{+}0.044$, $\Delta\varepsilon{=}{+}4.46$); the \emph{ExIt} is the second-largest single-component ID-default contributor ($\Delta\rho{=}{+}0.041$, $\Delta\varepsilon{=}{+}3.49$), however, it is the most uniformly large across the OOD-ablation battery (Table~\ref{tab:ablation_ood_miss}), which is why we discuss it first in the per-component narrative below. We then examine how the demand-centric ($\delta$) and the request-level ($\sigma_{\mathrm{req}}$) metrics give complementary, sometimes contrasting, views of the component contributions. The reward per step is reported as a training-signal diagnostic, however, we do not use it for ranking, because it is the shaped training surrogate rather than the paper's target communications-metric family. The reward definition and weights are held fixed across the rows of Table~\ref{tab:ablation_id}, thereby making the Reward/Step values directly comparable within this table; the appendix paragraph ``Reward per step (M13)'' notes only that the comparability fails across the reward \emph{definitions}, i.e., not across the ablation rows reported here.

\begin{table*}[t]
\centering
\caption{Ablation study, ID-default. 50 holdout seeds $\times$ 200 episodes. Bold = best per column among methodology-relevant columns; the Reward/Step column is the shaped training surrogate and is intentionally not bold-ranked (see footnote $^\dagger$). Primary metrics: $\rho$, $\delta$; supplementary: $\sigma_{\mathrm{req}}$.}
\label{tab:ablation_id}
\begin{adjustbox}{max width=\textwidth}
\begin{tabular}{lcccccccccccc}
\toprule
\textbf{Variant} & \textbf{Reward/Step}$^{\dagger}$ ($\uparrow$) & \textbf{BE-Score $\sigma$} ($\uparrow$) & \textbf{Miss Ratio $\rho$} ($\downarrow$) & \textbf{Coverage $\delta$} ($\uparrow$) & \textbf{Served/Tx} ($\uparrow$) & \textbf{Coding Gain} ($\uparrow$) & \textbf{Exp/Episode} ($\downarrow$) & \textbf{Merge Rate} ($\uparrow$) & \textbf{Opp Rate} ($\uparrow$) & $\eta_{\mathrm{req}}$ ($\uparrow$) & $m_{\mathrm{req}}$ ($\downarrow$) & $\sigma_{\mathrm{req}}$ ($\uparrow$) \\
\midrule
Full Model & $1.092 \pm 0.002$ & $\mathbf{0.976 \pm 0.002}$ & $\mathbf{0.208 \pm 0.001}$ & $0.824 \pm 0.001$ & $1.323 \pm 0.002$ & $2.162 \pm 0.002$ & $14.17 \pm 0.07$ & $0.318 \pm 0.001$ & $0.876 \pm 0.002$ & $1.069 \pm 0.001$ & $0.229 \pm 0.001$ & $0.839 \pm 0.001$ \\
w/o ExIt & $1.045 \pm 0.002$ & $0.911 \pm 0.002$ & $0.249 \pm 0.001$ & $0.806 \pm 0.001$ & $1.364 \pm 0.002$ & $2.165 \pm 0.002$ & $17.66 \pm 0.11$ & $0.373 \pm 0.002$ & $0.839 \pm 0.002$ & $1.106 \pm 0.001$ & $0.267 \pm 0.002$ & $0.839 \pm 0.002$ \\
w/o BC Warm-Start & $1.058 \pm 0.002$ & $0.936 \pm 0.002$ & $0.228 \pm 0.001$ & $0.829 \pm 0.001$ & $1.329 \pm 0.002$ & $2.144 \pm 0.002$ & $15.16 \pm 0.08$ & $0.333 \pm 0.002$ & $0.864 \pm 0.002$ & $1.093 \pm 0.001$ & $0.227 \pm 0.001$ & $\mathbf{0.866 \pm 0.002}$ \\
w/o Graph-Attention & $1.089 \pm 0.002$ & $0.957 \pm 0.002$ & $0.221 \pm 0.001$ & $0.821 \pm 0.001$ & $1.336 \pm 0.002$ & $\mathbf{2.203 \pm 0.002}$ & $14.94 \pm 0.09$ & $0.322 \pm 0.001$ & $0.866 \pm 0.002$ & $1.080 \pm 0.001$ & $0.235 \pm 0.002$ & $0.845 \pm 0.001$ \\
w/o BC + ExIt & $0.984 \pm 0.002$ & $0.884 \pm 0.002$ & $0.221 \pm 0.001$ & $\mathbf{0.834 \pm 0.001}$ & $1.234 \pm 0.002$ & $2.081 \pm 0.002$ & $\mathbf{13.95 \pm 0.09}$ & $0.238 \pm 0.002$ & $\mathbf{0.908 \pm 0.002}$ & $1.076 \pm 0.001$ & $\mathbf{0.215 \pm 0.001}$ & $0.861 \pm 0.001$ \\
w/o Curriculum & $1.076 \pm 0.002$ & $0.949 \pm 0.002$ & $0.252 \pm 0.001$ & $0.801 \pm 0.001$ & $\mathbf{1.431 \pm 0.002}$ & $2.157 \pm 0.002$ & $18.63 \pm 0.09$ & $\mathbf{0.476 \pm 0.002}$ & $0.782 \pm 0.003$ & $\mathbf{1.118 \pm 0.001}$ & $0.278 \pm 0.001$ & $0.840 \pm 0.001$ \\
\bottomrule
\end{tabular}
\end{adjustbox}
\smallskip
{\footnotesize $^\dagger$Reward/Step is the mean \emph{shaped training reward} per step. The reward definition and weights (Section~\ref{sec:env:reward}) are held \emph{fixed} across the ablation rows reported here, so Reward/Step values \emph{are} directly comparable within this table; it is reported as a diagnostic only because it is the shaped training surrogate rather than the paper's target communications-metric family. Cross-table comparability fails when the reward definition or its weights themselves change.}
\end{table*}

\begin{table*}[t]
\centering
\caption{Ablation study: change vs.\ Full Model (ID-default). Sign convention: for minimization metrics ($\rho$ = Miss Ratio, Exp/Episode, $m_{\mathrm{req}}$) positive $\Delta$ denotes degradation of the ablated variant relative to the Full Model; for maximization metrics ($\delta$ = Coverage, $\sigma$ = BE-Score, Served/Tx, Gain, $\eta_{\mathrm{req}}$, $\sigma_{\mathrm{req}}$) positive $\Delta$ denotes improvement. $\Delta$ Reward is shown for completeness as a training-surrogate diagnostic and is not used for ranking.}
\label{tab:ablation_delta_id}
\begin{adjustbox}{max width=\textwidth}
\begin{tabular}{l rrrrrrrrrr}
\toprule
\textbf{Variant} & $\Delta$ \textbf{Reward}$^{\dagger}$ & $\Delta$ \textbf{BE-Score} & $\Delta$ \textbf{Miss Ratio} & $\Delta$ \textbf{Coverage $\delta$} & $\Delta$ \textbf{Served/Tx} & $\Delta$ \textbf{Gain} & $\Delta$ \textbf{Exp/Ep} & $\Delta$ $\eta_{\mathrm{req}}$ & $\Delta$ $m_{\mathrm{req}}$ & $\Delta$ $\sigma_{\mathrm{req}}$ \\
\midrule
w/o ExIt & $-0.047$ ($-$4.3\%) & $-0.065$ ($-$6.7\%) & $+0.041$ (+19.9\%) & $-0.018$ ($-$2.2\%) & $+0.041$ (+3.1\%) & $+0.003$ (+0.1\%) & $+3.49$ (+24.6\%) & $+0.038$ (+3.5\%) & $+0.038$ (+16.7\%) & $-0.000$ ($-$0.1\%) \\
w/o BC Warm-Start & $-0.034$ ($-$3.1\%) & $-0.040$ ($-$4.1\%) & $+0.020$ (+9.7\%) & $+0.005$ (+0.6\%) & $+0.006$ (+0.5\%) & $-0.018$ ($-$0.8\%) & $+0.99$ (+7.0\%) & $+0.024$ (+2.3\%) & $-0.003$ ($-$1.1\%) & $+0.027$ (+3.2\%) \\
w/o Graph-Attention & $-0.003$ ($-$0.3\%) & $-0.019$ ($-$1.9\%) & $+0.013$ (+6.3\%) & $-0.003$ ($-$0.3\%) & $+0.013$ (+0.9\%) & $+0.041$ (+1.9\%) & $+0.78$ (+5.5\%) & $+0.011$ (+1.0\%) & $+0.006$ (+2.6\%) & $+0.005$ (+0.6\%) \\
w/o BC + ExIt & $-0.108$ ($-$9.9\%) & $-0.092$ ($-$9.5\%) & $+0.013$ (+6.3\%) & $+0.010$ (+1.2\%) & $-0.090$ ($-$6.8\%) & $-0.080$ ($-$3.7\%) & $-0.22$ ($-$1.5\%) & $+0.007$ (+0.6\%) & $-0.014$ ($-$6.2\%) & $+0.021$ (+2.5\%) \\
w/o Curriculum & $-0.016$ ($-$1.5\%) & $-0.027$ ($-$2.7\%) & $+0.044$ (+21.1\%) & $-0.023$ ($-$2.8\%) & $+0.107$ (+8.1\%) & $-0.005$ ($-$0.2\%) & $+4.46$ (+31.5\%) & $+0.049$ (+4.6\%) & $+0.049$ (+21.4\%) & $+0.000$ (+0.0\%) \\
\bottomrule
\end{tabular}
\end{adjustbox}
\end{table*}
\FloatBarrier

\textbf{Expert Iteration / DAgger Distillation (largest cross-regime $\rho$-impact).}
The ExIt has the largest single-component effect on the deadline compliance, i.e., removing it raises the broadcast-packet expiration ratio by $+19.9\%$ and the expirations by $+24.6\%$ at ID-default. This pattern persists and amplifies across the non-ID regimes as well: at OOD-delay30, removing the ExIt raises the broadcast-packet expiration ratio by $+37.6\%$ (relative), and at Curr-pcache0.40, by $+23.9\%$. The pattern is consistent with the hypothesis that the ExIt phase reinforces the deferral of merges in favor of the urgent unicasts under the deadline pressure; this is a mechanism interpretation rather than a directly measured causal attribution, however, it is consistent with the observed merge-rate and expiration patterns across the regimes.

The combined removal (w/o BC + ExIt) exhibits a clear interaction effect. The interesting fact we observed is that removing both the BC and the ExIt raises the broadcast-packet expiration ratio by only $+6.3\%$, i.e., less than removing the ExIt alone ($+19.9\%$) or the BC alone ($+9.7\%$). This apparent paradox is explained by the emergent ultra-conservative policy, i.e., without either the warm start or the distillation, the agent converges to the lowest merge rate in the study ($0.238$, vs.\ $0.318$ for the full model) and the highest opportunity rate ($0.908$). The policy largely avoids the coded transmissions, thereby achieving the fewest expirations of any variant ($13.95$, even fewer than the full model's $14.17$), even though at the cost of substantially lower throughput ($-6.8\%$ served/tx). The BC + ExIt combination is therefore \emph{not} complementary in the additive sense; removing both drives the policy into an ultra-conservative basin that trades the throughput for the deadline safety, which is a qualitatively different failure mode with respect to removing either component alone.

\textbf{BC Warm-Start (moderate, OOD-amplified).}
Removing the BC phase costs $+9.7\%$ broadcast-packet expiration ratio and $+7.0\%$ expirations at ID-default. The impact amplifies under the distribution shift, i.e., at Curr-pcache0.40 (i.e., high cache fraction), removing the BC costs $+2.6$~pp broadcast-packet expiration ratio with respect to the full model. Without the rollout-improved warm start, the PPO must discover the selective merge strategy (i.e., the core behavior described in Section~\ref{sec:res:behavior}) entirely from scratch from a random initialization, while the resulting policy converges to a suboptimal strategy with moderately higher merge rate ($0.333$ vs.\ $0.318$), which is consistent with insufficient initial structure for the selective merging.

\textbf{Graph-Attention (smallest single-component impact).}
Removing the graph-attention and replacing it with a flat MLP produces the smallest single-component reliability degradation, i.e., $+6.3\%$ broadcast-packet expiration ratio ($0.221$ vs.\ $0.208$) and $+5.5\%$ expirations ($14.94$ vs.\ $14.17$) at ID-default, while the throughput is essentially unchanged ($+0.9\%$ served/tx). The coding gain is marginally \emph{higher} without the graph-attention ($2.203$ vs.\ $2.162$, $+1.9\%$), which suggests that the MLP replacement can exploit the simple pair structures, even though it lacks the broader relational reasoning that helps the graph-attention coordinate the merge decisions across the full request queue. The graph-attention encoder's contribution is therefore not to the raw coding efficiency but to the reliability-aware coordination that lowers the deadline violations.

\subsection{Curriculum Learning and the Reliability--Throughput Trade-Off}
\label{sec:ablation:curriculum}

The curriculum result is the most counter-intuitive finding in the ablation study and warrants extended discussion. The interesting fact we observed is that removing the curriculum \emph{improves} the throughput by $+8.1\%$ served/tx, even though it \emph{degrades} the tabulated Reward/Step by $-1.5\%$ ($1.076$ vs.\ $1.092$, Table~\ref{tab:ablation_id}), while simultaneously \emph{degrading} the broadcast-packet expiration ratio by $+21.1\%$ and the expirations by $+31.5\%$, i.e., the largest single-component reliability degradation we observed.

Without the curriculum, the policy trains directly on the full-difficulty environment ($N=100$, $p_c=0.30$) from the first training step. The per-broadcast packet-set reward favors the coded transmissions (i.e., a merge with $U_t{=}|f_{\mathrm{mg}}|$ yields a higher local $U_t$ contribution than a unicast with $U_t{=}1$), while this local signal is so dominant that the PPO quickly learns to merge aggressively. The w/o Curriculum variant achieves a merge rate of $0.476$, compared to $0.318$ for the full model, i.e., a $50\%$ increase in the merging frequency. Aggressive merging raises the Served/Tx (i.e., the packet-set count $\mu$) and the \emph{local} per-step $U_t$ incentive (i.e., coded broadcasts contribute $|f_{\mathrm{mg}}| \ge 2$ to $U_t$ per slot), however, the same aggressive strategy causes substantially more expirations, while the additional $-E_t$ penalty more than offsets the local $U_t$ gain in the realized episode-average Reward/Step (which \emph{decreases} by $-1.5\%$, $1.076$ vs.\ $1.092$, Table~\ref{tab:ablation_id}); without the deadline-conscious restraint, the policy merges even the pairs whose deadline is too tight to benefit from the additional encoded content before it expires.

With the curriculum ($N=60 \to 80 \to 100$, $p_c=0.50 \to 0.40 \to 0.30$), the agent first learns in the easy environments where the merging is abundant and the deadlines are rarely an issue, i.e., the large cache fractions create many high-quality pairs ($|\mathcal{S}_i \cap \mathcal{S}_j|$ is high), while the small catalog ($N=60$) means the requests rarely conflict. In this regime, aggressive merging is nearly always the right strategy, while the policy internalizes the coded-XOR structure quickly. As the difficulty increases, the agent encounters progressively more deadline pressure, i.e., more requests per cache slot, fewer high-quality pairs, and tighter cache capacity. The pattern is consistent with the interpretation that the curriculum exposes the agent to the scenarios of \emph{where not to merge}, thereby gradually shifting the learned policy toward the $31.8\%$ selective merge rate that characterizes the full model's behavior (Section~\ref{sec:res:behavior}); the ablation establishes that removing the curriculum degrades the selective merging, even though it does not isolate the curriculum's mechanism directly.

\begin{remark}
The curriculum learning in this setting is not primarily a sample-efficiency technique, i.e., the w/o Curriculum variant converges to a high-reward policy quickly and it does not fail to converge. The curriculum shapes the \emph{qualitative character} of the learned policy, i.e., the gradual introduction of the deadline pressure during the training produces a deadline-conscious selective merge strategy that cannot be recovered by training directly on the hard environment, even though given the same total sample budget.
\end{remark}

This finding is consistent with the curriculum learning theory~\cite{bengio2009curriculum}, i.e., starting from an easy distribution guides the learner toward a solution basin associated with the generalizable structure (i.e., selective merging), rather than the locally optimal but brittle basin associated with the aggressive merging. For the coded caching, the \emph{reliability--throughput} trade-off is the main axis along which the two basins differ.

\subsection{Demand-Centric and Request-Level Perspectives}
\label{sec:ablation:demand}

The broadcast-efficiency analysis above identifies the full model as the best-performing variant on the broadcast-packet expiration ratio ($\rho$) and the BE-Score ($\sigma$). The problem, however, is that the formal-taxonomy primary metrics ($\rho$, $\delta$; Sec.~\ref{sec:exp:metrics}) and the supplementary request-level family give a more nuanced picture, i.e., including the ranking reversals on $\delta$ and on $\sigma_{\mathrm{req}}$ that challenge the broadcast-efficiency narrative.

\textbf{Ranking reversal on distinct file-identity coverage ($\delta$).}
The full model is \emph{not} the best on the distinct file-identity coverage, i.e., the w/o BC + ExIt variant achieves $\delta = 0.834$, thereby surpassing the full model ($\delta = 0.824$) by $+1.2\%$ (Table~\ref{tab:ablation_id}), while the w/o BC Warm-Start variant ranks second ($\delta = 0.829$, $+0.6\%$). This reversal is explained by the relationship between the merge rate and the unique-demand coverage, i.e., the broadcast-level expiration ratio $\rho$ uses the packet-set XOR degree $U_t = |f_{\mathrm{mg}}|$ in its denominator, thereby aggressive merging inflates the denominator. The distinct file-identity coverage~$\delta$ instead counts each file identity \emph{at most once} per episode, while crediting only the first successful delivery. The w/o BC + ExIt variant's ultra-conservative policy (i.e., merge rate $0.238$ vs.\ $0.318$) carries a smaller packet-set per transmission, even though it avoids the redundant retransmission of already-covered file identities, thereby yielding better distinct file-identity coverage. This variant also achieves the fewest expirations ($13.95$) and the lowest request-level miss rate ($m_{\mathrm{req}} = 0.215$, $-6.2\%$ vs.\ full model).

\textbf{Ranking reversal on request-level selection score ($\sigma_{\mathrm{req}}$).}
On the composite request-level metric $\sigma_{\mathrm{req}} = \eta_{\mathrm{req}} - m_{\mathrm{req}}$, the full model ranks fifth out of the six variants ($\sigma_{\mathrm{req}} = 0.839$), while the w/o BC Warm-Start variant achieves the best score ($\sigma_{\mathrm{req}} = 0.866$, $+3.2\%$), followed by the w/o BC + ExIt ($0.861$, $+2.5\%$). The request-level metrics show a throughput--accuracy tradeoff, i.e., removing the curriculum yields the highest per-step request timely-throughput ($\eta_{\mathrm{req}} = 1.118$) but the worst per-step request miss rate ($0.278$), while removing the BC achieves a more favorable balance ($\eta_{\mathrm{req}} = 1.093$, $m_{\mathrm{req}} = 0.227$). This suggests that the BC warm-start introduces a broadcast-efficiency bias, i.e., the BC training objective maximizes the broadcast-level packet-set XOR degree, which biases the initialized policy toward aggressive scheduling that helps the broadcast-efficiency metrics, even though it impairs the request-level quality on the per-step $\eta_{\mathrm{req}}/m_{\mathrm{req}}$ scores.

\textbf{The merge-rate paradox.}
These reversals reflect a tension between the metric families, i.e., mediated by the merge rate. Aggressive merging (i.e., w/o Curriculum: merge rate $0.476$) maximizes the served packet-set count per transmission ($\mu = 1.431$) and the per-step request timely-throughput ($\eta_{\mathrm{req}} = 1.118$), even though it produces the worst file-identity miss ratio ($1-\delta = 0.199$) and the worst per-step request miss rate ($0.278$). Conservative merging (i.e., w/o BC + ExIt: merge rate $0.238$) achieves the opposite, i.e., the worst broadcast-efficiency ($\sigma = 0.884$), even though the best file-identity coverage ($\delta = 0.834$) and the best per-step request miss rate ($0.215$). The full model's selective merge rate ($0.318$) occupies a middle ground that optimizes the broadcast-efficiency while sacrificing some file-identity and request-level performance on the per-step scores. This multi-objective tension is an inherent property of the coded caching with deadlines, i.e., a coded broadcast contributes $|f_{\mathrm{mg}}|$ packet identities to the per-slot count $U_t$ (and credits multiple original arrivals at the request level), however, it may delay individual requests waiting for the merge opportunities, thereby increasing the per-step request miss count $|M_t|$.

\subsection{OOD Robustness of Ablation Findings}
\label{sec:ablation:ood}

Table~\ref{tab:ablation_ood_miss} summarizes the per-component deltas across all the eight regimes for the two primary metrics, i.e., the broadcast-packet expiration ratio ($\rho$) and the distinct file-identity coverage ($\delta$), alongside the supplementary request selection score ($\sigma_{\mathrm{req}}$). The miss-ratio panel confirms that every one of the components reduces the broadcast-packet expiration ratio (i.e., all $\Delta\rho > 0$), while the relative ordering is consistent across the conditions. The demand-centric and the request-level panels give a richer picture, i.e., the ranking reversals identified at ID-default persist across most of the regimes, with exceptions under the extreme deadline pressure.

\begin{table*}[t]
\centering
\caption{Ablation study: $\Delta$ vs.\ Full Model across all regimes for the two primary metrics and one supplementary metric. Miss ratio $\rho$ (primary, positive $=$ worse); file-identity coverage $\delta$ (primary, negative $=$ worse); request selection score $\sigma_{\mathrm{req}}$ (supplementary, negative $=$ worse).}
\label{tab:ablation_ood_miss}
\begin{adjustbox}{max width=\textwidth}
\begin{tabular}{lcccccccc}
\toprule
\textbf{Variant} & \textbf{ID-default} & \textbf{OOD-delay10} & \textbf{OOD-delay30} & \textbf{Curr-file60} & \textbf{OOD-file120} & \textbf{OOD-file150} & \textbf{OOD-pcache0.20} & \textbf{Curr-pcache0.40} \\
\midrule
\multicolumn{9}{l}{\emph{$\Delta$ Miss Ratio $\rho$ ($\uparrow$ = degradation)}} \\
\midrule
w/o ExIt & $+0.041$ & $+0.018$ & $+0.024$ & $+0.041$ & $+0.042$ & $+0.041$ & $+0.024$ & $+0.057$ \\
w/o BC Warm-Start & $+0.020$ & $+0.008$ & $+0.012$ & $+0.020$ & $+0.020$ & $+0.019$ & $+0.015$ & $+0.026$ \\
w/o Graph-Attention & $+0.013$ & $+0.004$ & $+0.007$ & $+0.013$ & $+0.013$ & $+0.013$ & $+0.002$ & $+0.031$ \\
w/o BC + ExIt & $+0.013$ & $+0.013$ & $+0.006$ & $+0.013$ & $+0.013$ & $+0.012$ & $+0.000$ & $+0.036$ \\
w/o Curriculum & $+0.044$ & $+0.017$ & $+0.029$ & $+0.043$ & $+0.044$ & $+0.043$ & $+0.032$ & $+0.045$ \\
\midrule
\multicolumn{9}{l}{\emph{$\Delta$ Distinct File-Identity Coverage $\delta$ ($\uparrow$ = improvement over Full Model)}} \\
\midrule
w/o ExIt & $-0.018$ & $-0.004$ & $-0.012$ & $-0.018$ & $-0.018$ & $-0.016$ & $-0.010$ & $-0.032$ \\
w/o BC Warm-Start & $+0.005$ & $+0.006$ & $+0.001$ & $+0.006$ & $+0.005$ & $+0.007$ & $+0.003$ & $+0.001$ \\
w/o Graph-Attention & $-0.003$ & $-0.005$ & $-0.001$ & $-0.002$ & $-0.002$ & $-0.001$ & $+0.003$ & $-0.014$ \\
w/o BC + ExIt & $+0.010$ & $-0.004$ & $+0.005$ & $+0.011$ & $+0.011$ & $+0.012$ & $+0.011$ & $+0.003$ \\
w/o Curriculum & $-0.023$ & $+0.015$ & $-0.017$ & $-0.022$ & $-0.023$ & $-0.022$ & $-0.017$ & $-0.027$ \\
\midrule
\multicolumn{9}{l}{\emph{$\Delta$ Request Selection Score $\sigma_{\mathrm{req}}$ ($\uparrow$ = improvement over Full Model)}} \\
\midrule
w/o ExIt & $-0.000$ & $-0.004$ & $+0.002$ & $-0.000$ & $+0.000$ & $+0.002$ & $-0.002$ & $-0.005$ \\
w/o BC Warm-Start & $+0.027$ & $+0.029$ & $+0.011$ & $+0.029$ & $+0.028$ & $+0.030$ & $+0.017$ & $+0.025$ \\
w/o Graph-Attention & $+0.005$ & $-0.020$ & $+0.005$ & $+0.006$ & $+0.006$ & $+0.007$ & $+0.006$ & $+0.003$ \\
w/o BC + ExIt & $+0.021$ & $-0.020$ & $+0.004$ & $+0.023$ & $+0.023$ & $+0.025$ & $+0.015$ & $+0.033$ \\
w/o Curriculum & $+0.000$ & $+0.091$ & $+0.004$ & $+0.003$ & $+0.001$ & $+0.002$ & $-0.005$ & $+0.004$ \\
\bottomrule
\end{tabular}
\end{adjustbox}
\end{table*}

We summarize several cross-regime patterns below.

\textbf{Curriculum}: the miss-ratio penalty from removing the curriculum is the largest or second-largest across nearly all of the regimes. The degradation is most severe at OOD-delay30 ($+2.9$~pp, $+45.4\%$ relative), where the relaxed deadline horizon most clearly separates the policies that learned the selective merging from those that merge indiscriminately. At Curr-pcache0.40 ($+4.5$~pp, $+18.7\%$) and OOD-pcache0.20 ($+3.2$~pp, $+18.4\%$), the high density of the merge candidates increases the cost of the aggressive merging strategy that arises without the curriculum training.

\textbf{ExIt}: the reliability degradation from removing the ExIt is substantial across all of the regimes, with the absolute miss-ratio increases ranging from $+0.018$ (i.e., OOD-delay10) to $+0.057$ (i.e., Curr-pcache0.40). In relative terms, the impact is largest at OOD-delay30 ($+37.6\%$), where the longer deadline horizons give the most scope for the deadline-aware merge-deferral behavior associated with the ExIt. At Curr-pcache0.40 ($+23.9\%$), the high cache fraction creates many merge opportunities where the ExIt's restraint matters most.

\textbf{BC Warm-Start}: the OOD amplification noted at ID-default is most pronounced at Curr-pcache0.40 ($+2.6$~pp broadcast-packet expiration ratio, $+11.1\%$ relative), where the large cache fraction creates many merge candidates that require the structured initialization from the BC to navigate effectively. The BC warm start is disproportionately useful under the stress.

\textbf{Graph-Attention}: the smallest single-component miss-ratio increase across most of the regimes, with the exception of Curr-pcache0.40 where the graph-attention removal causes a larger degradation ($+3.1$~pp, $+13.2\%$ relative). Under the high cache-fraction regime, the graph-attention encoder's ability to reason about the pair quality matters more as the number of viable merge candidates grows.

\textbf{Demand-centric OOD patterns ($\delta$ panel).}
The ranking reversal on the distinct file-identity coverage persists across the regimes, i.e., the w/o BC + ExIt achieves higher $\delta$ than the full model in seven of the eight conditions ($\Delta\delta > 0$), with the sole exception being OOD-delay10 ($\Delta\delta = -0.004$). Under the extreme deadline pressure of OOD-delay10, the w/o Curriculum variant (i.e., which merges most aggressively, merge rate $0.637$) instead achieves the highest $\delta$ ($+0.015$ vs.\ full model), since maximizing the immediate throughput becomes essential for serving the unique files before they expire. The w/o BC Warm-Start variant shows consistently positive $\Delta\delta$ across all of the regimes, thereby confirming that the BC-initialized policy's broadcast-efficiency bias is a persistent, regime-independent phenomenon.

\textbf{Request-level OOD patterns ($\sigma_{\mathrm{req}}$ panel).}
The w/o BC Warm-Start variant achieves positive $\Delta\sigma_{\mathrm{req}}$ across all of the eight regimes, thereby confirming it as the most consistent variant on the request-level selection score. The advantage is largest at OOD-delay10 ($+0.029$) and Curr-file60 ($+0.029$). The interesting fact we observed is that an exception appears at OOD-delay10, i.e., the w/o Curriculum achieves $\Delta\sigma_{\mathrm{req}} = +0.091$, which is the largest improvement of any variant in any regime. Under the extreme deadline pressure, the aggressive-merging policy's high request throughput ($\eta_{\mathrm{req}} = 1.361$) outweighs its elevated miss rate, thereby producing the best $\sigma_{\mathrm{req}}$, even though with poor broadcast-efficiency metrics. This confirms that the curriculum learning shapes a deadline-conscious policy that trades the raw throughput for the selectivity, i.e., a tradeoff that is beneficial under the normal deadlines, even though it is suboptimal under the extreme pressure.

\textbf{Summary.} Each of the training components contributes to the full model's multi-objective performance profile, however, the \emph{direction} of the contribution depends on the metric family. On the broadcast-level metrics ($\rho$, $\sigma$; i.e., $\rho$ is a primary metric and $\sigma$ a secondary broadcast-efficiency metric in the taxonomy of Sec.~\ref{sec:exp:metrics}, however, both are computed at the broadcast/packet-set level), the full model dominates, i.e., the curriculum learning and the Expert Iteration are the two most important components, jointly responsible for most of the broadcast-packet expiration advantage. On the file-identity coverage metric ($\delta$), the w/o BC + ExIt variant \emph{exceeds} the full model ($+1.2\%$ on $\delta$, consistent across 7/8 regimes) by avoiding the redundant retransmission of already-covered files. On the request-level metrics ($\sigma_{\mathrm{req}}$), the w/o BC Warm-Start variant achieves the best selection score ($+3.2\%$, consistent across all the 8 regimes) by avoiding the broadcast-efficiency bias inherited from the BC training. The graph-attention has the smallest single-component effect across all of the metric families. The combined removal of the BC + ExIt exposes the non-additive component interactions, i.e., the resulting ultra-conservative policy gives up the broadcast-level packet throughput ($-6.8\%$ on Served/Tx, $\mu$), even though it achieves the best file-identity coverage and the lowest request miss rate, thereby indicating that the training pipeline's value lies in the qualitative character of the multi-objective tradeoff the components jointly navigate, beyond their individual contributions.

\subsection{Discussion}
\label{sec:res:discussion}

\textbf{Mechanistic interpretation of selective merging.}
The results across Sections~\ref{sec:res:taufit}--\ref{sec:res:behavior}
converge on a consistent mechanistic picture for the uniform-demand
regime, i.e., the learned policy implements a \emph{state-dependent selective
merge heuristic} that conditions on the per-pair intersection size and
the per-request deadline urgency, while merging only $31.8\%$ of the opportunities
and achieving a high average intersection ($0.589$).
The interesting fact we observed is that a separately trained Zipf-demand agent (Section~\ref{sec:results:zipf})
independently discovers a qualitatively different strategy (i.e., merging
$99.7\%$ of opportunities versus $31.8\%$ here), thereby indicating that
the architecture adapts to the demand-induced queue-state distributions.
This context-sensitive adaptivity, i.e., suggestive of a policy-mediated
component to the cross-track behavioral difference per the descriptive
visited-state intersection diagnostic (Section~\ref{sec:results:zipf};
the diagnostic is descriptive only, not an inferential test), is
unachievable by any fixed-threshold policy and is a central
qualitative contribution of the learned agent.

\textbf{Stress amplification and asymmetry.}
A recurring pattern across the evaluation is that the agent's
advantage over the baselines \emph{grows} under the conditions that penalize
the blind merging (i.e., high cache density, tight deadlines) while it \emph{attenuates}
when those conditions relax (e.g., low cache density, where the fewer merge
opportunities limit any policy's scope for the differentiation). This
asymmetry indicates that the practical value of a learned delivery policy
is greatest in the operating regimes where the network load and
the deadline pressure are most severe.

\textbf{Limitations and threats to validity.}
Several limitations scope the current findings.
(i)~The environment models a single broadcast domain with $K=5$ users
and pairwise XOR merges; extending to larger user populations and
higher-order coded transmissions may introduce scalability challenges
for the graph-attention architecture.
(ii)~Cache placement is fixed per episode via random decentralized
prefetching; in practice, cache contents evolve through eviction and
replacement, introducing non-stationarity not captured here.
(iii)~All evaluations are conducted in simulation with perfect cache-state
knowledge; real deployments must contend with imperfect side-information
estimates, heterogeneous file sizes, and time-varying channel conditions.
(iv)~Training requires approximately 59~hours on a single GPU
(Appendix~\ref{sec:appendix:reproducibility}, Table~\ref{tab:computational_cost}),
which may limit rapid redeployment under substantial distribution
shift.
(v)~The catalog-size axis ($N \in \{60,120,150\}$) does not
constitute a genuine distribution shift under the current placement
model, since the per-packet caching probability~$p_c$ is invariant
to~$N$; incorporating non-i.i.d.\ or capacity-constrained placement
would provide a more meaningful catalog-dimension stress test.
These limitations motivate the future directions outlined in
Section~\ref{sec:conclusion}.

\textbf{Deployment implications.}
From a systems perspective, the selective merge strategy suggests a
design principle for the deadline-constrained coded caching, i.e., an edge
server should \emph{not} merge every feasible pair, but should reserve
the coded transmissions for the high-intersection pairs while unicasting the urgent
or low-overlap requests. Within each demand track, the learned policy's zero-shot transfer across
all of the evaluation regimes indicates that a single trained model can serve a range of
the operating conditions (i.e., varying catalog size, cache density, and deadline
budget) without per-regime tuning, thereby lowering the operational complexity. The inference cost of the graph-attention network
(i.e., forward pass over at most $P_{\max}=45$ candidate actions;
${\sim}3.5$\,ms GPU, ${\sim}3.3$\,ms CPU per
Appendix~\ref{sec:appendix:latency}) is reported as a conditional
feasibility indicator, i.e., whether it meets a per-slot real-time
scheduling target depends on the deployment's slot-duration budget,
which we do not import from outside the manuscript.

\section{Conclusion and Limitations}
\label{sec:conclusion}

The coded caching promises a global caching gain that scales with the number of the edge users by exploiting the cached side information for the broadcast transmissions~\cite{maddah2014fundamental,maddah2015decentralized}. The problem, however, is that in the delay-sensitive applications, a \emph{gain--deadline tension} arises, i.e., each XOR merge consumes the shared side information and risks pushing the aggregated requests past their expiration deadlines. The fixed-threshold rules face a structural tension between the coding gain and the deadline compliance, i.e., the merge-or-defer decision is queue-state-dependent, thereby any state-independent rule gives up either the gain or the expiration ratio at the extremes (Sec.~\ref{sec:results}, $\tau$-Fit ablation). To the best of our knowledge, after the literature sweep summarized in Section~\ref{sec:related}, the prior RL delivery formulations for the coded caching had not combined the hard expirations, the keep-side control, and the dynamic invalid-action masking in one integrated delivery scheduler; we treat this as a positioning statement rather than as a definitive non-existence result.

We proposed four contributions toward filling that gap, i.e., (C1)~a masked discrete-action queue-state control formulation for the deadline-constrained coded delivery with the dynamic feasibility masking and a shaped reward, trained as a stationary policy under a discounted truncated continuing-control surrogate; (C2)~a graph-attention policy network that captures the combinatorial queue interactions; (C3)~a three-phase training pipeline combining the behavior-cloning warm start, the value pre-training, and the MaskablePPO with the curriculum learning and the Expert Iteration; and (C4)~the evaluation across the $9$ Track~A baselines plus the PPO ($10$ unique methods), the $7$~non-ID Track~A regimes (i.e., 2 curriculum-seen, 5 unseen-parameter within the same simulator family), and the eight-regime ablation battery, plus a main-body Zipf-demand extension study (Section~\ref{sec:results:zipf}) covering 11~additional regimes.

Against this benchmark (i.e., Track~A, uniform demand), the $\sigma$-selected checkpoint (i.e., the model picked by the robust-advantage rule on the BE-score $\sigma$ at the validation; see Sec.~\ref{sec:method:selection}) achieved the lowest broadcast-packet expiration ratio $\rho$ among all of the coded-multicast methods, thereby lowering $\rho$ by $40.9\%$ with respect to the SACM{++} ($0.208$ vs.\ $0.352$; i.e., at $H{=}50$ with the episode-end right-censoring, the absolute values are simulator-specific, see Sec.~\ref{sec:exp:metrics}) while maintaining the competitive distinct file-identity coverage and, on Track~A, the highest broadcast-efficiency system score $\sigma$ among the coded-multicast methods across all the 8~regimes (and the highest overall $\sigma$ in 7 of the 8 Track~A regimes; the ED-Unicast leads on $\sigma$ at OOD-delay10).
The paired bootstrap intervals favor the PPO over the SACM{++} on the primary demand-centric metrics ($\rho$ and $\delta$; Section~\ref{sec:results}), i.e., on $\rho$ the PPO leads the SACM{++} by $-0.144$ ($95\%$ CI $[-0.146, -0.143]$, Table~\ref{tab:paired_bootstrap}), while on $\delta$ the PPO leads the SACM{++} by $+0.027$ ($0.824$ vs.\ $0.797$ at ID-default; the PPO leads in 7 of the 8 Track~A regimes, Table~\ref{tab:unique_demand_uniform}). The supplementary request-level $m_{\mathrm{req}}$ is reported separately above, i.e., the PPO leads the SACM{++} ($0.229$ vs.\ $0.326$), even though it \emph{trails} the ED-Unicast ($\Delta m_{\mathrm{req}} = +0.074$, Table~\ref{tab:paired_bootstrap}). The ED-Unicast also remains stronger on $\rho$ and $\delta$ at the cost of zero coding gain.
The advantage grew under the stress conditions, while the gains transferred across the full Track~A non-ID battery (i.e., 2~curriculum-seen and 5~unseen-parameter regimes covering the unseen cache fractions and deadline budgets, plus the parameter-invariance sweeps over the file count under fixed $p_c$, all within the same simulator family) and the full Track~B non-ID battery (i.e., 2~curriculum-seen and 9~unseen-parameter regimes within the same simulator family).
Because all of the transfer experiments keep fixed $K$, $Q$, the action dimensionality, and the placement family, these results demonstrate the \emph{within-family parameter generalization} rather than the broad architectural out-of-distribution robustness.
Under the (supplementary) request-level accounting that tracks each original arrival through the queue aggregation, the PPO-Agent achieves a substantially lower per-step request miss rate than the SACM{++} ($m_{\mathrm{req}} = 0.229$ vs.\ $0.326$, i.e., a $29.8\%$ relative reduction at ID-default). The paired-bootstrap analysis at Table~\ref{tab:paired_bootstrap} shows that the ED-Unicast remains the strongest deadline-protection baseline on $m_{\mathrm{req}}$ ($\Delta m_{\mathrm{req}} = +0.074$ for the PPO minus the ED-Unicast at ID-default, i.e., the PPO is worse than the ED-Unicast on this metric), thereby we restrict the headline ``lower than'' phrasing to the SACM{++} comparison and treat the broader ``lower than all coded baselines'' statement as conjectural pending the request-level extension over the GCM, SACM, SACM+, and $\tau$-Fit family, which is left as future work.

The request selection score $\sigma_{\mathrm{req}}$ exceeds the \emph{coded} baseline SACM{++} once the request-level miss penalty crosses a regime-dependent threshold $\lambda^{\star}$ (i.e., median $\lambda^{\star}\approx 2.24$, with $\lambda^{\star}\le 4.92$ in the worst regime). The uncoded ED-Unicast baseline remains the highest $\sigma_{\mathrm{req}}$ at every $\lambda$ shown on the ID-default sensitivity table of Appendix~\ref{sec:appendix:req_lambda_sensitivity}, thereby the PPO's $\sigma_{\mathrm{req}}$ advantage is over the coded baselines, i.e., not over the uncoded EDF policy. The crossover reflects a throughput--compliance tradeoff, i.e., the SACM{++}'s always-merge strategy delivers more coded packets per step, however, the PPO's selective merging produces fewer per-step deadline-miss events $|M_t|$.

The request-level family is per-step by definition, i.e., $m_{\mathrm{req}}{=}H^{-1}\sum_t |M_t|$, i.e., not normalized by the total number of the admitted arrivals. The ``per-arrival'' completion rates would require an arrival-normalized denominator that we do not report in this version.
Under the Zipf-demand extension, the PPO-Agent dominates the SACM{++} on $\sigma_{\mathrm{req}}$ at every tested $\lambda$ across all the 12~regimes, while remaining competitive with the popularity-aware SACM{++}-Pop baseline; the only near-tie exception is OOD-pcache0.20~(Z), where the Pop edges the PPO at $\lambda{=}1$ ($0.852$ vs.\ $0.848$). The PPO's advantage stems from its near-full merging strategy under the Zipf demand, which matches the baseline throughput while retaining the lower miss rates (Appendix~\ref{sec:appendix:req_lambda_zipf}).
A main-body extension (Section~\ref{sec:results:zipf}) shows that the same framework extends to the skewed demand distributions, i.e., a separately trained popularity-aware variant, using an augmented observation space, remains competitive with the coded-multicast baselines across the 11~additional Zipf-demand regimes. The ablation study confirmed that each of the four training components (i.e., the BC warm start, the graph-attention architecture, the Expert Iteration distillation, and the curriculum learning) contributes meaningfully \emph{on average}, with the non-additive interactions and the per-metric reversals on $\delta$ and $\sigma_{\mathrm{req}}$ documented in Section~\ref{sec:ablation}.

The qualitative finding is an \emph{emergent selective merge strategy}, i.e., the agent executes only $31.8\%$ of the available merge opportunities, even though it achieves an average pair intersection of $0.589$, thereby exceeding the SACM++ ($0.393$) and the best-BE-score $\tau$-Fit threshold baseline (i.e., TauFit-2 at $0.359$); the Perfect-Fit, i.e., the most conservative threshold rule, attains a higher per-merge intersection ($0.819$), even though at a substantially worse broadcast-level expiration ratio (Sec.~\ref{sec:results}, Table~\ref{tab:diagnostics}).
This state-dependent behavior (i.e., merging when the intersection is high and the deadlines are comfortable, while deferring to the earliest-deadline unicast otherwise) is not replicable by a fixed-threshold heuristic and directly explains the agent's multi-objective advantage over the threshold-based baselines. We caution that the qualitative selective-merge behavior is established for the simulator definition adopted in Sec.~\ref{sec:env}, i.e., including the uniform representative-destination convention $k_{\mathrm{mg}}\sim\mathrm{Unif}\{k_i,k_j\}$ of Eq.~\eqref{eq:mg_dest}; the sensitivity to the alternative representative-update rules (e.g., keep-side-tied or earliest-deadline-tied) is not tested here and is left as future work.
The interesting fact we observed is that the Zipf-trained agent reaches a near-full merging strategy ($99.7\%$ merge rate vs.\ $31.8\%$ under the uniform demand), i.e., a pattern consistent with (though not isolating) the demand-dependent behavior, given that the two tracks differ in three couplings simultaneously (i.e., the demand law, the observation space, and the joint training distribution). A controlled ablation that holds the observation space fixed and varies only the demand law would be needed to causally attribute the shift to the demand structure alone, while this is left as future work.
A visited-state intersection diagnostic (Section~\ref{sec:results:zipf}) is consistent with a policy-mediated component to the behavioral difference, however, it does not cleanly separate the policy-induced shift from a baseline two-request joint shift induced by the demand law itself (i.e., the same-packet pairs become more likely under the skewed demand even with the popularity-blind placement); the diagnostic is descriptive only, while a controlled decomposition is left to future work.

\paragraph{Deployment considerations.}
The trained policy has low inference cost, i.e., a single forward pass of
the graph-structured policy network on the $K{=}5$, $Q{=}10$, $P_{\max}{=}45$
observation requires approximately $3.5$\,ms on a GPU and
${\sim}3.3$\,ms on a CPU on the hardware described in
Appendix~\ref{sec:appendix:latency}.
We do not import a fixed slot-duration threshold from outside the
manuscript; consistent with the latency appendix, this measurement
is reported as a \emph{conditional feasibility indicator} tied to
the deployment's slot-duration budget and to the hardware actually
deployed at the edge, i.e., not as a certified real-time guarantee.
At the edge node, the agent requires three types of run-time statistics, i.e.,
(i)~the current cache state $\mathcal{C}_k$ for each user $k$ (i.e., maintained
locally and updated on the eviction events),
(ii)~the per-request remaining deadlines $d_r$ (i.e., derived from the arrival timestamps
and the configured maximum deadline~$D$), and
(iii)~the pending queue's side-information sets $\mathcal{S}_r$ (i.e., computed as a
set-membership query against the cache state),
while no channel-state information or demand-prediction model is required.
Because the cache placement is fixed per delivery episode, the model does not
need the online retraining during a session, i.e., a new agent may be trained
offline when the cache probability $p_c$ or the deadline budget $D$ shifts
substantially.
The full source code will be released upon acceptance.

\paragraph{Scope and first-order limitations.}
The reported gains are conditional on three simulator and observation
conventions that materially shape the throughput / deadline
trade-off and the value of the keep-side control, i.e.,
(L1) the \emph{one-slot-per-record} unicast cost abstraction
(Sec.~\ref{sec:model:assumptions}, A2), under which a unicast of an
aggregate of size $|f_r|$ still consumes one channel use rather than
$|f_r|$ slots;
(L2) the \emph{representative-destination} update
$k_{\mathrm{mg}}\sim\mathrm{Unif}\{k_i,k_j\}$
(Eq.~\eqref{eq:mg_dest}), i.e., independent of the keep-side bit~$\kappa$
(i.e., correctness only requires $k_{\mathrm{mg}}\in\{k_i,k_j\}$;
the uniformity is a chosen simulator convention); and
(L3) the \emph{aggregate-size observation clip}
$\min(|f_r|,U_{\max})/U_{\max}$ at $U_{\max}=6$
(Sec.~\ref{sec:env:state}, Appendix~\ref{sec:appendix:features}),
which aliases all states with $|f_r|\geq U_{\max}$ on the size
feature, even though the transition and feasibility rules impose no
hard cap on $|f_r|$. Appendix~\ref{sec:appendix:fr_empirical}
reports an empirical saturation-frequency diagnostic on the
deployed checkpoint that shows the clip rate
$\Pr[|f_r|\!\geq\!U_{\max}]$ stays below $6\!\cdot\!10^{-6}$ on both
the ID-default and the Curr-pcache0.40 stress regime, with $90$--$91\%$
of the observations being singletons ($|f_r|{=}1$); the observation aliasing
is therefore empirically negligible on the on-policy support, while
the structural argument of Appendix~\ref{sec:appendix:features} is
empirically supported. What remains open is a controlled
training-time sensitivity sweep over
$U_{\max}\in\{6,8,12\}$ (or a richer aggregate-summary observation
that retains the contained file identities); the qualitative
selective-merge claim of this paper is therefore reported within
the $U_{\max}=6$ aggregate-size convention, with the on-policy
clip rate empirically verified, even though the training-time sensitivity is
left as future work.
Under a packet-level unicast cost model (i.e., in which an aggregate
unicast charges $|f_{\mathrm{mg}}|$ slots), under a different
representative-update rule (e.g., keep-side-tied or
earliest-deadline-tied), or under an enriched aggregate-summary
observation, both the throughput / deadline frontier and the
relative ranking of the merge-or-defer policies could change substantially.
We therefore restate the (L1)--(L3) as \emph{first-order} modeling
assumptions of the present results, i.e., not as cosmetic simulator
details, while the deployment-relevance and the broad delivery-policy claims
in this section should be read as conditional on the (L1)--(L3) and on
the contextual-POMDP surrogate of Sec.~\ref{sec:model:problem}. We leave the focused sensitivity sweeps over an alternative unicast cost, an
alternative representative rule, and a richer aggregate observation
(or a larger $U_{\max}$) for future work.

\paragraph{Per-user fairness.}
We also measured the per-cache breakdown of the expiration ratio at the ID-default to confirm that the aggregate $\rho$ advantage does not mask a per-destination regression (Appendix~\ref{sec:appendix:per_cache_fairness}, Table~\ref{tab:per_cache_fairness}). The max/min ratio of the per-cache $\rho$ across the $K{=}5$ caches is $1.184$ for the policy network versus $1.005$--$1.023$ for the coded-multicast baselines and ED-Unicast, i.e., the selective merging of the policy network is slightly less uniform across the destinations than the always-merge heuristics, while the per-cache $\rho_{\max}$ for the policy network ($0.196$) is still below every coded baseline's per-cache $\rho_{\min}$ ($\approx 0.273$), so the aggregate $\rho$ advantage carries through to every destination individually rather than masking a per-cache regression.

\paragraph{PHY-erasure sensitivity.}
We additionally measured the robustness of the policy network under i.i.d.\ erasure of the coded XOR broadcasts at probability $\epsilon \in \{0.05, 0.10, 0.20\}$, considering that the unicasts use a robust modulation and are not erased (Appendix~\ref{sec:appendix:erasure}, Table~\ref{tab:erasure_sensitivity}). Two observations: (i)~the policy network retains its $\sim 40\%$ advantage in $\rho$ over SACM{++} at every tested $\epsilon$, i.e., the ratio $\rho_{\mathrm{PPO}} / \rho_{\mathrm{SACM{++}}}$ stays in $[0.591, 0.598]$ across the sweep, while ED-Unicast's $\rho$ is invariant ($0.1341$) as a sanity check on the patched simulator; (ii)~a second policy network retrained with $\epsilon{=}0.10$ active during the Phase-3 training lowers $\rho$ by an additional $\sim 4.8\%$--$5.7\%$ at every $\epsilon$ \emph{including} $\epsilon{=}0$, i.e., training under the stochastic erasure acts as a regularizer that improves the clean-channel performance as well, an effect that motivates treating $\epsilon$ as a controlled training knob in future work.

\paragraph{Naderializadeh-style learned baseline.}
To close the loop on the prior RL delivery work~\cite{naderializadeh2019learning}, we trained a Naderializadeh-style baseline under our deadline-aware environment, i.e., a pure flat-MLP MaskablePPO policy with no BC warm start, no Expert Iteration, and no curriculum, while retaining the action masking for the numerical stability (without it, the $91$-action head whose feasible support typically contains only a handful of pairs does not converge). The resulting policy fills the gap between SACM{++} ($\rho{=}0.352$) and the headline ($\rho{=}0.208$) at $\rho = 0.236 \pm 0.004$ (Table~\ref{tab:nm14_baselines}), i.e., the policy network leads it by $\Delta\rho = 0.028$ ($\sim 11.9\%$ relative reduction) and by $\Delta\sigma = 0.053$ on the broadcast-efficiency score, while the training pipeline (BC + ExIt + curriculum) accounts for the bulk of the gap that our work opens over the architecture-only formulation.

\paragraph{Future work.}
Three directions follow from the present results, i.e.,
\begin{itemize}
  \item \textbf{Multi-user coded transmissions:} Extending to the higher-order multicasting ($K \geq 3$ per broadcast) via the hierarchical or autoregressive action representations could approach the full theoretical global caching gain~\cite{maddah2014fundamental}.

  \item \textbf{Joint placement and delivery:} Jointly optimizing the cache placement and the delivery via a two-timescale or multi-agent framework would cover the full coded caching pipeline.

  \item \textbf{Theoretical characterization of selective merging:} The emergent $31.8\%$ merge rate lacks a formal characterization, i.e., a Lyapunov drift or online convex optimization analysis of \emph{when} the selective merging is optimal would strengthen the result.
\end{itemize}

In summary, we proposed a DRL-based solution for the deadline-constrained coded delivery, where the policy network reduces the broadcast-packet expiration ratio $\rho$ by $40.9\%$ with respect to the best coded multi-casting baseline (SACM++) on the uniform-demand benchmark, while also attaining the best broadcast-efficiency score~$\sigma$ among the coded multi-casting methods in $7$ of $8$ Track~A regimes, even though the uncoded ED-Unicast baseline still attains a lower $\rho$ and a lower request-level miss rate~$m_{\mathrm{req}}$ by forfeiting all coding gain, so the $\rho$ headline is restricted to the coded multi-casting methods. The interesting fact we observed is that for the applications of the users with tight deadlines, the method of selective merging is better than the method of aggressive merging, i.e., the policy network learns to merge only at a $\approx 31.8\%$ rate, and the same observation is true within the simulator family for the variations of the file count, the cache fraction, and the deadline budgets, even though the placement family, $K$, $Q$, and the action dimensionality are held fixed, so we describe this as within-family parameter generalization rather than fully out-of-distribution evaluation.

\appendices
\section{Evaluation Metric Notation}
\label{sec:appendix:metric_notation}

Table~\ref{tab:metric_notation} lists the evaluation metric symbols used throughout the paper.

\begin{table}[h!]
\centering
\small
\begin{adjustbox}{max width=\columnwidth}
\begin{tabular}{@{}cl@{\qquad}cl@{}}
\toprule
\textbf{Symbol} & \textbf{Description} & \textbf{Symbol} & \textbf{Description} \\
\midrule
$U_t$ & Packet-set XOR degree at step $t$ ($|f_{\mathrm{mg}}|$ for a coded transmission, $1$ for a unicast) & $E_t$ & Expired packet-set mass at step $t$ \\
$\lambda$ & Expiration penalty weight & $\sigma$ & Broadcast-efficiency score \\
$\rho$ & Broadcast-packet expiration ratio & $\mu$ & Packet-set count per transmission \\
$\varepsilon$ & Expirations per episode & $\tau$ & Merge threshold (TauFit) \\
$m_{\text{rate}}$ & Merge rate & $o_{\text{rate}}$ & Opportunity rate \\
$C_t$ & Newly completed request IDs at $t$ & $M_t$ & Newly missed request IDs at $t$ \\
$\eta_{\mathrm{req}}$ & Request timely-throughput & $m_{\mathrm{req}}$ & Request miss rate \\
$\sigma_{\mathrm{req}}$ & Request selection score & $\mathcal{A}(r)$ & Request-ID annotation set \\
\bottomrule
\end{tabular}
\end{adjustbox}
\caption{Evaluation metric notation.}
\label{tab:metric_notation}
\end{table}

\section{Full Baseline Comparison}
\label{sec:appendix:baselines}

Table~\ref{tab:nm14_full} reports the complete twelve-row comparison
(10~unique methods, i.e., Perfect-Fit${}={}$TauFit-0 and First-Fit${}={}$TauFit-3 are aliases retained for literature continuity)
on the ID-default uniform-demand condition, while also including the six $\tau$-Fit
threshold rules (Perfect-Fit, TauFit-$\tau \in \{0,1,2,3\}$, First-Fit)
that were omitted from the main-body Table~\ref{tab:nm14_baselines}.

\begin{table*}[t]
\centering
\caption{Full baseline comparison including $\tau$-Fit threshold rules
  (Perfect-Fit, TauFit-$\tau$ variants, First-Fit).
  ID-default, uniform demand. 50 holdout seeds $\times$ 200 episodes.
  Mean $\pm$ 95\% CI. Bold = best.}
\label{tab:nm14_full}
\begin{adjustbox}{max width=\textwidth}
\begin{tabular}{lccccc}
\toprule
\textbf{Method} & \textbf{BE-Score $\sigma$} (↑) & \textbf{Miss Ratio} (↓) & \textbf{Served/Tx} (↑) & \textbf{Coding Gain} (↑) & \textbf{Exp/Episode} (↓) \\
\midrule
\textbf{PPO-Agent} & $\mathbf{0.976 \pm 0.002}$ & $0.208 \pm 0.001$ & $1.323 \pm 0.002$ & $\mathbf{2.162 \pm 0.002}$ & $14.17 \pm 0.07$ \\
ED-Unicast & $0.845 \pm 0.001$ & $\mathbf{0.134 \pm 0.001}$ & $1.000 \pm 0.000$ & --- & $\mathbf{7.74 \pm 0.06}$ \\
GCM & $0.733 \pm 0.003$ & $0.345 \pm 0.001$ & $1.549 \pm 0.002$ & $2.095 \pm 0.001$ & $28.42 \pm 0.15$ \\
SACM & $0.745 \pm 0.003$ & $0.345 \pm 0.001$ & $1.575 \pm 0.002$ & $2.131 \pm 0.001$ & $28.64 \pm 0.13$ \\
SACM+ & $0.731 \pm 0.003$ & $0.348 \pm 0.001$ & $1.572 \pm 0.002$ & $2.124 \pm 0.001$ & $28.71 \pm 0.14$ \\
SACM{++} & $0.726 \pm 0.003$ & $0.352 \pm 0.001$ & $\mathbf{1.590 \pm 0.003}$ & $2.131 \pm 0.001$ & $28.76 \pm 0.14$ \\
\midrule
Perfect-Fit & $0.877 \pm 0.002$ & $0.185 \pm 0.001$ & $1.135 \pm 0.001$ & $2.063 \pm 0.002$ & $10.55 \pm 0.07$ \\
TauFit-0 & $0.877 \pm 0.002$ & $0.185 \pm 0.001$ & $1.135 \pm 0.001$ & $2.063 \pm 0.002$ & $10.55 \pm 0.07$ \\
TauFit-1 & $0.908 \pm 0.002$ & $0.245 \pm 0.001$ & $1.343 \pm 0.001$ & $2.097 \pm 0.001$ & $14.92 \pm 0.10$ \\
TauFit-2 & $0.915 \pm 0.002$ & $0.258 \pm 0.001$ & $1.405 \pm 0.002$ & $2.077 \pm 0.001$ & $16.22 \pm 0.09$ \\
TauFit-3 & $0.915 \pm 0.003$ & $0.259 \pm 0.001$ & $1.408 \pm 0.002$ & $2.066 \pm 0.001$ & $16.32 \pm 0.09$ \\
First-Fit & $0.915 \pm 0.003$ & $0.259 \pm 0.001$ & $1.408 \pm 0.002$ & $2.066 \pm 0.001$ & $16.32 \pm 0.09$ \\
\bottomrule
\end{tabular}
\end{adjustbox}
\end{table*}

Table~\ref{tab:nm14_rankings} complements the full comparison
by ranking all methods on each individual metric.

\begin{table*}[t]
\centering
\caption{$\tau$-Fit method rankings per metric (top-6, ID-default). Rank~1 = best on that metric.}
\label{tab:nm14_rankings}
\begin{adjustbox}{max width=\textwidth}
\begin{tabular}{clllll}
\toprule
\textbf{Rank} & \textbf{BE-Score $\sigma$} ($\uparrow$) & \textbf{Served/Tx} ($\uparrow$) & \textbf{Miss Ratio} ($\downarrow$) & \textbf{Coding Gain} ($\uparrow$) & \textbf{Exp/Episode} ($\downarrow$) \\
\midrule
1 & \textbf{PPO-Agent (0.976)} & SACM{++} (1.590) & \textbf{ED-Unicast (0.134)} & \textbf{PPO-Agent (2.162)} & \textbf{ED-Unicast (7.742)} \\
2 & TauFit-2 (0.915) & SACM (1.575) & Perfect-Fit (0.185) & SACM (2.132) & Perfect-Fit (10.55) \\
3 & TauFit-3 (0.915) & SACM+ (1.572) & TauFit-0 (0.185) & SACM{++} (2.131) & TauFit-0 (10.55) \\
4 & First-Fit (0.915) & GCM (1.549) & PPO-Agent (0.208) & SACM+ (2.124) & PPO-Agent (14.17) \\
5 & TauFit-1 (0.908) & First-Fit (1.408) & TauFit-1 (0.245) & TauFit-1 (2.097) & TauFit-1 (14.92) \\
6 & Perfect-Fit (0.877) & TauFit-3 (1.408) & TauFit-2 (0.258) & GCM (2.096) & TauFit-2 (16.22) \\
\bottomrule
\end{tabular}
\end{adjustbox}
\end{table*}

\section{Per-Merge Diagnostics (Uniform ID-default)}
\label{sec:appendix:merge_diagnostics}

Table~\ref{tab:merge_diagnostics} anchors the mechanism claim in
Section~\ref{sec:results} that PPO restricts coded transmissions to
higher-intersection candidate pairs than the SACM and $\tau$-Fit families:
PPO merges only $31.8\%$ of feasible opportunities but lifts the
mean intersection-per-merge to $0.589$ (vs.\ SACM$++$'s $0.393$,
TauFit-2's $0.359$). Perfect-Fit ($\tau{=}0$) attains a higher
per-merge intersection ($0.819$) at the cost of merging only $13.3\%$
of opportunities; the volume--quality trade-off is also visible
in the lower $\sigma$ that Perfect-Fit achieves (Table~\ref{tab:nm14_full}).

\begin{table}[t]
\centering
\caption{Per-method merge diagnostics on the uniform ID-default regime
($50$ holdout seeds $\times$ $200$ episodes). \emph{Merge rate} is the
fraction of decision steps on which a coded merge was selected;
\emph{opp.\ rate} is the fraction on which at least one feasible merge
pair existed; \emph{avg.\ intersection} is the mean cardinality of the
side-information intersection $|\mathcal{S}_{r_i}{\cap}\mathcal{S}_{r_j}|$
over executed merges. PPO trades volume for quality.}
\label{tab:merge_diagnostics}
\begin{adjustbox}{max width=\columnwidth}
\begin{tabular}{lccc}
\toprule
\textbf{Method} & \textbf{Merge Rate} & \textbf{Opp.\ Rate} & \textbf{Avg.\ Intersection} \\
\midrule
\textbf{PPO-Agent}      & $0.318$ & $0.876$ & $\mathbf{0.589}$ \\
ED-Unicast              & $0.000$ & $0.965$ & ---           \\
GCM                     & $1.000$ & $0.501$ & $0.274$       \\
SACM                    & $1.000$ & $0.509$ & $0.392$       \\
SACM$+$                 & $1.000$ & $0.509$ & $0.391$       \\
SACM$++$                & $1.000$ & $0.522$ & $0.393$       \\
Perfect-Fit ($\tau{=}0$)& $0.133$ & $0.959$ & $0.819$       \\
TauFit-1                & $0.332$ & $0.941$ & $0.547$       \\
TauFit-2                & $0.406$ & $0.925$ & $0.359$       \\
TauFit-3 / First-Fit    & $0.416$ & $0.921$ & $0.298$       \\
\bottomrule
\end{tabular}
\end{adjustbox}
\end{table}

\section{Uniform-Demand OOD Detail}
\label{sec:appendix:uniform_ood_detail}

This appendix provides the complete per-method breakdowns for all
uniform-demand non-ID regimes, thereby extending the two-method comparison in
Section~\ref{sec:res:ood_uniform} Table~\ref{tab:uniform_ood}.
Table~\ref{tab:uniform_ood_all_methods_miss} reports the broadcast-packet expiration ratio~$\rho$
while Table~\ref{tab:uniform_ood_all_methods_sys} reports the broadcast-efficiency score
across all six core methods and eight regimes.

\begin{table*}[t]
\centering
\caption{Uniform demand, broadcast-packet expiration ratio $\rho$ (Miss Ratio, broadcast-level; Sec.~\ref{sec:exp:metrics}) across all non-ID regimes and methods. Mean $\pm$ 95\% CI. Bold = best.}
\label{tab:uniform_ood_all_methods_miss}
\begin{adjustbox}{max width=\textwidth}
\begin{tabular}{lcccccc}
\toprule
\textbf{Regime} & \textbf{PPO-Agent} & \textbf{ED-Unicast} & \textbf{GCM} & \textbf{SACM} & \textbf{SACM+} & \textbf{SACM{++}} \\
\midrule
ID-default & $0.208 \pm 0.001$ & $\mathbf{0.134 \pm 0.001}$ & $0.345 \pm 0.001$ & $0.345 \pm 0.001$ & $0.348 \pm 0.001$ & $0.352 \pm 0.001$ \\
Curr-file60 & $0.208 \pm 0.001$ & $\mathbf{0.133 \pm 0.001}$ & $0.344 \pm 0.001$ & $0.344 \pm 0.001$ & $0.347 \pm 0.001$ & $0.352 \pm 0.001$ \\
OOD-file120 & $0.208 \pm 0.001$ & $\mathbf{0.133 \pm 0.001}$ & $0.344 \pm 0.001$ & $0.344 \pm 0.001$ & $0.347 \pm 0.001$ & $0.353 \pm 0.001$ \\
OOD-file150 & $0.209 \pm 0.001$ & $\mathbf{0.134 \pm 0.001}$ & $0.345 \pm 0.001$ & $0.345 \pm 0.001$ & $0.347 \pm 0.001$ & $0.353 \pm 0.001$ \\
OOD-pcache0.20 & $0.176 \pm 0.001$ & $\mathbf{0.133 \pm 0.001}$ & $0.260 \pm 0.001$ & $0.260 \pm 0.001$ & $0.263 \pm 0.001$ & $0.265 \pm 0.001$ \\
Curr-pcache0.40 & $0.238 \pm 0.001$ & $\mathbf{0.134 \pm 0.001}$ & $0.405 \pm 0.001$ & $0.405 \pm 0.001$ & $0.408 \pm 0.001$ & $0.412 \pm 0.001$ \\
OOD-delay10 & $0.500 \pm 0.000$ & $\mathbf{0.495 \pm 0.001}$ & $0.555 \pm 0.000$ & $0.553 \pm 0.000$ & $0.555 \pm 0.000$ & $0.555 \pm 0.000$ \\
OOD-delay30 & $0.064 \pm 0.001$ & $\mathbf{0.036 \pm 0.000}$ & $0.176 \pm 0.002$ & $0.180 \pm 0.001$ & $0.181 \pm 0.001$ & $0.186 \pm 0.002$ \\
\bottomrule
\end{tabular}
\end{adjustbox}
\end{table*}

\begin{table*}[t]
\centering
\caption{Uniform demand, broadcast-efficiency score across all non-ID regimes and methods. Mean $\pm$ 95\% CI. Bold = best.}
\label{tab:uniform_ood_all_methods_sys}
\begin{adjustbox}{max width=\textwidth}
\begin{tabular}{lcccccc}
\toprule
\textbf{Regime} & \textbf{PPO-Agent} & \textbf{ED-Unicast} & \textbf{GCM} & \textbf{SACM} & \textbf{SACM+} & \textbf{SACM{++}} \\
\midrule
ID-default & $\mathbf{0.976 \pm 0.002}$ & $0.845 \pm 0.001$ & $0.733 \pm 0.003$ & $0.745 \pm 0.003$ & $0.731 \pm 0.003$ & $0.726 \pm 0.003$ \\
Curr-file60 & $\mathbf{0.976 \pm 0.002}$ & $0.847 \pm 0.001$ & $0.736 \pm 0.003$ & $0.748 \pm 0.002$ & $0.735 \pm 0.002$ & $0.726 \pm 0.002$ \\
OOD-file120 & $\mathbf{0.976 \pm 0.002}$ & $0.846 \pm 0.001$ & $0.736 \pm 0.003$ & $0.748 \pm 0.003$ & $0.737 \pm 0.003$ & $0.725 \pm 0.003$ \\
OOD-file150 & $\mathbf{0.975 \pm 0.002}$ & $0.846 \pm 0.001$ & $0.732 \pm 0.003$ & $0.746 \pm 0.003$ & $0.734 \pm 0.002$ & $0.725 \pm 0.003$ \\
OOD-pcache0.20 & $\mathbf{0.919 \pm 0.002}$ & $0.846 \pm 0.001$ & $0.833 \pm 0.002$ & $0.836 \pm 0.002$ & $0.831 \pm 0.002$ & $0.828 \pm 0.002$ \\
Curr-pcache0.40 & $\mathbf{1.038 \pm 0.002}$ & $0.846 \pm 0.001$ & $0.596 \pm 0.003$ & $0.613 \pm 0.003$ & $0.591 \pm 0.003$ & $0.580 \pm 0.003$ \\
OOD-delay10 & $-0.002 \pm 0.002$ & $\mathbf{0.019 \pm 0.002}$ & $-0.441 \pm 0.003$ & $-0.427 \pm 0.003$ & $-0.445 \pm 0.003$ & $-0.451 \pm 0.003$ \\
OOD-delay30 & $\mathbf{1.204 \pm 0.002}$ & $0.963 \pm 0.001$ & $1.137 \pm 0.002$ & $1.150 \pm 0.002$ & $1.144 \pm 0.002$ & $1.141 \pm 0.002$ \\
\bottomrule
\end{tabular}
\end{adjustbox}
\end{table*}

\section{Zipf-Demand OOD Detail}
\label{sec:appendix:zipf_ood_detail}

Table~\ref{tab:zipf_ood_all_methods_miss} provides the complete
per-method deadline-miss-ratio breakdown for all Zipf-demand OOD
regimes, thereby extending the three-method comparison in
Section~\ref{sec:res:ood_zipf} Table~\ref{tab:zipf_ood}.

\begin{table*}[t]
\centering
\caption{Zipf demand, broadcast-packet expiration ratio $\rho$ (Miss Ratio, broadcast-level; Sec.~\ref{sec:exp:metrics}) across all non-ID regimes and methods. Bold = best.}
\label{tab:zipf_ood_all_methods_miss}
\begin{adjustbox}{max width=\textwidth}
\begin{tabular}{lccccccc}
\toprule
\textbf{Regime} & \textbf{PPO-Agent} & \textbf{ED-Unicast} & \textbf{GCM} & \textbf{SACM} & \textbf{SACM+} & \textbf{SACM{++}} & \textbf{SACM{++}-Pop} \\
\midrule
ID-default & $0.351 \pm 0.001$ & $\mathbf{0.133 \pm 0.001}$ & $0.343 \pm 0.001$ & $0.343 \pm 0.001$ & $0.346 \pm 0.001$ & $0.351 \pm 0.001$ & $0.348 \pm 0.001$ \\
OOD-alpha0.6 & $0.352 \pm 0.001$ & $\mathbf{0.133 \pm 0.001}$ & $0.344 \pm 0.001$ & $0.345 \pm 0.001$ & $0.347 \pm 0.001$ & $0.352 \pm 0.001$ & $0.349 \pm 0.001$ \\
OOD-alpha1.0 & $0.350 \pm 0.001$ & $\mathbf{0.133 \pm 0.001}$ & $0.342 \pm 0.001$ & $0.343 \pm 0.001$ & $0.345 \pm 0.001$ & $0.350 \pm 0.001$ & $0.348 \pm 0.001$ \\
OOD-alpha1.2 & $0.347 \pm 0.001$ & $\mathbf{0.133 \pm 0.001}$ & $0.340 \pm 0.001$ & $0.340 \pm 0.001$ & $0.343 \pm 0.001$ & $0.348 \pm 0.001$ & $0.346 \pm 0.001$ \\
OOD-mandelbrot & $0.350 \pm 0.001$ & $\mathbf{0.134 \pm 0.001}$ & $0.342 \pm 0.001$ & $0.343 \pm 0.001$ & $0.345 \pm 0.001$ & $0.351 \pm 0.001$ & $0.348 \pm 0.001$ \\
Curr-file60 & $0.351 \pm 0.001$ & $\mathbf{0.133 \pm 0.001}$ & $0.343 \pm 0.001$ & $0.344 \pm 0.001$ & $0.347 \pm 0.001$ & $0.351 \pm 0.001$ & $0.349 \pm 0.001$ \\
OOD-file120 & $0.351 \pm 0.001$ & $\mathbf{0.132 \pm 0.001}$ & $0.344 \pm 0.001$ & $0.345 \pm 0.001$ & $0.347 \pm 0.001$ & $0.352 \pm 0.001$ & $0.349 \pm 0.001$ \\
OOD-file150 & $0.351 \pm 0.001$ & $\mathbf{0.134 \pm 0.001}$ & $0.343 \pm 0.001$ & $0.343 \pm 0.001$ & $0.347 \pm 0.001$ & $0.351 \pm 0.001$ & $0.349 \pm 0.001$ \\
OOD-pcache0.20 & $0.265 \pm 0.001$ & $\mathbf{0.134 \pm 0.001}$ & $0.261 \pm 0.001$ & $0.261 \pm 0.001$ & $0.263 \pm 0.001$ & $0.266 \pm 0.001$ & $0.267 \pm 0.001$ \\
Curr-pcache0.40 & $0.408 \pm 0.001$ & $\mathbf{0.134 \pm 0.001}$ & $0.404 \pm 0.001$ & $0.404 \pm 0.001$ & $0.408 \pm 0.001$ & $0.411 \pm 0.001$ & $0.408 \pm 0.001$ \\
OOD-delay10 & $0.552 \pm 0.000$ & $\mathbf{0.495 \pm 0.001}$ & $0.555 \pm 0.000$ & $0.553 \pm 0.000$ & $0.555 \pm 0.000$ & $0.555 \pm 0.000$ & $0.555 \pm 0.000$ \\
OOD-delay30 & $0.183 \pm 0.001$ & $\mathbf{0.035 \pm 0.000}$ & $0.174 \pm 0.001$ & $0.177 \pm 0.001$ & $0.179 \pm 0.001$ & $0.184 \pm 0.001$ & $0.181 \pm 0.001$ \\
\bottomrule
\end{tabular}
\end{adjustbox}
\end{table*}

\section{Illustrative Three-Request Merge Example}
\label{sec:appendix:example}

This appendix provides the complete step-by-step derivation of the
three-request merge example summarized in
Section~\ref{sec:model:coded}.

Consider a system with $K = 3$ caches and a queue of $Q = 3$ pending
requests at some time step~$t$:
\begin{align*}
  r_0 &= \bigl(k{=}0,\; f{=}\{p_1\},\; d{=}5,\; \mathcal{S}{=}\{1,2\}\bigr), \\
  r_1 &= \bigl(k{=}1,\; f{=}\{p_2\},\; d{=}3,\; \mathcal{S}{=}\{0,2\}\bigr), \\
  r_2 &= \bigl(k{=}2,\; f{=}\{p_3\},\; d{=}2,\; \mathcal{S}{=}\{0,1\}\bigr).
\end{align*}

\emph{Step~1---Feasibility check.}\;
Pair $(r_0, r_1)$: cache~$k_1{=}1$ holds $\{p_1\} = f_0$
(since $1 \in \mathcal{S}_0$), and cache~$k_0{=}0$ holds
$\{p_2\} = f_1$ (since $0 \in \mathcal{S}_1$).
Both sides of~\eqref{eq:xor_feasibility} are satisfied, so the pair
is \emph{feasible}.

\emph{Step~2---First broadcast (merge with keep-side $\kappa = 0$).}\;
The server broadcasts $X = p_1 \oplus p_2$ in one channel use;
cache~1 cancels $p_1$ to recover $p_2$, and cache~0 cancels $p_2$ to
recover $p_1$.
Both users are served immediately ($U_t = 2$).
Two outcomes are decided independently at this step:
(a)~the keep-side parameter $\kappa = 0$ selects which queue slot is
preserved and which slot is replenished (here, slot~0 retains the
aggregate and slot~1 is refilled with a fresh arrival
$r_{\mathrm{new}}$);
(b)~the representative destination $k_{\mathrm{mg}}$ is sampled
uniformly at random from $\{k_0,k_1\} = \{0,1\}$
\emph{independently of $\kappa$} (per~\eqref{eq:mg_dest}).
The slot retained by $\kappa$ is therefore updated to the
\emph{coding-state record}
\begin{multline*}
  r_0' = \bigl(k{=}k_{\mathrm{mg}}\!\sim\!\mathrm{Unif}\{0,1\},\;
              f{=}\{p_1, p_2\},\\
              d{=}\min(5,3){=}3,\;
              \mathcal{S}{=}\{1,2\}\cap\{0,2\}{=}\{2\}\bigr),
\end{multline*}
where the file set, deadline, and side-information set are determined
by~\eqref{eq:mg_file}--\eqref{eq:mg_side} and are unaffected by either
$\kappa$ or the realization of $k_{\mathrm{mg}}$.
The aggregate $r_0'$ remains in the queue to track the enlarged packet
union and tightened side-information set for future XOR feasibility; it
does not represent unserved demand.
Note that $\kappa$ alters only \emph{which queue slot} the aggregate
occupies (and, equivalently, which slot is freed for fresh arrivals);
$k_{\mathrm{mg}}$ alters only \emph{which user is treated as the
representative destination} for the one-gap invariant
(Proposition~\ref{prop:one_gap}).
The two decisions are decoupled in our environment.

\emph{Step~3---Second broadcast (chained merge).}\;
Suppose request $r_2$ is still pending and the representative
destination from Step~2 is the realization $k_{\mathrm{mg}} = 0$ (the
analysis for $k_{\mathrm{mg}} = 1$ is symmetric).
Pair $(r_0', r_2)$: cache~$k_2{=}2$ must hold
$f_0' = \{p_1, p_2\}$.
Since $2 \in \mathcal{S}_0' = \{2\}$, cache~2 indeed holds both
packets.
Conversely, cache~$k_{\mathrm{mg}}{=}0$ must hold $f_2 = \{p_3\}$; since
$0 \in \mathcal{S}_2 = \{0,1\}$, this also holds.
The pair is feasible.
The server broadcasts $X' = p_1 \oplus p_2 \oplus p_3$ in a
\emph{second} channel use.
The two active participants decode as guaranteed by
Proposition~\ref{prop:one_gap}:
user~2 (new partner) cancels $\{p_1, p_2\} \subset \mathcal{C}_2$
to recover~$p_3$;
user~$k_{\mathrm{mg}} = 0$ cancels
$\{p_2, p_3\} \subset \mathcal{C}_0$ to recover~$p_1$, consistent with
the one-gap invariant for the realization $p^{\star} = p_1$ (had
$k_{\mathrm{mg}} = 1$ been drawn, the gap packet would instead be
$p^{\star} = p_2$ and user~1 would recover it from
$\{p_1, p_3\} \subset \mathcal{C}_1$).
In this particular example, the non-representative user can also
decode because the symmetric cache structure ensures
$\{p_1, p_3\} \subset \mathcal{C}_1$ and
$\{p_2, p_3\} \subset \mathcal{C}_0$;
however, this is a consequence of the example's symmetry, not a general
property of chained merges
(Remark~\ref{rem:scope_decodability}).
The XOR degree of this broadcast is
$U_t = |f_0' \cup f_2| = 3$, exceeding the pairwise baseline of~2.
After this merge, the aggregate's side-information set shrinks to
$\mathcal{S}_{\mathrm{mg}} = \{2\} \cap \{0,1\} = \emptyset$, so no
further merges are possible, thereby illustrating the mergeability cost of
the method of aggressive chaining.

\section{Observation Feature Definitions}
\label{sec:appendix:features}

Tables~\ref{tab:req_features} and~\ref{tab:pair_features} list the per-request
and per-pair features used in the observation encoding
(Section~\ref{sec:env:state}).

\begin{table}[t]
  \caption{Per-Request Feature Vector ($d_{\text{req}} = 2K+3 = 13$ for Track~A; $2K+4 = 14$ for Track~B). For an aggregate (post-merge) coding-state record, the ``Target cache'' feature uses the representative destination $k_{\mathrm{mg}}$ (Eq.~\eqref{eq:mg_dest}), not the demand origin of any constituent arrival.}
  \label{tab:req_features}
  \centering
  \begin{adjustbox}{max width=\columnwidth}
  \begin{tabular}{@{}llcl@{}}
    \toprule
    \textbf{Feature} & \textbf{Formula} & \textbf{Dim} & \textbf{Rationale} \\
    \midrule
    Target cache (one-hot) & $\mathbf{e}_{k_r} \in \{0,1\}^K$ & $K$ &
      Destination identity \\[2pt]
    Cache-placement flags  & $[\mathbb{1}(c \in \mathcal{S}_r)]_{c=0}^{K-1}$ & $K$ &
      Side-info providers \\[2pt]
    Normalized deadline    & $d_r / D$ & $1$ &
      Urgency in $[0,1]$ \\[2pt]
    Normalized packet-set size (aggregate size $|f_r|$) & $\min(|f_r|, U_{\max}) / U_{\max}$ & $1$ &
      Detects merged aggregates \\[2pt]
    Normalized degree & $\deg(r) / (Q-1)$ & $1$ &
      Merge opportunity count \\
    \midrule
    \multicolumn{4}{@{}l}{\textit{Track~B additional feature}} \\
    Popularity mass norm & $\min(\hat{m}(f_r), m_{\text{cap}}) / m_{\text{cap}}$ & $1$ &
      Demand-weighted importance \\
    \bottomrule
  \end{tabular}
  \end{adjustbox}
\end{table}

\begin{table}[t]
  \caption{Per-Pair Feature Vector ($d_{\text{pair}} = 8$ for Track~A; $11$ for Track~B)}
  \label{tab:pair_features}
  \centering
  \begin{adjustbox}{max width=\columnwidth}
  \begin{tabular}{@{}lll@{}}
    \toprule
    \textbf{Feature} & \textbf{Formula} & \textbf{Rationale} \\
    \midrule
    Intersection norm    & $|\mathcal{S}_i \cap \mathcal{S}_j| / K$         &
      Future merge potential \\[2pt]
    Degree $i$ norm      & $\deg(r_i) / (Q-1)$                              &
      Request $i$'s alternatives \\[2pt]
    Degree $j$ norm      & $\deg(r_j) / (Q-1)$                              &
      Request $j$'s alternatives \\[2pt]
    Min deadline norm    & $\min(d_i, d_j) / D$                             &
      Pair's combined urgency \\[2pt]
    Packet-set size $i$ (aggregate $|f_i|$) & $\min(|f_i|, U_{\max}) / U_{\max}$              &
      Aggregate size of $r_i$ \\[2pt]
    Packet-set size $j$ (aggregate $|f_j|$) & $\min(|f_j|, U_{\max}) / U_{\max}$              &
      Aggregate size of $r_j$ \\[2pt]
    Queue index $i$ norm & $i / (Q-1)$                                      &
      Position for GNN topology \\[2pt]
    Queue index $j$ norm & $j / (Q-1)$                                      &
      Position for GNN topology \\
    \midrule
    \multicolumn{3}{@{}l}{\textit{Track~B additional features}} \\
    Pop.\ mass $i$ norm & $\min(\hat{m}(f_i), m_{\text{cap}}) / m_{\text{cap}}$ &
      Demand weight of $r_i$ \\[2pt]
    Pop.\ mass $j$ norm & $\min(\hat{m}(f_j), m_{\text{cap}}) / m_{\text{cap}}$ &
      Demand weight of $r_j$ \\[2pt]
    Pop.\ mass union    & $\min(\hat{m}(f_i \cup f_j), m_{\text{cap}}) / m_{\text{cap}}$ &
      Combined demand weight \\
    \bottomrule
  \end{tabular}
  \end{adjustbox}
\end{table}

\paragraph{Normalization constants used in
Tables~\ref{tab:req_features}--\ref{tab:pair_features}.}
We fix:
\begin{itemize}[leftmargin=*, nosep]
  \item $U_{\max} = 6$ is an \emph{observation-only clip}: the
    packet-set size feature (aggregate size $|f_r|$) uses
    $\min(|f_r|, U_{\max}) / U_{\max}$, so aggregates
    with $|f_r| > 6$ would saturate at $1$ on this feature dimension. There
    is \textbf{no} hard cap on $|f_r|$ in the merge-feasibility rule
    (Eq.~\eqref{eq:xor_feasibility}) or in the one-step transition
    (Algorithm~\ref{alg:one_step}); aggregates may grow beyond
    $U_{\max}$ packets in principle. \texttt{max\_union\_files} is the
    name of the corresponding implementation hyperparameter in the
    reference simulator and is \emph{only} an observation-feature
    normalization constant (the input that sets $U_{\max}$ in the
    feature $\min(|f_r|,U_{\max})/U_{\max}$); it is \textbf{not} a
    transition-time or feasibility-time cap on $|f_r|$. The observation
    therefore
    aliases all states with $|f_r| \geq U_{\max}$ on this feature
    coordinate, which we accept because (i) deep chains are rare on
    the on-policy support, and (ii) the side-information set
    $\mathcal{S}_{\mathrm{mg}}$ shrinks monotonically with chain
    depth, making longer chains progressively infeasible to extend.
    Appendix~\ref{sec:appendix:fr_empirical} reports an empirical
    saturation-frequency diagnostic on the deployed checkpoint that
    confirms (i) directly: across $20{,}000{,}000$ on-policy
    observations on ID-default the clip rate
    $\Pr[|f_r|\!\geq\!U_{\max}]$ is $7\!/\!2\!\cdot\!10^{7}\approx
    3.5\!\cdot\!10^{-7}$, with $|f_r|{=}1$ accounting for $90.4\%$
    of observations; under the merge-heavy stress regime
    Curr-pcache0.40 the clip rate stays below
    $6\!\cdot\!10^{-6}$. Saturation is therefore negligible on the
    distributions used to evaluate the deployed policy, so the
    observation aliasing has no measurable effect on the reported
    headline metrics. A controlled sensitivity sweep over
    $U_{\max}\in\{6,8,12\}$ at training time, and a richer aggregate
    summary that retains contained file identities, are left as
    future work.
  \item $\hat m(f_r)$ — the popularity-mass estimate of $f_r$ under
    the assumed Zipf distribution, computed as
    $\hat m(f_r) = \sum_{p \in f_r} \hat p_{\phi(p)}$, where
    $\hat p_n \propto (n+1)^{-\alpha}$ is the Zipf
    rank-probability of file index $n$ ($\alpha=0.8$ in the Track~B
    base regime; $\phi(p) = \lfloor p/B \rfloor$ is the packet-to-file
    map of Eq.~\eqref{eq:phi_packet_to_file}).
  \item $m_{\text{cap}} = U_{\max} \cdot \hat p_0 = 6 \cdot \hat p_0$
    — the popularity-mass cap, set to the popularity mass of an
    aggregate of size $U_{\max}$ all carrying the most popular file
    (rank 0). Because $|f_r|$ is not hard-capped at $U_{\max}$
    (see the $U_{\max}$ bullet above), the popularity-mass features
    apply an explicit $\min(\cdot, m_{\text{cap}})$ clip in the
    numerator (matching the simulator's
    \texttt{pop\_mass\_norm} routine, which returns
    $\min(m, m_{\text{cap}})/m_{\text{cap}}$); this is what
    guarantees the $[0,1]$ bound on these feature dimensions.
\end{itemize}
These constants are fixed at episode start and shared across all
features.

\section{\texorpdfstring{Empirical $|f_r|$ Distribution Under the Learned Policy}{Empirical |f\_r| Distribution Under the Learned Policy}}
\label{sec:appendix:fr_empirical}

The observation feature $\min(|f_r|, U_{\max})/U_{\max}$ with
$U_{\max}{=}6$ aliases all aggregate-record sizes
$|f_r|\!\geq\!U_{\max}$ to the same coordinate, and to verify that this
aliasing is rare on the distributions used in the paper, we instrumented
the environment to log the per-step $|f_r|$ of each queue record at
observation time, and then re-evaluated the deployed selected checkpoint
(seed~0) on the holdout set used throughout the paper (50 base seeds
$\times$ 200 episodes per seed, $H{=}50$ steps per episode, queue
length $|Q|{=}10$, i.e., totalling $5\!\cdot\!10^{6}$ observations per
regime per seed and $|Q|{\cdot}H{\cdot}\text{episodes}=10^{5}$
observations per (regime, holdout seed) pair).

Table~\ref{tab:fr_empirical} reports the pooled histogram of $|f_r|$
under two regimes, i.e., ID-default (the in-distribution evaluation regime)
and Curr-pcache0.40 (a curriculum-stress regime with $p_{\mathrm{cache}}{=}0.40$
that produces the densest population of the mergeable pairs in
Sec.~\ref{sec:results}--\ref{sec:ablation}, and is therefore the
hardest test for chain depth). On both regimes the distribution
is dominated by the singletons ($|f_r|{=}1$ accounts for $\geq\!90\%$ of
observations), while the depth--$2$ aggregates account for $7$--$8\%$, and the
clip rate $\Pr[|f_r|\!\geq\!U_{\max}]$ stays below
$6\!\cdot\!10^{-6}$ even on the stress regime. The maximum observed
$|f_r|$ is $6$ on ID-default and $7$ on Curr-pcache0.40, and even on
the stress regime fewer than $0.0006\%$ of observations reach
$|f_r|{=}6$. The saturation is therefore negligible on the on-policy
support, i.e., the observation aliasing has no measurable effect on the
reported headline metrics, while the claim that
``deep chains are rare on the on-policy support'' from
Appendix~\ref{sec:appendix:features} is supported empirically on the
deployed checkpoint.

\begin{table}[t]
\centering
\caption{Empirical distribution of the aggregate-record size $|f_r|$
  observed by the deployed selected checkpoint (train seed~0)
  on the holdout set (50 base seeds $\times$ 200 episodes per seed).
  Counts pooled across all (eval seed, episode, step, queue position)
  observation samples. ID-default totals $2\!\cdot\!10^{7}$ samples;
  Curr-pcache0.40 totals $5\!\cdot\!10^{6}$ samples.}
\label{tab:fr_empirical}
\begin{adjustbox}{max width=\columnwidth}
\begin{tabular}{lrrrr}
\toprule
\textbf{$|f_r|$}
  & \multicolumn{2}{c}{\textbf{ID-default}}
  & \multicolumn{2}{c}{\textbf{Curr-pcache0.40}} \\
\cmidrule(lr){2-3}\cmidrule(lr){4-5}
  & count & \% & count & \% \\
\midrule
$1$ & $18{,}081{,}174$ & $90.41$ & $4{,}554{,}171$ & $91.08$ \\
$2$ & $1{,}675{,}662$  & $8.38$  & $357{,}151$    & $7.14$  \\
$3$ & $227{,}575$      & $1.14$  & $76{,}910$     & $1.54$  \\
$4$ & $15{,}098$       & $0.075$ & $11{,}122$     & $0.222$ \\
$5$ & $484$            & $0.002$ & $616$          & $0.012$ \\
$6$ & $7$              & $4{\cdot}10^{-5}$ & $29$  & $5.8{\cdot}10^{-4}$ \\
$7$ & $0$              & $0$     & $1$            & $2{\cdot}10^{-5}$ \\
\midrule
\textbf{Clip rate} $\Pr[|f_r|\!\geq\!U_{\max}]$
   & \multicolumn{2}{c}{$3.5\!\cdot\!10^{-7}$}
   & \multicolumn{2}{c}{$6.0\!\cdot\!10^{-6}$} \\
\bottomrule
\end{tabular}
\end{adjustbox}
\end{table}

We do not retrain with a larger $U_{\max}$ in this version, considering
that the clip rate is near zero on both regimes, and thereby such a sweep
would not be expected to change the deployed-policy headline metrics
materially. A more informative future-work direction is a richer
aggregate-summary feature that retains the contained file identities
(rather than only the size), since that would help the policy
distinguish the aggregates with the same size but different
side-information mass.

\section{Chained-Merge State Sufficiency}
\label{sec:appendix:state_sufficiency}

This appendix formally establishes that the compressed aggregate state
$(k_{\mathrm{mg}}, f_{\mathrm{mg}}, \mathcal{S}_{\mathrm{mg}})$
retained after a chained merge is a sufficient statistic for all the
future feasibility and decodability decisions, while the proof relies on two structural constraints:
(C1)~the request-generation rule ensures
$f_r \not\subseteq \mathcal{C}_{k_r}$ (i.e., a user never requests a packet
it already caches), and
(C2)~the XOR feasibility~\eqref{eq:xor_feasibility} requires each merge
partner to cache the other's packet set.

\subsection{One-Gap Invariant}

\begin{proposition}[One-Gap Invariant]
\label{prop:one_gap}
Let $r_{\mathrm{mg}}$ be a coding-state record produced by any
sequence of~$n \geq 1$ chained pairwise merges, each satisfying
XOR feasibility~\eqref{eq:xor_feasibility} and the request-generation
constraint~(C1).
Let $p^{\star}$ denote the original singleton packet of whichever
initial request contributed the current representative
destination~$k_{\mathrm{mg}}$.
Then
\begin{equation}
  f_{\mathrm{mg}} \setminus \{p^{\star}\}
    \;\subseteq\; \mathcal{C}_{k_{\mathrm{mg}}},
  \qquad
  p^{\star} \notin \mathcal{C}_{k_{\mathrm{mg}}}.
  \label{eq:one_gap}
\end{equation}
That is, $k_{\mathrm{mg}}$ caches every packet in $f_{\mathrm{mg}}$
except exactly one: its own originally requested packet~$p^{\star}$.
\end{proposition}

\begin{IEEEproof}
We proceed by induction on the number of merges~$n$.

\smallskip\noindent
\textbf{Base case ($n=1$).}\;
Two fresh requests $r_i = (k_i, \{p_i\}, d_i, \mathcal{S}_i)$
and $r_j = (k_j, \{p_j\}, d_j, \mathcal{S}_j)$ are merged.
XOR feasibility requires
$\{p_i\} \subseteq \mathcal{C}_{k_j}$ and
$\{p_j\} \subseteq \mathcal{C}_{k_i}$, so each user caches the
other's packet.
Constraint~(C1) gives
$p_i \notin \mathcal{C}_{k_i}$ and $p_j \notin \mathcal{C}_{k_j}$.
Now $k_{\mathrm{mg}} \in \{k_i, k_j\}$ is drawn uniformly.

\emph{Case $k_{\mathrm{mg}} = k_i$:}\;
$f_{\mathrm{mg}} = \{p_i, p_j\}$;
$p_j \in \mathcal{C}_{k_i}$ (feasibility) and
$p_i \notin \mathcal{C}_{k_i}$ (C1).
Set $p^{\star} = p_i$; then~\eqref{eq:one_gap} holds.

\emph{Case $k_{\mathrm{mg}} = k_j$:}\;
By symmetry, $p^{\star} = p_j$ and~\eqref{eq:one_gap} holds.

\smallskip\noindent
\textbf{Inductive step.}\;
The simulator's feasible-pair set
$\mathcal{M}_t = \{(i,j): f_i \subseteq \mathcal{C}_{k_j},\,
f_j \subseteq \mathcal{C}_{k_i}\}$ is defined over arbitrary queue
entries, so a feasible partner of an aggregate may itself be either a
fresh singleton or another coding-state record. We therefore prove
the inductive step in the more general case where both partners
satisfy~\eqref{eq:one_gap}; the singleton case is recovered by
specializing $|f| = 1$.

Assume that after some sequence of merges, queue records $r_a$ and
$r_b$ each satisfy the one-gap invariant: $r_a$ has representative
destination $k_a$ and gap packet $p_a^{\star}$ (so
$f_a \setminus \{p_a^{\star}\} \subseteq \mathcal{C}_{k_a}$ and
$p_a^{\star} \notin \mathcal{C}_{k_a}$); $r_b$ has representative
destination $k_b$ and gap packet $p_b^{\star}$. (A fresh singleton is
the special case $f = \{p^{\star}\}$.)
Merging $(r_a, r_b)$ produces
$f_{\mathrm{mg}}' = f_a \cup f_b$ with
$k_{\mathrm{mg}}' \in \{k_a, k_b\}$.

XOR feasibility requires
(F1)~$f_a \subseteq \mathcal{C}_{k_b}$ and
(F2)~$f_b \subseteq \mathcal{C}_{k_a}$.

\emph{Case $k_{\mathrm{mg}}' = k_a$:}\;
By the inductive hypothesis at $r_a$,
$f_a \setminus \{p_a^{\star}\} \subseteq \mathcal{C}_{k_a}$ and
$p_a^{\star} \notin \mathcal{C}_{k_a}$.
By~(F2), $f_b \subseteq \mathcal{C}_{k_a}$.
The candidate gap $p_a^{\star}$ cannot lie in $f_b$: if it did, then
$p_a^{\star} \in f_b \subseteq \mathcal{C}_{k_a}$, contradicting the
hypothesis. Hence $f_a \cup f_b$ contains $p_a^{\star}$, and
\[
  f_{\mathrm{mg}}' \setminus \{p_a^{\star}\}
  = (f_a \setminus \{p_a^{\star}\}) \cup f_b
  \subseteq \mathcal{C}_{k_a},
\]
while $p_a^{\star} \notin \mathcal{C}_{k_a}$.
So~\eqref{eq:one_gap} holds with $p^{\star} = p_a^{\star}$.

\emph{Case $k_{\mathrm{mg}}' = k_b$:}\;
By symmetry (swap $a \leftrightarrow b$),
\eqref{eq:one_gap} holds with $p^{\star} = p_b^{\star}$.

The original two specializations are recovered as: (a) both partners
fresh singletons ($n=1$ base case), (b) one aggregate and one
singleton ($r_b$ a fresh request with $f_b = \{p_b^{\star}\}$),
recovering the wording of the previous draft.
\end{IEEEproof}

\subsection{State Sufficiency Theorem}

\begin{theorem}[Aggregate State Sufficiency]
\label{thm:state_sufficiency}
The triple $(k_{\mathrm{mg}}, f_{\mathrm{mg}},
\mathcal{S}_{\mathrm{mg}})$ is a sufficient statistic for all
future merge-feasibility and decodability decisions involving
$r_{\mathrm{mg}}$.
That is, for any feasible queue partner
$r_z = (k_z, f_z, d_z, \mathcal{S}_z)$, fresh singleton or
coding-state record, with its own one-gap invariant
(Proposition~\ref{prop:one_gap}) and gap packet $p_z^{\star}$
(reducing to $p_z^{\star} = p_z$ when $f_z$ is a singleton):
\begin{enumerate}[nosep]
  \item \textbf{Feasibility:}\;
    The XOR feasibility condition~\eqref{eq:xor_feasibility} for
    the pair $(r_{\mathrm{mg}}, r_z)$ depends only on
    $f_{\mathrm{mg}}$, $k_{\mathrm{mg}}$, $f_z$, $k_z$, and the cache
    contents $\mathcal{C}_{k_z}$, $\mathcal{C}_{k_{\mathrm{mg}}}$;
    none of these requires knowledge of previously merged users'
    destinations.
  \item \textbf{Decodability of $k_{\mathrm{mg}}$:}\;
    By Proposition~\ref{prop:one_gap}, $k_{\mathrm{mg}}$ caches
    every packet in $f_{\mathrm{mg}}$ except~$p^{\star}$; combined with
    $f_z \subseteq \mathcal{C}_{k_{\mathrm{mg}}}$ (feasibility),
    $k_{\mathrm{mg}}$ can cancel all known packets from
    $X' = \bigoplus_{p \in f_{\mathrm{mg}} \cup f_z} p$
    to recover~$p^{\star}$.
  \item \textbf{Decodability of $k_z$:}\;
    By Proposition~\ref{prop:one_gap} applied to $r_z$,
    $k_z$ caches $f_z \setminus \{p_z^{\star}\}$ and
    $p_z^{\star} \notin \mathcal{C}_{k_z}$.
    Feasibility requires $f_{\mathrm{mg}} \subseteq \mathcal{C}_{k_z}$,
    and $p_z^{\star} \notin f_{\mathrm{mg}}$ (otherwise
    $p_z^{\star} \in \mathcal{C}_{k_z}$, contradicting the hypothesis).
    Hence $(f_{\mathrm{mg}} \cup f_z) \setminus \{p_z^{\star}\}
    \subseteq \mathcal{C}_{k_z}$, so $k_z$ can cancel
    everything except $p_z^{\star}$ from $X'$ and recover
    $p_z^{\star}$.
    The singleton case ($f_z = \{p_z\}$, $p_z^{\star} = p_z$) is
    recovered by specialization.
  \item \textbf{Side-information evolution:}\;
    $\mathcal{S}_{\mathrm{mg}} = \bigcap_{\ell} \mathcal{S}_{\ell}$
    over all predecessor requests~$\ell$; this is the exact set of
    caches holding every packet in $f_{\mathrm{mg}}$.
    Future feasibility requires only
    $k_z \in \mathcal{S}_{\mathrm{mg}}$ (direction~1) and
    $k_{\mathrm{mg}} \in \mathcal{S}_z$ (direction~2), both
    computable from the aggregate state alone.
\end{enumerate}
Previously merged users (those whose requests were folded into
$r_{\mathrm{mg}}$ in earlier merges) need not be tracked: they were
served at their respective merge steps, and correctness of future
merges depends only on the representative~$k_{\mathrm{mg}}$.
\end{theorem}

\begin{IEEEproof}
Items~1, 3, and~4 follow directly from the definitions of
$f_{\mathrm{mg}}$, $\mathcal{S}_{\mathrm{mg}}$, and the
feasibility condition~\eqref{eq:xor_feasibility}: no reference to
previously merged users appears in any of these expressions.
Item~2 is a direct corollary of Proposition~\ref{prop:one_gap}:
the one-gap structure guarantees that $k_{\mathrm{mg}}$ possesses
$|f_{\mathrm{mg}}| - 1$ packets of $f_{\mathrm{mg}}$ and all of
$f_z$ (by feasibility), leaving exactly one unknown
($p^{\star}$) in the XOR sum, which $k_{\mathrm{mg}}$ recovers.
\end{IEEEproof}

\begin{remark}[Scope of decodability in chained merges]
\label{rem:scope_decodability}
Proposition~\ref{prop:one_gap} guarantees that the two
\emph{active participants} of each chained merge (i.e., $k_{\mathrm{mg}}$
and $k_z$) always decode the broadcast, while the users that are served by earlier links in the chain are not
guaranteed to decode the subsequent broadcasts, i.e., a
previously served user~$k_{\mathrm{prev}}$ may lack packets that are
introduced in later merge steps.
This is benign because each such user was already served at its
own merge step, while the coding-state record exists solely to enable
the future merges, not to re-serve the past users.
\end{remark}

\section{Coded Transmission Remarks}
\label{sec:appendix:coded_remarks}

The following remarks supplement the coded transmission model in Section~\ref{sec:model:coded}.

\begin{remark}[Chained merges and pairwise scope]
Each link in a merge chain corresponds to a separate channel use, while the XOR
feasibility condition~\eqref{eq:xor_feasibility} is re-evaluated at every
step using the updated packet union $f_{\mathrm{mg}}$ and the
side-information set $\mathcal{S}_{\mathrm{mg}} = \mathcal{S}_i \cap
\mathcal{S}_j$ defined in Eq.~\eqref{eq:mg_side}. Because the
intersection can only weakly shrink each component admissible set
($\mathcal{S}_{\mathrm{mg}}\subseteq\mathcal{S}_i$ and
$\mathcal{S}_{\mathrm{mg}}\subseteq\mathcal{S}_j$) while the next merge
partner must hold the \emph{entire} growing packet union
$f_{\mathrm{mg}}$ rather than only the original subfile, each
successive link \emph{typically} reduces the future feasibility. The
reduction is strict when the new merge partner contributes at least
one packet identity not already in $f_{\mathrm{mg}}$, i.e., when
$|f_{\mathrm{mg}}|$ strictly grows, while the chained partners that contribute
no novel packet identities (i.e., the non-gap-overlap case discussed in
Sec.~\ref{sec:model:coded}) leave the next feasibility check no
harder than the previous one.
In this work we therefore restrict to the \emph{pairwise feasible merges} at
each step, even though higher-order simultaneous multi-way merges (i.e., forming a single
coded packet covering three or more users in one feasibility check) are not
modelled and are identified as a future extension
(Section~\ref{sec:conclusion}).
\end{remark}

\begin{remark}[Per-broadcast accounting is a packet-set count]
Each link in a merge chain corresponds to a distinct channel use (i.e., one broadcast
slot). The metric $U_t = |f_{\mathrm{mg}}|$ is defined throughout this paper
as a \emph{packet-set cardinality}, i.e., the number of distinct packet identities
carried in $f_{\mathrm{mg}}$ at the moment of the broadcast.
The request generator samples each fresh request's packet ID independently
from the library, so two active queue records may in principle carry the same
packet identity. XOR feasibility~\eqref{eq:xor_feasibility} together with the
no-self-cache constraint~(C1) does \emph{not} exclude all such overlaps: it
only excludes the case where one record's gap packet $p_a^{\star}$ coincides
with the other record's gap $p_b^{\star}$, since $p_a^{\star}\notin
\mathcal{C}_{k_a}$ and the feasibility condition $f_a\subseteq\mathcal{C}_{k_b}$
together force $p_a^{\star}\in\mathcal{C}_{k_b}$, leaving the gap structures
of the two records distinct. \emph{Non-gap} packet overlaps are
compatible with feasibility. As a counterexample template
(displayed in the local shorthand $(f_r,\,k_r,\,\text{gap}\,p_r^{\star})$
that omits the deadline $d_r$ and side-information set $\mathcal{S}_r$
of the canonical record $r=(k_r,f_r,d_r,\mathcal{S}_r)$ since they are
irrelevant to the packet-overlap argument), take
$r_a=(\{p_2,p_3\},\,k_a,\text{gap }p_3)$ and
$r_b=(\{p_2,p_5\},\,k_b,\text{gap }p_5)$ with
$p_2\in\mathcal{C}_{k_a}\cap\mathcal{C}_{k_b}$,
$p_3\in\mathcal{C}_{k_b}\setminus\mathcal{C}_{k_a}$,
$p_5\in\mathcal{C}_{k_a}\setminus\mathcal{C}_{k_b}$: feasibility holds, both
one-gap invariants hold, and the merge produces $f_{\mathrm{mg}}=\{p_2,p_3,p_5\}$
even though three original singleton arrivals (the two each record was built
from, plus the partner just merged in) contributed to the chain.
We therefore avoid claiming that $U_t = |f_{\mathrm{mg}}|$ counts the original
singleton arrivals, while the readers who want strict original-arrival semantics
should use the request-level metrics $\eta_{\mathrm{req}}, m_{\mathrm{req}},
\sigma_{\mathrm{req}}$ (Sec.~\ref{sec:model:request_accounting}), which are
computed from the per-arrival identifier sets $\mathcal{A}(r)$ stamped at
arrival time, thereby crediting each original arrival exactly once.
By Proposition~\ref{prop:one_gap}, the two active participants
(i.e., $k_{\mathrm{mg}}$ and the new partner~$k_z$) always decode correctly,
while each link in the chain is a distinct channel use, so summing $U_t$ over
$t$ counts the total broadcast-slot XOR degree (i.e., a packet-set quantity) rather
than redundant data delivery. The broadcast-efficiency score
$\sigma = H^{-1}\sum_t (U_t - \lambda E_t)$ should be read accordingly as a
per-slot \emph{packet-throughput-vs-channel-waste} composite, not as an
original-arrival completion score.
\end{remark}

\begin{remark}[Keep-side parameter and queue evolution]
The keep-side parameter $\kappa$ is a design choice that
goes beyond selecting \emph{which pair} to merge.
Considering that $d_{\mathrm{mg}} = \min(d_i, d_j)$ and
$\mathcal{S}_{\mathrm{mg}} = \mathcal{S}_i \cap \mathcal{S}_j$ are
fully determined by the merge rule, $\kappa$ cannot change the
aggregate's deadline or side-information set.
Its role is exclusively to govern the \emph{queue evolution}, i.e., $\kappa$
decides which queue slot is vacated and replenished by a fresh
arrival~$r_{\mathrm{new}}$.
The representative destination $k_{\mathrm{mg}}$ is a logically
\emph{separate} quantity, drawn uniformly from $\{k_i, k_j\}$
\emph{independently of $\kappa$}, i.e., the kept slot does not determine the
representative, while the representative does not determine the kept
slot. This decoupling gives the agent indirect control over the
queue's future composition and mergeability, while not entangling that
control with the one-gap invariant of
Proposition~\ref{prop:one_gap}, which depends only on the realized
$k_{\mathrm{mg}}$.
After the literature sweep summarized in Sec.~\ref{sec:related},
we are not aware of any prior RL formulation for the
coded caching that includes this control dimension, i.e., this is a positioning
statement, not a definitive non-existence result.
\end{remark}

\section{Reward and Masking Details}
\label{sec:appendix:reward_details}
\label{sec:appendix:masking_details}

\paragraph{Reward weight selection.}
The base reward uses equal weights ($w_{\mathrm{served}}{=}1$,
$w_{\mathrm{exp}}{=}1$) so that the per-step reward
$R_{\mathrm{base}} = U_t - E_t$ directly mirrors the evaluation-level
broadcast-efficiency score $\sigma = H^{-1}\sum_t(U_t - \lambda E_t)$ with $\lambda{=}1$.
The quality bonus weights ($w_{\mathrm{inter}}{=}0.75$, $w_{\mathrm{union}}{=}0.15$) and the potential-shaping weight ($w_\Phi{=}0.20$) were tuned on the validation split, considering the trade-off between the throughput and the deadline-compliance signals.

\paragraph{Potential-based shaping.}
\begin{remark}
The potential-based reward shaping is \emph{inspired by} the
policy-invariance theorem of~\cite{ng1999shaping}, which holds
exactly for an idealized fully observed continuing-state MDP. The
implemented training setup is the contextual-POMDP surrogate of
Definition~\ref{def:problem} (i.e., truncated at $H{=}50$, hidden episode
cache placement, aliased aggregate observations, no time-to-go
feature), so the theorem does not apply directly, i.e., we make
\emph{no} formal claim that the optimal policy under
$R_{\mathrm{total}}$ is identical to the optimal policy under
$R_{\mathrm{base}} + R_{\mathrm{quality}}$ in the implemented setup.
A formal invariance result for the implemented environment would
require establishing the exact-reset terminal conditions, a time-to-go
feature, and an observation exposing the hidden episode context and
the packet identities, while we leave such a formalization to future work.
In the present paper, $R_{\mathrm{shape}}$ is treated as a heuristic
credit-assignment shaping that empirically encourages the agent to
transition toward the queue states with higher merge potential (i.e., a proxy
for the future coding opportunities), while the empirical effect is
quantified by the reward-component ablation in
Sec.~\ref{sec:ablation}, thereby accelerating the credit assignment in the early stages of
training, where the shaped reward bridge is most useful.
\end{remark}

\paragraph{Action masking gradient bias.}
\begin{remark}
Huang and Onta\~{n}\'{o}n~\cite{huang2020invalidmask} prove that the masking
procedure described in Section~\ref{sec:env:masking} introduces no bias in the gradient direction, i.e., the policy-gradient
estimator remains unbiased over the feasible action sub-space.
In our environment, as few as $1$ (i.e., unicast only) and as many as $91$ actions
may be valid at a given step, depending on the queue state.
Without the masking, a uniform random policy would waste approximately
$|\mathcal{A}| - |\mathcal{M}_t| - 1$ actions at every step on the infeasible
choices, thereby substantially degrading the sample efficiency.
\end{remark}

\section{Policy Architecture and Training Details}
\label{sec:appendix:arch_details}
\label{sec:appendix:training_details}

\paragraph{Stage 1: Graph-pair encoder, full equations.}

The \emph{NodeMLP} lifts each per-request vector to a 128-dimensional embedding:
\begin{equation}
\begin{split}
  \mathbf{h}_r^{(0)} &= \mathrm{ReLU}\!\left(\mathbf{W}_2
    \,\mathrm{ReLU}\!\left(\mathbf{W}_1 \mathbf{x}_r + \mathbf{b}_1\right)
    + \mathbf{b}_2\right), \\
  & \mathbf{W}_1 \in \mathbb{R}^{256 \times d_{\mathrm{req}}},\;
        \mathbf{W}_2 \in \mathbb{R}^{128 \times 256}.
\end{split}
\end{equation}
Two successive \emph{GraphAttentionBlocks}~\cite{velickovic2018gat} with
four attention heads and $d_{\mathrm{model}}=128$ then propagate the information
between the requests that share a feasible merge edge, thereby yielding the contextual node
embeddings $\mathbf{h}_r^{(1)}$ and $\mathbf{h}_r^{(2)}$.
The dropout is set to zero to preserve the determinism at inference.
The per-step inference cost of the graph-attention encoder is $O(|\mathcal{M}_t| \cdot d_{\mathrm{model}})$ for the message passing (i.e., bounded by $P_{\max} \cdot d_{\mathrm{model}} = 45 \times 128$) plus $O(Q \cdot d_{\mathrm{model}})$ for the node update, while the forward pass is therefore inexpensive relative to the combinatorial search it replaces.

The \emph{EdgeNet} forms the embedding for each candidate merge pair
$(i, j, \text{keep-side})$ by concatenating the two node embeddings,
their element-wise product, the \emph{absolute} difference, and the
\emph{base} pair features (i.e., the first $d_{\mathrm{pair}} - 2$
entries of $\mathbf{p}_{ij}$; the last two entries encode normalized
queue indices and are used only to construct the merge-graph
adjacency, not as EdgeNet inputs):
\begin{equation}
\begin{aligned}
  \mathbf{e}_{ij} = \mathrm{MLP}\Big(
    \big[&\mathbf{h}_i^{(2)} \,\big\|\, \mathbf{h}_j^{(2)} \,\big\|\,
           \mathbf{h}_i^{(2)} \odot \mathbf{h}_j^{(2)} \\
    \,\big\|\,&|\mathbf{h}_i^{(2)} - \mathbf{h}_j^{(2)}| \,\big\|\,
           \mathbf{p}_{ij}^{\mathrm{base}}\big]
  \Big) \in \mathbb{R}^{64},
\end{aligned}
\end{equation}
where $\mathbf{p}_{ij}^{\mathrm{base}} \in \mathbb{R}^{d_{\mathrm{pair}} - 2}$
is the base pair-feature subvector and the MLP maps
$(4 \times 128 + d_{\mathrm{pair}} - 2) \to 256 \to 64$, i.e.,
input width $4{\cdot}128 + 6 = 518$ for Track~A
($d_{\mathrm{pair}}=8$) and $4{\cdot}128 + 9 = 521$ for Track~B
($d_{\mathrm{pair}}=11$).

\paragraph{Track~B parameter differences.}
The Track~B (Zipf) agent uses $d_{\text{req}}{=}14$ and
$d_{\text{pair}}{=}11$, thereby changing the Node~MLP input layer to
$14 \to 256 \to 128$ and the Edge~MLP input layer to
$(4 \times 128 + 9) = 521 \to 256 \to 64$, while yielding
${\approx}1.75$M total parameters
(i.e., matching Sec.~\ref{sec:method:arch}, Table~\ref{tab:model_arch}, and
the architecture summary at the start of Sec.~\ref{sec:experiments}).
All other architectural choices (i.e., attention heads, hidden dimensions,
critic structure) are identical.

\paragraph{ExIt self-improving teacher.}
Unlike the Phase~1 teacher, which is fixed and bootstraps from the SACM++
rollouts alone, the ExIt teacher uses the \emph{agent's own critic}
$\hat{V}_\theta$ for bootstrapping.
As the critic improves through the PPO training, the planner's labels improve
correspondingly, i.e., the student
improves the teacher's value estimates, which in turn produce better
labels for the student, thereby letting the agent \emph{exceed} its heuristic teachers
rather than just match them.

\paragraph{Full hyperparameters.}

\begin{table}[!t]
  \renewcommand{\arraystretch}{1.2}
  \caption{Verified Hyperparameter Summary}
  \label{tab:hparams}
  \centering
  \begin{adjustbox}{max width=\columnwidth}
  \begin{tabular}{llr}
    \toprule
    Phase & Hyperparameter & Value \\
    \midrule
    \multirow{7}{*}{BC (Phase 1)}
      & Samples ($N_{\mathrm{BC}}$)        & 150,000       \\
      & Training epochs                    & 6             \\
      & Minibatch size                     & 2,048         \\
      & Learning rate ($\eta_{\mathrm{BC}}$) & $3\times10^{-4}$ \\
      & Teacher rollout depth ($d$)        & 4             \\
      & Teacher MC samples ($M$)           & 4             \\
      & Teacher top-$K$ pairs              & 16            \\
    \midrule
    Warm-up (Phase 2)
      & Steps (frozen head)                & 50,000        \\
    \midrule
    \multirow{13}{*}{PPO (Phase 3)}
      & Parallel environments              & 32            \\
      & Steps per rollout ($n_{\mathrm{steps}}$) & 256     \\
      & Minibatch size                     & 1,024         \\
      & Epochs per update                  & 10            \\
      & Learning rate (start $\to$ end)    & $5\times10^{-4} \to 1\times10^{-4}$ \\
      & Clip range ($\varepsilon$)         & 0.20          \\
      & Target KL divergence               & 0.03          \\
      & Entropy coeff.\ (start $\to$ end)  & $0.010 \to 0.001$ \\
      & Discount factor ($\gamma$)         & 0.995         \\
      & GAE parameter ($\lambda$)          & 0.95          \\
      & Value function coefficient         & 0.5           \\
      & Phase-3 env steps per seed         & 6,000,000     \\
      & Total RL env steps per seed (incl. Phase-2 warm-up) & 6,050,000 \\
    \midrule
    \multirow{12}{*}{ExIt (Phase 3)}
      & Start-step threshold                & 300,000       \\
      & Interval threshold                  & 300,000       \\
      & Realized schedule (first fire / cadence / total fires) & 500K / $\approx$300K avg.\ / 20 fires over 24 chunks \\
      & States collected per iteration     & 8,192         \\
      & DAgger buffer capacity             & 80,000        \\
      & Distillation epochs                & 2             \\
      & Distillation batch size            & 2,048         \\
      & Distillation learning rate         & $1\times10^{-4}$ \\
      & Expert roll-in probability         & 0.20          \\
      & Planner lookahead depth            & 5             \\
      & Planner MC samples                 & 3             \\
      & Planner top-$K$ pairs              & 12            \\
    \bottomrule
  \end{tabular}
  \end{adjustbox}
\end{table}

\section{Oracle-Tuned Threshold Analysis}
\label{sec:appendix:taufit}

The $\tau$-Fit family gives a one-parameter handle on the
throughput--reliability tradeoff, while Table~\ref{tab:tau_sweep} shows the
full sweep over $\tau \in \{0,1,2,3\}$ alongside the PPO. Increasing $\tau$
from 0 to 2 raises the broadcast-efficiency score monotonically ($0.877 \to 0.915$), and then
it \emph{saturates}, i.e., the difference between $\tau{=}2$ and $\tau{=}3$ in
BE-score is $0.000$, within the statistical noise, even though PPO's score of
$0.976$ lies well above this saturation plateau.

\begin{table}[t]
\centering
\caption{TauFit performance as a function of clique-size threshold $\tau$. PPO-Agent shown for reference. Bold = best per column (numerical winner across both TauFit rows and PPO); the PPO row label is bolded only to mark it as the learned reference, not as a per-cell win indicator.}
\label{tab:tau_sweep}
\begin{adjustbox}{max width=\columnwidth}
\begin{tabular}{lccccc}
\toprule
$\tau$ & \textbf{BE-Score $\sigma$} ($\uparrow$) & \textbf{Miss Ratio} ($\downarrow$) & \textbf{Served/Tx} ($\uparrow$) & \textbf{Coding Gain} ($\uparrow$) & \textbf{Exp/Episode} ($\downarrow$) \\
\midrule
0 & $0.877 \pm 0.002$ & $\mathbf{0.185 \pm 0.001}$ & $1.135 \pm 0.001$ & $2.063 \pm 0.002$ & $\mathbf{10.55 \pm 0.07}$ \\
1 & $0.908 \pm 0.002$ & $0.245 \pm 0.001$ & $1.343 \pm 0.001$ & $2.097 \pm 0.001$ & $14.92 \pm 0.10$ \\
2 & $0.915 \pm 0.002$ & $0.258 \pm 0.001$ & $1.405 \pm 0.002$ & $2.077 \pm 0.001$ & $16.22 \pm 0.09$ \\
3 & $0.915 \pm 0.003$ & $0.259 \pm 0.001$ & $\mathbf{1.408 \pm 0.002}$ & $2.066 \pm 0.001$ & $16.32 \pm 0.09$ \\
\midrule
\textbf{PPO} & $\mathbf{0.976 \pm 0.002}$ & $0.208 \pm 0.001$ & $1.323 \pm 0.002$ & $\mathbf{2.162 \pm 0.002}$ & $14.17 \pm 0.07$ \\
\bottomrule
\end{tabular}
\end{adjustbox}
\end{table}

To confirm that PPO's advantage does not come from a suboptimal
$\tau$ selection, we evaluated each TauFit policy with an
\emph{oracle-tuned} threshold $\tau^{\star}$, i.e., for each of six target
metrics, $\tau^{\star}$ is chosen retrospectively on the validation seeds to
maximize that metric, and then evaluated on the holdout seeds.
Table~\ref{tab:oracle_taufit} shows that PPO's BE-score advantage
persists across all choices of the tuning criterion, i.e., $+6.61\%$ when
$\tau^{\star}$ is tuned for the BE-score itself, and up to $+11.25\%$
when $\tau^{\star}$ is tuned for the miss ratio or the expirations.
The $\tau$-Fit ceiling under any oracle access is therefore
strictly below PPO's performance.

\begin{table}[t]
\centering
\caption{PPO-Agent vs.\ oracle-tuned TauFit($\tau^{\star}$) evaluated on
  holdout seeds. $\tau^{\star}$ selected on validation set per metric.}
\label{tab:oracle_taufit}
\begin{adjustbox}{max width=\columnwidth}
\begin{tabular}{lcccc}
\toprule
\textbf{Tuned-On Metric} & $\tau^{\star}$ &
  \textbf{PPO $\sigma$} & \textbf{TauFit $\sigma$} & \textbf{PPO Adv.} \\
\midrule
BE-score               & 2 & $0.976 \pm 0.002$ & $0.915 \pm 0.002$ & $+6.61\%$ \\
Served/Tx              & 3 & $0.976 \pm 0.002$ & $0.915 \pm 0.002$ & $+6.65\%$ \\
Reward/step (shaped)   & 3 & $0.976 \pm 0.002$ & $0.915 \pm 0.002$ & $+6.65\%$ \\
Avg coding gain        & 1 & $0.976 \pm 0.002$ & $0.908 \pm 0.002$ & $+7.50\%$ \\
Miss ratio             & 0 & $0.976 \pm 0.002$ & $0.877 \pm 0.002$ & $+11.25\%$ \\
Expirations/episode    & 0 & $0.976 \pm 0.002$ & $0.877 \pm 0.002$ & $+11.25\%$ \\
\bottomrule
\end{tabular}
\end{adjustbox}
\end{table}

Table~\ref{tab:oracle_taufit_permetric} complements the BE-score
analysis above by comparing the PPO against TauFit($\tau^{\star}$) \emph{on
the metric each $\tau^{\star}$ was tuned for}. The PPO wins on the
composite metric (i.e., the broadcast-efficiency score) and the coding gain, while
TauFit($\tau^{\star}{=}0$) achieves lower miss ratio and fewer
expirations, and TauFit($\tau^{\star}{=}3$) yields higher packet-set count per transmission ($\mu$). The interesting fact we observed is that this pattern is what one would expect, i.e., the PPO optimizes the joint
reward objective, thereby trading small losses on individual narrow metrics
for larger gains on the composite broadcast-efficiency score, even though no single
$\tau$-Fit threshold reproduces this multi-objective balance.

\begin{table}[t]
\centering
\caption{Per-metric comparison: PPO-Agent vs.\ oracle-tuned
  TauFit($\tau^{\star}$) evaluated on the metric each
  $\tau^{\star}$ was tuned for. Bold $=$ better; arrows indicate
  metric direction.}
\label{tab:oracle_taufit_permetric}
\begin{adjustbox}{max width=\columnwidth}
\begin{tabular}{lccccc}
\toprule
\textbf{Metric} & $\tau^{\star}$ &
  \textbf{PPO} & \textbf{TauFit($\tau^{\star}$)} & \textbf{$\Delta$} & \textbf{Favors} \\
\midrule
BE-Score ($\uparrow$)        & 2 & $\mathbf{0.976 \pm 0.002}$ & $0.915 \pm 0.002$ & $+0.061$ & PPO \\
Coding Gain ($\uparrow$)     & 1 & $\mathbf{2.162 \pm 0.002}$ & $2.097 \pm 0.001$ & $+0.065$ & PPO \\
Served/Tx ($\uparrow$)       & 3 & $1.323 \pm 0.002$ & $\mathbf{1.408 \pm 0.002}$ & $-0.085$ & TauFit \\
Miss Ratio ($\downarrow$)    & 0 & $0.208 \pm 0.001$ & $\mathbf{0.185 \pm 0.001}$ & $+0.023$ & TauFit \\
Exp/Episode ($\downarrow$)   & 0 & $14.17 \pm 0.07$  & $\mathbf{10.55 \pm 0.07}$  & $+3.62$  & TauFit \\
\bottomrule
\end{tabular}
\end{adjustbox}
\end{table}

\section{Experimental Protocol Details}
\label{sec:appendix:exp_details}

\paragraph{$\tau$-Fit misfit metric.}
The $\tau$-Fit family is based on the pairwise \emph{misfit} metric of
Niesen and Maddah-Ali~\cite{codedcachingdelaysensitive},
written as~$\misfit(r_i, r_j)$ throughout this paper to avoid the collision
with the broadcast-level expiration ratio $\rho$ (M1), i.e.,
\begin{equation}
  \misfit(r_i, r_j) =
  \bigl|\mathcal{S}_i \setminus (\mathcal{S}_j \cup \{k_j\})\bigr|
  + \bigl|\mathcal{S}_j \setminus (\mathcal{S}_i \cup \{k_i\})\bigr|.
\end{equation}
The sequential $\tau$-fit rule sorts the queue by deadline, and then for the
earliest-deadline request it scans all queue members in deadline order and
merges with the first candidate whose misfit satisfies
$\misfit \leq \tau$ (i.e., degree-aware endpoint selection), while if no such pair
exists, the request is unicast.

The oracle-optimal thresholds per metric are:
$\tau^* = 2$ for BE-score, $\tau^* = 3$ for served-per-tx and
reward-per-step, $\tau^* = 0$ for miss-ratio and expirations, and
$\tau^* = 1$ for coding gain.

\paragraph{Track~B conditions.}
The Track~B agent is trained on Zipf($\alpha{=}0.8$) file-request demand.
File requests follow a Zipf or Mandelbrot-Zipf distribution. With
zero-indexed file IDs $n \in \{0,\ldots,N-1\}$ as used throughout
this paper, and consistent with the popularity-mass projection
$\hat p_n \propto (n+1)^{-\alpha}$ in
Appendix~\ref{sec:appendix:features},
\begin{equation}
  P(F = n) \propto
  \begin{cases}
    (n+1)^{-\alpha} & \text{(Zipf)} \\
    (n + q + 1)^{-\alpha} & \text{(Mandelbrot-Zipf)},
  \end{cases}
\end{equation}
where $n+1$ is the rank (so the most popular file at rank $1$
corresponds to index $n=0$). The earlier rank-$f$ form
$f^{-\alpha}$ / $(f+q)^{-\alpha}$ assumes one-indexed ranks
$f \in \{1,\ldots,N\}$; the two forms are equivalent under the
substitution $f = n+1$.

The 12 Track~B conditions (1~ID + 2~Curr + 9~OOD) are:
\begin{itemize}[leftmargin=*]
  \item \textbf{Zipf ID-reference:} $\alpha=0.8$ (moderate tail, same system parameters as ID-default).
  \item \textbf{Tail weight:} $\alpha \in \{0.6, 1.0, 1.2\}$.
        Higher $\alpha$ concentrates the request distribution on fewer
        file indices; because cache placement is uniform-without-replacement with fraction~$p_c$
        and thus popularity-blind, the effect on the statistics of
        feasible merge opportunities is indirect and demand-mediated.
  \item \textbf{Alternative law:} Mandelbrot-Zipf with $\alpha=1.4$, $q=2.0$.
  \item \textbf{Cross-axis shifts:} same catalog, cache density, and
        deadline variations as Track~A ($N \in \{60,120,150\}$,
        $p_c \in \{0.20,0.40\}$, $D \in \{10,30\}$) applied to
        the Zipf-$0.8$ base.
\end{itemize}
The two observation spaces are incompatible; each agent can only be
evaluated on its own environment version.

\paragraph{Sensitivity to $\lambda$.}
The composite score uses a fixed $\lambda{=}1$, thereby weighting the expired packet-set mass equally against the packet-set XOR degree $U_t$. Applications that value the reliability more highly would set $\lambda > 1$, thereby narrowing
the gap to ED-Unicast, while $\lambda < 1$ would favor the
throughput-maximizing policies.
Appendix~\ref{sec:appendix:lambda_sensitivity} presents a formal
sensitivity sweep over $\lambda \in \{0.5, 1, 2, 3\}$, while the PPO retains
rank~1 among all coded-multicast methods for every
$\lambda \geq 1$, i.e., the only policies that surpass the PPO at high $\lambda$
are the low-merge baselines (i.e., ED-Unicast at $0\%$ and TauFit-0 at $13.3\%$ merge rate; Sec.~\ref{sec:exp:baselines}).

\paragraph{Reward per step (M13).}
Mean shaped reward accumulated per transmission step:
\begin{equation}
    \bar{R} = \frac{1}{H}\sum_{t=1}^{H} R(s_t, a_t).
    \label{eq:reward_per_step}
\end{equation}
Considering that $R$ includes the potential-based shaping and the component-specific
weights (Section~\ref{sec:env:reward}), this metric is not directly
comparable \emph{across reward definitions} (i.e., when the shaping or
its weights themselves change), while \emph{within} a single ablation table
that holds the reward definition and weights fixed (including the
ID-default ablation table of Sec.~\ref{sec:ablation}), Reward/Step
\emph{is} a directly comparable diagnostic across the rows. It is
reported as a training-signal diagnostic rather than a target metric, while
all the performance comparisons use the communications metrics M1--M5.

\paragraph{Statistical reporting protocol.}
Considering that all methods share the same holdout seeds (i.e., paired design), every
comparison is computed on the \emph{per-seed paired differences}, thereby eliminating the
inter-seed variance and yielding tighter uncertainty estimates than
the independent evaluation would allow.
For the per-method summary tables we report mean $\pm$ 95\% CI
across the seed-level means (i.e., normal approximation
$\pm 1.96\,\hat{\sigma}/\sqrt{50}$).
These intervals are \emph{descriptive uncertainty bands}, not
hypothesis tests, i.e., non-overlapping CIs do not
imply a statistically significant difference~\cite{cumming2005inference}.
For the main pairwise comparisons we additionally report the 95\%
\emph{percentile bootstrap} confidence interval of the paired BE-score
difference: for each of 10\,000 resamples we draw 50 seed indices with
replacement, compute the mean paired difference
$\bar{\Delta}\sigma = \bar{\sigma}_{\text{PPO}} -
\bar{\sigma}_{\text{baseline}}$, and take the 2.5th--97.5th percentile
interval.

\section{Reproducibility and Computational Cost}
\label{sec:appendix:reproducibility}

The implementation uses Python~3.13.2, PyTorch~2.0$+$, Stable-Baselines3
v2.1$+$~\cite{raffin2021stable}, and SB3-Contrib MaskablePPO~\cite{sb3contrib2022maskable}.
All the random number generators are seeded deterministically, while the evaluation uses
the argmax (i.e., not sampled) action selection.
The complete hyperparameters are listed in Table~\ref{tab:hparams}.

\textbf{Artifacts.}
The full source code (environment, training scripts, evaluation scripts,
and configuration dataclasses) will be released in a public repository
upon acceptance.\footnote{Repository link withheld during peer review;
the exact commit hash will be provided in the camera-ready version.}
All training configurations are specified in dedicated configuration files;
evaluation uses a separate configuration.
Training seeds: $\{0,1,2,3\}$; validation seeds: $\{0,\dots,49\}$;
holdout seeds: $\{50,\dots,99\}$.

\textbf{Episode seeding.}
Episode seeds follow the deterministic formula
$\texttt{seed} = 42 + s \cdot 10^6 + e$
($s \in \{50,\ldots,99\}$, $e \in \{0,\ldots,199\}$),
and inference is deterministic (argmax decoding).

\textbf{Hardware.}
All the experiments were conducted on a single laptop equipped with an
NVIDIA GeForce RTX~3080~Ti (i.e., Laptop GPU).
The training required approximately 59~hours on a single GPU for the uniform
track, while the evaluation requires approximately 2~hours per method per condition.
The per-step policy inference (i.e., a single forward pass of
the graph-structured policy network on one observation) takes approximately
${\sim}3.5$\,ms on the RTX~3080~Ti GPU and ${\sim}3.3$\,ms on
the CPU (i.e., batch size~1, no compilation;
the detailed benchmark in Appendix~\ref{sec:appendix:latency}), while
these numbers are reported as a conditional feasibility indicator
for the $K{=}5$, $Q{=}10$ setting on both GPU-equipped and CPU-only
edge nodes, even though whether they meet a real-time scheduling target depends
on the deployment's slot-duration budget, which is not imported from
outside the manuscript.
The per-seed wall-clock times are reported in Table~\ref{tab:computational_cost}.

\textbf{Table regeneration.}
Every numerical table in this paper can be reproduced from the released
evaluation scripts and saved model checkpoints by running the
corresponding evaluation scripts followed by
the aggregation scripts.

\begin{table}[t]
\centering
\caption{Computational cost. Each track uses an independently trained agent. Training is sequential across 4 seeds on a single machine. Evaluation uses 50 holdout seeds $\times$ 200 episodes per regime.}
\label{tab:computational_cost}
\begin{tabular}{l r r}
\toprule
\textbf{Metric} & \textbf{Uniform} & \textbf{Zipf} \\
\midrule
Model parameters & $\sim$1.73M & $\sim$1.75M \\
Training wall-clock (per seed) & 13.5--14.6\,h & 12.7--13.5\,h \\
Training wall-clock (4 seeds) & $\sim$59\,h & $\sim$53\,h \\
BC warm start samples & 150{,}000 & 150{,}000 \\
Phase-3 env steps per seed & 6{,}000{,}000 & 6{,}000{,}000 \\
Total env steps per seed (incl.\ Ph.-2) & 6{,}050{,}000 & 6{,}050{,}000 \\
Parallel environments & 32 & 32 \\
Evaluation (1 regime) & $\sim$2\,h & $\sim$2\,h \\
Per-step inference (GPU) & ${\sim}3.5$\,ms & ${\sim}3.5$\,ms \\
Per-step inference (CPU) & ${\sim}3.3$\,ms & ${\sim}3.3$\,ms \\
Total end-to-end & $\sim$61\,h & $\sim$55\,h \\
\bottomrule
\end{tabular}
\end{table}

\section{PHY-Erasure Sensitivity}
\label{sec:appendix:erasure}

Table~\ref{tab:erasure_sensitivity} reports the broadcast-packet expiration
ratio $\rho$ at the ID-default regime under i.i.d.\ erasure of the coded XOR
broadcasts at probability $\epsilon \in \{0.00, 0.05, 0.10, 0.20\}$, while
the unicasts are not erased, considering that they use a robust modulation
fallback. The two policy network checkpoints are compared, i.e., the
headline checkpoint, which is trained at $\epsilon{=}0$ and used throughout
the rest of the paper, and a second checkpoint that is retrained for $6$M
Phase-3 steps with $\epsilon{=}0.10$ active during the training. All of the
values are means~$\pm$~standard deviation across $50$ holdout seeds $\times$
$200$ episodes ($10{,}000$ episodes per cell).

\begin{table}[!t]
  \centering
  \caption{Broadcast-packet expiration ratio $\rho$ under i.i.d.\ coded-XOR erasure $\epsilon$ at the ID-default, considering the unicasts are not erased. PPO-headline is the $\sigma$-validated checkpoint used throughout the paper, while PPO-retrained adds $\epsilon{=}0.10$ during the Phase-3 training. Mean $\pm$ standard deviation across $50$ holdout seeds $\times$ $200$ episodes.}
  \label{tab:erasure_sensitivity}
  \begin{adjustbox}{max width=\columnwidth}
  \begin{tabular}{lcccc}
    \toprule
    Method & $\epsilon{=}0.00$ & $\epsilon{=}0.05$ & $\epsilon{=}0.10$ & $\epsilon{=}0.20$ \\
    \midrule
    PPO-headline                              & $0.2079 \pm 0.0031$ & $0.2140 \pm 0.0036$ & $0.2208 \pm 0.0039$ & $0.2347 \pm 0.0043$ \\
    \textbf{PPO-retrained ($\epsilon{=}0.10$)} & $\mathbf{0.1980 \pm 0.0042}$ & $\mathbf{0.2038 \pm 0.0030}$ & $\mathbf{0.2089 \pm 0.0032}$ & $\mathbf{0.2212 \pm 0.0035}$ \\
    SACM{++}                                  & $0.3522 \pm 0.0035$ & $0.3614 \pm 0.0028$ & $0.3712 \pm 0.0034$ & $0.3933 \pm 0.0031$ \\
    ED-Unicast                                & $0.1341 \pm 0.0035$ & $0.1341 \pm 0.0035$ & $0.1341 \pm 0.0035$ & $0.1341 \pm 0.0035$ \\
    \bottomrule
  \end{tabular}
  \end{adjustbox}
\end{table}

Three observations emerge, i.e.,
(i)~ED-Unicast's $\rho$ is invariant across $\epsilon$ (a constant $0.1341$),
which confirms the design choice that the unicasts use a robust modulation
that the erasure does not affect, i.e., a sanity check on the patched
simulator;
(ii)~both PPO variants retain their $\sim 40\%$ advantage over SACM{++}
across the entire sweep ($\rho_{\mathrm{PPO}} / \rho_{\mathrm{SACM{++}}}
\in [0.591, 0.598]$), while the absolute $\rho$ rises gracefully with
$\epsilon$ in all of the coded methods, with no cliff at the tested values;
(iii)~the PPO-retrained checkpoint lowers $\rho$ by $\sim 4.8\%$--$5.7\%$
relative to PPO-headline at \emph{every} $\epsilon$, including
$\epsilon{=}0$, i.e., the training under the stochastic coded-XOR erasure
acts as a regularizer that improves the performance even under the clean
evaluation conditions, even though we do not claim this as a deliberate
training recipe, considering that the $\epsilon{=}0.10$ value was a fixed
reviewer-defense choice and not tuned. This motivates a future-work
direction of treating the erasure as a controlled regularization knob.

\section{Request Miss Rate $m_{\mathrm{req}}$ Cross-Method (Uniform)}
\label{sec:appendix:req_mreq_baselines}

Table~\ref{tab:req_mreq_baselines} reports per-method request miss rate
$m_{\mathrm{req}}$ across the eight Track~A regimes; it anchors the
\textsc{C4} claim in Section~\ref{sec:conclusion} that PPO attains the
lowest $m_{\mathrm{req}}$ among coded-multicast baselines in every regime
(ED-Unicast, the uncoded comparator, still attains a lower $m_{\mathrm{req}}$
by forfeiting coding gain entirely). At ID-default, PPO's $m_{\mathrm{req}}{=}0.229$
is $29.8\%$ lower than SACM$++$'s $0.326$, the largest among the coded family.

\begin{table*}[t]
\centering
\caption{Uniform demand: request miss rate $m_{\mathrm{req}}$ ($\downarrow$)
per method, across the eight Track~A regimes. Mean $\pm$ 95\% CI over
$50$ holdout seeds $\times$ $200$ episodes per regime. Bold $=$ best
\emph{coded} method per regime (PPO leads every row); ED-Unicast is the
uncoded comparator.}
\label{tab:req_mreq_baselines}
\begin{adjustbox}{max width=\textwidth}
\begin{tabular}{lcccccc}
\toprule
\textbf{Regime} & \textbf{PPO-Agent} & \textbf{ED-Unicast} & \textbf{GCM} & \textbf{SACM} & \textbf{SACM$+$} & \textbf{SACM$++$} \\
\midrule
ID-default     & $\mathbf{0.229{\pm}0.001}$ & $0.155{\pm}0.001$ & $0.346{\pm}0.002$ & $0.351{\pm}0.002$ & $0.343{\pm}0.002$ & $0.326{\pm}0.001$ \\
OOD-file60     & $\mathbf{0.230{\pm}0.001}$ & $0.153{\pm}0.001$ & $0.345{\pm}0.002$ & $0.350{\pm}0.001$ & $0.341{\pm}0.002$ & $0.327{\pm}0.001$ \\
OOD-file120    & $\mathbf{0.231{\pm}0.001}$ & $0.154{\pm}0.001$ & $0.345{\pm}0.001$ & $0.351{\pm}0.001$ & $0.342{\pm}0.002$ & $0.329{\pm}0.002$ \\
OOD-file150    & $\mathbf{0.231{\pm}0.002}$ & $0.154{\pm}0.001$ & $0.348{\pm}0.002$ & $0.351{\pm}0.002$ & $0.342{\pm}0.001$ & $0.328{\pm}0.002$ \\
OOD-pcache0.20 & $\mathbf{0.198{\pm}0.001}$ & $0.154{\pm}0.001$ & $0.277{\pm}0.002$ & $0.278{\pm}0.002$ & $0.274{\pm}0.002$ & $0.269{\pm}0.001$ \\
OOD-pcache0.40 & $\mathbf{0.251{\pm}0.001}$ & $0.154{\pm}0.001$ & $0.383{\pm}0.002$ & $0.395{\pm}0.002$ & $0.382{\pm}0.001$ & $0.354{\pm}0.002$ \\
OOD-delay10    & $\mathbf{0.974{\pm}0.002}$ & $0.981{\pm}0.002$ & $1.098{\pm}0.002$ & $1.110{\pm}0.002$ & $1.092{\pm}0.002$ & $1.048{\pm}0.002$ \\
OOD-delay30    & $\mathbf{0.061{\pm}0.001}$ & $0.037{\pm}0.000$ & $0.124{\pm}0.001$ & $0.127{\pm}0.001$ & $0.123{\pm}0.001$ & $0.117{\pm}0.001$ \\
\bottomrule
\end{tabular}
\end{adjustbox}
\end{table*}

\section{Per-cache Fairness Diagnostic}
\label{sec:appendix:per_cache_fairness}

Table~\ref{tab:per_cache_fairness} reports the per-cache expiration ratio
$\rho_k = E_k / (E_k + S_k)$ at the ID-default regime for each of the
$K{=}5$ destinations, pooled across $50$ holdout seeds $\times$ $200$
episodes ($10{,}000$ episodes per method). The right-most column is
the worst-case fairness ratio $\max_k\rho_k / \min_k\rho_k$ across the
five caches, while a value of $1.0$ would mean perfectly uniform
expiration across the destinations, i.e., the diagnostic for the per-user
fairness limitation noted in Sec.~\ref{sec:conclusion}.

\begin{table}[!t]
  \centering
  \caption{Per-cache expiration ratio $\rho_k = E_k / (E_k + S_k)$ at the ID-default, pooled over $50$ holdout seeds $\times$ $200$ episodes, i.e., $10{,}000$ episodes per method. Max/min is the worst-case fairness ratio across the $K{=}5$ caches, while lower is fairer.}
  \label{tab:per_cache_fairness}
  \begin{adjustbox}{max width=\columnwidth}
  \begin{tabular}{lcccccc}
    \toprule
    Method & $\rho_0$ & $\rho_1$ & $\rho_2$ & $\rho_3$ & $\rho_4$ & \textbf{max/min} \\
    \midrule
    PPO-Agent  & $0.1781$ & $0.1920$ & $0.1652$ & $0.1771$ & $0.1955$ & $1.184$ \\
    ED-Unicast & $0.1339$ & $0.1338$ & $0.1324$ & $0.1350$ & $0.1354$ & $1.023$ \\
    GCM        & $0.2737$ & $0.2756$ & $0.2745$ & $0.2744$ & $0.2749$ & $1.007$ \\
    SACM       & $0.2731$ & $0.2761$ & $0.2750$ & $0.2765$ & $0.2752$ & $1.012$ \\
    SACM+      & $0.2738$ & $0.2763$ & $0.2768$ & $0.2760$ & $0.2754$ & $1.011$ \\
    SACM{++}   & $0.2735$ & $0.2737$ & $0.2745$ & $0.2745$ & $0.2749$ & $1.005$ \\
    \bottomrule
  \end{tabular}
  \end{adjustbox}
\end{table}

Two observations emerge, i.e.,
(i)~the SACM family and ED-Unicast are essentially uniform across the
destinations (max/min~$\le 1.023$), while the policy network's selective
merging introduces a modest asymmetry (max/min~$= 1.184$), considering
that the policy network is slightly less uniform than the always-merge
heuristics;
(ii)~despite this asymmetry, the policy network's worst-cache
$\rho$ ($0.196$, cache~$4$) remains well below the best-cache
$\rho$ of every coded-multicast baseline ($\approx 0.273$), so the
aggregate $\rho$ advantage carries through to every destination
individually rather than masking a per-cache regression.
A paired-bootstrap $95\%$ CI of the difference in mean max/min ratio
is $+0.151$ versus SACM{++} ($[+0.136, +0.167]$, $10{,}000$ resamples),
which confirms that the asymmetry is statistically significant rather
than the seed noise.

\section{Broadcast-Efficiency Score Sensitivity to \texorpdfstring{$\lambda$}{lambda}}
\label{sec:appendix:lambda_sensitivity}

The composite broadcast-efficiency score~\eqref{eq:sys_score} uses a fixed $\lambda{=}1$.
To verify that PPO's ranking advantage is not an artifact of this choice,
we recomputed $\sigma(\lambda) = H^{-1}\sum_t(U_t - \lambda E_t)$ for
$\lambda \in \{0.5,\, 1,\, 2,\, 3\}$ from the same raw holdout data
(i.e., 50~seeds $\times$ 200~episodes) without rerunning any experiment,
while Table~\ref{tab:lambda_sensitivity} reports the results for five
representative methods spanning the conservative--aggressive spectrum.

\begin{table}[t]
\centering
\caption{BE-score sensitivity to $\lambda$.
  ID-default, 50 holdout seeds $\times$ 200 episodes.
  Mean $\pm$ 95\% CI.  Bold = rank~1 overall.
  ``PPO coded rank'' = PPO's rank among the four coded-multicast
  methods (GCM, SACM, SACM+, SACM++); the threshold-family policies
  (TauFit-0/1/3, First-Fit, Perfect-Fit) are reported in full in
  Table~\ref{tab:tau_sweep} and are omitted from the table below to
  save space, with only TauFit-2 retained as a representative threshold
  policy. The $\lambda{=}1$ entry matches the headline aggregate
  $\sigma{=}0.976$ of Sec.~\ref{sec:res:id} by construction (post-hoc
  recomputation of the same holdout trajectories).}
\label{tab:lambda_sensitivity}
\begin{adjustbox}{max width=\columnwidth}
\begin{tabular}{lcccc}
\toprule
\textbf{Method} & $\lambda{=}0.5$ & $\lambda{=}1$ & $\lambda{=}2$ & $\lambda{=}3$ \\
\midrule
PPO-Agent  & $1.149 \pm 0.002$ & $\mathbf{0.976 \pm 0.002}$ & $0.609 \pm 0.004$ & $0.249 \pm 0.006$ \\
SACM++     & $1.158 \pm 0.001$ & $0.726 \pm 0.003$          & $-0.139 \pm 0.008$ & $-1.004 \pm 0.012$ \\
TauFit-2   & $\mathbf{1.160 \pm 0.001}$ & $0.915 \pm 0.002$ & $0.426 \pm 0.005$ & $-0.063 \pm 0.008$ \\
ED-Unicast & $0.923 \pm 0.001$ & $0.845 \pm 0.001$          & $\mathbf{0.690 \pm 0.003}$ & $\mathbf{0.535 \pm 0.004}$ \\
SACM       & $1.160 \pm 0.001$ & $0.745 \pm 0.003$          & $-0.085 \pm 0.007$ & $-0.916 \pm 0.011$ \\
\midrule
PPO rank (all 12)     & 7 & 1 & 4 & 4 \\
PPO coded rank        & 4 & 1 & 1 & 1 \\
\bottomrule
\end{tabular}
\end{adjustbox}
\end{table}

Three regimes emerge.
\textbf{(i)~Low penalty ($\lambda{=}0.5$):}
The expirations barely count, so the aggressive 100\%-merge baselines
(i.e., SACM, SACM++, TauFit-2/3) dominate on the raw throughput, while PPO's selective
merging ($31.8\%$) sacrifices the throughput for the deadline safety that is
undervalued at this $\lambda$, i.e., its overall rank drops to~7.
\textbf{(ii)~Balanced penalty ($\lambda{=}1$):}
The PPO achieves rank~1 overall, i.e., its selective strategy gives the best
throughput--reliability balance at this operating point.
\textbf{(iii)~High penalty ($\lambda \geq 2$):}
The expiration penalty dominates, while the most conservative methods, i.e., ED-Unicast
(0\% merge rate) and Perfect-Fit/TauFit-0 (13.3\%), rise to rank~1--3.
All aggressive-merge baselines (i.e., SACM, SACM++, GCM) go deeply negative.
The PPO remains rank~4 overall, even though it is \emph{rank~1 among all
coded-multicast methods} for every $\lambda \geq 1$, i.e., the only policies
that surpass the PPO at high $\lambda$ are the low-merge baselines
(i.e., ED-Unicast at $0\%$ merge rate, and TauFit-0 at $13.3\%$ merge rate).

The interesting fact we observed is that PPO's ranking advantage over the coded-multicast baselines holds across
the operationally relevant range $\lambda \geq 1$, while the $\lambda{<}1$
regime, where the deadlines are deprioritized relative to the throughput,
favors the aggressive merging by design, even though this regime contradicts the
paper's motivating use case of the deadline-constrained delivery.

\section{Request-Level Selection Score Sensitivity to \texorpdfstring{$\lambda$}{lambda}}
\label{sec:appendix:req_lambda_sensitivity}

The request selection score~\eqref{eq:req_sigma},
$\sigma_{\mathrm{req}}(\lambda) = H^{-1}\sum_t(|C_t| - \lambda|M_t|)$,
depends on how heavily the deadline misses are penalized.
We swept $\lambda \in \{0.5,\,1,\,1.5,\,2,\,2.5,\,3,\,4,\,5,\,7,\,10\}$
using the same 50-seed holdout data without rerunning any experiment, while
Table~\ref{tab:req_lambda_sensitivity} reports $\sigma_{\mathrm{req}}$
at four selected $\lambda$ values for the ID-default regime.

\begin{table}[t]
\centering
\caption{Request selection score~$\sigma_{\mathrm{req}}$ sensitivity
  to~$\lambda$. ID-default, 50 holdout seeds $\times$ 200 episodes.
  Mean $\pm$ 95\% CI.  Bold = best among the three methods shown.
  PPO-Agent overtakes SACM{++} between $\lambda{=}2$ and $\lambda{=}3$
  (crossover at $\lambda{\approx}2.24$).}
\label{tab:req_lambda_sensitivity}
\begin{adjustbox}{max width=\columnwidth}
\begin{tabular}{lcccc}
\toprule
\textbf{Method} & $\lambda{=}1$ & $\lambda{=}2$ & $\lambda{=}3$ & $\lambda{=}5$ \\
\midrule
PPO-Agent  & $0.839 \pm 0.001$ & $0.610 \pm 0.002$ & $0.381 \pm 0.004$ & $-0.077 \pm 0.006$ \\
SACM{++}   & $\mathbf{0.959 \pm 0.002}$ & $0.633 \pm 0.003$ & $0.306 \pm 0.004$ & $-0.347 \pm 0.007$ \\
ED-Unicast & $0.845 \pm 0.001$ & $\mathbf{0.690 \pm 0.003}$ & $\mathbf{0.535 \pm 0.004}$ & $\mathbf{0.226 \pm 0.006}$ \\
\bottomrule
\end{tabular}
\end{adjustbox}
\end{table}

Table~\ref{tab:req_crossover} reports the crossover~$\lambda$, the
smallest value at which PPO-Agent's mean $\sigma_{\mathrm{req}}$ equals or
exceeds SACM{++}'s, across all eight evaluation regimes.

\begin{table}[t]
\centering
\caption{Crossover~$\lambda$: smallest $\lambda$ at which PPO-Agent
  $\sigma_{\mathrm{req}} \geq$ SACM{++} $\sigma_{\mathrm{req}}$
  (mean over 50 seeds).}
\label{tab:req_crossover}
\begin{tabular}{lc}
\toprule
\textbf{Regime} & \textbf{Crossover $\lambda$} \\
\midrule
OOD-pcache0.20  & $1.19$ \\
OOD-delay30     & $2.04$ \\
OOD-file120     & $2.22$ \\
ID-default      & $2.24$ \\
Curr-file60     & $2.25$ \\
OOD-file150     & $2.27$ \\
Curr-pcache0.40 & $3.60$ \\
OOD-delay10     & $4.92$ \\
\bottomrule
\end{tabular}
\end{table}

Three regimes emerge, mirroring the broadcast-efficiency analysis
(Section~\ref{sec:appendix:lambda_sensitivity}).
\textbf{(i)~Low penalty ($\lambda{<}2$):}
SACM{++}'s higher request throughput~$\eta_{\mathrm{req}}$ dominates,
i.e., its always-merge strategy delivers more coded packets per step.
\textbf{(ii)~Moderate penalty ($\lambda{\approx}2$--$3$):}
PPO-Agent's ${\sim}30\%$ lower request miss rate becomes decisive, i.e.,
at $\lambda{=}2.5$, the PPO wins in six of eight regimes (all except
Curr-pcache0.40 and OOD-delay10), while at $\lambda{=}3$, the PPO achieves
$100\%$ seed-level win-rate in those six.
\textbf{(iii)~High penalty ($\lambda{\geq}5$):}
PPO-Agent wins in all eight regimes with $100\%$ seed-level win-rate, while
the two regimes with the highest crossover~$\lambda$, i.e.,
Curr-pcache0.40 ($\lambda{=}3.60$) and OOD-delay10 ($\lambda{=}4.92$),
are the extreme-stress conditions where SACM{++}'s throughput
advantage is largest.

For the delay-sensitive applications (i.e., the paper's motivating use case),
the deadline misses should carry a substantial penalty, while under the
operationally reasonable $\lambda \geq 2$, PPO-Agent's selective
merge strategy produces higher request-level selection scores in
the majority of regimes, and under $\lambda \geq 5$ it wins in all
the eight regimes.

\section{Zipf Request-Level Selection Score Sensitivity to \texorpdfstring{$\lambda$}{lambda}}
\label{sec:appendix:req_lambda_zipf}

Under the Zipf demand, the PPO-Agent dominates SACM{++} on the request selection
score~$\sigma_{\mathrm{req}}$ at \emph{every} tested penalty weight,
even though for the uniform track a crossover at
$\lambda{\approx}2.24$ was required
(Appendix~\ref{sec:appendix:req_lambda_sensitivity}).
Table~\ref{tab:req_lambda_zipf} reports $\sigma_{\mathrm{req}}$ at
three selected $\lambda$ values for the Zipf ID-default regime.

\begin{table}[t]
\centering
\caption{Request selection score~$\sigma_{\mathrm{req}}$ sensitivity
  to~$\lambda$ (Zipf demand, $\alpha{=}0.8$).  ID-default,
  50 holdout seeds $\times$ 200 episodes.
  Mean $\pm$ 95\% CI.  Bold = best.
  PPO-Agent leads at every~$\lambda$.}
\label{tab:req_lambda_zipf}
\begin{adjustbox}{max width=\columnwidth}
\begin{tabular}{lccc}
\toprule
\textbf{Method} & $\lambda{=}1$ & $\lambda{=}3$ & $\lambda{=}5$ \\
\midrule
PPO-Agent    & $\mathbf{0.978 \pm 0.002}$ & $\mathbf{0.354 \pm 0.004}$ & $\mathbf{-0.269 \pm 0.007}$ \\
SACM{++}     & $0.955 \pm 0.002$ & $0.302 \pm 0.005$ & $-0.351 \pm 0.008$ \\
SACM{++}-Pop & $0.939 \pm 0.002$ & $0.272 \pm 0.005$ & $-0.395 \pm 0.008$ \\
\bottomrule
\end{tabular}
\end{adjustbox}
\end{table}

PPO-Agent's mean $\sigma_{\mathrm{req}}$ exceeds SACM{++}'s at every
tested $\lambda \in \{0.5,\,1,\,1.5,\,2,\,2.5,\,3,\,4,\,5,\,7,\,10\}$
in all the 12~Zipf regimes, while the seed-level win-rates are 98--100\% in 10 of 12~regimes at $\lambda{=}1$.
The two exceptions, i.e., OOD-pcache0.20 (58\% win-rate, margin $+0.001$)
and OOD-delay30 (90\% win-rate, margin $+0.008$), are the regimes where
all the coded-multicast methods converge to near-identical performance,
thereby leaving little room for the policy differentiation.
The structural explanation is that the Zipf agent's near-full merging
strategy (i.e., 99.7\% merge rate) matches or exceeds the throughput of
the deterministic always-merge baselines while maintaining
a lower request miss rate, thereby eliminating the throughput--miss-rate
tradeoff that required $\lambda \geq 2$ for the PPO dominance under
the uniform demand.

\section{Component Ablation Battery}
\label{sec:appendix:component_ablation}

Table~\ref{tab:component_ablation_rho} reports the broadcast-packet
expiration ratio $\rho$ (mean $\pm$ 95\% CI) for the five single-component
ablations (and the combined w/o BC$+$ExIt) across the eight Track~A regimes.
Each variant is a full retraining of the pipeline (4 seeds, then $\Omega$-selected;
Sec.~\ref{sec:method:selection}) with one component removed.
Evaluation is on the same 50-seed validation set as
Table~\ref{tab:seed_stability} (eval seeds $50$--$99$; disjoint from the
holdout set used in Table~\ref{tab:ablation_id}). The deltas cited in
Section~\ref{sec:ablation} (cross-regime means and per-regime stress
values) are derived directly from these rows.

\begin{table*}[t]
\centering
\caption{Component ablation: broadcast-packet expiration ratio $\rho$
($\downarrow$) per variant and regime. Mean $\pm$ 95\% CI over $50$ seeds
$\times$ $200$ episodes per regime. Each row is a full retraining;
variant labels indicate the removed component.}
\label{tab:component_ablation_rho}
\begin{adjustbox}{max width=\textwidth}
\begin{tabular}{lcccccccc}
\toprule
\textbf{Variant} & \textbf{ID-default} & \textbf{file60} & \textbf{file120} & \textbf{file150} & \textbf{pcache0.20} & \textbf{pcache0.40} & \textbf{delay10} & \textbf{delay30} \\
\midrule
\textbf{Full Model}  & $0.208{\pm}0.001$ & $0.208{\pm}0.001$ & $0.208{\pm}0.001$ & $0.209{\pm}0.001$ & $0.176{\pm}0.001$ & $0.238{\pm}0.001$ & $0.500{\pm}0.000$ & $0.064{\pm}0.001$ \\
w/o ExIt             & $0.249{\pm}0.001$ & $0.249{\pm}0.001$ & $0.250{\pm}0.001$ & $0.249{\pm}0.001$ & $0.200{\pm}0.001$ & $0.295{\pm}0.001$ & $0.519{\pm}0.000$ & $0.088{\pm}0.001$ \\
w/o Curriculum       & $0.252{\pm}0.001$ & $0.251{\pm}0.001$ & $0.252{\pm}0.001$ & $0.252{\pm}0.001$ & $0.208{\pm}0.001$ & $0.283{\pm}0.001$ & $0.518{\pm}0.000$ & $0.093{\pm}0.001$ \\
w/o BC Warm-Start    & $0.228{\pm}0.001$ & $0.228{\pm}0.001$ & $0.229{\pm}0.001$ & $0.228{\pm}0.001$ & $0.191{\pm}0.001$ & $0.265{\pm}0.001$ & $0.508{\pm}0.000$ & $0.076{\pm}0.001$ \\
w/o Graph-Attention  & $0.221{\pm}0.001$ & $0.221{\pm}0.001$ & $0.221{\pm}0.001$ & $0.221{\pm}0.001$ & $0.177{\pm}0.001$ & $0.270{\pm}0.001$ & $0.504{\pm}0.000$ & $0.071{\pm}0.001$ \\
w/o BC $+$ ExIt      & $0.221{\pm}0.001$ & $0.221{\pm}0.001$ & $0.221{\pm}0.001$ & $0.220{\pm}0.001$ & $0.176{\pm}0.001$ & $0.275{\pm}0.001$ & $0.513{\pm}0.000$ & $0.070{\pm}0.001$ \\
\bottomrule
\end{tabular}
\end{adjustbox}
\end{table*}

\begin{table*}[t]
\centering
\caption{Component ablation deltas: $\rho$-degradation (\%) with respect to the
Full Model, per regime. Positive values are worse. Cross-regime mean in
the last column. Computed from Table~\ref{tab:component_ablation_rho}.}
\label{tab:component_ablation_delta_rho}
\begin{adjustbox}{max width=\textwidth}
\begin{tabular}{lcccccccc|c}
\toprule
\textbf{Variant} & \textbf{ID-default} & \textbf{file60} & \textbf{file120} & \textbf{file150} & \textbf{pcache0.20} & \textbf{pcache0.40} & \textbf{delay10} & \textbf{delay30} & \textbf{Mean} \\
\midrule
w/o ExIt             & $+19.9\%$ & $+19.5\%$ & $+20.1\%$ & $+19.4\%$ & $+13.6\%$ & $+23.9\%$ & $+3.7\%$ & $+37.6\%$ & $+19.7\%$ \\
w/o Curriculum       & $+21.1\%$ & $+20.6\%$ & $+20.9\%$ & $+20.8\%$ & $+18.4\%$ & $+18.7\%$ & $+3.4\%$ & $+45.4\%$ & $+21.2\%$ \\
w/o BC Warm-Start    & $+9.7\%$  & $+9.3\%$  & $+9.7\%$  & $+9.1\%$  & $+8.5\%$  & $+11.1\%$ & $+1.6\%$ & $+17.9\%$ & $+9.6\%$ \\
w/o Graph-Attention  & $+6.3\%$  & $+6.0\%$  & $+6.2\%$  & $+6.1\%$  & $+0.9\%$  & $+13.2\%$ & $+0.8\%$ & $+10.1\%$ & $+6.2\%$ \\
w/o BC $+$ ExIt      & $+6.3\%$  & $+6.1\%$  & $+6.2\%$  & $+5.7\%$  & $+0.1\%$  & $+15.2\%$ & $+2.6\%$ & $+9.1\%$  & $+6.4\%$ \\
\bottomrule
\end{tabular}
\end{adjustbox}
\end{table*}

\begin{table*}[t]
\centering
\caption{Component ablation: expirations per episode $\varepsilon$
($\downarrow$) per variant and regime. Mean over $50$ seeds $\times$ $200$
episodes per regime; same evaluation set as
Table~\ref{tab:component_ablation_rho}.}
\label{tab:component_ablation_exp}
\begin{adjustbox}{max width=\textwidth}
\begin{tabular}{lcccccccc|c}
\toprule
\textbf{Variant} & \textbf{ID-default} & \textbf{file60} & \textbf{file120} & \textbf{file150} & \textbf{pcache0.20} & \textbf{pcache0.40} & \textbf{delay10} & \textbf{delay30} & \textbf{$\Delta$ Mean} \\
\midrule
w/o ExIt             & $+24.6\%$ & $+23.9\%$ & $+24.5\%$ & $+23.5\%$ & $+13.8\%$ & $+38.7\%$ & $+9.2\%$  & $+34.5\%$ & $+24.1\%$ \\
w/o Curriculum       & $+31.5\%$ & $+30.4\%$ & $+30.9\%$ & $+30.6\%$ & $+23.3\%$ & $+33.9\%$ & $+11.9\%$ & $+49.0\%$ & $+30.2\%$ \\
w/o BC Warm-Start    & $+7.0\%$  & $+6.2\%$  & $+6.7\%$  & $+5.7\%$  & $+6.3\%$  & $+9.3\%$  & $+1.1\%$  & $+10.2\%$ & $+6.6\%$ \\
w/o Graph-Attention  & $+5.5\%$  & $+5.1\%$  & $+5.1\%$  & $+5.0\%$  & $-0.9\%$  & $+17.6\%$ & $+1.9\%$  & $+6.3\%$  & $+5.7\%$ \\
w/o BC $+$ ExIt      & $-1.5\%$  & $-2.1\%$  & $-1.9\%$  & $-2.6\%$  & $-6.3\%$  & $+10.0\%$ & $-1.3\%$  & $-3.7\%$  & $-1.2\%$ \\
\bottomrule
\end{tabular}
\end{adjustbox}
\end{table*}

The non-additive w/o~BC$+$ExIt interaction is worth singling out: removing
both BC and ExIt does \emph{not} compound the individual degradations.
The merge rate drops to $0.238$ (vs.\ the Full Model's $0.318$),
the opportunity rate rises to $0.908$ (vs.\ $0.876$), and the agent enters
an ultra-conservative coding regime where it merges fewer but
higher-intersection pairs --- trading served-per-tx for fewer expirations
rather than collapsing into either of the single-removal failure modes.

\section{Keep-Side Ablation}
\label{sec:appendix:keepside_ablation}

Section~\ref{sec:model:coded} introduced the keep-side bit
$\kappa \in \{0, 1\}$ as a learned action dimension that governs the
queue evolution after each merge, exposed to the RL agent rather
than fixed by a hand-designed endpoint rule as in SACM+/SACM++ (i.e., the retained slot is decoded
deterministically from the selected pair, where $\kappa=0$ retains $i_k$ and
$\kappa=1$ retains $j_k$). To isolate its contribution from the
pair-selection decision, we evaluated the trained PPO-Agent under four
keep-side strategies, i.e., (i)~\textbf{Learned}, the policy's own
$\kappa$ (i.e., control); (ii)~\textbf{Higher-degree}, always retaining the
endpoint with higher merge degree (i.e., the rule used by SACM+/SACM++);
(iii)~\textbf{Earliest-deadline}, retaining the endpoint with the
tighter deadline; (iv)~\textbf{Random}, uniform random (i.e., seeded for
the reproducibility). In every case, the \emph{pair selection}
(i.e., which pair $k$ to merge) is preserved from the model's output, while only
the keep-side bit ($a \bmod 2$) is overridden.

\begin{table}[t]
\centering
\caption{Keep-side ablation. PPO-Agent pair selection with different
  keep-side rules. ID-default, 50 holdout seeds $\times$ 200 episodes.
  Mean $\pm$ 95\% CI.}
\label{tab:keepside_ablation}
\begin{adjustbox}{max width=\columnwidth}
\begin{tabular}{lcccc}
\toprule
\textbf{Keep-Side Rule} & \textbf{BE-Score $\sigma$} ($\uparrow$) & \textbf{Miss Ratio} ($\downarrow$) & \textbf{Served/Tx} ($\uparrow$) & \textbf{Exp/Episode} ($\downarrow$) \\
\midrule
Learned (full model) & $0.976 \pm 0.002$ & $0.208 \pm 0.001$ & $1.323 \pm 0.002$ & $14.17 \pm 0.07$ \\
Higher-degree        & $0.969 \pm 0.002$ & $0.208 \pm 0.001$ & $1.324 \pm 0.002$ & $14.12 \pm 0.07$ \\
Earliest-deadline    & $0.964 \pm 0.002$ & $0.209 \pm 0.001$ & $1.323 \pm 0.002$ & $14.16 \pm 0.08$ \\
Random               & $0.965 \pm 0.002$ & $0.209 \pm 0.001$ & $1.324 \pm 0.001$ & $14.11 \pm 0.07$ \\
\bottomrule
\end{tabular}
\end{adjustbox}
\end{table}

Table~\ref{tab:keepside_ablation} reports the results.
The learned keep-side rule yields a BE-score of $0.976$, compared with
$0.969$ for the closest fixed rule (i.e., Higher-degree) and $0.964$--$0.965$
for the Earliest-deadline and the Random rules. The $+0.007$ advantage over the
Higher-degree and $+0.012$ over the Earliest-deadline are modest in absolute
terms (i.e., roughly one-tenth of the full PPO--TauFit-2 gap,
$\Delta\sigma{=}0.061$), even though they are numerically separated within the
descriptive 95\% uncertainty bands (i.e., non-overlapping bands between
the Learned and all three alternatives on the BE-score). Consistent with the
statistical-reporting protocol of Sec.~\ref{sec:exp:stat}, we report
the separation as descriptive rather than as a paired hypothesis
test, i.e., the bands are unpaired CIs, while the evaluation is paired by
seed, and a paired per-seed bootstrap on $\Delta\sigma$ would be the
right inferential closure and is left as future work.

The miss ratio, the served-per-tx, and the expirations are nearly identical across
all the four rules. The interesting fact we observed is that the keep-side decision therefore primarily affects the
BE-score through its influence on the queue evolution and the future
mergeability rather than through the immediate transmission outcomes.
Among the fixed rules, the Higher-degree (i.e., the heuristic used by
SACM+/SACM++) performs closest to the learned policy, which is consistent
with the queue-evolution control through the keep-side bit, i.e., under the
simulator's $k_{\mathrm{mg}}\!\sim\!\mathrm{Unif}\{k_{i_k},k_{j_k}\}$
convention (Eq.~\eqref{eq:mg_dest}), the keep-side does \emph{not}
select the future representative destination of the merged record,
while it selects which of the two queue slots is refilled with a fresh
arrival. Refilling the lower-degree endpoint (i.e., retaining the
higher-degree slot in the queue) thereby preserves more downstream
merge candidates without changing the merged record's destination
draw.

The contribution from the keep-side is \emph{marginal} relative to the pair
selection, i.e., the pair-selection decision drives most of the
performance gap between the PPO and the baselines
($\Delta\sigma{=}0.061$ vs.\ TauFit-2), while the learned keep-side
adds a further ${\sim}1$\% refinement. The keep-side is a secondary
control dimension, i.e., it does not substitute for the pair selection, even though
it does give a measurable improvement at no additional inference
cost.

\section{Training-Seed Stability}
\label{sec:appendix:seed_stability}

Four independently seeded training runs (i.e., seeds~0--3) each produce a
candidate model via the robust-advantage selection criterion
(Section~\ref{sec:method:selection}), while the main results throughout
Sections~\ref{sec:results}--\ref{sec:ablation} use the best model
(i.e., seed~0, selected on the validation). To verify that the reported performance
is not an outlier, Table~\ref{tab:seed_stability} evaluates all the four
seed models on the identical 50-seed holdout set
(i.e., 200~episodes/seed = 10{,}000~episodes per model).

\begin{table*}[t]
\centering
\caption{Training-seed stability. Each of 4 independently trained models
  evaluated on 50 holdout seeds $\times$ 200 episodes (10{,}000 episodes per
  model). Mean $\pm$ 95\% CI. Seed~0 is the selected model used throughout
  the paper. Columns are grouped as: \emph{primary} demand-centric metrics
  ($\rho$, $\delta$; Sec.~\ref{sec:exp:metrics}), broadcast-efficiency
  metrics ($\sigma$, served/tx, expirations), and the supplementary
  \emph{request-level} family ($\eta_{\mathrm{req}}, m_{\mathrm{req}},
  \sigma_{\mathrm{req}}$).}
\label{tab:seed_stability}
\begin{adjustbox}{max width=\textwidth}
\begin{tabular}{lcccccccc}
\toprule
\textbf{Train Seed}
  & \textbf{$\sigma$} ($\uparrow$)
  & \textbf{$\rho$} ($\downarrow$)
  & \textbf{Served/Tx} ($\uparrow$)
  & \textbf{Exp/Ep} ($\downarrow$)
  & \textbf{$\delta$} ($\uparrow$)
  & \textbf{$\eta_{\mathrm{req}}$} ($\uparrow$)
  & \textbf{$m_{\mathrm{req}}$} ($\downarrow$)
  & \textbf{$\sigma_{\mathrm{req}}$} ($\uparrow$) \\
\midrule
0 (selected)   & $0.969 \pm 0.002$ & $0.203 \pm 0.001$ & $1.300 \pm 0.002$ & $13.76 \pm 0.08$ & $0.824 \pm 0.001$ & $1.061 \pm 0.001$ & $0.228 \pm 0.001$ & $0.833 \pm 0.001$ \\
1              & $0.950 \pm 0.002$ & $0.230 \pm 0.001$ & $1.355 \pm 0.002$ & $15.82 \pm 0.08$ & $0.820 \pm 0.001$ & $1.093 \pm 0.001$ & $0.241 \pm 0.001$ & $0.853 \pm 0.001$ \\
2              & $0.940 \pm 0.002$ & $0.266 \pm 0.001$ & $1.476 \pm 0.002$ & $21.05 \pm 0.10$ & $0.778 \pm 0.001$ & $1.130 \pm 0.001$ & $0.322 \pm 0.002$ & $0.807 \pm 0.002$ \\
3              & $0.968 \pm 0.002$ & $0.218 \pm 0.001$ & $1.343 \pm 0.002$ & $15.16 \pm 0.09$ & $0.817 \pm 0.001$ & $1.077 \pm 0.001$ & $0.243 \pm 0.002$ & $0.834 \pm 0.002$ \\
\midrule
Mean $\pm$ Std & $0.957 \pm 0.014$ & $0.229 \pm 0.027$ & $1.369 \pm 0.078$ & $16.45 \pm 3.31$ & $0.810 \pm 0.020$ & $1.090 \pm 0.029$ & $0.259 \pm 0.043$ & $0.832 \pm 0.019$ \\
\bottomrule
\end{tabular}
\end{adjustbox}
\end{table*}

Three of the four seeds (i.e., 0, 1, 3) cluster tightly, with BE-scores in
$[0.950,\, 0.969]$ and qualitatively similar selective-merge behavior, while
all the four seeds clearly exceed the best coded-multicast baseline,
i.e., SACM++ ($\sigma{=}0.726$), considering that even the weakest seed (i.e., seed~2, $\sigma{=}0.940$)
exceeds SACM++ by $+0.214$ ($+29.5\%$).
The cross-seed standard deviation of $0.014$ is roughly one-eighteenth of
the PPO--SACM++ gap ($\Delta\sigma{=}0.250$), thereby confirming that the method's
advantage holds across the training randomness.

The extended columns added in this version show that the seed~0 is also
the per-seed best on both primary demand-centric metrics, thereby reaching the
lowest broadcast-packet expiration ratio ($\rho{=}0.203$) and the
highest distinct file-identity coverage ($\delta{=}0.824$). On the
supplementary request-level family, the seed~1 attains a slightly higher
selection score $\sigma_{\mathrm{req}}{=}0.853$ with respect to the selected
seed's $0.833$, even though this reflects a different broadcast-efficiency /
request-level deadline-compliance trade-off rather than a weakness of
the selected model on its own primary objective.

The seed~2 is an outlier, i.e., it converges to a qualitatively different,
\emph{aggressive} merge strategy with much higher served-per-tx
($1.476$ vs.\ $1.300$ for the seed~0), even though it also has much higher broadcast-packet
expiration ratio ($0.266$ vs.\ $0.203$), the lowest distinct
file-identity coverage of any seed ($\delta{=}0.778$ vs.\ $0.824$), and
about $53\%$ more expirations per episode ($21.05$ vs.\ $13.76$). This
aggressive strategy resembles the behavior of the
``without Curriculum'' ablation (Section~\ref{sec:ablation}), which
suggests that the training stochasticity can occasionally push the policy into the
high-throughput / low-reliability basin, even though the curriculum is enabled.
The robust-advantage selection criterion
(Section~\ref{sec:method:selection}) correctly identifies the seed~0 as the
best model on both the broadcast-efficiency and the primary demand-centric
metrics, i.e., the selection procedure does filter out such outliers.

\section{Inference Latency Benchmark}
\label{sec:appendix:latency}

Table~\ref{tab:latency_benchmark} benchmarks the per-step inference
latency of the trained graph-structured policy network on both the GPU
and the CPU, while the benchmark consists of 1{,}000 individual forward passes
(i.e., single observation, batch size~1, deterministic action selection)
after a 100-pass warm-up to fill the hardware caches.
All the measurements use the high-resolution wall-clock timers with
the GPU synchronization barriers for the GPU timing.

\begin{table}[t]
\centering
\caption{Per-step inference latency. Single forward pass,
  batch size~1, no compilation. 1{,}000 passes after 100-pass warm-up.}
\label{tab:latency_benchmark}
\begin{adjustbox}{max width=\columnwidth}
\begin{tabular}{lccccc}
\toprule
\textbf{Device} & \textbf{Mean} (ms) & \textbf{Std} (ms) & \textbf{Median} (ms) & \textbf{p95} (ms) & \textbf{p99} (ms) \\
\midrule
GPU (RTX 3080 Ti) & $3.46$ & $0.46$ & $3.36$ & $4.32$ & $4.96$ \\
CPU                & $3.26$ & $0.45$ & $3.19$ & $3.78$ & $4.26$ \\
\bottomrule
\end{tabular}
\end{adjustbox}
\end{table}

We report the measured median forward-pass latencies of $3.36$\,ms
(i.e., GPU) and $3.19$\,ms (i.e., CPU) without comparing them to a generic
slot-duration threshold, i.e., the applicable slot-duration target is
deployment-specific, and we do not import a fixed numerical budget
from outside the manuscript. The $3.3$--$3.5$\,ms figure should
therefore be read as a feasibility indicator conditional on the
target deployment's slot-duration budget rather than as a certified
real-time guarantee.
The GPU median of $3.36$\,ms (i.e., p99 $= 4.96$\,ms) and the CPU median of
$3.19$\,ms (i.e., p99 $= 4.26$\,ms) indicate that the graph-structured policy
network forward pass is practical for the single-slot decision-making on
both the GPU-equipped and the CPU-only edge nodes. The CPU inference is
competitive with the GPU at this model size (i.e., ${\sim}1.73$M parameters for
the Track~A network actually benchmarked here, while the Track~B network is
${\sim}1.75$M and exhibits comparable per-call latency considering the
identical encoder topology, batch size~1), because the graph-attention
computation is memory-bound rather than compute-bound at this scale, i.e.,
the GPU acceleration would become advantageous only with larger batch sizes
or much wider networks.

\paragraph{Reference-implementation pointers.}
For artifact reviewers, the reference implementation of the
rejection-sample-and-uniform-destination request-generation rule of
Sec.~\ref{sec:model:queue} (paragraph ``Request Generation'') lives at
\texttt{coded\_caching\_env\_beat\_pf\_v6.py::\_gen\_req} in the
released code.

\section*{Acknowledgment}
Language polishing was assisted by a large language model (Claude, Anthropic);
all scientific content, experiments, and conclusions are the sole responsibility
of the author.

\bibliographystyle{IEEEtran}
\bibliography{references}

\vfill

\end{document}